\documentclass[12pt,english]{amsart}
\usepackage[T1]{fontenc}
\usepackage[latin9]{inputenc}
\usepackage{babel}
\usepackage{units}
\usepackage{mathtools}
\usepackage{amstext}
\usepackage{amsthm}
\usepackage{amssymb}
\usepackage{graphicx}
\usepackage{geometry}
\usepackage{wasysym}
\PassOptionsToPackage{normalem}{ulem}
\usepackage{ulem}
\usepackage[pdfusetitle,
 bookmarks=true,bookmarksnumbered=false,bookmarksopen=false,
 breaklinks=true,pdfborder={0 0 1},backref=false,colorlinks=false]
 {hyperref}

\makeatletter

\newcommand*\LyXZeroWidthSpace{\hspace{0pt}}

\numberwithin{equation}{section}
\numberwithin{figure}{section}
\theoremstyle{plain}
\newtheorem{thm}{\protect\theoremname}[section]
\theoremstyle{definition}
\newtheorem{defn}[thm]{\protect\definitionname}
\theoremstyle{definition}
\newtheorem{example}[thm]{\protect\examplename}
\theoremstyle{plain}
\newtheorem{assumption}[thm]{\protect\assumptionname}
\theoremstyle{plain}
\newtheorem{cor}[thm]{\protect\corollaryname}
\theoremstyle{remark}
\newtheorem{rem}[thm]{\protect\remarkname}
\theoremstyle{plain}
\newtheorem{prop}[thm]{\protect\propositionname}
\theoremstyle{plain}
\newtheorem{lem}[thm]{\protect\lemmaname}
\theoremstyle{remark}
\newtheorem*{rem*}{\protect\remarkname}
\theoremstyle{remark}
\newtheorem{notation}[thm]{\protect\notationname}

\usepackage{xurl}

\usepackage{color}

\usepackage{cite}

\makeatother

\providecommand{\assumptionname}{Assumption}
\providecommand{\corollaryname}{Corollary}
\providecommand{\definitionname}{Definition}
\providecommand{\examplename}{Example}
\providecommand{\lemmaname}{Lemma}
\providecommand{\notationname}{Notation}
\providecommand{\propositionname}{Proposition}
\providecommand{\remarkname}{Remark}
\providecommand{\theoremname}{Theorem}

\begin{document}


\global\long\def\Z{\mathbb{Z}}%
\global\long\def\TT{\mathbb{T}}%
\global\long\def\R{\mathbb{R}}%
\global\long\def\C{\mathbb{C}}%
\global\long\def\N{\mathbb{N}}%
\global\long\def\Q{\mathbb{Q}}%
\global\long\def\B{\mathcal{B}}%
\global\long\def\D{\mathcal{D}}%
\global\long\def\P{\mathcal{P}}%
\global\long\def\M{\mathcal{M}}%
\global\long\def\rmi{\mathbf{\textrm{i}}}%
\global\long\def\rme{\mathbf{\textrm{e}}}%
\global\long\def\rmd{\mathbf{\textrm{d}}}%
\global\long\def\set#1#2{\left\{  #1~:~#2\right\}  }%

\global\long\def\A{\mathcal{A}}%


\global\long\def\E{\mathcal{E}}%
\global\long\def\Ev{\mathcal{E}_{v}}%
\global\long\def\V{\mathcal{V}}%
\global\long\def\va{v_{a}}%
\global\long\def\pa{\Gamma_{\omega}|_{\left[0,n\right]}}%
\global\long\def\dpa{G_{\omega}|_{\left[0,n\right]}}%
\global\long\def\TG{\Gamma_{\omega}^{(n)}}%
\global\long\def\gra{\Gamma_{a}}%
\global\long\def\gri#1{\Gamma_{#1}}%
\global\long\def\gr{\Gamma}%
\global\long\def\Ta{T_{a}}%
\global\long\def\Ti#1{\mathrm{T}_{#1}}%


\global\long\def\spec#1{\mathrm{Spec}\left(#1\right)}%
\global\long\def\NHE#1{N_{#1}}%
\global\long\def\NHNE#1{N_{#1}^{n}}%
\global\long\def\GL#1{\mathcal{GL}\left(\NHE{#1}\right)}%
\global\long\def\ebad{\mathcal{E}_{\text{bad}}}%


\global\long\def\Aa{\mathcal{A}_{\alpha}}%
\global\long\def\Ha{H_{\omega}|_{\left[0,n\right]}}%
\global\long\def\Ja{\mathbf{\mathcal{J}_{\alpha}}}%
\global\long\def\Jan#1{\mathcal{J}_{\alpha_{#1}}}%
\global\long\def\Dam{\Delta_{\alpha}^{\text{metric}}}%
\global\long\def\S{\mathfrak{S}_{\Omega}}%
\global\long\def\Aan#1{\mathcal{A}_{\alpha_{#1}}}%


\global\long\def\ct#1{\#_{a}^{#1}\left(\omega\right)}%
\global\long\def\freq#1{\nu_{#1}}%
\global\long\def\sft{\mathrm{S}}%
\global\long\def\UH{\mathcal{UH}}%
\global\long\def\NUH{\mathcal{NUH}}%
\global\long\def\T{\mathbb{\mathcal{T}}}%


\global\long\def\mnw{\mathcal{M}_{n}\left(\omega,E\right)}%
\global\long\def\mw#1{\mathcal{M}_{#1}\left(\omega,E\right)}%
\global\long\def\M{\mathcal{M}}%
\global\long\def\mtw#1{\mathcal{M}\left(#1\omega,E\right)}%
\global\long\def\MNV#1#2{\mathcal{M}_{#1}^{H_{\alpha}^{#2}}}%
\global\long\def\MNJ#1{\mathcal{M}_{#1}^{\Ja}}%


\global\long\def\fe{f_{E}}%
\global\long\def\fwe{f_{\omega,E}}%


\global\long\def\sc{n_{\omega,t}}%
\global\long\def\sch{n_{\omega,t}^{\textrm{horiz}}}%
\global\long\def\sca{n^{\left(a\right)}}%


\global\long\def\zc{Z_{\omega,t}}%
\global\long\def\zch{Z_{\omega,t}^{\textrm{horiz}}}%
\global\long\def\zca{Z^{\left(a\right)}}%

\begin{center}
{\large\textbf{SPECTRAL PROPERTIES OF APERIODIC METRIC AND DISCRETE
GRAPHS}}{\large\par}
\par\end{center}

\vspace{140bp}

\begin{center}
{\Large\emph{Gilad Sofer}}{\Large\par}
\par\end{center}

\thispagestyle{empty}

\newpage{}
\title[]{{\large SPECTRAL PROPERTIES OF APERIODIC METRIC AND DISCRETE GRAPHS}}

\maketitle
\vspace{70bp}

\begin{center}
{\Large Research Thesis In Partial Fulfillment of The Requirements
for the Degree of Doctor of Philosophy}{\Large\par}
\par\end{center}

\vspace{70bp}

\begin{center}
{\Large\emph{Gilad Sofer}}\vspace{250bp}
\par\end{center}

\begin{center}
{\Large Submitted to the Senate of the Technion - Israel Institute
of Technology }{\Large\par}
\par\end{center}

\vspace{5bp}

\begin{center}
{\Large Elul, 5786 , Haifa, August, 2026 }{\Large\par}
\par\end{center}

\thispagestyle{empty}

\newpage{}

~\vspace{50bp}

\begin{center}
{\Large The Research Thesis Was Done Under The Supervision of Ram Band
in the Faculty of Mathematics}{\Large\par}
\par\end{center}

\vspace{60bp}

\begin{center}
{\Large The Generous Financial Help of the Technion, the Miram and
Aaron Gutwirth Memorial Fellowship, the Irwin and Joan Jacobs Fellowship,
and the Daniel Fellowship is Gratefully Acknowledged}{\Large\par}
\par\end{center}

\vspace{60bp}

\begin{center}
The author of this thesis states that the research, including the
collection, processing and presentation of data, addressing and comparing
to previous research, etc., was done entirely in an honest way, as
expected from scientific research that is conducted according to the
ethical standards of the academic world. Also, reporting the research
and its results in this thesis was done in an honest and complete
manner, according to the same standards.\vspace{60bp}
\par\end{center}

Some results in this thesis have been published as articles by the
author together with collaborators:
\begin{enumerate}
\item R. Band, G. Sofer. Johnson--Schwartzman gap labelling for metric
and discrete decorated graphs, arXiv:2604.08496.
\end{enumerate}
\thispagestyle{empty}

\newpage{}

\tableofcontents{}

\thispagestyle{empty}

\newpage{}

\listoffigures

\thispagestyle{empty}\newpage{}

\section*{Abstract}

\vspace{120bp}

In this thesis, we study the spectral properties of dynamically defined
aperiodic metric and discrete graphs. Our goal is to determine to
what extent spectral properties of discrete one-dimensional ergodic
Schrödinger operators persist when the aperiodicity is manifested
through the geometry rather than through a potential. The graphs considered
here are inspired by one-dimensional aperiodic tilings, and are called
tiling graphs and decorated $\Z$-graphs. 

For a large family of metric tiling graphs equipped with the standard
Laplacian, we show that the spectrum is of zero Lebesgue measure,
and is a generalized Cantor set up to a possible discrete set of energies.
For decorated $\Z$-graphs, we further show that for a Baire-generic
and Lebesgue almost-sure choice of the decoration edge lengths, the
spectrum is a generalized Cantor set.

We then study the integrated density of states (IDS) for metric and
discrete decorated $\Z$-graphs. We prove a gap labelling theorem,
which characterizes the set of possible values taken by the IDS inside
spectral gaps. Specifically, we show that the gap labels are contained
in the Schwartzman group associated with the dynamical system generating
the graph, up to a geometric scaling factor.

Lastly, we consider the Dry Ten Martini Problem for discrete Sturmian
decorated $\Z$-graphs, asking whether all possible values predicted
by the gap labelling theorem are indeed attained by the IDS inside
spectral gaps. We answer this question negatively, by identifying
a large set of gap labels which are not attained due to jump discontinuities
of the IDS. We then show that away from these jump discontinuities,
the periodic approximants for Sturmian graphs display the same combinatorial
structure as the standard Sturmian Hamiltonians, and use this to obtain
an explicit characterization of the realized gap labels for Sturmian
comb graphs.

\setcounter{page}{1}

\newpage{}

\section*{Nomenclature}

$\C$ -- The set of complex numbers.

$\R$ -- The set of real numbers.

$\Q$ -- The set of rational numbers.

$\Z$ -- The set of integers.

$\N$ -- The set of natural numbers.

$\langle\cdot,\cdot\rangle$ -- Inner product.

$\|\cdot\|$ -- Norm.

$L^{2}(X)$ -- The Hilbert space of square-integrable functions on
$X$.

$H^{1}(\Gamma)$ -- The Sobolev space $W^{1,2}(\Gamma)$.

$H^{2}(\Gamma)$ -- The Sobolev space $W^{2,2}(\Gamma)$.

$\mathrm{Dom}(H)$ -- Domain of the operator $H$.

$\spec H$ -- Spectrum of the operator $H$.

$\text{Spec}_{ac}\left(H\right)$ -- Absolutely continuous spectrum.

$\A^{\Z}$ -- Space of bi-infinite sequences over the alphabet $\A$.

$\sft$ -- Shift operator.

$\Omega$ -- A subshift.

$\omega=(\omega(n))_{n\in\Z}$ -- Bi-infinite sequence.

$V_{W}$ -- Cylinder set corresponding to the finite word $W$.

$\ct N$ -- Number of occurrences of the letter $a$ in $\omega|_{\left[0,N-1\right]}$.

$\nu_{a}$ -- Frequency of the letter $a\in\A$.

$G=(\V,\E)$ -- A combinatorial graph with vertex set $\V$ and edge
set $\E$.

$\Ev$ -- Set of edges incident to the vertex $v$.

$\deg(v)$ -- Degree of the vertex $v$.

$\Gamma$ -- A metric graph.

$\Gamma_{\omega}$ -- Metric tiling graph or decorated $\Z$-graph
associated with $\omega$.

$G_{\omega}$ -- Discrete decorated $\Z$-graph associated with $\omega$.

$\Gamma_{\Omega}:=(\Gamma_{\omega})_{\omega\in\Omega}$ -- Family
of metric graphs.

$G_{\Omega}:=(G_{\omega})_{\omega\in\Omega}$ -- Family of discrete
graphs.

$\Ta$ -- Tile corresponding to $a\in\A$.

$v_{1}^{a},v_{2}^{a}$ -- The left/right attachment vertex of $\Ta$,
respectively.

$\TG$ -- The tile $\Ti{\omega(n)}$.

$\Gamma_{a}$ -- Decoration of type $a\in\A$.

$\va$ -- The base vertex of the decoration $\Gamma_{a}$ .

$\ell_{e}$ -- Length of the edge $e$.

$|\Gamma|$ -- Total length of a metric graph.

$\beta_{\Gamma}$ -- First Betti number of $\Gamma$.

$\overline{L}(\Gamma_{\Omega})$ -- Normalized length of a metric
graph family.

$\overline{V}(G_{\Omega})$ -- Average number of vertices for discrete
graph family.

$tr\left(T\right)$ -- The trace of the linear operator $T$.

$H_{\omega}$ -- Schrödinger operator on $\Gamma_{\omega}$.

$H_{\Omega}:=(H_{\omega})_{\omega\in\Omega}$ -- Operator family.

$\Delta_{\omega}$ -- Normalized discrete Laplacian.

$\Delta_{\Omega}:=(\Delta_{\omega})_{\omega\in\Omega}$ -- Discrete
operator family.

$H_{\alpha}^{V}$ -- Classical Sturmian Hamiltonian with coupling
constant $V$.

$H_{\alpha_{n}}^{V}$ -- $n$th periodic approximant of $H_{\alpha}$.

$\Ja$ -- Jacobi operator on Sturmian decorated $\Z$-graph.

$\Jan n$ -- $n$th periodic approximant of $\Ja$.

$\Aa$ -- Adjacency operator on Sturmian decorated $\Z$-graph.

$\Delta_{\alpha}$ -- Normalized discrete Laplacian on Sturmian decorated
$\Z$-graph.

$\Delta_{\alpha}^{\mathrm{metric}}$ -- Kirchhoff Laplacian on equilateral
metric Sturmian decorated $\Z$-graph.

$o(\Gamma_{\omega})$ -- The distinguished origin of the graph $\Gamma_{\omega}$.

$\Gamma_{\omega}^{\pm}$ -- The right/left half-graphs obtained by
cutting $\Gamma_{\omega}$ at the origin.

$m_{\omega}^{\pm}(z)$ -- Weyl--Titchmarsh $m$-functions.

$\Lambda_{\omega}^{\pm}$ -- Dynamical Dirichlet-to-Neumann map.

$\mw{}$ -- One-step transfer matrix.

$\mnw$ -- $n$-step transfer matrix.

$\mathcal{M}_{a}\left(E\right)$ -- One-step transfer matrix of a
Sturmian decoration.

$\mathcal{M}_{n}^{\Ja}\left(E\right)$ -- Transfer matrix of the
Sturmian operator $\Jan n.$

$L_{\Omega}(E)$ -- Lyapunov exponent.

$\mathcal{Z}$ -- Zero set of Lyapunov exponent.

$\B$ -- Exceptional set where the transfer matrix is not defined.

$\UH$ -- Set of uniformly hyperbolic energies.

$\NUH$ -- Set of non-uniformly hyperbolic energies.

$\mathcal{R}^{\pm}$ -- Restriction maps from a bi-infinite sequence
to its positive/negative half-line.

$\D(\mathcal{Z})$ -- Set of $\omega\in\Omega$ whose $m$-functions
are reflectionless on $\mathcal{Z}$.

$\NHE{H_{\Omega}}(E)$ -- Integrated density of states.

$N_{H_{\omega}}^{(n)}(E)$ -- Finite-volume normalized counting function.

$\GL{\Omega}$ -- Set of gap labels.

$X_{\Omega}$ -- Suspension space of the subshift $\Omega$.

$\mathbb{T}$ -- The one-dimensional torus $\R/\Z$.

$\tilde{\phi}_{x}\left(t\right)$ -- Lift of a function from $\mathbb{T}$
to $\R$.

$S_{\Omega}$ -- Schwartzman homomorphism.

$\S$ -- Schwartzman group.

$\Gamma_{\omega}(t)$ -- Truncated right half-graph $\{x\in\Gamma_{\omega}^{+}:d(o(\Gamma_{\omega}),x)\le tL\}$.

$s_{\omega}(t)$ -- Propagation front at distance $tL$ from the
origin.

$\phi$ -- The generalized Prüfer-angle function on the suspension
space $X_{\Omega}$.

$f_{\omega,E}$ -- Generalized eigenfunction in $L^{2}\left(\Gamma_{\omega}^{+}\right)$.

$Z_{\omega,t}$ -- Zero counting function of $f_{\omega,E}$ on $\Gamma_{\omega}(t)$.

$Z_{\omega,t}^{\mathrm{horiz}}$ -- Zero counting function of $f_{\omega,E}$
on the horizontal path $[0,tL]$.

$Z^{(a)}(E)$ -- Nodal counting function on a decoration of type
$a$.

$f_{E}$ -- Canonical solution on a single decoration $\Gamma_{a}$
at energy $E$.

$m_{a}(E)$ -- Effective Robin parameter induced by the decoration
of type $a$.

$n^{(a)}(E)$ -- Spectral counting function of the effective decoration
operator $H|_{\Gamma_{a}}$.

$\sigma^{(a)}(E)$ -- Nodal surplus of the decoration $\Gamma_{a}$
at energy $E$.

$H|_{\Gamma_{a}}$ -- Effective operator on the decoration $\Gamma_{a}$
with the induced Robin condition at the base vertex.

$H_{\omega}|_{\Gamma_{\omega}(t)}$ -- Restriction of $H_{\omega}$
to the truncated graph $\Gamma_{\omega}(t)$ with the induced Robin
boundary conditions.

$H_{\omega}|_{[0,tL]}$ -- Operator on the horizontal subgraph $[0,tL]$
with the induced Robin-type vertex conditions.

$n_{\omega,t}^{\mathrm{horiz}}(E)$ -- Spectral counting function
of $H_{\omega}|_{[0,tL]}$.

$(H_{\tau})_{\tau\in[0,\pi]}$ -- The one-parameter family of operators
continuously decoupling $\Gamma_{\omega}(t)$ into its horizontal
and decoration parts.

$C(G_{\Omega})$ -- Conversion factor between the metric and discrete
IDS.

$\#_{a}^{t}(\omega)$ -- Number of decorations of type $a$ in $\Gamma_{\omega}(t)$.

$\partial Q$ -- Boundary of a finite subset $Q\subset\Z$.

$n_{\omega}^{Q}(E)$ -- Finite-volume spectral counting function.

$N_{\omega}^{Q}(E)$ -- Finite-volume normalized spectral counting
function.

$n_{\omega,D}^{Q}(E)$ -- Spectral counting function of the decoupled
finite-volume operator.

$\xi_{\omega}^{Q}(E)$ -- Spectral shift function.

$H_{\alpha}^{V_{0},V_{1}}$ -- Auxiliary Sturmian Hamiltonian with
two effective potential values.

$V(E)$ -- The pair of effective potentials $(f_{0}(E),f_{1}(E))$.

$\E_{\mathrm{bad}}$ -- Set of bad energies, where the effective
potentials coincide.

$\mathcal{T}(H)$ -- Spectral tree associated with the operator $H$.

$\mathcal{GL}\left(\mathcal{T}(\Ja)\right)$ -- Gap labels read from
the spectral tree of $\Ja$.

$\gamma$ -- Infinite directed path in a spectral tree.

$N_{H}(\gamma)$ -- IDS value associated with the path $\gamma$
in $\mathcal{T}(H)$.

$t_{n}^{\alpha}(E)$ -- Trace of the $n$-step transfer matrix of
the Sturmian periodic approximant.

$I_{1}\prec I_{2}$ -- Strict ordering relation between spectral
bands (or $\preceq$ for weak).

$\mathcal{J}_{G}$ -- Finite-dimensional operator on a decoration
$G$.

$\A_{G}$ -- Adjacency operator on a single decoration $G$.

$\mathcal{G}_{\mathcal{J}_{G}}(E)$ -- Diagonal Green\textquoteright s
function of the pointed graph $(G,v)$ for the operator $\mathcal{J}_{G}$.

$e_{v}$ -- Standard basis vector at the distinguished vertex $v$.

\newpage{}

\section{Introduction\label{sec:Introduction}}

The goal of this thesis is to study how aperiodic order, when encoded
geometrically, affects the spectral structure of Laplacians on  graphs.
More concretely, we study metric (quantum) and discrete graphs whose
combinatorial or metric structure is generated by minimal aperiodic
dynamical systems. Our main question of interest is:

\textit{To what extent do the fundamental spectral results for one-dimensional
ergodic Schrödinger operators still hold when the aperiodicity is
transferred from a potential to the geometry of the underlying space?
What new phenomena arise from such a transition?}

Understanding this transition from potentials to geometry is useful
both for extending one-dimensional spectral theory to network-like
models, and for identifying genuinely geometric mechanisms which have
no analogue in the classical potential setting. The results of this
thesis show that some of the classical phenomena indeed persist in
this geometric setting, while simultaneously new phenomena may arise
due to effects coming from the local geometry.

\subsection*{The classical setting: ergodic Schrödinger operators}

The spectral theory of ergodic Schrödinger operators has been a central
topic in mathematical physics since the 1970s. Given a dynamical system
$(\Omega,\sft)$ and a function $f:\Omega\to\mathbb{R}$ (usually
continuous or measurable), one considers operators on $\ell^{2}(\mathbb{Z})$
of the form

\begin{equation}
(H_{\omega}\psi)(n)=\psi(n+1)+\psi(n-1)+f(\sft^{n}\omega)\psi(n).\label{eq:ergodic-S}
\end{equation}
Such operators usually model the behavior of a quantum particle in
an aperiodic one-dimensional medium, including both structured aperiodic
systems, such as quasicrystals, and disordered media. In the latter
setting, Anderson localization may occur. Localization has also been
studied for quantum graphs, see e.g. \cite{Schanz2000a,Damanik2020a},
and \cite{aizenman2015random} for general background on random operators
and localization. The systems considered in this thesis are instead
structured aperiodic systems, and exhibit a rather different spectral
picture.

Despite their simple form, ergodic operators exhibit a rich range
of spectral phenomena. This is best illustrated through a classical
model that will serve as a guiding reference throughout the thesis,
known as Sturmian Hamiltonians. Fix $\alpha\in(0,1)\setminus\mathbb{Q}$
and define\footnote{More generally, one usually defines $\omega_{\alpha,\theta}(n)=\chi_{[1-\alpha,1)}(n\alpha+\theta\bmod1)$
with $\theta\in[0,1)$, as described in Example \ref{exa:Sturmian}.} 
\begin{equation}
\omega_{\alpha}(n)=\chi_{[1-\alpha,1)}(n\alpha\bmod1).\label{eq:sturmian-sequence}
\end{equation}
The associated operator
\begin{equation}
(H_{\alpha}^{V}\psi)(n)=\psi(n+1)+\psi(n-1)+V\omega_{\alpha}(n)\psi(n)\label{eq:S-H}
\end{equation}
was introduced by Kohmoto, Kadanoff, and Tang \cite{Kohmoto1983}.
These operators possess many interesting features. For every $V\neq0$,
the spectrum is a Cantor set of zero Lebesgue measure \cite{Suetoe1987,Bellissard1989},
and the spectral measures are purely singular continuous (see Figure
\ref{fig: butterfly-1-1}). More generally, the phenomenon of zero
measure Cantor spectrum holds in broad generality for subshifts satisfying
Boshernitzan\textquoteright s condition \cite{Damanik2006a} (see
Definition \ref{def:Boshernitzan}). The problem of establishing Cantor
spectrum for a family of aperiodic operators has been studied for
many important aperiodic models. A central example is the \textit{Ten
Martini Problem}, originally asked by Kac and later named by Simon
\cite{SIMON1982463}, which asks whether the spectrum of the almost
Mathieu operator is a Cantor set for all parameters. Its resolution
by Avila and Jitomirskaya \cite{Avila2009} illustrates the subtlety
(and difficulty) of proving Cantor spectrum, even for concrete aperiodic
models.

\begin{figure}
\includegraphics[scale=0.6]{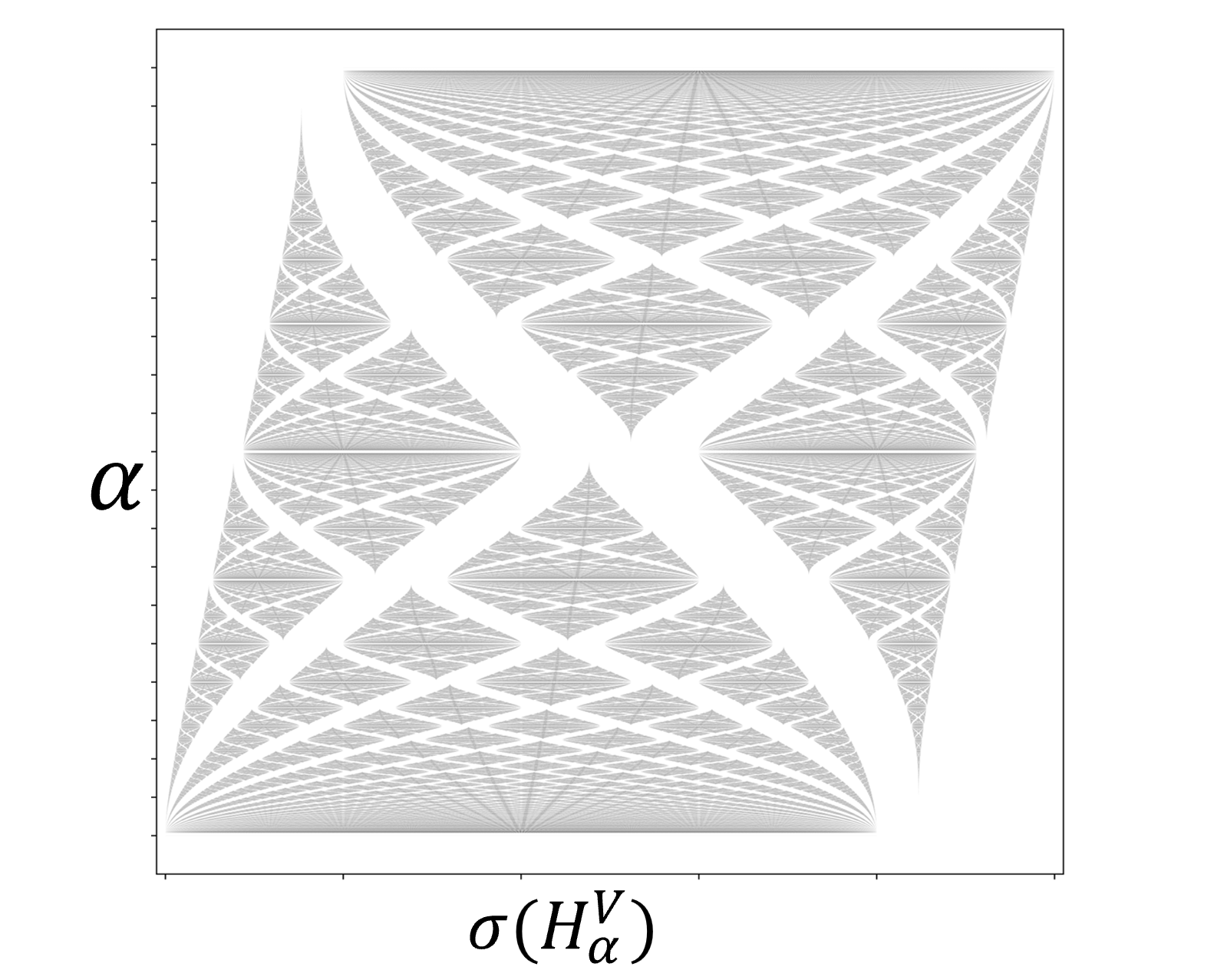}

\caption[Kohmoto butterfly.]{The \textquotedblleft Kohmoto butterfly\textquotedblright{} -- $\protect\spec{H_{\alpha}^{V}}$
as a function of $\alpha$ for $V$ fixed, displaying the fractal
nature of the spectrum.  \label{fig: butterfly-1-1}}
\end{figure}

A central object in the theory of aperiodic Schrödinger operators
is the \emph{integrated density of states} (IDS). For an ergodic
family $(H_{\omega})_{\omega\in\Omega}$, the IDS is defined as
\begin{equation}
N(E)=\lim_{n\to\infty}\frac{1}{n}\#\{\lambda\in\mathrm{Spec}(H_{\omega}|_{[0,n-1]}):\lambda\le E\}.\label{eq:IDS}
\end{equation}
whenever the limit exists (see Figure \ref{fig:IDS-pic}). Here $H_{\omega}|_{[0,n-1]}$
denotes the restriction of $H_{\omega}$ to $\ell^{2}(\{0,...,n-1\})$
obtained by setting the values outside this interval equal to zero,
and the eigenvalues are counted with multiplicity. Under standard
ergodicity assumptions, the limit exists for almost every $\omega$
and is independent of $\omega$ \cite{Shubin1979,Pastur1992}. The
IDS is monotone non-decreasing and constant on spectral gaps. If $E$
lies in a gap, the constant value $N(E)$ is called the \emph{gap label}.
Gap labels carry physical meaning (for instance in the Integer Quantum
Hall Effect \cite{Avron2003}) and may be interpreted as topological
invariants of the underlying dynamical system.

\begin{figure}
\includegraphics[scale=0.4]{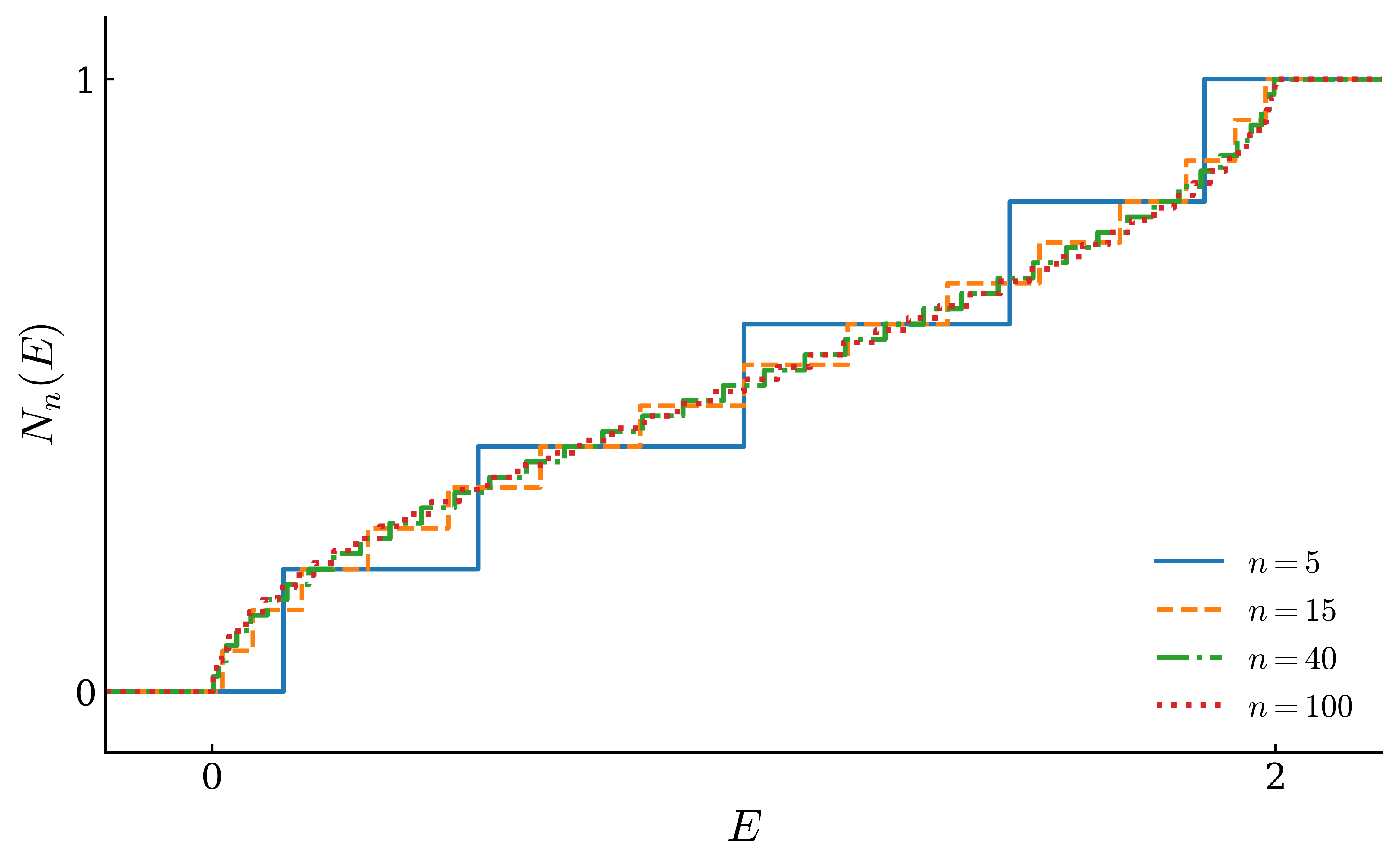}

\caption[Integrated density of states.]{A sequence of normalized counting functions $N_{n}\left(E\right):=\frac{1}{n}\#\{\lambda\in\mathrm{Spec}(H|_{[0,n-1]}):\lambda\le E\}$
converging to the integrated density of states. \label{fig:IDS-pic}}
\end{figure}

Gap labelling theorems (GLT) describe the set of all possible gap
labels. In many one-dimensional ergodic settings, this set is determined
by the $K$--theory of a crossed-product $C^{*}$--algebra \cite{Bellissard1992},
or equivalently via the Schwartzman homomorphism associated with the
suspension flow \cite{Schwartzman1957,Kellendonk1995,Johnson1982}.
For Sturmian Hamiltonians, there exist GLT which give the following
predicted gap labels:
\begin{equation}
\{N(E):E\notin\mathrm{Spec}(H_{\alpha}^{V})\}\subset\{\alpha n+m:m,n\in\mathbb{Z}\}\cap[0,1].\label{eq:GL1-2}
\end{equation}
However, the GLT does not guarantee that any of these allowed gap
labels are actually realized as gap values by the IDS. The question
of whether all gap labels permitted by the GLT are attained (equivalently,
whether ``all gaps are open'') is known as the \emph{Dry Ten Martini Problem}
(DTMP), by analogy with the Ten Martini Problem for the almost Mathieu
operator. For the almost Mathieu operator, the DTMP was resolved by
Avila, You, and Zhou \cite{Avila2023} for all parameter values except
for one (see also \cite{Avila2009a,Puig2004,Liu2015,Borgnia2021}).
For the Sturmian DTMP, the final resolution was proved recently by
Band, Beckus and Loewy \cite{Band2024}, who showed that all gaps
allowed by the GLT are open for every irrational $\alpha$ and all
$V\neq0$. This built on substantial progress by Damanik, Gorodetski,
Mei, Raymond, Yessen, and collaborators \cite{Raymond1995,Damanik2016,band2024review,Damanik2011,Mei2014}.

\subsection*{From potentials to geometry}

In this thesis, we transfer this framework from potentials to geometry.
Instead of encoding the dynamical system in a potential on $\mathbb{Z}$,
we encode it in the combinatorial or metric structure of the graph.
Starting from a bi-infinite path on $\Z$, we attach finite graphs
(decorations) according to a symbolic sequence generated by a minimal
aperiodic subshift. The resulting objects are metric and discrete
\emph{decorated $\Z$-graphs}, or more generally \emph{tiling graphs}
(see Figure \ref{fig: TilingGraphs-1}).

\begin{figure}
\includegraphics[scale=0.5]{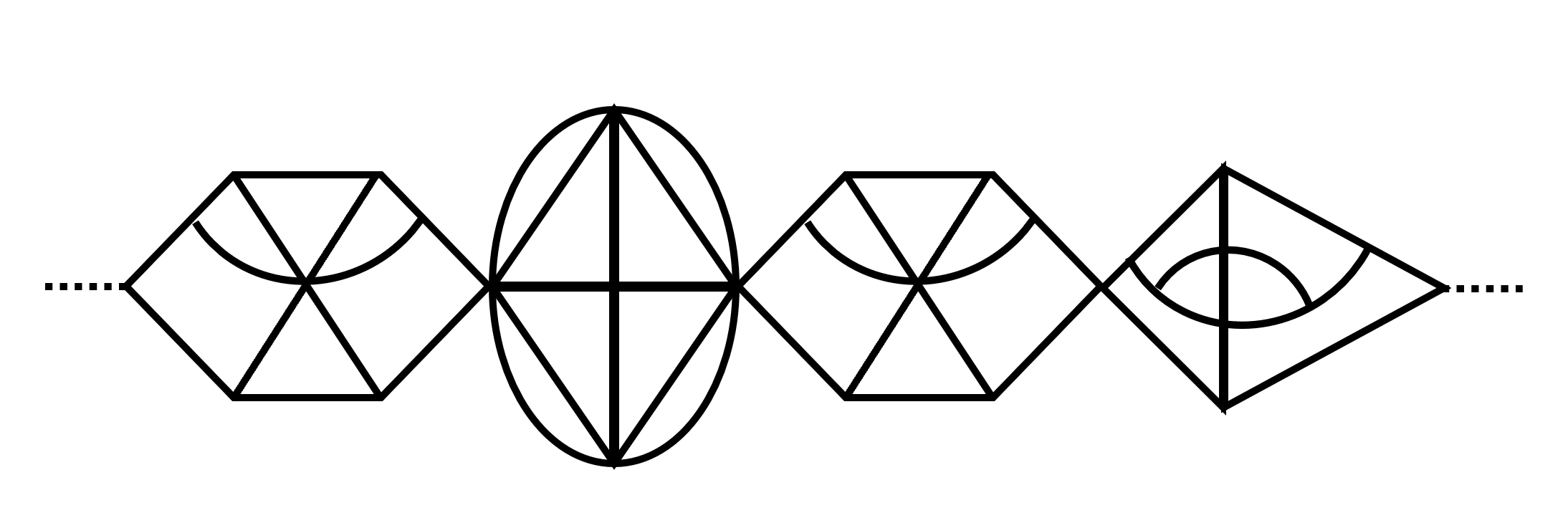}

\caption[An aperiodic tiling graph.]{Example of an aperiodic tiling graph. \label{fig: TilingGraphs-1}}
\end{figure}

This geometric viewpoint is natural for several reasons. First, quantum
graphs provide a standard framework for studying wave propagation
in network-like structures \cite{BerKuc_graphs,Kottos2000,Gnutzmann2006a,Kottos1997,Kuchment2004}.
Second, in certain higher-dimensional quasicrystal models, especially
aperiodic tiling and Delone-set models, aperiodicity is naturally
encoded in the geometry of the underlying space rather than in a scalar
potential \cite{Baake2013,bellissard2006spaces,sadun2008topology,Kellendonk1995}.
Third, tiling graphs offer a flexible intermediate setting between
strictly one-dimensional Schrödinger operators and more complicated
network models.

At this point, the analogy with the classical setting breaks down:
decorations may support compactly supported eigenfunctions (flat bands),
and these affect both the IDS and the gap structure. This leads to
a more concrete formulation of the central question of this thesis:

\textit{How do dynamical properties of the subshift interact with
local geometric features in determining the spectral structure?}

The key idea is that spectral properties are governed by two competing
mechanisms: global dynamical structure, encoded by the subshift, and
local geometric structure, encoded by the tiles. Understanding their
interaction is the main theme of this work.

The thesis consists of three main parts. The first part deals with
the global structure of the spectrum. Inspired by Kotani theory and
the work of Damanik--Lenz \cite{Damanik2006a}, we study families
of metric and discrete tiling graphs generated by minimal subshifts
satisfying Boshernitzan\textquoteright s condition (see Definition
\ref{def:Boshernitzan}). Under mild non-degeneracy assumptions on
the tiles, we prove that the spectrum has zero Lebesgue measure. In
particular, there is no absolutely continuous spectrum. In the metric
setting, we further show that for decorated $\Z$-graphs with a Baire-generic
choice of edge lengths, the spectrum is a generalized Cantor set.
Thus, at the level of global spectral type and measure, aperiodic
geometry behaves much like an aperiodic potential. This extends earlier
results for more symmetric graph models (such as equilateral antitrees
\cite{Damanik2020}) to genuinely aperiodic geometric models, where
the spectral analysis cannot be reduced to an effective one-dimensional
Sturm--Liouville problem. A key ingredient in the proof is a Borg--Marchenko
type theorem for tiling graphs, which shows that the symbolic sequence
is uniquely determined by the Weyl--Titchmarsh $m$-function. This
connects inverse spectral theory for quantum graphs with dynamical
rigidity arguments through Kotani theory and an oracle theorem, which
allows us to relate the local geometric structure to dynamical properties
of the subshift.

The second part of the thesis addresses gap labelling. For uniquely
ergodic decorated $\Z$-graphs, we establish Johnson--Schwartzman
type gap labelling theorems for both metric and discrete graphs. The
set of allowed gap labels is given by the Schwartzman group of the
underlying subshift, scaled by geometric normalization factors (average
length in the metric case and average number of vertices in the discrete
case). Thus, the allowed gap labels are still determined by the dynamics,
as in the classical Schrödinger setting; the geometry enters only
through normalization. The proof requires going beyond classical arguments
based on Sturm\textquoteright s oscillation theorem, since these graphs
contain cycles. We do so by analyzing the non-trivial nodal count
of these graphs, and with further tools from the spectral analysis
on metric graphs.

The final part of the thesis concerns which of these allowed labels
are actually realized. For the classical Sturmian Hamiltonians, all
allowed labels are realized. In contrast, for decorated graphs, local
geometry can force certain gaps to close. Compactly supported eigenfunctions
produce isolated eigenvalues of infinite multiplicity (flat bands),
leading to jump discontinuities in the IDS. Each jump removes an entire
interval of allowed labels. We analyze this phenomenon in detail for
Sturmian comb graphs, where it can be made completely explicit. By
analyzing the corresponding transfer matrices, we show that the band
structure of the associated periodic approximations displays the same
combinatorial structure as the classical Sturmian Hamiltonian. This
allows us to compute the gap labels explicitly. The conclusion is
precise: the only missing labels are those removed by the jump discontinuities
of the IDS coming from flat bands. Thus, the geometric analogue of
the DTMP has a fundamentally different answer: some gaps necessarily
close for structural geometric reasons.

The results can be summarized as follows:
\begin{itemize}
\item The \emph{dynamical system} determines the allowed gap labels and
the global structure of the spectrum away from flat bands.
\item The \emph{local geometry} determines whether flat bands occur and
which of the allowed labels are realized.
\end{itemize}
The main interest here is how these two mechanisms, dynamical and
geometric, appear together in the spectral data.

The thesis is organized as follows. Section \ref{sec:System-and-main}
introduces tiling graphs and decorated\textit{ $\Z$}-graphs, together
with the associated dynamical systems, and presents informal statements
of the main results, as well as an outline of how the proofs are organized
within the thesis. Section \ref{sec:Kotani} introduces the necessary
background and tools for stating precisely and proving zero measure
Cantor spectrum for aperiodic tiling graphs, with the proofs presented
in Section \ref{sec:Kotani-proofs}. Similarly, Section \ref{sec:GLT}
presents the relevant tools for proving the gap labelling theorems,
which are proven in Section \ref{sec:GLT-proofs}. Section \ref{sec:example-comb-graphs}
is devoted to a detailed case study of metric Sturmian comb graphs.
It is meant both to illustrate the results of the previous sections
and to serve as a primer for resolving the DTMP in the following subsections,
showing explicitly how periodic approximations, compactly supported
eigenfunctions, IDS jumps, and closed gaps appear in this setting.
The relevant tools for solving the DTMP for Sturmian combs are presented
in Section \ref{sec:DTMP}, while the relevant proofs and explicit
computation of the realized gap labels appear in Sections \ref{sec:DTMP-proofs}
and \ref{sec:DTMP-computation}. Section \ref{sec:Discussion} contains
further discussion and possible extensions. The thesis concludes with
several technical appendices.

\newpage{}

\section{\label{sec:System-and-main}Models and main results}

In this section we introduce the models under consideration and summarize
the main results informally, with emphasis on their meaning, postponing
precise formulations and technical assumptions to later sections.

\subsection{Dynamical systems\label{subsec:Dynamics}}

The graphs we consider in this work are defined via one-dimensional
dynamical systems which determine their geometric structure. We now
introduce the relevant definitions.

Let $\mathcal{A}$ be a finite set, which we call an alphabet. We
consider the space of bi-infinite sequences $\mathcal{A}^{\mathbb{Z}}$,
equipped with the product topology, as induced by the following metric:
\begin{equation}
d\left(\omega,\omega'\right)=\sum_{n\in\Z}\frac{1-\delta_{\omega\left(n\right),\omega'\left(n\right)}}{2^{\left|n\right|}},\label{eq:disc-metric}
\end{equation}
where $\delta_{i,j}$ is the Kronecker delta. The space $\mathcal{A}^{\mathbb{Z}}$
is naturally equipped with the \textit{shift} map:
\begin{align}
 & \sft:\mathcal{A}^{\mathbb{Z}}\rightarrow\mathcal{A}^{\mathbb{Z}},\label{eq:shift}\\
 & \sft\omega\left(n\right)=\omega\left(n+1\right).\label{eq:shift2}
\end{align}

\begin{defn}
\label{def:Subshift}A \textit{subshift} is a closed, shift-invariant
subset $\Omega\subset\mathcal{A}^{\mathbb{Z}}$. $\Omega$ is called
\textit{minimal} if it contains no proper nontrivial subshifts, and
is called \textit{aperiodic }if it contains no periodic elements,
i.e., an element $\omega$ such that $\sft^{p}\omega=\omega$ for
some $p\neq0$. $\Omega$ is called \textit{uniquely ergodic} if there
exists a unique $\sft$-invariant probability measure on $\Omega$.

We define the letter counting function for $a\in\A$ on $\omega\in\Omega$
by
\begin{equation}
\ct N:=\#\set{n\in\left\{ 0,...,N-1\right\} }{\omega\left(n\right)=a}.\label{eq:counting}
\end{equation}
For a uniquely ergodic subshift $\Omega$, the letter frequencies
\begin{equation}
\freq a=\lim_{N\rightarrow\infty}\frac{\ct N}{N},\label{eq:frequency}
\end{equation}
are well-defined, independent of $\omega\in\Omega$, and satisfy $\sum_{a\in\mathcal{A}}\freq a=1$
(\cite[prop. 4.4]{Baake2013}, \cite{Oxtoby1952}).
\end{defn}

\begin{example}
\label{exa:Sturmian}Let $\alpha\in\left(0,1\right)\backslash\mathbb{Q}$.
For $\theta\in[0,1)$, we define the Sturmian sequence $\left(\omega_{\alpha,\theta}\left(n\right)\right)_{n\in\Z}$
on the alphabet $\A=\{0,1\}$ by
\begin{equation}
\omega_{\alpha,\theta}\left(n\right)=\chi_{[1-\alpha,1)}\left(n\alpha+\theta\text{ mod \ensuremath{1}}\right).\label{eq:sturm}
\end{equation}
 The Sturmian subshift is then defined as
\begin{equation}
\Omega_{\alpha}=\overline{\left\{ \omega_{\alpha,\theta}:\theta\in[0,1)\right\} }.\label{eq:Sturm-subshift}
\end{equation}
$\Omega_{\alpha}$ is a minimal, aperiodic, and uniquely ergodic subshift
over the alphabet $\left\{ 0,1\right\} $, with letter frequencies
$\alpha$ and $1-\alpha$ for $1$ and $0$, respectively \cite{Queffelec2010}.
\end{example}

\subsection{Metric and discrete graphs}

A \emph{metric graph} is a pair $\Gamma=(G,\vec{\ell})$, where $G=(\mathcal{V},\mathcal{E})$
is a combinatorial graph with vertex set $\mathcal{V}$ and edge set
$\mathcal{E}$, and $\vec{\ell}\in\mathbb{R}_{+}^{\left|\E\right|}$
is a vector assigning a positive length to each edge in $\mathcal{E}$.
Each edge $e\in\E$ is identified with the interval $\left[0,\ell_{e}\right]$,
which equips $\Gamma$ with the natural structure of a metric space.

For a vertex $v$, we write $e\sim v$ to indicate that the edge $e$
is incident to $v$. The degree of a vertex $v$, denoted by $\deg(v)$,
is defined as the number of edges incident to $v$, with loop edges
(edges that connect $v$ to itself) counted twice.

A \emph{quantum graph} is a metric graph $\Gamma$ equipped with a
self-adjoint differential operator $H$ acting on the Sobolev space
$H^{2}\left(\Gamma\right):=\oplus_{e\in\mathcal{E}}H^{2}\left(0,\ell_{e}\right)$.
In this work, the operator acts as the Laplacian on each edge, $H=-\frac{d^{2}}{dx^{2}}$,
together with the \textit{Neumann-Kirchhoff} boundary conditions imposed
at each vertex:
\begin{align}
 & f|_{e}\left(v\right)=f|_{e'}\left(v\right),\ensuremath{\quad}\ensuremath{\forall e,e'\sim v},\label{eq:-15-2}\\
 & \sum_{e\sim v}f'|_{e}\left(v\right)=0,\label{eq:-16-2}
\end{align}
where the derivatives are taken in the outward-pointing direction
from the vertex. These vertex conditions are also known as Kirchhoff,
continuity-Kirchhoff, or standard vertex conditions. For an extensive
introduction to quantum graphs, see \cite{Gnutzmann2006a,Kurasov_Book,Band2018,BerKuc_graphs,Berkolaiko2017a}.

In this work, $\Gamma$ will be an infinite, connected metric graph,
where the vertex degrees $\deg\left(v\right)$ are uniformly bounded
from above, and the edge lengths $\ell_{e}$ are uniformly bounded
from above and also uniformly bounded from below by some $\ell_{0}>0$.
Under these conditions, the associated Kirchhoff Laplacian is self-adjoint
and non-negative, see \cite[thm 1.4.19]{BerKuc_graphs}. 

\subsection{Tiling graphs and decorated\textit{ $\protect\Z$-graph}s}

We now introduce the class of \textit{aperiodic tiling graphs}, which
will be the main object studied in this thesis. To define those, we
fix a minimal subshift $\Omega$ over a finite alphabet $\mathcal{A}$.

Let $\left(\Ta\right)_{a\in\mathcal{A}}$ be a family of compact metric
graphs, which we call the \textit{tiles}. For each $\Ta$, we select
two distinct vertices $v_{1}^{a},v_{2}^{a}\in\V\left(\Ta\right)$.
Given $\omega\in\Omega$, construct an infinite metric graph $\Gamma_{\omega}$
by forming a bi-infinite chain of graphs $\left(\TG\right)_{n\in\Z}$,
such that
\begin{equation}
\TG=\Ti{\omega(n)},\label{eq:tiling-n}
\end{equation}
and the graphs $\TG,\Gamma_{\omega}^{\left(n+1\right)}$ are connected
to one another by identifying the vertices $v_{2}^{\omega\left(n\right)}\in\V\left(\TG\right)$
and $v_{1}^{\omega\left(n+1\right)}\in\V\left(\Gamma_{\omega}^{\left(n+1\right)}\right)$
(see Figure \ref{fig: TilingGraphs}). Each graph $\Gamma_{\omega}$
can be associated with a natural origin $o\left(\Gamma_{\omega}\right)\in\Gamma_{\omega}$,
defined by the vertex $v_{1}^{\omega\left(0\right)}\in\V\left(\Gamma_{\omega}^{\left(0\right)}\right)$.

\begin{figure}
\includegraphics[scale=0.5]{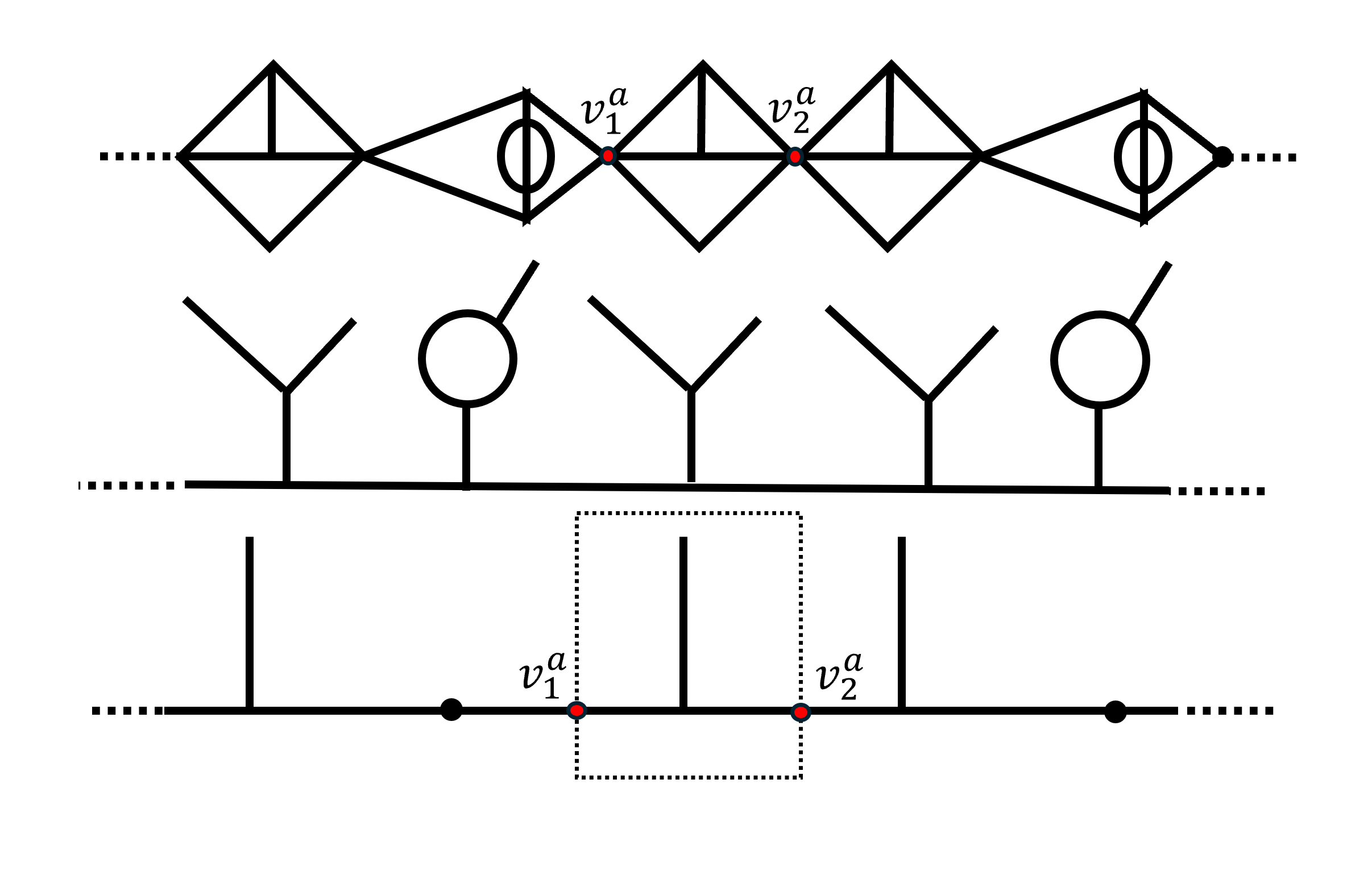}

\caption[Three examples of tiling graphs.]{Three examples of tiling graphs, where the bottom two graphs are decorated
$\protect\Z$-graphs. For the top and bottom graph the attachment
vertices $v_{1}^{a}$ and $v_{2}^{a}$ of the central tile are marked
in red. \label{fig: TilingGraphs}}
\end{figure}

The family of graphs $\Gamma_{\Omega}:=\left(\Gamma_{\omega}\right)_{\omega\in\Omega}$
is equipped with a naturally induced shift, 
\begin{align}
 & \sft:\Gamma_{\Omega}\rightarrow\Gamma_{\Omega},\label{eq:randomgraph1}\\
 & \sft\Gamma_{\omega}=\Gamma_{\sft\omega},\label{eq:randomgraph2}
\end{align}
where we slightly abuse the shift notation $\sft$. One can further
define
\begin{align}
 & \sft:C\left(\Gamma_{\omega}\right)\rightarrow C\left(\Gamma_{\sft\omega}\right),\label{eq:sft}\\
 & \left(\sft f\right)\left(x\right)=f\left(\sft^{-1}x\right),\label{eq:sft-1}
\end{align}
where we once again abuse the notation $\sft$. Equipping the tiling
graphs $\Gamma_{\omega}\in\Gamma_{\Omega}$ with the Kirchhoff Laplacian
$H_{\omega}$, one can consider the family $H_{\Omega}:=\left(H_{\omega}\right)_{\omega\in\Omega}$
as a dynamical system of operators. The family $H_{\Omega}$ is covariant,
i.e.,
\begin{equation}
H_{\sft\omega}=\sft^{-1}H_{\omega}\sft,\,\forall\omega\in\Omega,\label{eq:covariant}
\end{equation}
which implies that the operators $H_{\omega}$ and $H_{\sft\omega}$
are unitarily equivalent. As we shall see in Theorem \ref{thm:TMP},
minimality of $\Omega$ implies that $\spec{H_{\omega}}$ is in fact
independent of $\omega\in\Omega$, and so we can simply denote it
by $\spec{H_{\Omega}}$.

\subsubsection{Decorated $\protect\Z$-graph\emph{s}\label{def: decorated lines} }

An important subclass of tiling graphs on which we also focus is \textit{decorated
$\Z$-graph}\emph{s}. We first introduce the metric version.

To each $a\in\A$, we associate a compact metric graph $\gra$, which
we call a \textit{decoration} (it may consist of just a single vertex).
We also select a distinguished \textit{base vertex} $v_{a}\in\V\left(\gra\right)$
in each decoration. Given $L>0$, we construct a family of infinite
metric graphs $\Gamma_{\Omega}:=\left(\Gamma_{\omega}\right)_{\omega\in\Omega}$
as follows: for each $\omega\in\Omega$, we begin with the bi-infinite
chain graph with vertex set $L\Z$. To each vertex $Ln\in L\Z$, attach
the graph $\Gamma_{\omega\left(n\right)}$, by identifying the base
vertex $v_{a}\in\V\left(\Gamma_{\omega\left(n\right)}\right)$ with
the vertex $Ln$ (see Figure \ref{fig: TilingGraphs}). This produces
an infinite metric graph $\Gamma_{\omega}$, obtained by decorating
the chain graph $\Z$ with the graphs $\{\gra\}_{a\in\A}$ according
to $\omega\in\Omega$. This may be viewed as a particular case of
the tiling graph construction above, where the tiles consist of the
decorations together with a dangling edge attached on each side.

We define the normalized length associated with the graph family $\Gamma_{\Omega}$
by 
\begin{equation}
\overline{L}\left(\Gamma_{\Omega}\right):=L+\sum_{a\in\mathcal{A}}\freq a\ell_{a},\label{eq:norm-length}
\end{equation}
where $L$ is the horizontal distance between consecutive decorations,
$\freq a$ is the frequency of $a\in\A$, and $\ell_{a}$ is the total
length of the decoration $\gra$. Since the frequencies $\va$ are
independent of $\omega\in\Omega$ (Subsection \ref{subsec:Dynamics}),
the normalized length (\ref{eq:norm-length}) may also be expressed
through the average growth rate of geodesic balls (which is independent
of the choice of $\omega\in\Omega$):
\begin{equation}
\overline{L}\left(\Gamma_{\Omega}\right)=\lim_{r\rightarrow\infty}\frac{\left|\left.\Gamma_{\omega}\right|_{B\left(x,r\right)}\right|}{2\nicefrac{r}{L}},\quad\quad\forall\omega\in\Omega,\quad x\in\Gamma_{\omega},\label{eq:norm-length-2}
\end{equation}
where $B\left(x,r\right)$ is the geodesic ball of radius $r$ around
$x\in\Gamma_{\omega}$, and $\left|\cdot\right|$ is the standard
Lebesgue measure.
\begin{example}[Sturmian comb]
\label{exa: Sturm-comb} A particular family of decorated $\Z$-graphs
is obtained by taking the Sturmian subshift $\Omega_{\alpha}$ over
the alphabet $\left\{ 0,1\right\} $ from Example \ref{exa:Sturmian},
and choosing the first decoration to be a dangling edge of length
$\ell$, and the second decoration to be a single vertex (see bottom
of Figure \ref{fig: TilingGraphs}). In this case,
\begin{equation}
\overline{L}\left(\Gamma_{\Omega}\right)=L+\alpha\ell.\label{eq:L-Sturm}
\end{equation}
\end{example}

We also consider a discrete version of decorated $\Z$-graphs, constructed
in the same manner. Let $\left(\Omega,\sft\right)$ be a uniquely
ergodic subshift over an alphabet $\A$. Let $\left(G_{a}\right)_{a\in\A}$
be a set of discrete graphs (the possible decorations), each assigned
a base vertex $v_{a}\in\V\left(G_{a}\right)$. Form a family of \textit{discrete
decorated }$\Z$-\textit{graphs} $G_{\Omega}:=\left(G_{\omega}\right)_{\omega\in\Omega}$
as follows: the graph $G_{\omega}$ is obtained from the chain graph
$\Z$ by attaching to each vertex $n\in\Z$ the decoration $G_{\omega(n)}$,
via the identification of the vertex $n\in\Z$ with the base vertex
of $G_{\omega\left(n\right)}$. Each graph $G_{\omega}$ is equipped
with a graph operator denoted for now by $\mathcal{J}_{\omega}$.
We get the operator family $\mathcal{J}_{\Omega}:=\left(\mathcal{J}_{\omega}\right)_{\omega\in\Omega}$,
and as above, $\spec{\mathcal{J}_{\omega}}$ is almost-surely independent
of $\omega\in\Omega$, and is simply denoted by $\spec{\mathcal{J}_{\Omega}}$.

In this setting, the analogue of the normalized length will be the
average number of vertices:
\begin{align}
 & \overline{V}\left(G_{\Omega}\right):=\sum_{a\in\A}\freq a\left|\V\left(G_{a}\right)\right|.\label{eq:discrete-norm-l}
\end{align}

In this work, we will have several discrete graph operators of interest.
Our first operator will be the normalized discrete Laplacian (NDL)
$\Delta_{\omega}$, acting on $\ell^{2}\left(G_{\omega}\right)$ as:
\begin{equation}
\Delta_{\omega}\psi\left(v\right)=\psi\left(v\right)-\sum_{u\in\Ev}\frac{1}{\sqrt{\deg\left(v\right)\deg\left(u\right)}}\psi\left(u\right).\label{eq:norm-lap}
\end{equation}
Another operator of interest is the adjacency matrix $\mathcal{A}_{\omega}$,
acting as:
\begin{equation}
\mathcal{A}_{\omega}\psi\left(v\right)=\sum_{u\in\Ev}\psi\left(u\right).\label{eq:norm-lap-1}
\end{equation}
We will also study a more general class of Jacobi operators, introduced
in Section \ref{sec:DTMP}.

\subsection{Informal statement of main results\label{subsec:informal-results}}

For readability, we suppress some technical assumptions and defer
precise definitions to later sections, where the corresponding theorems
are stated formally. The precise theorem statements corresponding
to the informal results below are listed at the end of the subsection.

Our first main result concerns the global structure of the spectrum
for metric tiling graphs (including decorated $\Z$-graphs).

\subsubsection*{Main result 1 -- global structure of the spectrum}

For  tiling graphs $\left(\Gamma_{\omega}\right)_{\omega\in\Omega}$
generated by a minimal aperiodic subshift $\left(\Omega,\sft\right)$
(under mild and Baire-generic assumptions), $\spec{H_{\omega}}$ is
independent of $\omega\in\Omega$, and has zero Lebesgue measure.
In particular, there is no absolutely continuous spectrum. Furthermore,
up to a possible discrete set, $\spec{H_{\omega}}$ is a Cantor set.

\bigskip

In particular, this shows that aperiodic geometry generically preserves
much of the global spectral structure known from one-dimensional models.

We now turn to the IDS introduced in Section \ref{sec:Introduction}.
The next main result classifies the possible gap labels for metric
and discrete decorated $\Z$-graphs in terms of the Schwartzman group,
which is a dynamical invariant associated with the subshift.

\subsubsection*{Main result 2 -- Johnson--Schwartzman gap labelling for decorated
$\protect\Z$-graphs}

Let $\left(\Omega,\sft\right)$ be a uniquely ergodic subshift.
\begin{itemize}
\item If $\Gamma_{\Omega}$ is a family of metric decorated $\Z$-graphs
equipped with the Kirchhoff Laplacian $H_{\Omega}$, then the set
of gap labels for the IDS $\NHE{H_{\Omega}}$ satisfies
\begin{equation}
\mathcal{GL}\left(\NHE{H_{\Omega}}\right)\subset\frac{1}{\overline{L}\left(\Gamma_{\Omega}\right)}\S\cap[0,\infty),\label{eq:GL1-1}
\end{equation}
where $\S$ is the Schwartzman group defined in Subsection \ref{subsec:Schwartzman-group},
and $\overline{L}\left(\Gamma_{\Omega}\right)$ is the normalized
length.
\item If $G_{\Omega}$ is a family of discrete decorated $\Z$-graphs equipped
with the normalized discrete Laplacian $\Delta_{\Omega}$, then
\begin{equation}
\mathcal{GL}\left(\NHE{\Delta_{\Omega}}\right)\subset\frac{1}{\overline{V}\left(G_{\Omega}\right)}\S\cap\left[0,1\right],\label{eq:GL1-1-1}
\end{equation}
where $\overline{V}\left(G_{\Omega}\right)$ is the average number
of vertices.
\end{itemize}
\bigskip

This shows that, as in the classical Schrödinger setting, the set
of allowed gap labels is determined entirely by the underlying dynamics,
and the local geometric structure only enters through a scaling factor.

As we shall see, for some non-generic choices of the edge lengths,
$\spec{H_{\Omega}}$ may also contain isolated eigenvalues, and these
correspond to discontinuities (jumps) in the IDS, which affects the
realization of gap labels. In the particular case of the Sturmian
comb graph (defined in Example \ref{exa: Sturm-comb}), these eigenvalues
and IDS jumps can be fully characterized.

\subsubsection*{Main result 3 -- jump discontinuities of the IDS for metric Sturmian
combs}

Let $\Gamma_{\Omega_{\alpha}}$ be a metric Sturmian comb graph. If
$E\in\spec{H_{\Omega_{\alpha}}}$ is an eigenvalue, then $E$ is an
isolated point in $\spec{H_{\Omega_{\alpha}}}$ of infinite multiplicity,
with compactly supported eigenfunctions. The resulting jump in the
IDS may attain one of the following values
\begin{equation}
\Delta\NHE{H_{\Omega_{\alpha}}}\left(E\right)\in\frac{1}{L+\alpha\ell}\left\{ (a_{1}+1)\alpha-1,-a_{1}\alpha+1,\alpha\right\} ,\label{eq:IDS-JUMPS}
\end{equation}
where $\ell$ is the length of the decoration, $L$ is the decoration
spacing, and $a_{1}$ is the first digit in the continued fraction
expansion of $\alpha$:
\begin{equation}
\alpha=\frac{1}{a_{1}+\frac{1}{a_{2}+...}}.\label{eq:cf-expansion}
\end{equation}

\bigskip

In particular, these jumps are purely geometric in origin -- they
are caused by the existence of compactly supported eigenfunctions,
a phenomenon that does not occur in the classical one-dimensional
setting.

The two previous main results imply together that certain labels predicted
by the GLT are necessarily not attained by the IDS, since each jump
discontinuity removes an interval of allowed IDS values. This is unlike
the standard DTMP, where all predicted gaps are open. Our goal is
now to determine which of these allowed labels are actually realized.
Our next main result shows that the spectrum of discrete Sturmian
decorated $\Z$-graphs, equipped with a Jacobi operator, displays
a combinatorial structure similar to the usual Sturmian Hamiltonians,
which plays an important structural role in computing the gap labels:

\subsubsection*{Main result 4 -- combinatorial spectral structure for Sturmian decorated
$\protect\Z$-graphs}

Let $G_{\alpha}$ be a discrete Sturmian decorated $\Z$-graph equipped
with a Jacobi operator $\Ja$. Let $\alpha_{n}\in\Q$ be the $n$th
periodic approximant of $\alpha$, obtained by truncating the continued
fraction expansion (\ref{eq:cf-expansion}) at the $n$th digit. Then
away from a discrete set of exceptional energies, the periodic approximants
$\left(\Jan n\right)_{n\in\N}$ display the standard Sturmian combinatorial
structure of spectral bands (as described in Subsection \ref{subsec:FB}),
and this structure depends only on the continued fraction expansion
of $\alpha$.

\bigskip

In our final main result, we apply the combinatorial structure described
above in order to completely characterize which gaps are open for
Sturmian comb graphs equipped with the adjacency matrix.

\subsubsection*{Main result 5 -- the Dry Ten Martini Problem for the adjacency matrix
on Sturmian combs}

Let $\alpha\in\left(0,1\right)\backslash\Q$ and let $a_{1}$ be the
first continued fraction digit of $\alpha$ in (\ref{eq:cf-expansion}).
Let $G_{\alpha}$ be a discrete Sturmian comb, equipped with the adjacency
matrix $\mathcal{A}_{\alpha}$. Then the gap labels of $\mathcal{A}_{\alpha}$
are given by
\begin{equation}
\GL{\Aa}=\begin{cases}
\left\{ \frac{\alpha n+m}{1+\alpha}:m,n\in\Z\right\} \cap\left(\left[0,\frac{(a_{1}+1)\alpha}{2(1+\alpha)}\right]\cup\left[1-\frac{(a_{1}+1)\alpha}{2(1+\alpha)},1\right]\right), & a_{1}\ \text{odd},\\[2mm]
\left\{ \frac{\alpha n+m}{1+\alpha}:m,n\in\Z\right\} \cap\left(\left[0,\frac{2-a_{1}\alpha}{2(1+\alpha)}\right]\cup\left[\frac{(a_{1}+2)\alpha}{2(1+\alpha)},1\right]\right), & a_{1}\ \text{even}.
\end{cases}\label{eq:GLAa-1}
\end{equation}

\bigskip

This gives a complete solution to the DTMP in this setting, characterizing
exactly which predicted gaps are open. Unlike the case of Sturmian
potentials, where all predicted gaps are open, the geometric setting
exhibits a fundamentally different behavior: the set of open gaps
depends explicitly on arithmetic properties of the Sturmian frequency.
While the results above are derived specifically for Sturmian comb
graphs equipped with the adjacency matrix, in Section \ref{sec:DTMP-proofs}
we outline how the gap labels can be similarly derived for more general
decorated $\Z$-graphs, equipped with different Jacobi operators,
as well as equilateral metric decorated $\Z$-graphs.

\subsubsection*{Outline of theorems}

Main result $1$ corresponds to Theorems \ref{thm:TMP} and \ref{thm:Cantor}
and Corollary \ref{cor:generic-cantor}, which are proven in Section
\ref{sec:Kotani-proofs}. Main result $2$ corresponds to Theorems
\ref{thm:GLT} and \ref{thm:Discrete-GLT}, which are proven in Section
\ref{sec:GLT-proofs}. Main result $3$ corresponds to Theorem \ref{thm: IDS jumps Sturmian combs},
which is proven in Section \ref{sec:example-comb-graphs}. Main result
$4$ corresponds to Theorems \ref{thm:Backwards-type}, \ref{thm:bad-type},
and \ref{lem:Sturmian-type}, proven in Section \ref{sec:DTMP-proofs}.
Main result $5$ corresponds to Theorem \ref{thm:DTMP-Adj}, which
is proven in Section \ref{sec:DTMP-computation}.

A schematic flowchart of the thesis structure and logical dependencies
of the main results appears in Figure \ref{fig:flowchart}.

\begin{figure}
\includegraphics[scale=0.5]{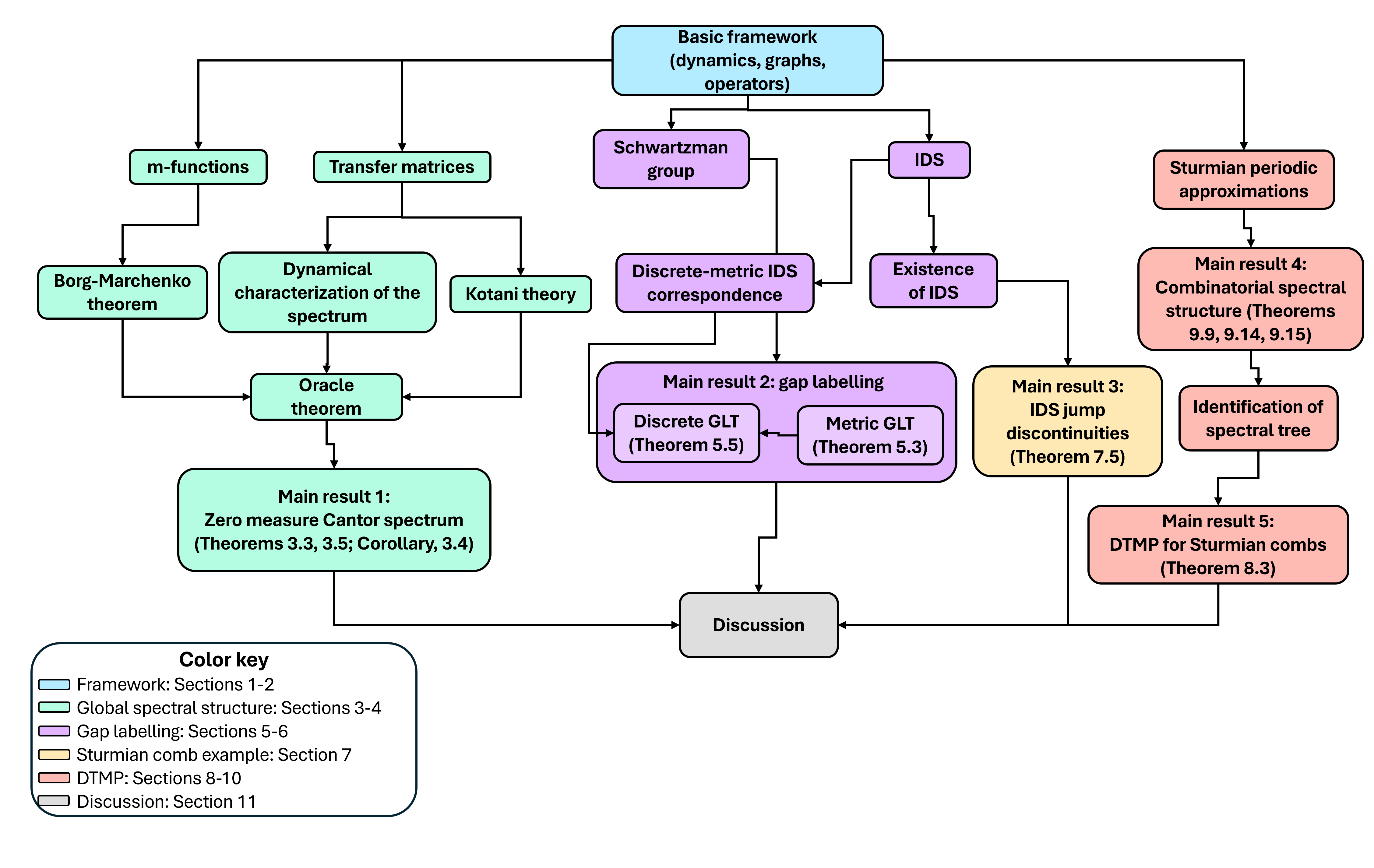}

\caption[Flowchart of thesis structure.]{Schematic flowchart of thesis structure and logical dependencies of
main results.  \label{fig:flowchart}}
\end{figure}

\newpage{}

\section{Zero measure Cantor spectrum -- preliminaries \label{sec:Kotani}}

\subsection{Introduction and statement of main results}

The goal of this section is to introduce the background necessary
to state the main results related to zero measure and Cantor spectrum
for aperiodic metric tiling graphs, as well as some of the main tools
used in the proofs. We continue to work with a minimal aperiodic subshift
$\Omega$ as introduced in Section \ref{sec:System-and-main}.
\begin{defn}
\label{def:Boshernitzan}1. Given a finite subword $W=W\left(0\right)...W\left(k-1\right)$
of an element in $\Omega$, we define the corresponding \textit{cylinder
set} by
\begin{equation}
V_{W}=\left\{ \omega\in\Omega:\omega|_{\left[0,k-1\right]}=W\right\} .\label{eq:bosh1}
\end{equation}
2. If $\mu$ is an $\sft$-invariant probability measure on $\Omega$,
we define the quantity
\begin{equation}
\eta_{\mu}\left(n\right)=\min\left\{ \mu\left(V_{W}\right):\left|W\right|=n\right\} .\label{eq:bosh2}
\end{equation}
3. A minimal subshift $\Omega$ is said to satisfy \textit{Boshernitzan's
condition} if there exists an ergodic probability measure $\mu$ on
$\Omega$ such that
\begin{equation}
\limsup_{n\rightarrow\infty}n\cdot\eta_{\mu}\left(n\right)>0.\label{eq:Boshernitzan}
\end{equation}
\end{defn}

This condition, introduced by Boshernitzan in \cite{Boshernitzan1984,Boshernitzan1985},
is a quantitative recurrence condition for  an aperiodic subshift.
It says that along an unbounded sequence of word lengths $n$, every
admissible word of length $n$ has frequency at least of order $1/n$.
Thus, although the subshift may be aperiodic, no allowed word may
become ``too rare'' along these scales. The condition holds for
many standard low-complexity subshifts, including Sturmian subshifts
and more generally linearly recurrent subshifts, see \cite{Damanik2006a}
for further examples and discussion. For minimal subshifts, Boshernitzan\textquoteright s
condition implies strong uniformity for the associated dynamics, as
discussed in Subsection \ref{subsec:Transfer-matrices}, and this
dynamical property plays an important role in obtaining zero measure
spectrum.

Throughout this section and the next, we make the following technical
assumption about the tiles:
\begin{assumption}
\label{assu:no-deg-2}At most one tile in $\left(\Ta\right)_{a\in\mathcal{A}}$
is an interval. Furthermore, each tile in $\left(\Ta\right)_{a\in\A}$
has no degree-$2$ vertices, except for possibly the boundary vertices
$v_{1}^{a}$ and $v_{2}^{a}$. In particular, if $v_{1}^{a}\in\Gamma_{\omega}$
lies in a non-interval tile, then every vertex adjacent to it has
degree at least $3$.
\end{assumption}

The assumption of excluding vertices of degree $2$ is a common one
when considering the Kirchhoff Laplacian. Indeed, a degree-$2$ vertex
with Kirchhoff conditions may be removed by replacing its two incident
edges with a single edge whose length is the sum of their lengths.
This does not affect the Laplacian, up to unitary equivalence.

The following theorem is the main result of this part.
\begin{thm}
\label{thm:TMP}Let $\Omega$ be a minimal aperiodic subshift over
an alphabet $\A$ satisfying Boshernitzan's condition. Let $\left(\Gamma_{\omega}\right)_{\omega\in\Omega}$
be an associated family of metric tiling graphs, with tiles satisfying
Assumption \ref{assu:no-deg-2}. Further assume that for each $a\in\A$,
the quantity
\begin{equation}
\ell_{\min}^{a}:=\min\left\{ \ell_{e}:e\sim v_{a}^{1}\right\} \label{eq:min-l-1}
\end{equation}
is realized by a unique edge, and that $\ell_{\min}^{a}\neq\ell_{\min}^{a'}$
for all $a\neq a'$. Then:

1. $\spec{H_{\omega}}\subset\R$ is independent of $\omega$ and has
Lebesgue measure zero. In particular, the absolutely continuous spectrum
$\text{Spec}_{ac}\left(H_{\omega}\right)$ is empty.

2. There exists a (possibly empty) set of isolated eigenvalues $F\subset\R$
such that $\spec{H_{\omega}}\backslash F$ is a generalized Cantor
set.
\end{thm}

By \textit{a generalized} Cantor set, we mean a closed, nowhere-dense
subset of $\R$ with no isolated points (not necessarily compact).
A precise description of the set $F$ will be given in Section \ref{sec:Kotani-proofs}.

Recall that a subset $A\subset\mathbb{R}_{+}^{n}$ is called Baire-generic
if it can be expressed as the countable intersection of dense open
subsets. Noting that the additional condition on the tile edge lengths
presented in Theorem \ref{thm:TMP} is satisfied for a Baire-generic
and Lebesgue almost-sure choice of the lengths, we immediately deduce
the following:
\begin{cor}
\label{cor:generic-cantor}For a Baire-generic and Lebesgue almost-sure
choice of the tile edge lengths, any resulting tiling graph satisfies
that $\spec{H_{\omega}}$ has Lebesgue measure zero and, apart from
a discrete set of isolated eigenvalues, is a generalized Cantor set.
\end{cor}

While this generic condition on the edge lengths is a sufficient one,
it is definitely not necessary, and we discuss this in Section \ref{sec:Discussion}.
For instance, it does not cover the important class of decorated $\Z$-graphs
(see Subsection \ref{def: decorated lines}). Nevertheless, for decorated
$\Z$-graphs, the conclusion can be strengthened.
\begin{thm}
\label{thm:Cantor}Let $\left(\Omega,\sft\right)$ be a minimal aperiodic
subshift satisfying Boshernitzan's condition, and let $\left(\Gamma_{\omega}\right)_{\omega\in\Omega}$
be an associated family of decorated $\Z$-graphs. Assume that the
decorations do not contain loop edges (i.e., an edge from a vertex
to itself). Then, for a Baire-generic choice of the decoration edge
lengths, $\spec{H_{\omega}}\subset\R$ is independent of $\omega$,
is a generalized Cantor set of Lebesgue measure zero.
\end{thm}

In particular, generically no discrete exceptional set appears. As
we shall see in Subsection \ref{subsec:Generic-Cantor-spectrum},
$\spec{H_{\omega}}$ may not be a Cantor set if $H_{\omega}$ admits
flat bands (i.e. infinitely degenerate eigenvalues with compactly
supported eigenfunctions), a case that may occur for particular (non-generic)
choices of the edge lengths, or in the presence of loop edges.

\subsection{Transfer matrices and the Lyapunov exponent\label{subsec:Transfer-matrices-and-Lyapunov}}

We now introduce the dynamical object that will connect the spectral
problem to ergodic theory, namely the transfer matrix cocycle. This
will be used to obtain a dynamical characterization of the spectrum
via the associated Lyapunov exponent. We give a brief review of those,
and refer to \cite{Damanik2022} for a more thorough overview.

As always, we fix the minimal subshift $\left(\Omega,\sft\right)$,
with an associated family of metric tiling graphs $\left(\Gamma_{\omega}\right)_{\omega\in\Omega}$.
Consider the following ODE on $\Gamma_{\omega}$:
\begin{equation}
-\frac{d^{2}}{dx^{2}}\phi_{\omega}=E\phi_{\omega}\,\,\,\,(E\in\mathbb{R}),\label{eq:ODE-transfer}
\end{equation}
with the Kirchhoff condition imposed at all vertices of $\Gamma_{\omega}$.
For a solution $\phi_{\omega}$ to (\ref{eq:ODE-transfer}), setting
$o:=o\left(\Gamma_{\omega}\right)$, we denote
\begin{equation}
\Phi_{\omega}:=\left(\begin{array}{c}
\phi_{\omega}\left(o\right)\\
\sum_{e\sim o_{+}}\phi_{\omega}|_{e}'\left(o\right)
\end{array}\right)\in\C^{2},\label{eq:initial-ode}
\end{equation}
where the sum is taken over edges incident to the origin within the
tile $\Gamma_{\omega}^{\left(0\right)}$ (hence $e\sim o_{+}$), and
the derivatives are taken in the positive direction. Standard ODE
theory then shows that for all but a discrete set of $E\in\mathbb{R}$
(described after Definition \ref{def:Bad-energies}), there exists
a unique matrix $\mw{}\in SL\left(2,\mathbb{C}\right)$ such that
for any solution $\phi_{\omega}$ to (\ref{eq:ODE-transfer}),
\begin{equation}
\Phi_{\sft\omega}=\mw{}\Phi_{\omega}.\label{eq:trans-mat}
\end{equation}
The matrix $\mw{}$ is known as the one-step transfer matrix (or monodromy
matrix), as it describes how boundary data propagates from one tile
to the next. One can similarly define the $n$-step transfer matrix
by
\begin{equation}
\mnw:=\begin{cases}
\mtw{\sft^{n-1}}\cdot...\cdot\mw{}, & n\geq1,\\
I_{2}, & n=0,\\
\mtw{\sft^{n}}^{-1}\cdot...\cdot\mtw{\sft^{-1}}^{-1}, & n\leq-1.
\end{cases}\label{eq:trans-mat-n}
\end{equation}

The transfer matrices allow us to define a dynamical system on $SL\left(2,\C\right)$
(known as a cocycle) through
\begin{equation}
\left(\omega,n\right)\mapsto\mnw.\label{eq:cocycle}
\end{equation}

Given $E\in\R\backslash\B$ and our cocycle $\mw{}$, fix the ergodic
invariant measure $\mu$ appearing in Boshernitzan\textquoteright s
condition. Then Kingman's ergodic theorem ensures that the quantity
\begin{equation}
L_{\omega}\left(E\right):=\lim_{n\rightarrow\infty}\frac{1}{n}\log\left\Vert \mnw\right\Vert ,\label{eq:Ly-exp}
\end{equation}
is identical for almost every $\omega\in\Omega$. We denote this almost-sure
value by $L_{\Omega}\left(E\right)$, and call $L_{\Omega}\left(E\right)$
the \textit{Lyapunov exponent} of the cocycle.

Finally, we define the exceptional set of energies for which the transfer
matrix is not defined.
\begin{defn}
\label{def:Bad-energies}Denote by $\B\subset\R$ the set of energies
$E\in\R$ for which the cocycle $\mw{}$ is not well-defined.
\end{defn}

These are exactly the values of $E\in\R$ such that there exists a
nontrivial solution to (\ref{eq:ODE-transfer}) with initial condition
$\Phi_{\omega}=\vec{0}$ for some $\omega\in\Omega$ (see e.g., \cite[thm. 2.1 and cor. 2.4]{Band2012a}).
This set is independent of $\omega$, and is discrete. 

\subsection{$m$-functions and the Borg--Marchenko theorem\label{subsec:m-functions}}

\subsubsection{Definitions\label{subsec:m-definitions}}

In this section, we introduce a main tool for the proof of Theorem
\ref{thm:TMP}, which is the $m$-function associated with a quantum
graph (also known as the Weyl--Titchmarsh function, or Dirichlet-to-Neumann
map). We then prove that the $m$-function of a tiling graph $\Gamma_{\omega}$
can be used to recover the sequence of tiles $\omega\left(n\right)$,
which will be used in the proof of the oracle theorem \ref{thm:oracle}.
Such an inverse spectral geometric result is commonly known as a Borg--Marchenko
type result, see \cite{Avdonin2008a,Kurasov2009,Kurasov_Book} for
some relevant results in the context of quantum graphs, or \cite{Marchenko1950}
for the original paper by Marchenko. In the process of the proof,
we shall also derive an interesting periodic orbit expansion formula
for the $m$-function, similar in spirit to the works above.

We now introduce the $m$-functions associated with $\Gamma_{\omega}$.
For this, we start by partitioning the infinite graph $\Gamma_{\omega}$
into two half-infinite graphs $\Gamma_{\omega}^{\pm}$, which are
obtained by cutting $\Gamma_{\omega}$ at the origin $o\left(\Gamma_{\omega}\right)$.

Let $z\in\C_{+}:=\left\{ z\in\C:\mathrm{Im}\left(z\right)>0\right\} $.
On each half-graph $\Gamma_{\omega}^{\pm}$, consider the differential
equation
\begin{equation}
-\frac{d^{2}}{dx^{2}}f_{\pm}=zf_{\pm},\label{eq:m-ODE}
\end{equation}
with the Kirchhoff condition imposed at all vertices except for $o\left(\Gamma_{\omega}^{\pm}\right)$,
where no boundary condition is imposed. Since $z\notin\R$, this equation
has a unique solution $f_{\omega}^{\pm}\left(\thinspace\cdot\thinspace;z\right)\in L^{2}\left(\Gamma_{\omega}^{\pm}\right)$,
up to a nonzero scalar multiple (the proof is similar to that of \cite[lem. 9.7, thm. 9.1]{Teschl2014}).
The $m_{\omega}^{\pm}$ functions are then defined as follows:
\begin{align}
 & m_{\omega}^{\pm}:\C_{+}\rightarrow\C,\label{eq:m2}\\
 & m_{\omega}^{\pm}\left(z\right)=\frac{\sum_{e\sim o_{\pm}}\left.\frac{df_{\omega}^{\pm}}{dx}\right|_{e}\left(o;z\right)}{f_{\omega}^{\pm}\left(o;z\right)},\label{eq:m3}
\end{align}
where we abbreviate notation by using $o$ instead of $o\left(\Gamma_{\omega}^{\pm}\right)$,
the notation $o_{\pm}$ is as in (\ref{eq:initial-ode}), and the
derivatives are taken in the outward-pointing direction from $o$
(the $\pm$ sign affects the chosen edges and the derivative directions).
By standard theory (see \cite[thm. 18.2]{Kurasov_Book}), $m_{\omega}^{\pm}\left(z\right)$
are Herglotz functions, i.e., they are analytic and their image is
contained in $\C_{+}$.

Our goal in this section is to prove the following Theorem:
\begin{thm}[Borg--Marchenko]
\label{thm:BM}Assume that the individual tiles $\left(\Ta\right)_{a\in\mathcal{A}}$
are given and satisfy the conditions of Theorem \ref{thm:TMP}. Then
for any $\omega\in\Omega$, the pair of corresponding $m$-functions
$m_{\omega}^{\pm}\left(z\right)$ uniquely determines the sequence
$\omega$.
\end{thm}

This result essentially states that one can reconstruct the graph
$\Gamma_{\omega}^{\pm}$ from the functions $m_{\omega}^{\pm}\left(z\right)$,
assuming that the basic building blocks (tiles) are known. The proof
relies on a periodic orbit expansion for $m_{\omega}^{\pm}\left(z\right)$,
which is developed next. From this point and in the following subsections,
we conveniently focus on the half-graph $\Gamma_{\omega}^{+}$ (and
the associated $m$-function $m_{\omega}^{+}\left(z\right)$), noting
that all the derivations for $m_{\omega}^{-}\left(z\right)$ are completely
identical. We also keep $\omega$ fixed, and thus from now on and
until the end of the section, we abbreviate and simply write $o$
instead of $o\left(\Gamma_{\omega}^{+}\right)$.

\newpage{}

\section{Zero measure Cantor spectrum -- proofs\label{sec:Kotani-proofs}}

This section is devoted to proving the results stated in Section \ref{sec:Kotani}.
We briefly outline the strategy of the proof, highlighting the main
steps. We first provide a dynamical characterization of $\spec{H_{\Omega}}$
as the zero set of an associated Lyapunov exponent (Theorem \ref{thm:Kotani-eq}),
reducing the spectral problem to a statement about the transfer matrix
cocycle. We then use this to prove an \textquotedblleft oracle theorem\textquotedblright{}
(Theorem \ref{thm:oracle}) -- showing that if the set where the
Lyapunov exponent vanishes has positive Lebesgue measure, then the
half-line restrictions $\omega|_{\N}$ and $\omega|_{-\N}$ determine
one another. To prove the oracle theorem, we first prove a Borg--Marchenko
result (Theorem \ref{thm:BM}), showing that the sequence of tiles
is uniquely determined by the associated $m$-function. This allows
us to pass from identities for the $m$-functions (given by Kotani
theory) to a deterministic relation between the two halves of the
sequence. Finally, we use the oracle theorem in order to show that
if $\spec{H_{\Omega}}$ is of positive Lebesgue measure, then only
finitely many local configurations can occur along orbits in the subshift
$\Omega$, and so $\Omega$ is in fact periodic. Cantor spectrum then
follows from simple topological arguments. Lastly, for the particular
case of decorated $\Z$-graphs, we show that one can generically exclude
the existence of flat bands, and that upon excluding those, the spectrum
is a Cantor set.

\subsection{The Borg--Marchenko theorem}

We begin by establishing the Borg--Marchenko result, which allows
us to recover the sequence of tiles from the associated $m$-function.
This will later allow us to relate spectral data to dynamical data
through the oracle theorem. The proof relies on a relation between
the $m$-function and the dynamical Dirichlet-to-Neumann map, which
will allow us to express the $m$-function via graph periodic orbits,
and in turn to reconstruct the sequence of tiles.

\subsubsection{The dynamical Dirichlet-to-Neumann map}

We start by introducing the reader to relevant definitions related
to the dynamical Dirichlet-to-Neumann map (also known as the dynamical
response operator), and refer to \cite[sec. 19]{Kurasov_Book} for
a more thorough background. We then use it to derive a periodic orbit
expansion for the $m$-function, which will be used to prove Theorem
\ref{thm:periodic-orbits}.

Consider the dynamical wave equation on $\Gamma_{\omega}^{+}$:
\begin{align}
 & \frac{\partial^{2}u}{\partial x^{2}}=\frac{\partial^{2}u}{\partial t^{2}},\label{eq:wave equation}\\
 & u\left(x,0\right)=0,\label{eq:wave2}\\
 & u_{t}\left(x,0\right)=0,\label{eq:wave3}
\end{align}
with the Kirchhoff condition imposed at all vertices except for $o$,
which is subject to the boundary control:
\begin{equation}
u\left(o,t\right)=b\left(t\right),\label{eq:wave4}
\end{equation}
with $b\in L_{\text{loc}}^{2}\left(0,\infty\right)$. The \textit{dynamical
Dirichlet-to-Neumann (DDTN) map} $\Lambda_{\omega}^{+}$ is the mapping:
\begin{equation}
b\left(\cdot\right)\xmapsto{\Lambda_{\omega}^{+}}\partial_{n}u\left(o,\cdot\right):=\sum_{e\sim o_{+}}\frac{\left.\partial u\right|_{e}\left(o,\cdot\right)}{\partial x}.\label{eq:ddtn}
\end{equation}
By \cite[eq. (19.12)]{Kurasov_Book}, the DDTN map and the $m$-function
are in one-to-one correspondence, given by:
\begin{equation}
m_{\omega}^{+}\left(-s^{2}\right)=\frac{\widehat{\left(\Lambda_{\omega}^{+}b\right)}\left(s\right)}{\widehat{b}\left(s\right)},\label{eq:DDTN-M-correspondence}
\end{equation}
where $\widehat{\left(\cdot\right)}$ is the Laplace transform:
\begin{equation}
\hat{f}\left(s\right)=\int_{0}^{\infty}f\left(t\right)e^{-st}dt\label{eq:Laplace}
\end{equation}
 Thus, knowing $\Lambda_{\omega}^{+}$ allows one to compute $m_{\omega}^{+}$,
and vice versa.

\subsubsection{TPOs}

In order to state the periodic orbit expansion for $m_{\omega}^{+}\left(z\right)$,
we first introduce a definition of a \textit{tailed} periodic orbit
(see also Figure~\ref{fig: periodic orbit}).
\begin{defn}
\label{def:periodic orbits} ~
\begin{enumerate}
\item Let $\Gamma$ be a metric graph. A \textit{tailed periodic orbit}
(TPO) $p$ around a vertex $v\in\Gamma$ is a finite sequence of directed
edges in $\Gamma$, $\left(v_{1},v_{2}\right)$,$\left(v_{2},v_{3}\right),...,\left(v_{n},v_{n+1}\right)$,
with $v_{1}=v_{n}=v$ (see Figure~\ref{fig: periodic orbit}). We
think of such TPOs as periodic orbits around $v$, along with a chosen
terminal direction $\left(v_{n},v_{n+1}\right)$ (thought of as the
``tail'' of the orbit). In this work, we mostly consider TPOs around
the origin of $\Gamma_{\omega}^{\pm}$ (so $v=o$).
\item For a TPO $p$, we denote its total length by $\ell_{p}:=\sum_{i=1}^{n-1}\left|\left(v_{i},v_{i+1}\right)\right|$,
where $\left|\cdot\right|$ denotes the edge length (we do not include
the tail $\left(v_{n},v_{n+1}\right)$ in the sum). We denote its
topological length by $\left|p\right|:=n-1$, i.e., the number of
directed edges in the sequence $p$, excluding the tail $\left(v_{n},v_{n+1}\right)$.
\item Let $\Gamma_{\omega}^{+}$ be a half-infinite tiling graph, with origin
vertex $o$. The set of TPOs around the origin for $\Gamma_{\omega}^{+}$
is denoted by $\mathcal{P}_{\omega}^{+}$.
\end{enumerate}
\end{defn}

\begin{figure}
\includegraphics[scale=0.55]{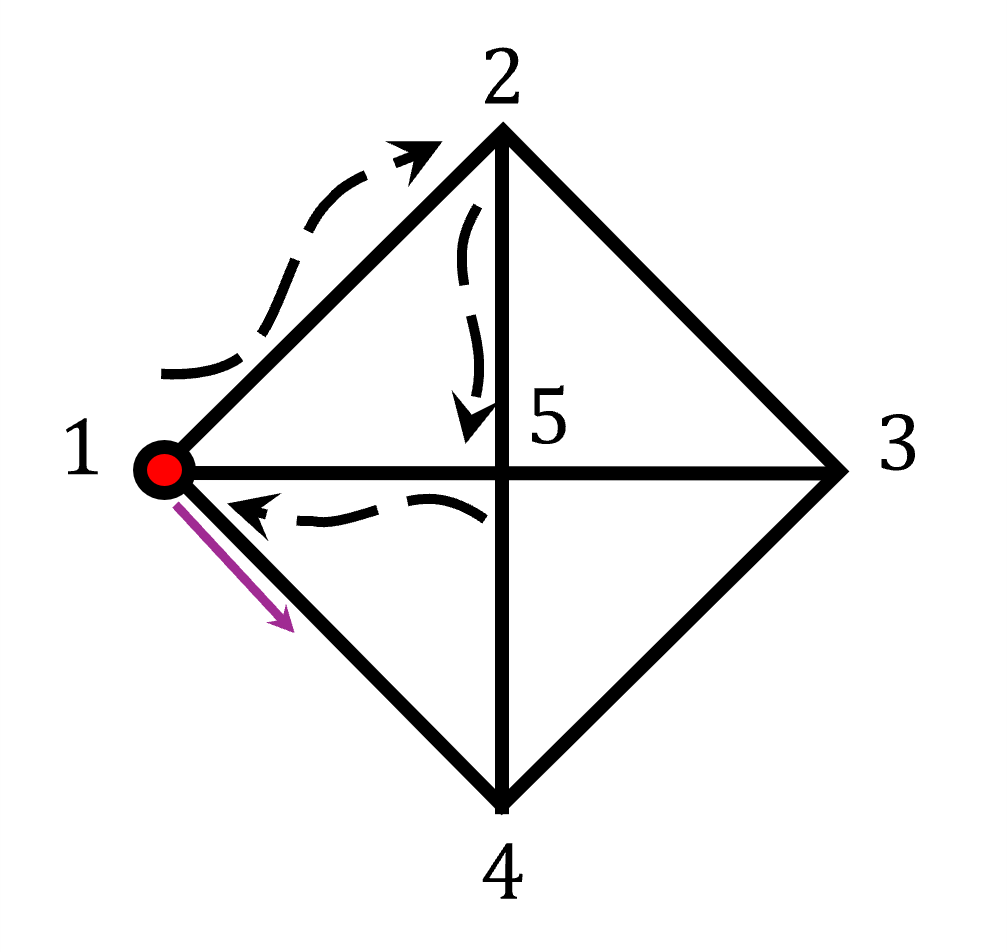}

\caption[Tailed periodic orbits on a kite graph.]{A square kite graph with vertex set $\left\{ 1,2,3,4,5\right\} $,
together with a TPO $p=\left(\left(1,2\right),\left(2,5\right),\left(5,1\right),\left(1,4\right)\right)$
around the marked vertex $1$ (the terminal direction is marked in
purple).  \label{fig: periodic orbit}}
\end{figure}

We stated the definition of TPOs for $\Gamma_{\omega}^{+}$; the associated
definitions for $\Gamma_{\omega}^{-}$ are the same. Similarly, we
shall prove all our results for $m_{\omega}^{+}$, keeping in mind
that the proofs for $m_{\omega}^{-}$ are identical.
\begin{thm}
\label{thm:periodic-orbits} The function $m_{\omega}^{+}\left(z\right)$
can be expanded in terms of the tailed periodic orbits of $\Gamma_{\omega}^{+}$:
\begin{equation}
m_{\omega}^{+}\left(z\right)=-\sqrt{-z}\sum_{p\in\mathcal{P}_{\omega}^{+}}A_{p}e^{-\sqrt{-z}\ell_{p}}.\label{eq:periodic-m}
\end{equation}
\end{thm}

The coefficients $A_{p}$ appearing in the theorem will be described
explicitly in the proof. This expansion is the key ingredient in reconstructing
the sequence $\omega|_{\N}$.

We remark that in many cases (see e.g. \cite[thm 9.22]{Teschl2014}
and \cite[thm 3.5]{Damanik2020}), Borg--Marchenko results on the
line rely on first proving asymptotic estimates for $m\left(z\right)$
as the spectral parameter $z$ tends to infinity. One advantage of
the TPO expansion in Theorem \ref{thm:periodic-orbits} is that it
yields the Borg--Marchenko Theorem \ref{thm:BM} directly, without
first passing through such asymptotics. At the same time, the expansion
can also be used to recover asymptotic estimates for $m^{\pm}\left(z\right)$
as $z\rightarrow\infty$.

\subsubsection{Proof of TPO expansion}

Our next goal is to obtain an expansion for $\Lambda_{\omega}^{+}$
(and hence for $m_{\omega}^{+}$) in terms of tailed periodic orbits\@.
Before doing so, we present the scattering matrix of a quantum graph
(which was first introduced in \cite{Kottos1999,Kottos1997}):
\begin{defn}
\label{def:Scattering-mat}Let $\Gamma$ be a quantum graph and let
$v\in\V\left(\Gamma\right)$. On each edge $e$ incident to $v$,
one can write the solutions to the eigenvalue equation $-\frac{d^{2}}{dx^{2}}g=k^{2}g$
in the form
\begin{equation}
g|_{e}\left(x\right)=a_{e}\rme^{\rmi kx}+a_{\hat{e}}\rme^{\rmi k\left(\ell_{e}-x\right)}.\label{eq:scatterwave}
\end{equation}
Here, we think of $e$ as a directed edge pointing towards a vertex
$v$, and denote its reversal by $\hat{e}$. Accordingly, $a_{e}$
and $a_{\hat{e}}$ are the amplitudes of the incoming and outgoing
waves at $v$, respectively. For $g$ to also satisfy the Kirchhoff
vertex condition at $v$, the vectors $\vec{a}_{\text{in}}:=\left(a_{e_{1}},...,a_{e_{\text{deg\ensuremath{\left(v\right)}}}}\right)$,
$\vec{a}_{\text{out}}:=\left(a_{\hat{e}_{1}},...,a_{\hat{e}_{\text{deg\ensuremath{\left(v\right)}}}}\right)$
must satisfy a particular linear relation. The scattering matrix associated
with $v$ is the unique unitary matrix $S_{v}\left(k\right)\in U\left(\deg\left(v\right)\right)$
that expresses this linear relation:
\begin{equation}
\vec{a}_{\text{out}}=S_{v}\left(k\right)\rme^{\rmi k\vec{L}}\vec{a}_{\text{in}},\label{eq:a=00003DUa}
\end{equation}
where $\vec{L}=diag\left(\ell_{e_{1}},...,\ell_{e_{\deg\left(v\right)}}\right)$
is the diagonal matrix of edge lengths.
\end{defn}

The scattering matrix generally depends on the imposed vertex conditions.
Given two directed edges $e,e'$, we write that $e\rightarrow e'$
at $v$ if the vertex $v$ is both the end vertex of $e$ and the
starting vertex of $e'$. One can show that for the Kirchhoff condition,
$S_{v}\left(k\right)=:S_{v}$ is independent of $k$ and is given
by (cf. \cite[eq. (3.6)]{Gnutzmann2006a}),
\begin{equation}
\left[S_{v}\right]_{e',e}=\begin{cases}
\frac{2}{\deg\left(v\right)}-1, & e'=\hat{e},\\
\frac{2}{\deg\left(v\right)}, & e\rightarrow e'\text{ at \ensuremath{v} and }e'\neq\hat{e}\\
0, & \text{otherwise.}
\end{cases},\label{eq:kirch-scar}
\end{equation}

The scattering matrix is a common tool in the theory of quantum graphs,
not only for studying the spectrum of the Laplacian, but also for
describing the propagation of waves on $\Gamma$ (as we shall do now).
\begin{proof}[Proof of Theorem \ref{thm:periodic-orbits}]
 We show that the DDTN map can be represented as convolution against
a generalized kernel:
\begin{equation}
\Lambda_{\omega}^{+}b\left(t\right)=\left(k_{\omega}^{+}*b\right)\left(t\right):=\int_{0}^{\infty}k_{\omega}^{+}\left(t-s\right)b\left(s\right)ds.\label{eq:ddtnker}
\end{equation}
Combining (\ref{eq:DDTN-M-correspondence}) with the fact that the
Laplace transform is multiplicative with respect to convolution gives
the simple relation between the kernel $k_{\omega}^{+}\left(t\right)$
and the $m$-function:
\begin{equation}
m_{\omega}^{+}\left(-s^{2}\right)=\frac{\widehat{\left(k_{\omega}^{+}*b\right)}\left(s\right)}{\widehat{b}\left(s\right)}=\widehat{k_{\omega}^{+}\left(s\right)},\label{eq:kernel-M-correspondence}
\end{equation}
and so to prove the theorem we develop a TPO expansion for the integral
kernel.

The key idea is that the solution $u\left(x,t\right)$ to the dynamical
wave equation (\ref{eq:wave equation}) consists of waves propagating
along the edges of the graph, with different waves corresponding to
different TPOs. We will systematically construct the expansion of
$k_{\omega}^{+}\left(s\right)$ by analyzing the solution to the dynamical
wave equation on $\Gamma_{\omega}^{+}$, keeping track of how it scatters
at each vertex. Due to the finite propagation speed of the wave equation
(which equals to $1$ by (\ref{eq:wave equation})), the solution
at time $t$ is always supported on the geodesic ball of radius $t$
around the origin, which allows us, for any $t>0$, to focus on the
behavior of the solution at some bounded neighborhood of the origin.

~

\textbf{\textit{Step 1: Short-time behavior of the wave equation.}}

Recall (see e.g. \cite{pinchover2005introduction}) that the d'Alembert
solution to the dynamical wave equation is given as a superposition
of an ingoing and outgoing wave on each edge:
\begin{equation}
u|_{e}\left(x,t\right)=F_{e}\left(t-x\right)+G_{e}\left(t+x\right).\label{eq:Dalambert}
\end{equation}
 Through direct substitution into (\ref{eq:wave equation}),(\ref{eq:wave2}),(\ref{eq:wave3}),(\ref{eq:wave4}),
one can verify that for $t\leq\min_{e\sim o}\left\{ \left|e\right|\right\} =:L$,
the solution to the wave equation on each edge is given by
\begin{equation}
u|_{e}\left(x,t\right)=\begin{cases}
b\left(t-x\right), & e\sim o,\\
0, & \text{otherwise}.
\end{cases}\label{eq:1wave}
\end{equation}
Thus, for small times, the wave propagates only along edges adjacent
to the origin, and the associated d'Alembert solution has profile
equal to that of the boundary control $b\left(t\right)$. This means
that at short times, the DDTN satisfies: 
\begin{align}
 & \Lambda_{\omega}^{+}b\left(t\right)=\partial_{n}u\left(o,t\right)=-\sum_{e\sim o_{+}}b'\left(t\right),\label{eq:2wave}
\end{align}
and so we may choose the representing kernel to be
\begin{equation}
k_{\omega}^{+}\left(t\right)=-\deg\left(o\right)\cdot\delta'\left(t\right)+...,\label{eq:3wave}
\end{equation}
where the additional terms vanish for $t<L$, and the equality above
is in the sense of distributions. Here, $\delta'$ is a delta prime
distribution, defined through
\begin{equation}
\left(\delta'*f\right):=f'\left(0\right).\label{eq:delta-prime}
\end{equation}

~

\textbf{\textit{Step 2: First scattering event.}}

Let $v$ be the vertex which is the nearest to the origin. Recalling
that we denoted $L:=\min_{e\sim o}\left\{ \left|e\right|\right\} $,
we have $L=\left|(o,v)\right|$. Then, once the wave $u|_{e}\left(x,t\right)$
reaches $v$ at time $t=L$, it scatters to all neighboring edges
$e'\sim v$ (see Figure \ref{fig: scattering}), where the associated
wave amplitudes at each edge are determined by the Kirchhoff condition
at $v$. It is simple to verify using Equations (\ref{eq:scatterwave}),(\ref{eq:a=00003DUa})
that for a plane wave solution to the wave equation of the form
\begin{equation}
w|_{e}\left(x,t\right)=\rme^{\rmi\left(t-kx\right)},\label{eq:plane-wave}
\end{equation}
the backscattered wave amplitude along $\hat{e}$ following this event
is determined via the scattering matrix element $\left[S_{v}\right]_{\hat{e},e}$,
i.e.,
\begin{equation}
w_{e}\left(x,t\right)=\rme^{{\rm \rmi}\left(t-kx\right)}+\left[S_{v}\right]_{\hat{e},e}\rme^{\rmi\left(\left(t-L\right)-\left(L-kx\right)\right)}.\label{eq:plane-scatter}
\end{equation}
Then, application of Fourier theory (writing $u|_{e}\left(x,t\right)$
as an infinite superposition of plane waves) shows that similarly,
our d'Alembert solution at the edge $e$ after the first scattering
event is given by
\begin{equation}
u|_{e}\left(x,t\right)=b\left(t-x\right)+\left[S_{v}\right]_{\hat{e},e}b\left(\left(t-L\right)-\left(L-x\right)\right),\label{eq:4wave}
\end{equation}
see also \cite[sec. 19.5]{Kurasov_Book} for additional details. Mainly,
the scattering process at this vertex introduces a secondary wave
that backscatters into its initial edge $e$ with amplitude $\left[S_{v}\right]_{\hat{e},e}$,
but with a reversed direction (see Figure \ref{fig: scattering}).

\begin{figure}
\includegraphics[scale=0.55]{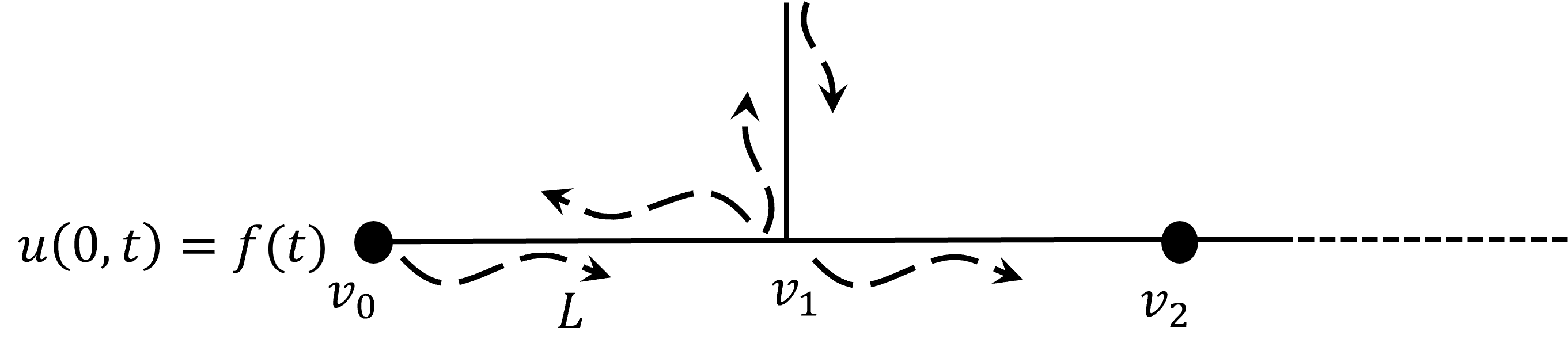}

\caption[Wave scattering on a comb graph.]{Wave scattering on a comb graph at short times. Initially, the solution
propagates along the directed edge $e=\left(v_{0},v_{1}\right)$,
until it reaches the vertex $v_{1}$ at time $t=L$. The solution
then scatters in three directions, and in particular backscatters
into the directed edge $\left(v_{1},v_{0}\right)$ with wave amplitude
$\left[S_{v_{1}}\right]_{\hat{e},e}=\frac{2}{\deg\left(v_{1}\right)}-1=-\frac{1}{3}$.
 \label{fig: scattering}}
\end{figure}
At time $t=2L$, the secondary wave reaches the origin $o$, and similarly
to before, scatters in all directions $e'\sim o$ with amplitude $\left[S_{o}\right]_{e',\hat{e}}$
due to the Kirchhoff condition at the origin. At short times after
this event\footnote{i.e. times with respect to which there is no additional scattering
event which makes another wave reach the origin.}, the normal derivative of $u\left(x,t\right)$ at the origin is given
by:
\begin{align}
 & \partial_{n}u\left(o,t\right)=-\sum_{e'\sim o_{+}}\left(b'\left(t\right)+\left[S_{o}\right]_{e',\hat{e}}\left[S_{v}\right]_{\hat{e},e}b'\left(t-2L\right)\right).\label{eq:5wave}
\end{align}
Thus, for $t\in\left[0,2L\right]$, we may choose the kernel to be
\begin{equation}
k_{\omega}^{+}\left(t\right)=-\deg\left(o\right)\delta'\left(t\right)-\left(\sum_{e'\sim o_{+}}\left[S_{o}\right]_{e',\hat{e}}\left[S_{v}\right]_{\hat{e},e}\right)\delta'\left(t-2L\right).\label{eq:6wave}
\end{equation}
Note that the terms $\delta'\left(t\right)$ and $\delta'\left(t-2L\right)$
correspond exactly to the two shortest TPOs around the origin, with
$\delta'\left(t\right)$ corresponding to the trivial TPOs of length
$0$.\\
~

\textbf{\textit{Step 3: Scattering in longer times.}}

Repeating this argument, each scattering event contributes a term
corresponding to a TPO. As time progresses beyond $t=2L$, the d'Alembert
solution $u\left(x,t\right)$ continues to propagate and undergoes
additional scattering events at vertices further away from the origin.
On edges adjacent to $o$, $u\left(x,t\right)$ will always consist
of a superposition of delayed waves that scattered across some TPO
through the origin. The term corresponding to a TPO $p$ will be of
the form $A_{p}b\left(t-\ell_{p}\right)$, where the delay time $\ell_{p}$
is equal to the length of the corresponding period $p$, and the amplitude
$A_{p}$ is the product of scattering coefficients along the TPO:
\begin{equation}
A_{p}:=\prod_{j=1}^{\left|p\right|}\left[S_{v_{j+1}}\right]_{\left(v_{j+1},v_{j+2}\right),\left(v_{j},v_{j+1}\right)}.\label{eq:scat-amp-2}
\end{equation}
On the level of the DDTN kernel $k_{\omega}^{+}\left(s\right)$, this
gives an infinite TPO expansion:
\begin{align}
 & \partial_{n}u\left(o,t\right)=-\sum_{p\in\mathcal{P}_{\omega}^{+}}A_{p}b'\left(t-\ell_{p}\right),\label{eq:2ddtn}\\
\Rightarrow\quad & k_{\omega}^{+}\left(t\right)=-\sum_{p\in\mathcal{P}_{\omega}^{+}}A_{p}\delta'\left(t-\ell_{p}\right).\label{eq:ddtn-ker-2}
\end{align}
This gives the precise expression for the integral kernel in (\ref{eq:ddtnker}).

Equation (\ref{eq:periodic-m}) can now be obtained by substituting
(\ref{eq:ddtn-ker-2}) into (\ref{eq:kernel-M-correspondence}) and
using the identity $\widehat{\delta'\left(t-\ell\right)}\left(s\right)=se^{-\ell s}$,
which overall gives:
\begin{equation}
m_{\omega}^{+}\left(z\right)=-\sqrt{-z}\sum_{p\in\mathcal{P}_{\omega}^{+}}A_{p}e^{-\sqrt{-z}\ell_{p}},\label{eq:m-final}
\end{equation}
which proves the expansion.
\end{proof}

\subsubsection{Proof of the Borg--Marchenko Theorem \ref{thm:BM}\label{subsec:bm-proof}}
\begin{proof}[Proof of Theorem \ref{thm:BM} (Borg--Marchenko)]
 Let $\left(\omega\left(n\right)\right)_{n=0}^{\infty}\in\A^{\N}$.
We use the TPO expansion (\ref{eq:periodic-m}) of $m_{\omega}^{+}\left(z\right)$
in order to inductively recover each element of $\omega$.

\uline{Base case -- \mbox{$n=0$}:}

By Assumption \ref{assu:no-deg-2}, there is at most one $\overline{a}\in\A$
such that $T_{\overline{a}}$ is an interval. Consider the TPO expansion
of $m_{\omega}^{+}\left(z\right)$. Firstly, note that if $\omega\left(0\right)=a\neq\overline{a}$,
then setting
\begin{equation}
\ell_{\min}^{a}:=\min\left\{ \ell_{e}:e\sim v_{a}^{1}\right\} ,\label{eq:lmin}
\end{equation}
the term $-A_{a}\sqrt{-z}e^{-\sqrt{-z}2\ell_{\min}^{a}}$ appears
in this TPO expansion with $A_{a}\neq0$. Indeed, in this case, there
are exactly $\deg\left(o\right)$ tailed periodic orbits $p\in\mathcal{P}_{\omega}^{+}$
of length $2\ell_{\min}^{a}$ -- these are exactly the TPOs that
traverse the (unique) edge $e_{\min}\sim o$ of length $\ell_{\min}^{a}$,
backscatter and return to the origin, and then scatter to some other
edge $e'\sim o$. If we denote by $v_{\min}$ the neighboring vertex
to $o$ such that $e_{\min}=\left\{ o,v_{\min}\right\} $, then direct
computation using (\ref{eq:kirch-scar}) and the second term in (\ref{eq:6wave})
shows that the total contribution of these shortest TPOs is given
by
\begin{align}
 & C_{a}=\sum_{e'\sim o_{+}}\left[S_{o}\right]_{e',\hat{e}}\left[S_{v_{\min}^{a}}\right]_{\hat{e},e}\label{eq:scattering-prefactor}\\
 & =\sum_{e_{\min}\neq e'\sim o_{+}}\left(\frac{2}{\deg\left(v_{\min}\right)}-1\right)\cdot\frac{2}{\deg\left(o\right)}+\left(\frac{2}{\deg\left(v_{\min}\right)}-1\right)\left(\frac{2}{\deg\left(o\right)}-1\right)\nonumber \\
 & =\left(\frac{2}{\deg\left(v_{\min}\right)}-1\right)\left(\frac{2\deg\left(o\right)-1}{\deg\left(o\right)}+\left(\frac{2}{\deg\left(o\right)}-1\right)\right)\nonumber \\
 & =\left(\frac{2}{\deg\left(v_{\min}\right)}-1\right)\left(\frac{\deg\left(o\right)+1}{\deg\left(o\right)}\right),\nonumber 
\end{align}
and $C_{a}\neq0$ because $\Ta$ is not an interval and hence Assumption
\ref{assu:no-deg-2} gives $\deg\left(v_{\min}\right)\geq3$.

Now, iterate over the list $\left(\ell_{\min}^{a}\right)_{\overline{a}\neq a\in\A}$
in increasing order. For each $\ell_{\min}^{a}$, examine the corresponding
term $-A_{a}\sqrt{-z}e^{-\sqrt{-z}2\ell_{\min}^{a}}$ in the TPO expansion
of $m_{\omega}^{+}\left(z\right)$ (with the prefactor $A_{a}$ being
possibly zero). By the discussion above, for each $\overline{a}\neq a\in\A$
such that $A_{a}=0$, we can immediately deduce that $\omega\left(0\right)\neq a$.
Suppose first that for some $\overline{a}\neq a\in\A$, the prefactor
$A_{a}$ is nonzero. Then there exists a TPO of length $2\ell_{\min}^{a}$
around $o$, and hence the unique edge of length $\ell_{\min}^{a}$
must lie in the first tile, implying that $\omega\left(0\right)=a$.
If, on the other hand, $A_{a}=0$ for all $a\neq\overline{a}$, then
the only remaining possibility is $\omega\left(0\right)=\overline{a}$.

\uline{Induction step:}

Suppose that $\omega\left(0\right),...,\omega\left(n-1\right)$ are
known, and we wish to recover $\omega\left(n\right)$. Let $p$ be
a TPO of minimal length that travels from the origin to the right
boundary vertex $v_{\omega\left(n-1\right)}^{2}$, and then returns
to the origin along the reversal of the same path (There might be
several such TPOs if $\deg\left(o\right)>1$, each corresponding to
a different tail). As in the base case, we iterate over the list $\left(\ell_{\min}^{a}\right)_{\overline{a}\neq a\in\A}$
in increasing order. Since we already know the first $n$ tiles, we
can list all TPOs of length $\ell_{p}+2\ell_{\min}^{a}$ in $\mathcal{P}_{\omega}^{+}$
passing only through the first $n$ tiles (if such exist), and explicitly
compute their contribution to the TPO expansion of $m_{\omega}^{+}\left(z\right)$
through a term of the form $-\tilde{A_{a}}\sqrt{-z}e^{-\sqrt{-z}\left(\ell_{p}+2\ell_{\min}^{a}\right)}$
(and again, $\tilde{A_{a}}$ may be $0$).

Similar to before, for each $a\neq\overline{a}$, focus on the term
$-A_{a}\sqrt{-z}e^{-\sqrt{-z}\left(\ell_{p}+2\ell_{\min}^{a}\right)}$
in the TPO expansion of $m_{\omega}^{+}\left(z\right)$. If for some
$\overline{a}\neq a\in\A$ we have that $A_{a}\neq\tilde{A_{a}}$,
then there must be additional TPOs of length $\ell_{p}+2\ell_{\min}^{a}$
which pass through the $n$th tile $\TG$ (leading to a contribution
that causes $A_{a}\neq\tilde{A_{a}}$). By the choice of $\ell_{p}$
and $\ell_{\min}^{a}$, such TPOs must consist of the shortest path
from the origin to $v_{\omega\left(n-1\right)}^{2}$, going along
and back the unique edge $e_{\min}$ of length $\ell_{\min}^{a}$
in $\TG$ adjacent to $v_{\omega\left(n\right)}^{1}$, and then returning
to the origin through the initial path (reversed). Since $\ell_{\min}^{a'}\neq\ell_{\min}^{a}$
for $a\neq a'\in\A$, we deduce that $\omega\left(n\right)=a$. If
instead $A_{a}=\tilde{A_{a}}$ for all $a\neq\overline{a}$, then
no non-interval tile can appear at position $n$, and hence the only
remaining possibility is $\omega\left(n\right)=\overline{a}$. This
completes the reconstruction of the sequence $\omega$.
\end{proof}
\begin{rem}
We emphasize that Theorem \ref{thm:BM} shows that one can reconstruct
$\Gamma_{\omega}$ from $m_{\omega}^{\pm}\left(z\right)$ only assuming
that the possible tiles $\left(\Ta\right)_{a\in\A}$ are already known.
One may ask when one can recover $\Gamma_{\omega}$ from $m_{\omega}^{\pm}\left(z\right)$
without a priori knowing the possible tiles. The proof of Theorem
\ref{thm:BM} shows that this is essentially equivalent to being able
to recover the individual (compact) tiles from their respective $m$-functions.
This question has been addressed in works such as \cite{Avdonin2008a,Avdonin2010,Kurasov2009}.
This can for instance be done assuming that the tiles $\left(\Ta\right)_{a\in\A}$
are trees, but in general counterexamples exist.
\end{rem}

The following is an immediate corollary of the TPO expansion for $m_{\omega}^{\pm}\left(z\right)$
and the Borg--Marchenko theorem:
\begin{cor}
\label{cor:loc-uni}A sequence $\left(\omega_{n}\right)_{n\in\Z}$
converges to $\omega$ in the product topology if and only if $m_{\omega_{n}}^{\pm}\left(z\right)\rightarrow m_{\omega}^{\pm}\left(z\right)$
locally uniformly for $z\in\C_{+}$. In other words, the mapping $\omega\mapsto m_{\omega}^{\pm}\left(z\right)$
is bijective and continuous.
\end{cor}

This result will be used in the proof of the oracle theorem (Theorem
\ref{thm:oracle}) to pass from identities of $m$-functions to continuity
of the corresponding restriction maps.

While the proof of the Borg--Marchenko result above was tailored
for tiling graphs which satisfy the generic edge lengths assumption
of Theorem \ref{thm:TMP}, a similar result can be proven for many
other families of tiles which do not satisfy this assumption, as long
as the TPOs of the different tiles are ``sufficiently distinguishable''.
One such class for which this can be done (under a similar generic
assumption) is the class of decorated $\Z$-graphs from Subsection
\ref{def: decorated lines}. 
\begin{assumption}
\label{assu:deco-deg}Let $\Gamma_{\omega}$ be a decorated $\Z$-graph
as in Subsection \ref{def: decorated lines}. Suppose that the decorations
$\left(\gra\right)_{a\in\A}$ satisfy that at most one decoration
is trivial (i.e., a single vertex), and that the decorations have
no vertices of degree $2$, except for possibly the base vertex $u_{a}$.
\end{assumption}

\begin{prop}
\label{prop:BM-decorated}Let $\Gamma_{\omega}$ be a decorated $\Z$-graph
as in Assumption \ref{assu:deco-deg}. Suppose that for each $a\in\A$,
the quantity
\begin{equation}
\ell_{\min}^{a}:=\min\left\{ \ell_{e}:e\sim v_{a}\right\} \label{eq:lmin-decorated-line}
\end{equation}
is realized by a unique edge, and that $\ell_{\min}^{a}\neq\ell_{\min}^{a'}$
for all $a\neq a'$. Then for each $\omega\in\Omega$, the function
$m_{\omega}^{+}\left(z\right)$ uniquely determines  $\left(\omega\left(n\right)\right)_{n=0}^{\infty}$.
\end{prop}

Since the geometric assumption in the proposition above holds Baire-generically
and Lebesgue almost-surely, we immediately have the following:
\begin{cor}
\label{cor:generic-BM-decorated}For a Baire-generic and Lebesgue
almost-sure choice of the decoration edge lengths, the Borg--Marchenko
property holds for decorated $\Z$-graphs.
\end{cor}

The proof of Proposition \ref{prop:BM-decorated} is completely analogous
to that of Theorem \ref{thm:BM}, and appears in Appendix \ref{sec:Appendices}.

\subsection{Dynamical characterization of the spectrum\label{subsec:Transfer-matrices}}

In this subsection, we provide a dynamical characterization of the
spectrum, by showing that $\spec{H_{\omega}}$ is equal (up to a possible
discrete subset) to the zero set of the Lyapunov exponent associated
with the cocycle. Such a result has been proven in the ergodic Schrödinger
operator setting for many systems, including Sturmian Hamiltonians,
the almost Mathieu operator (cf. \cite{Damanik2006}), as well as
in certain continuum models \cite{Damanik2014}. This characterization
will allow us to apply Kotani theory in the next subsection.

Recall from Subsection \ref{subsec:Transfer-matrices-and-Lyapunov}
that for each $\omega\in\Omega$ the cocycle $\mw{}$ is well-defined
for all but a discrete subset of energies $\B\subset\R$. In general,
$\B$ may intersect $\spec{H_{\omega}}$ nontrivially -- for instance,
this is the case for the aperiodic antitrees studied in \cite{Damanik2020}.
At such spectral points, the analysis of $\spec{H_{\omega}}$ through
the transfer matrix applied below breaks down. Nevertheless, since
the set $\B$ is discrete, it is enough to analyze $\spec{H_{\omega}}$
on $\R\backslash\B$ in order to prove that $\spec{H_{\omega}}$ is
of Lebesgue measure zero. In Subsection \ref{subsec:Generic-Cantor-spectrum},
we show that generically $\B\cap\spec{H_{\omega}}=\varnothing$ for
decorated $\Z$-graphs.
\begin{thm}
\label{thm:Kotani-eq} Let $(\Omega,S)$ be a minimal aperiodic subshift
satisfying Boshernitzan's condition, and let $\left(\Gamma_{\omega}\right)_{\omega\in\Omega}$
be the associated family of metric tiling graphs. Then there exists
a closed set~ $\spec{H_{\Omega}}\subset\R$ such that~ $\spec{H_{\omega}}=\spec{H_{\Omega}}$
for all $\omega\in\Omega$. Furthermore,
\begin{equation}
\spec{H_{\Omega}}\backslash\B=\mathcal{Z}:=\left\{ E\in\R\backslash\B:L_{\Omega}\left(E\right)=0\right\} .\label{eq:sigma=00003Dz}
\end{equation}
\end{thm}

Thus, apart from the exceptional set $\B$, the spectrum is exactly
the set of energies with zero Lyapunov exponent. This provides the
dynamical characterization needed for Kotani theory.

The proof of Theorem \ref{thm:Kotani-eq} appears at the end of this
subsection. To prove it, we first need to present several lemmas.
These lemmas are adaptations of existing techniques to quantum graphs.
Since this subsection mostly adapts many of the techniques presented
in previous works to the quantum graph setting, some of the results
will not be proven completely, and the proofs will only be outlined,
with reference to the relevant works.

We now introduce the standard decomposition of energies according
to the behavior of the cocycle.
\begin{defn}
Define
\begin{align}
 & \UH=\left\{ E\in\R\backslash\B:\exists\gamma>1,C>0;\left\Vert \mnw\right\Vert \geq C\gamma^{\left|n\right|},\forall\omega\in\Omega,\forall n\in\mathbb{Z}\right\} ,\label{eq:uh}\\
 & \NUH=\left\{ E\in\R\backslash\B:L_{\Omega}\left(E\right)>0\right\} \backslash\UH.\label{eq:NUH}
\end{align}
\end{defn}

It holds that $\R\backslash\B=\UH\sqcup\NUH\sqcup\mathcal{Z}$, and
this partition into uniformly hyperbolic, non-uniformly hyperbolic
and non-hyperbolic energies forms the basis to our spectral analysis.
\begin{lem}[{\cite[thm. 1.2]{Damanik2016a}}]
\label{fact:fact1} Let $E\in\UH$. Then there is no nontrivial solution
$\phi_{\omega}$ to the ODE (\ref{eq:ODE-transfer}) such that the
vector $\Phi_{n}$ in (\ref{eq:initial-ode}) is subexponentially
bounded in both directions in $\Gamma_{\omega}$, i.e. such that for
every $\varepsilon>0$ there exists $C_{\varepsilon}>0$ with
\begin{equation}
\|\Phi_{n}\|\le C_{\varepsilon}e^{\varepsilon|n|}\qquad\forall n\in\mathbb{Z}.\label{eq:subexponential}
\end{equation}
\end{lem}

\begin{lem}
\label{fact:fact2} For Lebesgue almost every $E\in\spec{H_{\omega}}$,
there is a solution to $-\frac{d^{2}\phi_{\omega}}{dx^{2}}=E\phi_{\omega}$
such that for all $\beta>0$,
\begin{equation}
e^{-\beta\cdot\left|x\right|}\phi_{\omega}\left(x\right)\in L^{2}\left(\Gamma_{\omega}^{\pm}\right),\label{eq:L2-decay-2}
\end{equation}
where $\left|x\right|$ denotes the geodesic distance of $x$ from
the origin $o$ with respect to the metric on $\Gamma_{\omega}$.
\end{lem}

That is, there exist subexponentially bounded generalized eigenfunctions.
\begin{proof}[Proof of Lemma \ref{fact:fact2}]
 We apply the result from \cite[thm. 1.1]{Monvel2003}. It remains
to verify its two hypotheses for the present metric-graph setting:\\
1. For all $\alpha>0$ and $x_{0}\in\Gamma_{\omega}$,
\begin{equation}
\lim_{R\rightarrow\infty}e^{-\alpha R}\left|B_{R}^{\omega}\left(x_{0}\right)\right|=0,\label{eq:Brx}
\end{equation}
where $B_{R}^{\omega}\left(x_{0}\right)$ is the geodesic ball of
radius $R$ around $x_{0}\in\Gamma_{\omega}$ and $\left|\cdot\right|$
denotes the Lebesgue measure.\\
2. For all $\omega\in\Omega$, the heat semigroup $e^{-tH_{\omega}}$
associated with the system is ultracontractive, i.e.,
\begin{equation}
e^{-tH_{\omega}}:L^{2}\left(\Gamma_{\omega}\right)\rightarrow L^{\infty}\left(\Gamma_{\omega}\right).\label{eq:heat-semigroup}
\end{equation}

The first condition holds since the volume of balls in $\Gamma_{\omega}$
grows piecewise linearly with $R$. The second condition follows from
\cite[sec. 4.2]{Bifulco2023}.
\end{proof}
\begin{lem}
\label{lem:Subexponential}Suppose that $\phi_{\omega}$ is a subexponentially
bounded solution of $-\frac{d^{2}\phi_{\omega}}{dx^{2}}=E\phi_{\omega}$
on $\Gamma_{\omega}$. That is, for all $\kappa>0$ there is $C>0$
such that 
\begin{equation}
\left|\phi_{\omega}\left(x\right)\right|\leq Ce^{\kappa\left|x\right|},\forall x\in\Gamma_{\omega}.\label{eq:subexp}
\end{equation}
Then $E\in\spec{H_{\omega}}$.
\end{lem}

\begin{proof}
We shall construct a Weyl sequence for $E$, similar to \cite[lem. 4.2]{Damanik2020}.
$H_{\omega}$ corresponds to the sesquilinear form (see \cite[thm. 1.4.11]{BerKuc_graphs})
\begin{equation}
h\left[u,v\right]=\left\langle u',v'\right\rangle _{L^{2}\left(\Gamma_{\omega}\right)},\text{dom}\left(h\right)=H^{1}\left(\Gamma_{\omega}\right).\label{eq:sesq1}
\end{equation}
Denote by $B_{R}^{\omega}\left(o\right)$ the geodesic ball of radius
$R$ around the origin $o\left(\Gamma_{\omega}\right)$. Let $\left(\rho_{n}\right)_{n\in\N}$
be a sequence of smooth functions, with $\rho_{n}:\Gamma_{\omega}\rightarrow\left[0,1\right]$
satisfying
\begin{equation}
\rho_{n}\left(x\right)=\begin{cases}
1, & x\in B_{n}^{\omega}\left(o\right),\\
0, & x\notin B_{n+1}^{\omega}\left(o\right),
\end{cases}\label{eq:mollifier}
\end{equation}
and moreover $\sup_{n\geq1}\left\Vert \rho_{n}'\right\Vert _{L^{\infty}}=M<\infty$.
Define:
\begin{equation}
\psi_{n}:=\frac{\phi_{\omega}\rho_{n}}{\left\Vert \phi_{\omega}\chi_{B_{n}^{\omega}\left(o\right)}\right\Vert _{L^{2}}}.\label{eq:Weyl-seq}
\end{equation}
Note that since $\rho_{n}$ is compactly supported and $\phi_{\omega}\in H_{\mathrm{loc}}^{1}(\Gamma_{\omega})$,
we have $\rho_{n}\phi_{\omega}\in H^{1}(\Gamma_{\omega})=\text{dom}(h)$
and namely $\psi_{n}\in\text{dom}(h)$. Furthermore, $\left\Vert \psi_{n}\right\Vert _{L^{2}}\geq1$
for all $n\in\N$, and we claim that $\left\{ \psi_{n}\right\} _{n\in\N}$
is a Weyl sequence, i.e. that
\begin{equation}
\sup_{\left\Vert u\right\Vert _{H^{1}\left(\Gamma_{\omega}\right)}\leq1}\left|h\left[u,\psi_{n}\right]-E\left\langle u,\psi_{n}\right\rangle _{L^{2}}\right|\underset{n\rightarrow\infty}{\rightarrow}0.\label{eq:Weyl-seq-1}
\end{equation}

Fix $u\in\text{dom}\left(h\right)$ with $\left\Vert u\right\Vert _{H^{1}}\leq1$.
Denoting $A_{n}^{\omega}:=B_{n+1}^{\omega}\left(o\right)\backslash B_{n}^{\omega}\left(o\right)$
a direct computation gives
\begin{align}
 & h\left[u,\rho_{n}\phi_{\omega}\right]-E\left\langle u,\rho_{n}\phi_{\omega}\right\rangle _{L^{2}}\nonumber \\
= & \left\langle u',\rho_{n}\phi_{\omega}'\right\rangle _{L^{2}\left(A_{n}^{\omega}\right)}-E\left\langle u,\rho_{n}\phi_{\omega}\right\rangle _{L^{2}\left(A_{n}^{\omega}\right)}+\left\langle u',\rho_{n}'\phi_{\omega}\right\rangle _{L^{2}\left(A_{n}^{\omega}\right)}.\label{eq:est-1}
\end{align}
Thus,
\begin{align}
 & \left|h\left[u,\rho_{n}\phi_{\omega}\right]-E\left\langle u,\rho_{n}\phi_{\omega}\right\rangle _{L^{2}}\right|\nonumber \\
\leq & _{\left(\ref{eq:est-1}\right)}\left|\left\langle u',\rho_{n}\phi_{\omega}'\right\rangle _{L^{2}\left(A_{n}^{\omega}\right)}\right|+E\left|\left\langle u,\rho_{n}\phi_{\omega}\right\rangle _{L^{2}\left(A_{n}^{\omega}\right)}\right|+\left|\left\langle u',\rho_{n}'\phi_{\omega}\right\rangle _{L^{2}\left(A_{n}^{\omega}\right)}\right|\nonumber \\
\leq & \left|\left\langle u',\phi_{\omega}'\right\rangle _{L^{2}\left(A_{n}^{\omega}\right)}\right|+E\left|\left\langle u,\phi_{\omega}\right\rangle _{L^{2}\left(A_{n}^{\omega}\right)}\right|+M\left|\left\langle u',\phi_{\omega}\right\rangle _{L^{2}\left(A_{n}^{\omega}\right)}\right|,\label{eq:est-3}
\end{align}
where in the third line we used the assumption that $\sup_{n\geq1}\left\Vert \rho_{n}'\right\Vert _{L^{\infty}}=M,\,\left|\rho_{n}\right|\leq1$.
Note that
\begin{equation}
\left|\left\langle u',\phi_{\omega}\right\rangle _{L^{2}\left(A_{n}^{\omega}\right)}\right|\leq\left\Vert u'\right\Vert _{L^{2}\left(A_{n}^{\omega}\right)}\left\Vert \phi_{\omega}\right\Vert _{L^{2}\left(A_{n}^{\omega}\right)}\leq\left\Vert u\right\Vert _{H^{1}\left(A_{n}^{\omega}\right)}\left\Vert \phi_{\omega}\right\Vert _{L^{2}\left(A_{n}^{\omega}\right)},\label{eq:est-3-1}
\end{equation}
and furthermore
\begin{equation}
\left|\left\langle u',\phi_{\omega}'\right\rangle _{L^{2}\left(A_{n}^{\omega}\right)}\right|+E\left|\left\langle u,\phi_{\omega}\right\rangle _{L^{2}\left(A_{n}^{\omega}\right)}\right|\apprle\left|\left\langle u,\phi_{\omega}\right\rangle _{H^{1}\left(A_{n}^{\omega}\right)}\right|\leq\left\Vert u\right\Vert _{H^{1}\left(A_{n}^{\omega}\right)}\left\Vert \phi_{\omega}\right\Vert _{H^{1}\left(A_{n}^{\omega}\right)},\label{eq:est-3-2}
\end{equation}
where $\apprle$ indicates inequality up to a multiplicative constant.
Combining those we overall get
\begin{equation}
\left|h\left[u,\rho_{n}\phi_{\omega}\right]-E\left\langle u,\rho_{n}\phi_{\omega}\right\rangle _{L^{2}}\right|\apprle\left\Vert u\right\Vert _{H^{1}\left(A_{n}^{\omega}\right)}\left\Vert \phi_{\omega}\right\Vert _{H^{1}\left(A_{n}^{\omega}\right)}.\label{eq:est-3-3}
\end{equation}
Moreover, since $-\frac{d^{2}\phi_{\omega}}{dx^{2}}=E\phi_{\omega}$,
we have $\left\Vert \phi_{\omega}\right\Vert _{H^{1}\left(A_{n}^{\omega}\right)}\apprle\left\Vert \phi_{\omega}\right\Vert _{L^{2}\left(A_{n}^{\omega}\right)}$.
Overall:
\begin{align}
 & \frac{\left|h\left[u,\rho_{n}\phi_{\omega}\right]-E\left\langle u,\rho_{n}\phi_{\omega}\right\rangle _{L^{2}}\right|}{\left\Vert \phi_{\omega}\chi_{B_{n}^{\omega}}\right\Vert _{L^{2}\left(\Gamma_{\omega}\right)}}\apprle\frac{\left\Vert u\right\Vert _{H^{1}\left(A_{n}^{\omega}\right)}\left\Vert \phi_{\omega}\right\Vert _{H^{1}\left(A_{n}^{\omega}\right)}}{\left\Vert \phi_{\omega}\chi_{B_{n}^{\omega}}\right\Vert _{L^{2}\left(\Gamma_{\omega}\right)}}\nonumber \\
\apprle & \frac{\left\Vert u\right\Vert _{H^{1}\left(A_{n}^{\omega}\right)}\left\Vert \phi_{\omega}\right\Vert _{L^{2}\left(A_{n}^{\omega}\right)}}{\left\Vert \phi_{\omega}\chi_{B_{n}^{\omega}}\right\Vert _{L^{2}\left(\Gamma_{\omega}\right)}}\apprle\frac{\left\Vert \phi_{\omega}\chi_{A_{n}^{\omega}}\right\Vert _{L^{2}}}{\left\Vert \phi_{\omega}\chi_{B_{n}^{\omega}}\right\Vert _{L^{2}\left(\Gamma_{\omega}\right)}}.\label{eq:est-4}
\end{align}
The fact that $\phi_{\omega}$ is subexponentially bounded in particular
implies that
\begin{equation}
\lim_{n\rightarrow\infty}\frac{\left\Vert \phi_{\omega}\chi_{A_{n}^{\omega}}\right\Vert _{L^{2}\left(\Gamma_{\omega}\right)}}{\left\Vert \phi_{\omega}\chi_{B_{n}^{\omega}}\right\Vert _{L^{2}\left(\Gamma_{\omega}\right)}}=0,\label{eq:est-5-lim}
\end{equation}
showing that $\psi_{n}$ is thus indeed a Weyl sequence.
\end{proof}
\begin{proof}[Proof of Theorem \ref{thm:Kotani-eq}]
 The fact that $\spec{H_{\omega}}$ is independent of $\omega$ follows
from standard arguments for minimal dynamical systems, see for instance
\cite[cor. 4.5]{Beckus2017}, \cite[prop. 2.1]{Damanik2014}, \cite[thm. 5]{Gruber2007}
and references therein for similar proofs. We wish to prove the equality
(\ref{eq:sigma=00003Dz}). Recall that $\mathbb{R}\backslash\B=\UH\sqcup\NUH\sqcup\mathcal{Z}$.
Boshernitzan's condition implies that in fact the Lyapunov exponent
$L_{\Omega}\left(E\right)$ converges uniformly in $\omega\in\Omega$,
see \cite[thm. 1]{Damanik2006}. This implies that $\NUH=\varnothing$
(see \cite[app. A]{Damanik2025}), and in particular $\mathbb{R}\backslash\B=\UH\sqcup\mathcal{Z}$.
We thus first show that
\begin{equation}
\spec{H_{\Omega}}\backslash\B\subset\left(\mathbb{R}\backslash\B\right)\backslash\left(\UH\right)=\mathcal{Z}.\label{eq:spec-nuh}
\end{equation}
Indeed, fix arbitrary $\omega\in\Omega$ and assume that $E_{0}\in\UH$.
Since $\UH$ is open, we can choose some open interval $I_{E_{0}}\subset\UH$
around $E_{0}$. We claim that $I_{E_{0}}\cap\spec{H_{\omega}}=\varnothing$.
Assume otherwise. Then, by Lemma \ref{fact:fact2}, there exists some
$E\in I_{E_{0}}\cap\spec{H_{\omega}}$ and a nontrivial solution $\phi_{\omega}$
of the ODE (\ref{eq:ODE-transfer}) which is subexponentially bounded
on both half-graphs $\Gamma_{\omega}^{\pm}$. Equivalently, the corresponding
vector $\Phi_{n}$ in (\ref{eq:initial-ode}) is subexponentially
bounded in both directions. This contradicts Lemma \ref{fact:fact1},
since $E\in I_{E_{0}}\subseteq\UH$. Therefore $I_{E_{0}}\cap\spec{H_{\omega}}=\varnothing$
and in particular $E_{0}\notin\spec{H_{\omega}}$. Since $E_{0}\in\UH$
was arbitrary, this gives that $\spec{H_{\omega}}\backslash\B=\spec{H_{\Omega}}\backslash\B\subset\mathcal{Z}$.

For the reverse inclusion, let $E\in\mathcal{Z}$. Choose nonzero
$\Phi_{0}$, and let $\phi_{\omega}$ be the corresponding solution
to (\ref{eq:ODE-transfer}). Since $L_{\Omega}(E)=0$ and the convergence
of the Lyapunov exponent is uniform in $\omega$, the vector $\Phi_{n}$
is subexponentially bounded as $|n|\to\infty$. This implies that
$\phi_{\omega}$ itself is subexponentially bounded on $\Gamma_{\omega}$.
By Lemma \ref{lem:Subexponential}, $E\in\spec{H_{\omega}}$. Hence
$\mathcal{Z}\subset\spec{H_{\Omega}}\backslash\B$.
\end{proof}

\subsection{Kotani theory and the oracle theorem}

In this subsection, we prove the oracle theorem, which shows that
on a set of energies of positive Lebesgue measure, the two half-line
restrictions of almost all $\omega\in\Omega$ determine one another.
This is the mechanism that will allow us to relate the Lebesgue measure
of the spectrum to a structural constraint on the subshift.

We begin with the following fundamental result by Kotani \cite{Kotani1987}:
\begin{thm}
\label{thm:-Kotani} Let $\left(\Omega,\sft,\mu\right)$ be an ergodic,
measure preserving dynamical system. For $\mu$-almost every $\omega\in\Omega$,
the functions $m_{\pm}^{\omega}$ are reflectionless on $\mathcal{Z}$.
That is, for almost all $\omega\in\Omega$ and Lebesgue almost every
$E\in\mathcal{Z}$,
\begin{equation}
m_{\omega}^{-}\left(E+i0\right)=-\overline{m_{\omega}^{+}\left(E+i0\right)},\label{eq:reflectionless}
\end{equation}
where $m_{\omega}^{\pm}\left(E+i0\right)=\lim_{\epsilon\rightarrow0^{+}}m_{\omega}^{\pm}\left(E+i\epsilon\right)$.
In particular, the two half-line $m$-functions determine each other
on $\mathcal{Z}$.
\end{thm}

\begin{rem}
The term \textit{reflectionless} comes from interpreting $m_{\omega}^{\pm}\left(E+i0\right)$
as the boundary Dirichlet-to-Neumann data of the two half-graphs $\Gamma_{\omega}^{\pm}$.
The relation $m_{\omega}^{-}\left(E+i0\right)=-\overline{m_{\omega}^{+}\left(E+i0\right)}$
means that the right and left half-line solutions can be matched at
the origin with the Kirchhoff condition, so that the interface between
the two half graphs produces no reflected component.
\end{rem}

In other words, Kotani's theorem states that the set
\begin{equation}
\D\left(\mathcal{Z}\right):=\left\{ \omega\in\Omega:\text{\ensuremath{m_{\omega}^{\pm}\left(z\right)} are reflectionless on }\mathcal{Z}\right\} ,\label{eq:Dz}
\end{equation}
is of full measure. We note that while Kotani's original result was
established for ergodic potentials on $\R$, it can be directly extended
to tiling graphs by following essentially the same steps, once the
appropriate relations between the $m$-function and Green's function
are established (see \cite[sec. 2]{Kotani1987}).
\begin{defn}
Denote by $\mathcal{R}_{\pm}$ the restriction maps of a sequence
$\omega$ from $\mathbb{Z}$ to $\pm\mathbb{N}$:
\begin{align}
 & \mathcal{R}_{\pm}:\A^{\Z}\rightarrow\A^{\pm\N},\label{eq:R1}\\
 & \mathcal{R}_{\pm}\left(\omega\left(n\right)\right)_{n\in\Z}=\left(\omega\left(n\right)\right)_{n\in\pm\N}.\label{eq:R2}
\end{align}
\end{defn}

\begin{thm}[Oracle theorem]
\label{thm:oracle} Suppose that $\mathcal{Z}$ has positive Lebesgue
measure. Then for the set $\D\left(\mathcal{Z}\right)$ in (\ref{eq:Dz}),
the map
\begin{align}
 & \mathcal{R}:\mathcal{R}_{-}\left(\D\left(\mathcal{Z}\right)\right)\rightarrow\mathcal{R}_{+}\left(\D\left(\mathcal{Z}\right)\right),\label{eq:R3}\\
 & \mathcal{R}_{-}\left(\omega\right)\mapsto\mathcal{R}_{+}\left(\omega\right),\omega\in\D\left(\mathcal{Z}\right),\label{eq:R4}
\end{align}
 is well-defined, bijective, and uniformly continuous. Consequently,
$\mathcal{R}^{-1}$ is uniformly continuous as well.
\end{thm}

In a nutshell, the oracle theorem states that if $\left|\mathcal{Z}\right|>0$,
then the (full-measure) set $\D\left(\mathcal{Z}\right)$ is ``deterministic'',
in the sense that one can predict one half of the sequence $\omega\in\D\left(\mathcal{Z}\right)$
by knowing the other (hence the name 'oracle'). The result is similar
to the one established in \cite[thm. 3.7]{Damanik2020}.
\begin{proof}
By the Borg--Marchenko Theorem \ref{thm:BM}, the restrictions $\mathcal{R}_{\pm}\left(\omega\right)$
determine $m_{\omega}^{\pm}$ uniquely, and this correspondence is
continuous by Corollary \ref{cor:loc-uni}. Moreover, since $m_{\omega}^{-}\left(E+i0\right)=-\overline{m_{\omega}^{+}\left(E+i0\right)}$
on the positive measure set $\mathcal{Z}$, and Herglotz functions
are uniquely determined by their boundary values on a set of positive
measure, we obtain a one-to-one correspondence between the functions
$m_{\omega}^{\pm}\left(z\right)$ for $\omega\in\D\left(\mathcal{Z}\right)$.
We thus see that $\mathcal{R}$ is well-defined, bijective, and continuous
by the chain of continuous bijections
\begin{equation}
\mathcal{R}_{+}\left(\omega\right)\iff m_{\omega}^{+}\left(z\right)\iff m_{\omega}^{-}\left(z\right)\iff\mathcal{R}_{-}\left(\omega\right).\label{eq:M-R-Correspondence}
\end{equation}
Due to \cite[lem. 5]{Kotani1985} and Corollary \ref{cor:loc-uni},
we also know that $\D\left(\mathcal{Z}\right)$ is closed, and thus
compact. Since $\mathcal{R}_{\pm}$ are homeomorphisms, $\mathcal{R}_{\pm}\left(\D\left(\mathcal{Z}\right)\right)$
are compact as well, showing that $\mathcal{R}$ is uniformly continuous
(and similarly $\mathcal{R}^{-1}$).
\end{proof}
The following lemma will allow us to apply the oracle theorem simultaneously
for entire orbits in $\Omega$:
\begin{lem}
\label{lem:shift-invariant}The set $\D\left(\mathcal{Z}\right)$
is shift-invariant.
\end{lem}

\begin{proof}
Let $z\in\C_{+}$. For $\omega\in\Omega$ fixed, we denote the elements
of the transfer matrix $\mw{}$ by $\mathcal{M}_{ij}$. Then, combining
the definition (\ref{eq:m3}) for $m_{\omega}^{\pm}\left(z\right)$
with definition (\ref{eq:trans-mat}) for the transfer matrix (and
carefully keeping track of the orientation of the derivative with
respect to the origin), direct computation shows that,
\begin{align}
 & m_{\sft\omega}^{-}\left(z\right)=\frac{\sum_{e\sim o\left(\Gamma_{\sft\omega}^{-}\right)}\frac{df_{\sft\omega}}{dx}|_{e}\left(o\left(\Gamma_{\sft\omega}^{-}\right);z\right)}{f_{\sft\omega}\left(o\left(\Gamma_{\sft\omega}^{-}\right);z\right)}=-\frac{\sum_{e\sim o\left(\Gamma_{\sft\omega}^{-}\right)}\frac{df_{\sft\omega}}{dx}|_{e}\left(o\left(\Gamma_{\sft\omega}^{+}\right);z\right)}{f_{\sft\omega}\left(o\left(\Gamma_{\sft\omega}^{+}\right);z\right)}\nonumber \\
= & -\frac{\mathcal{M}_{21}f_{\omega}\left(o\left(\Gamma_{\sft\omega}^{+}\right);z\right)+\mathcal{M}_{22}\sum_{e\sim o\left(\Gamma_{\omega}^{+}\right)}\frac{df_{\omega}}{dx}|_{e}\left(o\left(\Gamma_{\omega}^{+}\right);z\right)}{\mathcal{M}_{11}f_{\omega}\left(o\left(\Gamma_{\sft\omega}^{+}\right);z\right)+\mathcal{M}_{12}\sum_{e\sim o\left(\Gamma_{\omega}^{+}\right)}\frac{df_{\omega}}{dx}|_{e}\left(o\left(\Gamma_{\omega}^{+}\right);z\right)}\nonumber \\
= & -\frac{\mathcal{M}_{21}f_{\omega}\left(o\left(\Gamma_{\sft\omega}^{-}\right);z\right)-\mathcal{M}_{22}\sum_{e\sim o\left(\Gamma_{\omega}^{-}\right)}\frac{df_{\omega}}{dx}|_{e}\left(o\left(\Gamma_{\omega}^{-}\right);z\right)}{\mathcal{M}_{11}f_{\omega}\left(o\left(\Gamma_{\sft\omega}^{-}\right);z\right)-\mathcal{M}_{12}\sum_{e\sim o\left(\Gamma_{\omega}^{-}\right)}\frac{df_{\omega}}{dx}|_{e}\left(o\left(\Gamma_{\omega}^{-}\right);z\right)}\nonumber \\
= & -\frac{\mathcal{M}_{21}-\mathcal{M}_{22}m_{\omega}^{-}\left(z\right)}{\mathcal{M}_{11}-\mathcal{M}_{12}m_{\omega}^{-}\left(z\right)},\label{eq:mTw-1}
\end{align}
where we note that $f\left(o\left(\Gamma_{\omega}^{-}\right);z\right)\neq0$,
otherwise $m_{\omega}^{-}$ would have a pole at $z$ (which cannot
happen for $z\in\C_{+}$). Similarly (and noting that the value of
$f_{\omega}\left(v;z\right)$ does not depend on the orientation with
respect to $v$), one can show that
\begin{align}
m_{\sft\omega}^{+}\left(z\right)= & \frac{\mathcal{M}_{21}+\mathcal{M}_{22}m_{\omega}^{+}\left(z\right)}{\mathcal{M}_{11}+\mathcal{M}_{12}m_{\omega}^{+}\left(z\right)}.\label{eq:mTw}
\end{align}

Now, let $\omega\in\D\left(\mathcal{Z}\right)$, i.e., there exists
$\tilde{\mathcal{Z}}\subset\mathcal{Z}$ of full measure such that
$m_{\omega}^{-}\left(E+i0\right)=-\overline{m_{\omega}^{+}\left(E+i0\right)}$
for all $E\in\tilde{\mathcal{Z}}$. Taking $z=E+i\epsilon$ with $\epsilon\rightarrow0$
in (\ref{eq:mTw}) above, one obtains that for all $E\in\tilde{\mathcal{Z}}$,
\begin{align}
 & -\overline{m_{\sft\omega}^{+}\left(E+i0\right)}\nonumber \\
= & -\frac{\mathcal{M}_{21}\left(\omega,E+i0\right)+\mathcal{M}_{22}\left(\omega,E+i0\right)\overline{m_{\omega}^{+}\left(E+i0\right)}}{\mathcal{M}_{11}\left(\omega,E+i0\right)+\mathcal{M}_{12}\left(\omega,E+i0\right)\overline{m_{\omega}^{+}\left(E+i0\right)}}\nonumber \\
= & -\frac{\mathcal{M}_{21}\left(\omega,E+i0\right)-\mathcal{M}_{22}\left(\omega,E+i0\right)m_{\omega}^{-}\left(E+i0\right)}{\mathcal{M}_{11}\left(\omega,E+i0\right)-\mathcal{M}_{12}\left(\omega,E+i0\right)m_{\omega}^{-}\left(E+i0\right)}\nonumber \\
= & m_{\sft\omega}^{-}\left(E+i0\right),\label{eq:mp-mm-equality}
\end{align}
where we have used the fact that $\mathcal{M}\left(\omega,E+i0\right)$
is a real-valued matrix, as the ODE (\ref{eq:ODE-transfer}) has real-valued
solutions at the limit $\epsilon\rightarrow0$. This completes the
proof.
\end{proof}

\subsection{Proof of Theorem \ref{thm:TMP}}

We are now ready to prove Theorem \ref{thm:TMP} (we in fact prove
Theorem \ref{thm:Weaker-generalize-Cantor}, which immediately implies
Theorem \ref{thm:TMP}). The argument combines the dynamical characterization
of the spectrum with the oracle theorem in order to rule out aperiodicity
when the spectrum has positive measure.
\begin{thm}
\label{thm:Weaker-generalize-Cantor}$\spec{H_{\Omega}}$ has Lebesgue
measure zero. Moreover, any isolated point of $\spec{H_{\Omega}}$
(if exists) belongs to the discrete set $\B\cap\spec{H_{\Omega}}$
and is an eigenvalue. In particular, up to a possible discrete set
of eigenvalues, $\spec{H_{\Omega}}$ is a generalized Cantor set.
\end{thm}

\begin{rem}
\label{rem:thm1.7.4.5}This is precisely Theorem \ref{thm:TMP}: the
exceptional set $F$ of isolated points appearing there is exactly
the set of isolated points in $\B\cap\spec{H_{\Omega}}$.
\end{rem}

\begin{proof}
By Theorem \ref{thm:Kotani-eq}, $\spec{H_{\omega}}\backslash\B=\spec{H_{\Omega}}\backslash\B=\mathcal{Z}$
for all $\omega\in\Omega$. Since $\B$ is discrete, it suffices to
show that $\left|\mathcal{Z}\right|=0$ to deduce that $\left|\spec{H_{\Omega}}\right|=0$.
Suppose, for contradiction, that $\left|\mathcal{Z}\right|>0$. By
the oracle theorem, this implies that the map
\begin{equation}
\mathcal{R}:\left\{ \omega\left(n\right)\right\} _{n<0}\mapsto\left\{ \omega\left(n\right)\right\} _{n\geq0},\omega\in\D\left(\mathcal{Z}\right),\label{eq:shift-invariant}
\end{equation}
is well-defined and uniformly continuous. Moreover, by Kotani's theorem
\ref{thm:-Kotani}, the set $\D\left(\mathcal{Z}\right)\subset\Omega$
has full measure.

By uniform continuity of $\mathcal{R}$ and $\mathcal{R}^{-1}$, there
exists $k=k\left(\mathcal{A}\right)\in\mathbb{N}$ (independent of
$\omega$) such that for all $\omega\in\D\left(\mathcal{Z}\right)$,
the block $\omega|_{\left[-k,0\right]}$ uniquely determines both
$\omega\left(1\right)$ and $\omega\left(-k-1\right)$. In particular,
if we fix $\omega\in\D\left(\mathcal{Z}\right)$, then using the shift
invariance of the set $\D\left(\mathcal{Z}\right)$ (Lemma \ref{lem:shift-invariant}),
one can further repeat this argument for all shifts $\left\{ \sft^{n}\omega\right\} _{n\in\Z}$,
deducing that the block $\omega|_{\left[-k,0\right]}$ determines
the entire sequence $\omega$. In particular, each shift $\sft^{n}\omega$
is determined by its restriction to a block of length $k+1$. Thus,
\begin{equation}
\#\left\{ \sft^{n}\omega:n\in\Z\right\} \leq\left|\mathcal{A}\right|^{k+1}<\infty.\label{eq:shift-finite}
\end{equation}
Hence $\omega$ has a finite orbit under the shift, and is therefore
periodic, contradicting the aperiodicity of $\Omega$. This gives
that $\left|\mathcal{Z}\right|=\left|\spec{H_{\Omega}}\backslash\B\right|=\left|\spec{H_{\Omega}}\right|=0$.

In particular, $\spec{H_{\Omega}}$ contains no intervals. Since it
is also closed, then it is nowhere-dense in $\R$. Moreover, standard
results \cite[thm. 4.2.4]{Damanik2022} imply that the discrete spectrum
$\text{Spec}_{disc}\left(H_{\Omega}\right)\backslash\B$ is empty.
Hence, $\spec{H_{\Omega}}$ is closed, nowhere-dense, and contains
no isolated points, up to possible isolated points in $\B\cap\spec{H_{\Omega}}$
(which are automatically eigenvalues as isolated spectral points).
This shows that, apart from the possible discrete set $\B\cap\spec{H_{\Omega}}$,
the spectrum is closed, nowhere-dense, and has no isolated points,
and is therefore a generalized Cantor set, completing the proof.
\end{proof}
\begin{rem}
We emphasize that energies in $\B\cap\spec{H_{\Omega}}$ may correspond
either to isolated eigenvalues or to embedded eigenvalues. In the
latter case, these energies are accumulation points of the spectrum
and therefore do not affect its Cantor nature. The question of whether
attaching graph decorations or introducing internal vertex structures
necessarily leads to the opening of spectral gaps (and hence to isolated
eigenvalues) has been studied in various settings, see for instance
\cite{Kuchment2005a,Avron1994,Do2017}. These works identify conditions
under which such local structures produce genuine spectral gaps (and
not embedded eigenvalues).
\end{rem}

\subsection{Proof of Theorem \ref{thm:Cantor}\label{subsec:Generic-Cantor-spectrum}}

By Theorem \ref{thm:TMP}, apart from the set $\spec{H_{\Omega}}\cap\B$,
the spectrum is a generalized Cantor set. In particular, the spectrum
is a Cantor set if $\spec{H_{\Omega}}\cap\B=\varnothing$. We now
restrict to the class of decorated $\Z$-graphs, and show that $\spec{H_{\Omega}}\cap\B=\varnothing$
holds Baire-generically for a suitable choice of decoration edge lengths,
which will immediately prove Theorem \ref{thm:Cantor}. We conjecture
that a similar genericity result extends to arbitrary tiling graphs,
leading to generic Cantor spectrum. However, such a generalization
is not straightforward, and goes beyond the scope of this work.

For the remainder of this subsection, we fix a decorated $\Z$-graph
$\Gamma_{\omega}$ with no loops. As previously noted in Subsection
\ref{subsec:Transfer-matrices}, the discrete set $\spec{H_{\omega}}\cap\B$
consists of energies $E\in\spec{H_{\omega}}$ admitting a nontrivial
solution to (\ref{eq:ODE-transfer}) with initial condition $\Phi_{\omega}=\vec{0}$.
For decorated $\Z$-graphs, the vector $\Phi_{\omega}=\vec{0}$ is
evaluated at the center of an edge. Hence both the value and derivative
of the solution vanish at this point, and uniqueness for the ODE on
an interval implies that the solution vanishes identically on that
horizontal edge. At a spectral energy $E\in\spec{H_{\omega}}\cap\B$,
this gives rise to a compactly supported eigenfunction. In other words,
$E$ corresponds to a \textit{flat band}. 
\begin{lem}
\label{lem:compact}Let $\left(E,f\right)$ be an eigenpair of $H_{\omega}$
with $f$ compactly supported. Then there exists a decoration $\gra$
such that $\left(E,f|_{\gra}\right)$ is an eigenpair of the Laplacian
$H|_{\gra}$, with the Dirichlet condition $f\left(v_{a}\right)=0$
imposed at the base vertex $v_{a}$ and the Kirchhoff condition everywhere
else.
\end{lem}

\begin{proof}
Let $\left(E,f\right)$ be as above. Since the line graph itself does
not admit compactly supported eigenfunctions, $f$ must be nonzero
at some decoration. Choose the rightmost decoration $\Gamma_{\omega\left(n\right)}$
on which $f$ is nonzero, and write $\omega\left(n\right)=a$. Since
$f\equiv0$ to the right of $\Gamma_{\omega\left(n\right)}$, continuity
then implies that $f\left(u\right)=0$ at the base vertex $v_{a}$
of the decoration $\gra$. Now, by construction, $f|_{\gra}$ is a
solution to the eigenvalue equation $-f''=Ef$, satisfies the Dirichlet
condition at the base of $\gra$, and the Kirchhoff condition everywhere
else, as required.
\end{proof}
We now prove
\begin{thm}
\label{thm:Generic-assumption}There exists a Baire-generic set $A$
of decoration edge lengths for which $\B\cap\spec{H_{\Omega}}=\varnothing$
for all $\vec{\ell}\in A$. Equivalently, $H_{\Omega}$ admits no
flat bands.
\end{thm}

\begin{rem}
The generic exclusion of flat bands was already proven for various
settings, namely for periodic discrete graphs in \cite{Faust2025a,Faust2025}.
\begin{proof}
It suffices to construct a Baire-generic set of decoration edge lengths
such that the corresponding decorated $\Z$-graph does not admit compactly
supported eigenfunctions. Since $\B\cap\spec{H_{\Omega}}$ is discrete,
there are at most countably many energies admitting compactly supported
eigenfunctions. For each such eigenpair, we show that one can perturb
the edge lengths within a Baire-generic set in order to ensure that
there is no eigenvalue with compactly supported eigenfunctions in
a neighborhood of $E$. Taking the intersection of all such Baire-generic
subsets over all possible eigenpairs results in a Baire-generic subset
of edge lengths as desired.

Let $\left(E,f\right)$ be a compactly supported eigenpair. By Lemma
\ref{lem:compact}, there is some decoration $\gra$ such that $\left(E,f|_{\gra}\right)$
is an eigenpair for the Laplacian $H|_{\gra}$ defined above, and
in fact $f$ must vanish at the base of all decorations of the same
type. We may moreover assume that $\left(E,f|_{\gri{a'}}\right)$
is not an eigenfunction of $H|_{\gri{a'}}$ for any decoration $\gri{a'}$
of a different type, since there exists a Baire-generic set of edge
lengths for which the spectra of $H|_{\gra}$ on all decorations of
different types are disjoint. Indeed, this follows by independently
varying the edge lengths of the different decorations. Since $\gra$
is assumed to contain no loops, \cite[thm. 3.6]{Berkolaiko2017} (see
also \cite[thm. 3.2]{Alon2024a}) implies that there exists a Baire-generic
set of edge lengths for $\gra$ such that $f|_{e}\not\equiv0$ on
every edge in $\gra$. This means that the support of $f$ cannot
be contained inside a single decoration, otherwise the Kirchhoff condition
at the base vertex $u\in\gra$ would imply that $f\equiv0$ on at
least one edge of $\gra$. Thus, for $f$ to have compact support,
it must be supported on at least two copies of $\gra$, where it vanishes
at the base vertex. Choose two adjacent copies $\gra^{1},\gra^{2}$
of $\gra$ on the line such that $f|_{\gra^{i}}\not\equiv0$ for $i\in\left\{ 1,2\right\} $
(i.e. such that $f|_{\gra}\equiv0$ along every copy of $\gra$ between
$\gra^{1},\gra^{2}$). Let $\left[\gra^{1},\gra^{2}\right]$ denote
the subgraph of $\Gamma_{\omega}$ which lies between $\gra^{1},\gra^{2}$,
excluding these two decorations. Since $\gra^{1},\gra^{2}$ are chosen
to be adjacent, we may assume that no additional decorations of type
$\gra$ exist inside $\left[\gra^{1},\gra^{2}\right]$, as $f$ is
not supported on such decorations, and so their presence would not
affect the argument. As before, consider the Laplacian $H|_{\left[\gra^{1},\gra^{2}\right]}$
on the subgraph $\left[\gra^{1},\gra^{2}\right]$, with the Dirichlet
condition imposed at the two boundary points (which are the bases
of $\gra^{1},\gra^{2}$), and the Kirchhoff condition everywhere else.
It follows that $\left(E,f|_{\left[\gra^{1},\gra^{2}\right]}\right)$
is an eigenpair of $H|_{\left[\gra^{1},\gra^{2}\right]}$. Note that
in our setting, the condition that $\left(E,f|_{\gra}\right)$ and
$\left(E,f|_{\left[\gra^{1},\gra^{2}\right]}\right)$ are eigenpairs
is necessary for $E$ to correspond to a compactly supported eigenfunction
on $\Gamma_{\omega}$.

Now, choose an edge $e$ in $\gra$ adjacent to the base $u$ of $\gra$.
Then $f|_{e}\not\equiv0$ by the choice of edge length perturbation
made above, and since $f\left(u\right)=0$, it follows that $f'\left(u\right)\neq0$.
By \cite[prop. 4.1]{Berkolaiko2012}, the eigenvalues of $H|_{\gra}$
are strictly decreasing functions of the edge length $\ell_{e}$.
Therefore, we may fix all edge lengths within $\left[\gra^{1},\gra^{2}\right]$,
and perturb $\ell_{e}$ within a Baire-generic set so that $E\notin\spec{H|_{\gra}}$.
Since $e$ belongs to $\Gamma_{a}$ and not to $\left[\gra^{1},\gra^{2}\right]$,
this perturbation keeps $E$ in $\spec{H|_{\left[\gra^{1},\gra^{2}\right]}}$.
As $\spec{H|_{\gra}}$ and $\spec{H|_{\left[\gra^{1},\gra^{2}\right]}}$
are both discrete, there exists a neighborhood of $E$ in which these
two spectra are disjoint. By the necessary condition above, $\Gamma_{\omega}$
cannot admit any compactly supported eigenfunctions with energies
in this neighborhood, completing the proof.
\end{proof}
\end{rem}

The proof of Theorem \ref{thm:Cantor} now follows immediately:
\begin{proof}[Proof of Theorem \ref{thm:Cantor}]
For decorated $\Z$-graphs, the argument of Theorem \ref{thm:Weaker-generalize-Cantor}
applies in the same way, using the adapted Borg--Marchenko result
from Proposition \ref{prop:BM-decorated}, under the generic condition
on decoration edge lengths assumed there (see also Corollary \ref{cor:generic-BM-decorated}).
In this case, we immediately get that $\spec{H_{\Omega}}$ is of zero
Lebesgue measure. Furthermore, by Theorem \ref{thm:Generic-assumption},
we have that $\B\cap\spec{H_{\Omega}}=\varnothing$ under a Baire-generic
choice of the edge lengths (and since the intersection of generic
sets is generic, we can also ensure that the conditions of Proposition
\ref{prop:BM-decorated} still hold). For such generic choices of
the edge lengths, Theorem \ref{thm:Weaker-generalize-Cantor} then
also implies that $\spec{H_{\Omega}}$ is a generalized Cantor set.
\end{proof}
\newpage{}

\section{Gap labelling -- preliminaries \label{sec:GLT}}

In this section, we introduce the objects needed to formulate and
prove the gap labelling theorems. The first is the integrated density
of states, whose values inside spectral gaps are the gap labels. The
second is the Schwartzman group of the underlying dynamical system,
which will provide the dynamical constraint on these labels.

\subsection{Integrated density of states}

We first consider the metric setting. Let $\left(\Omega,\sft\right)$
be a uniquely ergodic subshift, equipped with its (unique) ergodic
measure $\mu$, and let $\Gamma_{\Omega}$ be an associated family
of metric decorated $\Z$-graphs. For $\omega\in\Omega$ and $n\in\mathbb{N}$,
we restrict $\Gamma_{\omega}$ to a compact graph by removing the
edges $(-1,0)\cdot L$ and $(n,n+1)\cdot L$ from $\Gamma_{\omega}$,
and denote by $\pa$ the resulting compact connected component (see
Figure \ref{fig: GammaN}). At the cut vertices, $0$ and $nL$, impose
Neumann-Kirchhoff vertex conditions (though the results below do not
depend on the vertex conditions as long as they render the operator
self-adjoint). The corresponding Kirchhoff Laplacian $\Ha$ is bounded
from below and has a compact resolvent, and thus has purely discrete
spectrum accumulating at infinity. Denote the associated normalized
spectral counting function by
\begin{equation}
N_{H_{\omega}}^{(n)}\left(E\right):=\frac{\#\left\{ \lambda\in\spec{\Ha}:\lambda\le E\right\} }{\left|\pa\right|}.\label{eq:truncation-IDS}
\end{equation}

\begin{figure}
\includegraphics[scale=0.55]{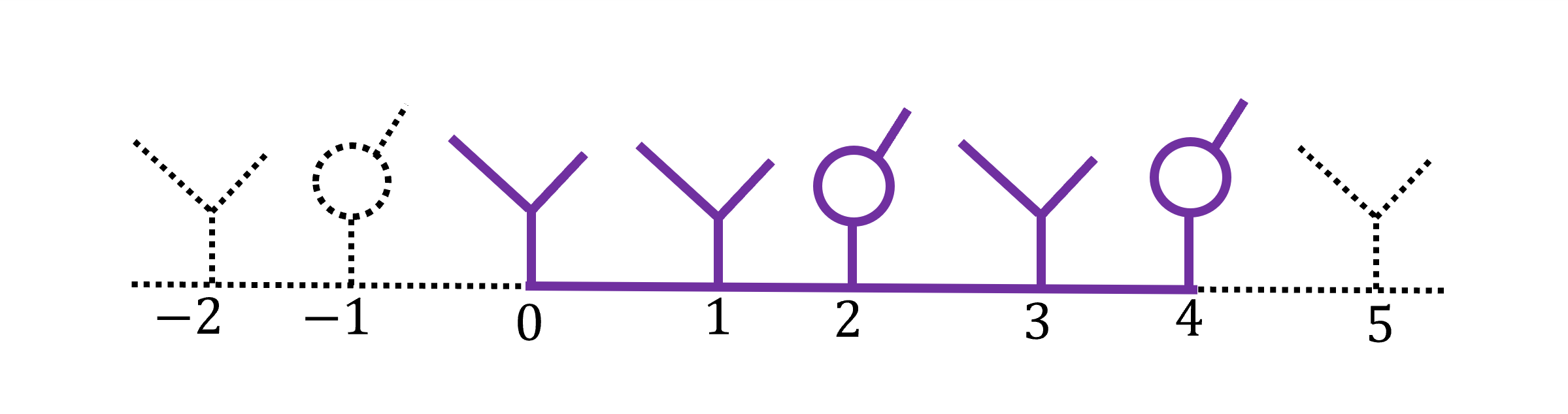}

\caption[Truncated decorated $\Z$-graph.]{The compact graph $\Gamma_{\omega}|_{\left[0,4\right]}$, constructed
by truncating the infinite graph $\Gamma_{\omega}$ and keeping five
decorations. \label{fig: GammaN}}
\end{figure}

\begin{prop}
\label{prop:IDS-existence}For $\mu$-almost all $\omega\in\Omega$,
the sequence of functions $N_{H_{\omega}}^{(n)}\left(E\right)$ converges
uniformly in $E\in\R$ as $n\rightarrow\infty$ to a limiting function
$\NHE{H_{\Omega}}:\mathbb{R}\rightarrow\mathbb{R}$, which is independent
of $\omega\in\Omega$.
\end{prop}

Similar to (\ref{eq:IDS}), we call the function $\NHE{H_{\Omega}}\left(E\right)$
the integrated density of states (IDS) of the family $H_{\Omega}$.
The proof of Proposition \ref{prop:IDS-existence} appears in Subsection
\ref{subsec:IDS existence}. Our main interest is in the values taken
by $\NHE{H_{\Omega}}$ on spectral gaps, namely the gap labels:
\begin{equation}
\mathcal{GL}\left(\NHE{H_{\Omega}}\right):=\left\{ \NHE{H_{\Omega}}\left(E\right):E\in\R\backslash\spec{H_{\Omega}}\right\} \subset\mathbb{R}.\label{eq:GL}
\end{equation}

The IDS in the discrete setting is completely analogous to the metric
case discussed above. Let $\left(G_{\omega}\right)_{\omega\in\Omega}$
be a family of discrete decorated $\Z$-graphs, equipped with the
NDL $\left(\Delta_{\omega}\right)_{\omega\in\Omega}$ (one may also
take a Jacobi operator as described in Section \ref{sec:DTMP}). Remove
from $G_{\omega}$ the two edges $(-1,0)$ and $(n,n+1)$ and denote
by $\dpa$ the resulting compact connected component. The resulting
operator $\left.\Delta_{\omega}\right|_{[0,n]}$ is a self-adjoint
matrix. We define the IDS as the limit of associated normalized spectral
counting functions,
\begin{equation}
\NHE{\Delta_{\Omega}}\left(E\right):=\lim_{n\rightarrow\infty}\frac{\#\left\{ \lambda\in\spec{\left.\Delta_{\omega}\right|_{[0,n]}}:\lambda\le E\right\} }{\left|\V_{\dpa}\right|},\label{eq:truncation-IDS-2}
\end{equation}
where the limit exists for $\mu$-almost all $\omega\in\Omega$ and
its value is independent of $\omega$ (the proof is similar to that
of Proposition\ \ref{prop:IDS-existence}).

\subsection{Schwartzman group\label{subsec:Schwartzman-group}}

We now introduce the dynamical object that will constrain the possible
gap labels in Theorems \ref{thm:GLT} and \ref{thm:Discrete-GLT},
namely the Schwartzman group associated with the dynamical system.
For more details, see \cite{Damanik2023b,Damanik[2023]copyright2023,Damanik2023,Damanik2022}
and references therein. It was first shown by Johnson \cite{Johnson1986}
(see also \cite{Johnson1982,Schwartzman1957}) that for certain Schrödinger
operators on $\R$, the IDS takes values in $\S$. Since then, and
especially in recent years, this approach (known as Johnson--Schwartzman
gap labelling) has been successfully extended to ergodic Schrödinger
operators on $\Z$, Jacobi matrices, and CMV matrices. It is known
that for one-dimensional systems and whenever both approaches (K-theory
and Johnson--Schwartzman) provide a well-defined label set, these
label sets agree \cite{Kellendonk}.

Let $\left(\Omega,\sft\right)$ be a uniquely ergodic subshift, equipped
with a (unique) invariant probability measure $\mu$. We associate
with this dynamical system a suspension space:
\begin{equation}
X_{\Omega}:=\Omega\times\left[0,1\right]/\left\{ \left(\omega,1\right)\sim\left(\sft\omega,0\right)\right\} .\label{eq:suspension}
\end{equation}
The space $X_{\Omega}$ is naturally endowed with the translation
flow in the second factor:
\begin{align}
 & \tau^{t}:X_{\Omega}\rightarrow X_{\Omega}\,\,\,\,\left(t\in\R\right),\label{eq:flow-1}\\
 & \tau^{t}\left(\omega,s\right)=\left(\omega,t+s\text{ mod \ensuremath{1}}\right),\label{eq:flow-2}
\end{align}
and with a probability measure $\eta$:
\begin{equation}
\int_{X_{\Omega}}fd\eta=\int_{0}^{1}\int_{\Omega}f\left(\left[\omega,t\right]\right)d\mu\left(\omega\right)dt.\label{eq:measure-prod}
\end{equation}

Let $C^{\sharp}\left(X_{\Omega}\right)$ denote the space of homotopy
classes of functions from $X_{\Omega}$ to the one-dimensional torus
$\TT:=\R/\Z$. For a given function $\phi:X_{\Omega}\rightarrow\TT$,
let $\phi_{x}$ denote the restriction of $\phi$ to the orbit of
a point $x=\left(\omega,s\right)\in X_{\Omega}$ under the flow $\tau^{t}$:
\begin{align}
 & \phi_{x}:\R\rightarrow\TT,\label{eq:lift1}\\
 & \phi_{x}\left(t\right)=\phi\left(\tau^{t}x\right).\label{eq:lift2}
\end{align}
Since $\R$ is the universal cover of $\TT$$,$ the function $\phi_{x}$
is naturally lifted to a map $\widetilde{\phi}_{x}\left(t\right):\R\rightarrow\R$.
With this in mind, define the Schwartzman homomorphism by
\begin{align}
 & S_{\Omega}:C^{\sharp}\left(X_{\Omega}\right)\rightarrow\mathbb{R},\label{eq:SH1}\\
 & S_{\Omega}\left(\left[\phi\right]\right)=\lim_{t\rightarrow\infty}\frac{\widetilde{\phi}_{x}\left(t\right)}{t},\label{eq:SH2}
\end{align}
where the limit above exists for $\eta$-almost every $x=(\omega,s)$
and is independent of $x$ \cite[thm. 3.9.13]{Damanik2022}. In other
words, the Schwartzman homomorphism is the average rate of rotation
of $\phi_{x}$ along the flow.

The image of $S_{\Omega}$ is a countable subgroup of $\R$, known
as the Schwartzman group $\S$. The Schwartzman group depends on the
full dynamical system $\left(\Omega,\sft,\mu\right)$, but for brevity
we denote it by $\S$. 
\begin{example}
\label{exa:SG-Sturmian} Let $\left(\Omega_{\alpha},\sft\right)$
be the Sturmian subshift from Example \ref{exa:Sturmian}. Its Schwartzman
group is given by 
\begin{equation}
\mathfrak{S}_{\Omega_{\alpha}}=\set{\alpha n+m}{m,n\in\Z},\label{eq:SG-Sturm}
\end{equation}
see \cite[thm. 10.9.3]{Damanik}.
\end{example}

\subsection{Statement of main results}

With the IDS and the Schwartzman group now in place, we can state
the gap-labelling theorems precisely.

Our first main result is a gap labelling theorem (GLT) for the metric
graph operator family $H_{\Omega}$, giving an upper bound on the
possible gap labels.
\begin{thm}
\label{thm:GLT}Let $\left(\Omega,\sft\right)$ be a uniquely ergodic
subshift, with an associated family of metric decorated $\Z$-graphs
$\Gamma_{\Omega}$, equipped with the Kirchhoff Laplacian $H_{\Omega}$.
Then,
\begin{equation}
\mathcal{GL}\left(\NHE{H_{\Omega}}\right)\subset\frac{1}{\overline{L}\left(\Gamma_{\Omega}\right)}\S\cap[0,\infty),\label{eq:GL1}
\end{equation}
where $\S$ is the Schwartzman group, and $\overline{L}\left(\Gamma_{\Omega}\right)$
is the normalized length (\ref{eq:norm-length}).
\end{thm}

Applying the theorem above to the Sturmian subshift, $\Omega_{\alpha}$,
we obtain the following:
\begin{cor}
\label{cor:GLT-Sturmian}For a metric Sturmian decorated $\Z$-graph,
the possible gap labels are given by
\begin{equation}
\mathcal{GL}\left(\NHE{H_{\Omega_{\alpha}}}\right)\subset\left\{ \frac{\alpha n+m}{L+\alpha\ell_{1}+\left(1-\alpha\right)\ell_{2}}:m,n\in\mathbb{\Z}\right\} \cap[0,\infty),\label{eq:AP-GL}
\end{equation}
where $\ell_{1},\ell_{2}$ are the total lengths of the decorations,
and $L$ is the horizontal distance between the decorations.
\end{cor}

This result will be the starting point for the DTMP studied in Sections
\ref{sec:DTMP}--\ref{sec:DTMP-computation}.

Our next main result is a GLT for the discrete decorated $\Z$-graphs:
\begin{thm}
\label{thm:Discrete-GLT} Let $\left(\Omega,\sft\right)$ be a uniquely
ergodic subshift, with an associated family of discrete decorated
$\Z$-graphs $G_{\Omega}$, equipped with the normalized discrete
Laplacian $\Delta_{\Omega}$. Then
\begin{equation}
\mathcal{GL}\left(\NHE{\Delta_{\Omega}}\right)\subset\frac{1}{\overline{V}\left(G_{\Omega}\right)}\S\cap\left[0,1\right],\label{eq:DGL}
\end{equation}
where $\overline{V}\left(G_{\Omega}\right)$ is the average number
of vertices (\ref{eq:discrete-norm-l}).
\end{thm}

Section \ref{sec:GLT-proofs} is devoted to the proofs: first the
existence of the IDS, then the metric GLT, and lastly its discrete
version.

\newpage{}

\section{Gap labelling -- proofs\label{sec:GLT-proofs}}

\subsection{Existence of IDS\label{subsec:IDS existence}}

In this subsection, we prove Proposition \ \ref{prop:IDS-existence},
namely that the IDS for metric decorated $\Z$-graphs is well-defined
and given by the limit of the spectral counting functions (\ref{eq:truncation-IDS}).
The proof proceeds by establishing a Pastur--Shubin-type trace formula,
which expresses the IDS as an ergodic average of local spectral quantities.
The proof for the discrete case is analogous and omitted.

In the following, we denote by $\mathcal{F}$ the set of finite subsets
of $\mathbb{Z}$. For any subset $Q\in\mathcal{F}$, let $\left.H_{\omega}\right|_{Q}$
represent the restriction of the operator $H_{\omega}$ to the compact
subgraph $\left.\Gamma_{\omega}\right|_{Q}$ of the decorated $\Z$-graph
$\Gamma_{\omega}$ induced by $Q$ (as in the previous subsection).
We impose the Dirichlet condition at the boundary vertices of $\left.\Gamma_{\omega}\right|_{Q}$
where the decorated $\Z$-graph $\Gamma_{\omega}$ is truncated, although
other self-adjoint boundary conditions would yield the same results.
We prove the following Pastur--Shubin-type trace formula, which gives
Proposition\ \ref{prop:IDS-existence} as an immediate corollary:
\begin{prop}
\label{prop:PS-trace-formula}Let $\left(\Omega,\sft\right)$ be a
uniquely ergodic subshift. Denote by $\mu$ the unique shift-invariant
probability measure on $\Omega$. For $\mu$-almost every $\omega\in\Omega$,
the sequence of normalized counting functions $N_{H_{\omega}}^{(n)}\left(E\right)$
in (\ref{eq:truncation-IDS}) converges uniformly in $E\in\R$ to
a limiting function $\NHE{H_{\Omega}}\left(E\right)$. For an arbitrary
nonempty finite $Q\in\mathcal{F}$, the function $\NHE{H_{\Omega}}$
can be expressed as
\begin{equation}
\NHE{H_{\Omega}}\left(E\right)=\frac{1}{\left|Q\right|\overline{L}\left(\Gamma_{\Omega}\right)}\int_{\Omega}tr\left[\chi_{\left.\Gamma_{\omega}\right|_{Q}}\chi_{(-\infty,E]}\left(H_{\omega}\right)\right]\rmd\mu\left(\omega\right),\label{eq:trace}
\end{equation}
where $\overline{L}\left(\Gamma_{\Omega}\right)$ denotes the average
metric length, as defined in (\ref{eq:norm-length}).
\end{prop}

Our proof relies on an adaptation of the method presented in \cite{Gruber2007},
which utilizes an ergodic theorem proven in \cite{Lenz2008a} (see
also \cite{Gruber2008}). Notably, the proof can be generalized to
many other graph families, including graphs with random potentials
and vertex conditions, higher-dimensional decorated graphs (i.e. $\Z^{d}$
with $d>1$), and tiling graphs, as studied in \cite{Band2024a}.

\subsubsection{Background and definitions}

We introduce the objects needed to apply the ergodic theorem, and
refer to \cite{Lenz2008a} for more details.

We denote the spectral counting function (and normalized spectral
counting function) for $H_{\omega}^{Q}$ by

\begin{align}
 & n_{\omega}^{Q}\left(E\right):=\#\left\{ \lambda\in\spec{\left.H_{\omega}\right|_{Q}}:\lambda\leq E\right\} ,\label{eq:nQ}\\
 & N_{\omega}^{Q}\left(E\right):=\frac{1}{\left|\left.\Gamma_{\omega}\right|_{Q}\right|}n_{\omega}^{Q}\left(E\right).\label{eq:NQ}
\end{align}
To decouple the graph into its decorations, we further introduce the
operator $\left.H_{\omega,D}\right|_{Q}$, obtained by imposing Dirichlet
conditions at the centers of all edges of the horizontal path. The
corresponding spectral counting function is
\begin{equation}
n_{\omega,D}^{Q}\left(E\right):=\#\left\{ \lambda\in\spec{\left.H_{\omega,D}\right|_{Q}}:\lambda\leq E\right\} .\label{eq:nQD}
\end{equation}
To simplify notation, let $n_{D}^{a}\left(E\right)$ denote the counting
function $n_{\omega,D}^{Q}$ when $\left.\Gamma_{\omega}\right|_{Q}=\Gamma_{a}$
is a single decoration of type $a\in\A$. The overall counting function
for the decoupled operator $H_{\omega,D}^{Q}$ can then be written
as:
\begin{equation}
n_{\omega,D}^{Q}\left(E\right)=\sum_{a\in\mathcal{A}}\#_{a}^{Q}\left(\omega\right)n_{D}^{a}(E),\label{eq:nQD2}
\end{equation}
 where $\#_{a}^{Q}\left(\omega\right)$ is the number of occurrences
of the letter $a$ in the subword $\omega|_{Q}$ (extending the definition
of the letter counting function from (\ref{eq:counting})). With the
above, we define the \textit{spectral shift} function
\begin{equation}
\xi_{\omega}^{Q}\left(E\right):=n_{\omega}^{Q}\left(E\right)-n_{\omega,D}^{Q}\left(E\right)=n_{\omega}^{Q}\left(E\right)-\sum_{a\in\mathcal{A}}\#_{a}^{Q}\left(\omega\right)n_{D}^{a}(E).\label{eq:S-Shift}
\end{equation}

Lastly, we provide a few definitions which are required for the proofs.
\begin{defn}
A \textit{van Hove} sequence is a sequence $\left(Q_{j}\right)_{j\in\mathbb{N}}\subset\mathcal{F}$
such that
\begin{equation}
\lim_{j\rightarrow\infty}\frac{\left|\partial Q_{j}\right|}{\left|Q_{j}\right|}=0,\label{eq:Hove}
\end{equation}
where the boundary $\partial Q$ is defined as
\begin{equation}
\partial Q:=\set{n\in Q}{n+1\notin Q\text{ or }n-1\notin Q}.\label{eq:bdry}
\end{equation}
\end{defn}

\begin{defn}
A function $b:\mathcal{F}\rightarrow[0,\infty)$ is called a \textit{boundary
term} if
\end{defn}

\begin{enumerate}
\item $b\left(Q\right)=b\left(m+Q\right)$ for all $m\in\mathbb{Z}$ and
$Q\in\mathcal{F}$,
\item there exists $D>0$ such that $b\left(Q\right)\leq D\left|Q\right|$
for all $Q\in\mathcal{F}$,
\item for any van Hove sequence $\left(Q_{j}\right)_{j\in\N}$, the following
holds:
\begin{equation}
\lim_{j\rightarrow\infty}\frac{b\left(Q_{j}\right)}{\left|Q_{j}\right|}=0.\label{eq:Boundary-map}
\end{equation}
\end{enumerate}
\begin{defn}
~Let $X$ be a Banach space.
\begin{enumerate}
\item A function $F:\mathcal{F}\rightarrow X$ is called almost-additive
if there exists a boundary term $b$ such that
\begin{equation}
\left\Vert F\left(\cup_{k=1}^{l}Q_{k}\right)-\sum_{k=1}^{l}F\left(Q_{k}\right)\right\Vert \leq\sum_{k=1}^{l}b\left(Q_{k}\right)\label{eq:AAfunction}
\end{equation}
 for all $l\in\mathbb{N}$ and pairwise disjoint sets $Q_{k}$.
\item For a subshift element $\omega\in\Omega$, $F$ is said to be $\omega$-equivariant
if $F\left(Q\right)$ depends only on the local pattern of $\omega$
at $Q$, i.e.,
\begin{equation}
F\left(Q\right)=F\left(m+Q\right),\label{eq:equivariant}
\end{equation}
 whenever $m\in\mathbb{Z}$ and $Q$ obeys $\omega|_{m+Q}=\omega|_{Q}$. 
\item $F$ is said to be bounded if there exists $C>0$ such that
\begin{equation}
\left\Vert F\left(Q\right)\right\Vert \leq C\left|Q\right|.\label{eq:bdd-function}
\end{equation}
\end{enumerate}
\end{defn}

\subsubsection{Proving the main result}

The following paraphrase of the ergodic theorem \cite[thm. 1]{Lenz2008a}
is a key ingredient in the proof of Proposition \ref{prop:PS-trace-formula}:
\begin{thm}
\label{thm:Ergodic}Let $\left(\Omega,\sft\right)$ be a uniquely
ergodic subshift over $\mathcal{A}$. Let $\left(X,\left\Vert \cdot\right\Vert \right)$
be a Banach space, and let $\left(Q_{j}\right)_{j\in\mathbb{N}}$
be a van Hove sequence. Suppose that $F:\mathcal{F}\rightarrow X$
is an $\omega$-equivariant, almost-additive bounded function. Then
the following limit exists:
\begin{equation}
\overline{F}:=\lim_{j\rightarrow\infty}\frac{F\left(Q_{j}\right)}{\left|Q_{j}\right|}.\label{eq:ergodic}
\end{equation}
\end{thm}

\begin{rem*}
\cite[thm. 1]{Lenz2008a} also assumes existence of all subword frequencies,
which here follows from unique ergodicity (see \cite[prop. 4.4]{Baake2013},
\cite{Oxtoby1952}).
\end{rem*}
The following lemma provides the function $F$ on which Theorem\ \ref{thm:Ergodic}
is applied.
\begin{lem}
\label{lem:Spectral shift}Let $\left(X,\left\Vert \cdot\right\Vert _{\infty}\right)$
be the Banach space of right-continuous bounded functions $f:\R\rightarrow\R$.
We define the function
\begin{align}
 & F:\mathcal{F}\rightarrow X,\label{eq:ss-ergodic}\\
 & \left(F\left(Q\right)\right)\left(E\right)=\frac{\xi_{\omega}^{Q}\left(E\right)}{\overline{L}\left(\Gamma_{\Omega}\right)},\label{eq:ss-ergodic-2}
\end{align}
where $\xi_{\omega}^{Q}$ is the spectral shift function (\ref{eq:S-Shift}).
Then $F$ is $\omega$-equivariant, bounded, and almost-additive,
and hence satisfies the assumptions of Theorem \ref{thm:Ergodic}.
\end{lem}

The proof is similar to \cite[lem. 22]{Gruber2007}. Boundedness follows
since $H_{\omega}^{Q}$ and $H_{\omega,D}^{Q}$ differ by a finite
rank perturbation. Similarly, almost-additivity holds since the disjoint
decomposition $Q=\sqcup_{k=1}^{l}Q_{k}$ yields finite rank perturbations
between the associated operators.

The proof of Proposition \ref{prop:PS-trace-formula} now follows,
using conceptually the same arguments as in \cite[thm. 3]{Gruber2007}.
\begin{proof}[Proof of Proposition \ref{prop:PS-trace-formula}]
 By Lemma \ref{lem:Spectral shift}, the function $\left(F\left(Q\right)\right)\left(E\right)=\xi_{\omega}^{Q}\left(E\right)/\overline{L}\left(\Gamma_{\Omega}\right)$
is $\omega$-equivariant, almost-additive and bounded. Applying (\ref{eq:S-Shift})
along a van Hove sequence $\left(Q_{j}\right)_{j\in\N}$, we get for
all $j\in\N$
\begin{equation}
n_{\omega}^{Q_{j}}\left(E\right)=\xi_{\omega}^{Q_{j}}\left(E\right)+\sum_{a\in\mathcal{A}}\#_{a}^{Q_{j}}\left(\omega\right)n_{D}^{a}\left(E\right).\label{eq:nQD3}
\end{equation}

Dividing both sides by $\left|\Gamma_{\omega}^{Q_{j}}\right|$ and
taking the limit $j\rightarrow\infty$, we obtain using (\ref{eq:truncation-IDS}):
\begin{align}
\NHE{H_{\Omega}}\left(E\right) & =\lim_{j\rightarrow\infty}\frac{n_{\omega}^{Q_{j}}\left(E\right)}{\left|\Gamma_{\omega}|_{Q_{j}}\right|}=_{\text{(\ref{eq:norm-length-2})}}\frac{1}{\overline{L}\left(\Gamma_{\Omega}\right)}\lim_{j\rightarrow\infty}\frac{n_{\omega}^{Q_{j}}\left(E\right)}{\left|Q_{j}\right|}\nonumber \\
 & =\frac{1}{\overline{L}\left(\Gamma_{\Omega}\right)}\lim_{j\rightarrow\infty}\left(\frac{\xi_{\omega}^{Q_{j}}\left(E\right)}{\left|Q_{j}\right|}+\frac{1}{\left|Q_{j}\right|}\sum_{a\in\mathcal{A}}\#_{a}^{Q_{j}}\left(\omega\right)n_{D}^{a}(E)\right)\nonumber \\
 & =\frac{1}{\overline{L}\left(\Gamma_{\Omega}\right)}\lim_{j\rightarrow\infty}\frac{\xi_{\omega}^{Q_{j}}\left(E\right)}{\left|Q_{j}\right|}+\frac{1}{\overline{L}\left(\Gamma_{\Omega}\right)}\sum_{a\in\mathcal{A}}\freq an_{D}^{a}\left(E\right),\label{eq:ergodic-IDS}
\end{align}
and the limit exists by Theorem\ \ref{thm:Ergodic}. The convergence
is uniform, since the Banach-space norm in Lemma\ \ref{lem:Spectral shift}
is $\left\Vert \cdot\right\Vert _{\infty}$.

We now prove formula (\ref{eq:trace}), beginning with independence
of the right-hand side from $Q$. It suffices to first consider $Q=\{m\}$,
i.e., $Q$ is a singleton. This follows from the translation invariance
of the ergodic measure $\mu$, since for all $m\in\Z$:
\begin{align}
 & \int_{\Omega}tr\left[\chi_{\Gamma_{\omega}^{\left\{ m\right\} }}\chi_{(-\infty,E]}\left(H_{\omega}\right)\right]d\mu\left(\omega\right)\nonumber \\
 & =\int_{\Omega}tr\left[\chi_{\Gamma_{\sft\omega}^{\left\{ m\right\} }}\chi_{(-\infty,E]}\left(H_{\sft\omega}\right)\right]d\mu\left(\omega\right)\nonumber \\
 & =_{\left(\ref{eq:covariant}\right)}\int_{\Omega}tr\left[\chi_{\Gamma_{\sft\omega}^{\left\{ m\right\} }}\chi_{(-\infty,E]}\left(\sft^{-1}H_{\omega}\sft\right)\right]d\mu\left(\omega\right)\nonumber \\
 & =\int_{\Omega}tr\left[\sft^{-1}\chi_{\Gamma_{\sft\omega}^{\left\{ m\right\} }}\chi_{(-\infty,E]}\left(H_{\omega}\right)\sft\right]d\mu\left(\omega\right)\nonumber \\
 & =\int_{\Omega}tr\left[\chi_{\Gamma_{\sft\omega}^{\left\{ m\right\} }}\chi_{(-\infty,E]}\left(H_{\omega}\right)\right]d\mu\left(\omega\right)\nonumber \\
 & =\int_{\Omega}tr\left[\chi_{\Gamma_{\omega}^{\left\{ m+1\right\} }}\chi_{(-\infty,E]}\left(H_{\omega}\right)\right]d\mu\left(\omega\right),\label{eq:point-invariance}
\end{align}
where the first equality in the last line follows from the cyclic
property of the trace. For arbitrary $Q$, the claim follows by writing
$\frac{1}{\left|Q\right|}\chi_{\Gamma_{\omega}^{Q}}=\frac{1}{\left|Q\right|}\sum_{m\in Q}\chi_{\Gamma_{\omega}^{\left\{ m\right\} }}$.

To prove (\ref{eq:trace}), we first show that
\begin{align}
\lim_{j\rightarrow\infty}\left[\frac{1}{\left|Q_{j}\right|\overline{L}\left(\Gamma_{\Omega}\right)}tr\left(\chi_{Q_{j}}f_{z}\left(H_{\omega}\right)\right)-\frac{1}{\left|\left.\Gamma_{\omega}\right|_{Q_{j}}\right|}tr\left(f_{z}\left(\left.H_{\omega}\right|_{Q_{j}}\right)\right)\right] & =0,\label{eq:trace-id}
\end{align}
for $f_{z}\left(t\right)=\left(t-z\right)^{-1}$ with $z\in\C\backslash\R$,
where we use $\chi_{Q_{j}}$ as an abbreviated notation for $\chi_{\left.\Gamma_{\omega}\right|_{Q_{j}}}$.

For a given box $Q=Q_{j}$, the graph $\Gamma_{\omega}$ naturally
splits into two components, $\left.\Gamma_{\omega}\right|_{Q}$ and
$\left.\Gamma_{\omega}\right|_{\Z\backslash Q}$. In this case, the
operators $H_{\omega}$ and $\left.H_{\omega}\right|_{Q}\oplus\left.H_{\omega}\right|_{\Z\backslash Q}$
only differ by the boundary conditions imposed at the set $\partial Q$.
We consider the operator
\begin{equation}
D:=f_{z}\left(H_{\omega}\right)-f_{z}\left(\left.H_{\omega}\right|_{Q}\oplus\left.H_{\omega}\right|_{\Z\backslash Q}\right).\label{eq:Res-D}
\end{equation}
Since it is the difference of resolvents of two self-adjoint operators
and $z\notin\R$, we get $\left\Vert D\right\Vert \leq2\left|Im\left(z\right)\right|^{-1}$.
In addition, by the second resolvent identity we get $\mathop{rank}D\leq\left|\partial Q\right|$.
Both bounds imply $\left|\mathop{tr}D\right|\leq2\left|\partial Q\right|\left|Im\left(z\right)\right|^{-1}$.
We use this bound to get
\begin{align}
 & \lim_{j\to\infty}\left|\frac{1}{|Q_{j}|\overline{L}(\Gamma_{\Omega})}tr\!\left\{ \chi_{Q_{j}}f_{z}(H_{\omega})\right\} -\frac{1}{|\Gamma_{\omega}|_{Q_{j}}|}tr\!\left\{ f_{z}\!\left(H_{\omega}|_{Q_{j}}\right)\right\} \right|\nonumber \\[0.5em]
 & \qquad=\lim_{j\to\infty}\frac{1}{|Q_{j}|\overline{L}(\Gamma_{\Omega})}\left|tr\!\left\{ \chi_{Q_{j}}f_{z}(H_{\omega})-f_{z}(H_{\omega}|_{Q_{j}})\right\} \right|\nonumber \\[0.5em]
 & \qquad=\lim_{j\to\infty}\frac{1}{|Q_{j}|\overline{L}(\Gamma_{\Omega})}\left|tr\!\left\{ \chi_{Q_{j}}\!\left(f_{z}(H_{\omega})-f_{z}\!\left(H_{\omega}|_{Q_{j}}\oplus H_{\omega}|_{\mathbb{Z}\setminus Q_{j}}\right)\right)\right\} \right|\nonumber \\[0.5em]
 & \qquad\le\frac{2}{\overline{L}(\Gamma_{\Omega})\,|Im\left(z\right)|}\lim_{j\to\infty}\frac{|\partial Q_{j}|}{|Q_{j}|}=0,\label{eq:trace-bound}
\end{align}
where in the last equality we used that $Q_{j}$ is van Hove.

A Stone-Weierstrass argument then upgrades $(\ref{eq:trace-id})$
to indicator functions $\chi_{(-\infty,E]}$. Applying this to $Q_{j}=\left[0,j\right]$
and recalling that the normalized spectral functions were defined
as $N_{H_{\omega}}^{(j)}(E)=\frac{1}{\left|\left.\Gamma_{\omega}\right|_{Q_{j}}\right|}\mathop{tr}\left\{ \chi_{(-\infty,E]}\left(\left.H_{\omega}\right|_{Q_{j}}\right)\right\} $
we get
\begin{align}
\int_{\Omega}\lim_{j\rightarrow\infty}N_{H_{\omega}}^{(j)}\left(E\right)\ \rmd\mu\left(\omega\right) & =\int_{\Omega}\lim_{j\rightarrow\infty}\frac{1}{\left|\left.\Gamma_{\omega}\right|_{Q_{j}}\right|}\mathop{tr}\left\{ \chi_{(-\infty,E]}\left(\left.H_{\omega}\right|_{Q_{j}}\right)\right\} \ \rmd\mu\left(\omega\right)\nonumber \\
 & =\int_{\omega\in\Omega}\lim_{j\rightarrow\infty}\frac{1}{\left|Q_{j}\right|\overline{L}\left(\Gamma_{\Omega}\right)}\mathop{tr}\left\{ \chi_{\left.\Gamma_{\omega}\right|_{Q_{j}}}\chi_{(-\infty,E]}\left(H_{\omega}\right)\right\} \ \rmd\mu\left(\omega\right)\nonumber \\
 & =\frac{1}{\left|Q\right|\overline{L}\left(\Gamma_{\Omega}\right)}\int_{\omega\in\Omega}\mathop{tr}\left\{ \chi_{\left.\Gamma_{\omega}\right|_{Q}}\chi_{(-\infty,E]}\left(H_{\omega}\right)\right\} \ \rmd\mu\left(\omega\right),\label{eq:PS-trace}
\end{align}
where in the second equality we used ($\ref{eq:trace-id}$) to replace
the finite-volume trace by the corresponding localized infinite-volume
trace, and the third equality follows from the $Q$-independence of
the trace. This completes the proof.
\end{proof}
\begin{cor}
\label{cor:freq-jump}Assume that the frequencies of all finite subwords
in $\Omega$ are positive. The IDS $\NHE{H_{\Omega}}$ has a jump
discontinuity at $E\in\R$ if and only if $E$ admits a compactly
supported eigenfunction.
\end{cor}

The proof follows the same arguments as in \cite[thm. 2]{Klassert2003}
and \cite[cor. 7]{Gruber2007}: a jump discontinuity of the IDS is
equivalent to the dimension of the eigenspace of $E\in\spec{\Ha}$
growing ``sufficiently quickly'' in $n$, which then allows one
to construct compactly supported eigenfunctions for $H_{\omega}$.

\subsection{Preparations for the metric GLT}

\subsubsection{Right propagation along $\Gamma_{\omega}$}

For the proof in the next subsection, we need to establish a notion
of propagation along $\Gamma_{\omega}$. Fix $\omega\in\Omega$, and
set the origin $o\left(\Gamma_{\omega}\right)$ to be the vertex with
$0$ coordinate of the $\Z$-graph (which is identified with the base
vertex of $\Gamma_{\omega(0)}$).

\begin{figure}
\includegraphics[scale=0.55]{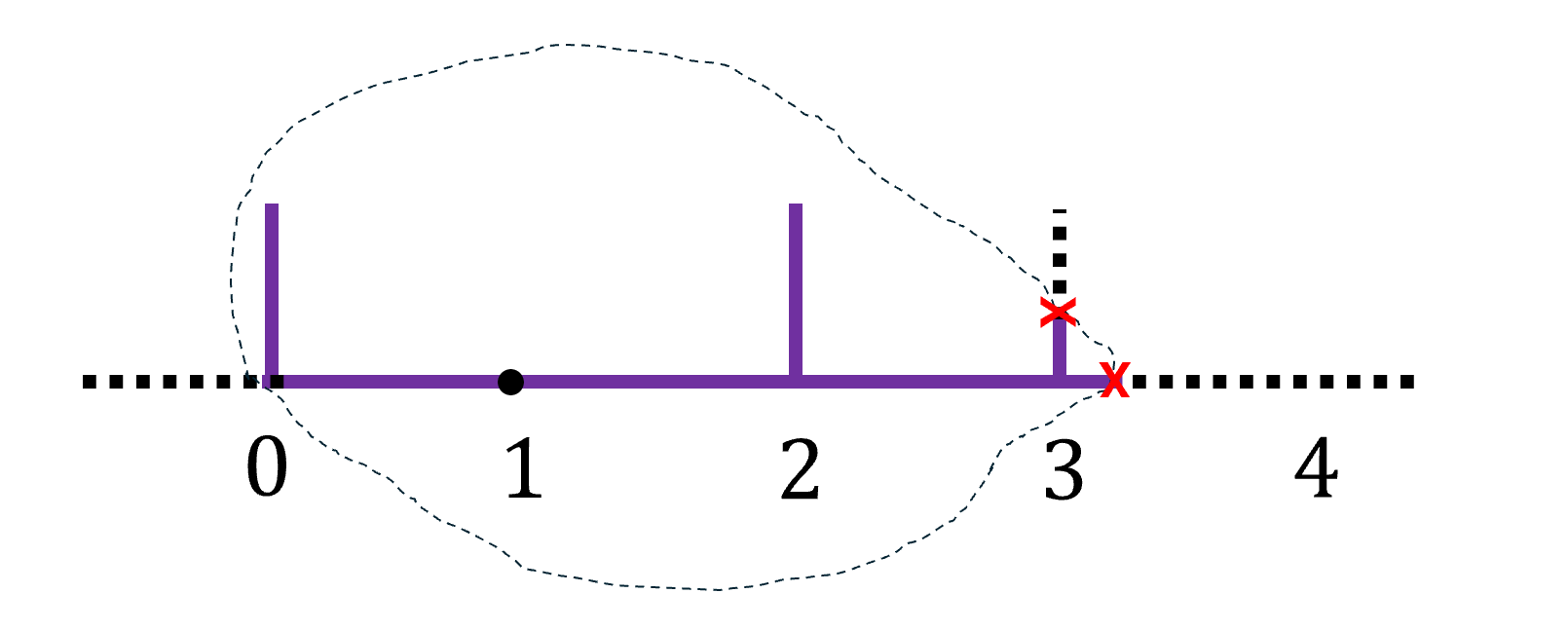}

\caption[The compact graph $\Gamma_{\omega}\left(t\right)$.]{The compact graph $\Gamma_{\omega}\left(t\right)$ for $t\approx3.3$.
Here, $s_{\omega}\left(3.3\right)$ consists of two points.  \label{fig: HorizontalCovering}}
\end{figure}

Assume first that for all $a\in\A$ the total metric length of $\gra$
is smaller than $L$ (the horizontal distance between adjacent decorations).
Under this assumption we set for $t\geq0$, 
\begin{equation}
s_{\omega}\left(t\right):=\set{x\in\Gamma_{\omega}^{+}}{d\left(o\left(\Gamma_{\omega}\right),x\right)=tL},\label{eq: spheres def}
\end{equation}
where $\Gamma_{\omega}^{+}$ is the right part of $\Gamma_{\omega}$,
as in Subsection \ref{subsec:m-definitions}. In particular we note
that $s_{\omega}\left(0\right)=\left\{ o\left(\Gamma_{\omega}\right)\right\} $
and for $k\in\Z$, $s_{\omega}\left(k\right)=\left\{ kL\right\} $.
We wish to maintain the property $s_{\omega}\left(k\right)=\left\{ kL\right\} $
at integer radii even when some decorations $\gra$ have total length
exceeding $L$. To do so, if needed, we can rescale the metric inside
each decoration $\Gamma_{a}$ used in (\ref{eq: spheres def}) by
a factor of $\frac{L}{2\left|\gra\right|}$, while leaving the horizontal
$\Z$-graph unchanged. The original metric graph and the operator
$H_{\omega}$ are not modified.

\subsubsection{The nodal surplus of a decoration\label{subsec:The-nodal-surplus}}

We begin by introducing the nodal count of metric graphs, which is
later used for the proof of Theorem~\ref{thm:GLT}.

Let $\gra$ be a decoration of type $a\in\A$ with base vertex $v_{a}$.
For an energy $E\in\mathbb{R}$, consider the differential equation
on the decoration,
\begin{equation}
-\frac{d^{2}}{dx^{2}}f=Ef,\label{eq:E-ODE}
\end{equation}
subject to the Kirchhoff conditions at all vertices of $\gra$, except
for $\va$, where we impose only a continuity condition (\ref{eq:-15-2}),
but no condition on the derivatives. Then, for all but a discrete
subset of $E\in\R$, prescribing the value $f\left(v_{a}\right)$
determines a unique solution. Equivalently, the space of solutions
is one-dimensional, and we fix a nonzero solution $f_{E}$. The exceptional
set is exactly the spectrum of the Kirchhoff Laplacian on $\Gamma_{a}$
with a Dirichlet condition imposed at $v_{a}$ (see e.g., \cite[thm. 2.1 and cor. 2.4]{Band2012a}).
For $E$ values outside this discrete set we denote
\begin{equation}
m_{a}\left(E\right):=\sum_{e\in\E_{v_{a}}}\frac{\fe'|_{e}\left(\va\right)}{\fe\left(\va\right)}.\label{eq:m-func}
\end{equation}
We use $m_{a}\left(E\right)$ to replace the interaction with the
rest of the graph by an effective boundary condition at $v_{a}$.
Namely, we introduce the following (energy dependent) Robin vertex
condition at $\va$:
\begin{align}
 & g|_{e}\left(v_{a}\right)=g|_{e'}\left(v_{a}\right),\quad\quad\forall e,e'\in\E_{v_{a}},\label{eq:m-Robin-1}\\
 & \sum_{e\in\E_{v_{a}}}g'|_{e}\left(\va\right)=m_{a}\left(E\right)g\left(\va\right).\label{eq:Robin2}
\end{align}
By construction, $\left(E,\fe\right)$ is an eigenpair of $-\frac{d^{2}}{dx^{2}}$
on $\gra$, with Kirchhoff condition imposed at all vertices except
for $\va$, where the Robin condition (\ref{eq:m-Robin-1}), (\ref{eq:Robin2})
is imposed. We denote the resulting operator by $\left.H\right|_{\Gamma_{a}}$,
keeping in mind that this operator depends on $E$ (but do not indicate
this in the notation for brevity).

Denoting the (non-normalized) spectral counting function of $\left.H\right|_{\Gamma_{a}}$
by
\begin{equation}
n^{\left(a\right)}\left(E\right):=\#\left\{ \lambda\in\spec{\left.H\right|_{\Gamma_{a}}}:\lambda\le E\right\} ,\label{eq:counting-1}
\end{equation}
we define the nodal surplus of $E$ in $\gra$ by
\begin{equation}
\sigma^{\left(a\right)}\left(E\right):=\#\left\{ \text{zeros of \ensuremath{f_{E}} in \ensuremath{\gra}}\right\} -\left(n^{\left(a\right)}\left(E\right)-1\right).\label{eq:surplus}
\end{equation}
Outside a discrete set of $E$ values, the eigenfunction $f_{E}$
does not vanish at any vertex of $\gra$ \cite[cor. 2.4]{Band2012a}.
Restricting $E$ to be outside the mentioned set, and using the unique
continuation of $f_{E}$ at every edge of $\gra$, we conclude that
the zero set of $f_{E}$ is discrete. Hence the surplus $\sigma^{\left(a\right)}\left(E\right)$
is well-defined for all such $E$ values.

The nodal surplus has been extensively studied for quantum graphs
beginning with \cite{Gnutzmann2004a}. For additional background see
\cite{Alon2022,Alon2018,Band2012a,Alon2020a,Berkolaiko2008}.

\subsection{Proof of Theorem \ref{thm:GLT}}

The proof proceeds in three main steps: first, for each gap of $H_{\omega}$
we define an appropriate function on $X_{\Omega}$ for which the Schwartzman
homomorphism will be evaluated. Second, we relate this function to
the nodal count of a generalized eigenfunction, and finally we combine
these results to express the IDS value at the gap in terms of the
Schwartzman group.

\subsubsection*{\textbf{Step 1: Defining an appropriate function on the suspension.}}

Fix $E\in\R\backslash\spec{H_{\Omega}}$. For $\omega\in\Omega$,
consider the differential equation on $\Gamma_{\omega}^{+}$, 
\begin{equation}
-\frac{d^{2}}{dx^{2}}u\left(x\right)=Eu\left(x\right),\label{eq:E-ODE-1}
\end{equation}
with the Kirchhoff condition imposed at all vertices, except at the
origin, where no boundary condition is imposed. Since $E\notin\spec{H_{\omega}}$,
this equation has a unique solution (up to a scalar multiple) which
lies in $L^{2}(\Gamma_{\omega}^{+})$, denoted $\fwe$ (as described
in Subsection \ref{subsec:m-functions}). In addition, the uniqueness
of the solution guarantees that, up to normalization, $\left.f_{\sft\omega,E}\right|_{\Gamma_{\omega}^{+}}=\left.\sft\fwe\right|_{\Gamma_{\omega}^{+}}$.
Each solution $f_{\omega,E}$ may also be extended to the left (i.e.,
to $\Gamma_{\omega}\backslash\Gamma_{\omega}^{+}$) by solving the
ODE (\ref{eq:E-ODE-1}). We may therefore adopt the notation $f_{\omega,E}$
for a function on the whole $\Gamma_{\omega}$ and we get that $\left.f_{\sft\omega,E}\right|_{\Gamma_{\omega}^{+}}\in L^{2}(\Gamma_{\omega}^{+})$
and $f_{\sft\omega,E}=\sft\fwe$.

We use the function $f_{\omega,E}$ to define a function from the
suspension space to the one-dimensional torus. Consider the following
form of the Cayley transform, 
\begin{equation}
\mathop{C}(t)=\frac{t+\rmi}{t-\rmi},\label{eq:Cayley}
\end{equation}
which maps the left to right oriented real line $\overline{\R}$ (augmented
with $\pm\infty$) onto the clockwise oriented unit circle. Using
this we define the following function on the suspension space: 

\begin{align}
 & \phi:X_{\Omega}\rightarrow\TT\label{eq:phi1}\\
 & \phi\left(\omega,t\right)=\frac{1}{2\pi}\mathop{Arg}\left[\mathop{C}\left(\sum_{x\in s_{\omega}\left(t\right)}\frac{\fwe'\left(x\right)}{\fwe\left(x\right)}\right)\right],\label{eq:phi2}
\end{align}
where $s_{\omega}\left(t\right)$ is given in (\ref{eq: spheres def}),
and $\mathop{Arg}$ is the argument function mapping complex numbers
onto the one-dimensional torus $\left[0,2\pi\right)$. In the sum
over $x\in s_{\omega}\left(t\right)$ above, special care is needed
to the case when $x$ is a vertex. The derivative $\fwe'\left(x\right)$
at a vertex $x$ is defined by parameterizing the elements of $s_{\omega}\left(t\right)$
as $x(t)$ and setting $\fwe'\left(x(t)\right):=\lim_{\tilde{t}\rightarrow t^{-}}\fwe'\left(x\left(\tilde{t}\right)\right)$.
For $t=0$ and $x=o(\Gamma)$, we similarly set $\fwe'\left(x\right):=\lim_{\tilde{x}\rightarrow x^{-}}\fwe'\left(\tilde{x}\right)$.
The function $\phi$ may be considered as a generalized Pr\"ufer
angle. Direct computation shows that $\phi$ is well-defined on $X_{\Omega}$,
as
\begin{align}
\phi\left(\omega,1\right) & =\frac{1}{2\pi}\mathop{Arg}\left[\mathop{C}\left(\sum_{x\in s_{\omega}\left(1\right)}\frac{\fwe'\left(x\right)}{\fwe\left(x\right)}\right)\right]\label{eq:phi-invariance}\\
 & =\frac{1}{2\pi}\mathop{Arg}\left[\mathop{C}\left(\sum_{x\in s_{\sft\omega}\left(0\right)}\frac{f_{\sft\omega,E}'\left(x\right)}{f_{\sft\omega,E}\left(x\right)}\right)\right]=\phi\left(\sft\omega,0\right),\nonumber 
\end{align}
where we have used that $f_{\sft\omega,E}=\sft\fwe$ and also the
equivalence between $s_{\omega}\left(1\right)$ in $\Gamma_{\omega}$
and $s_{\sft\omega}\left(0\right)$ in $\Gamma_{\sft\omega}$ (both
consist of a single point, which is the same up to the isomorphism
between $\Gamma_{\omega}$ and $\Gamma_{\sft\omega}$). We further
argue that $\phi$ is continuous on $X_{\Omega}$. First, we show
that $\phi(\omega,t)$ is continuous in $\omega\in\Omega$ using the
Weyl--Titchmarsh $m$-function of the half infinite graph $\Gamma_{\omega}^{+}$
from Subsection \ref{subsec:m-functions}. By definition of the $m$-function,
we have that $m_{\omega}^{+}(E)=\frac{\fwe'\left(o(\Gamma_{\omega}^{+})\right)}{\fwe\left(o(\Gamma_{\omega}^{+})\right)}$,
and hence $\phi(\omega,0)=\frac{1}{2\pi}\mathop{Arg}\left[\mathop{C}\left(m_{\omega}^{+}(E)\right)\right]$.
The $m$-function $m_{\omega}^{+}(E)$ is continuous in $\omega$
by Corollary \ref{cor:loc-uni}, so that $\phi(\omega,0)$ is also
continuous in $\omega$. The values $\fwe\left(x\right)$ and $\fwe'\left(x\right)$
depend continuously on $\frac{\fwe'\left(o(\Gamma_{w}^{+})\right)}{\fwe\left(o(\Gamma_{w}^{+})\right)}$
as solutions of the ODE (\ref{eq:E-ODE-1}) with the Robin boundary
condition $\frac{\fwe'\left(o(\Gamma_{w}^{+})\right)}{\fwe\left(o(\Gamma_{w}^{+})\right)}$
at the origin. Therefore the right-hand side of (\ref{eq:phi2}) depends
continuously on $m_{\omega}^{+}(E)=\frac{\fwe'\left(o(\Gamma_{w}^{+})\right)}{\fwe\left(o(\Gamma_{w}^{+})\right)}$
and by the argument above we conclude that $\phi(\omega,t)$ is continuous
in $\omega$. Next, we show that $\phi(\omega,t)$ is also continuous
in $t$. Clearly the expression $\sum_{x\in s_{\omega}\left(t\right)}\frac{\fwe'\left(x\right)}{\fwe\left(x\right)}$
is continuous in $t$, when $s_{\omega}\left(t\right)$ does not contain
any vertex of $\Gamma_{\omega}^{+}$. In addition, the Kirchhoff vertex
conditions (\ref{eq:-15-2}),(\ref{eq:-16-2}) ensure the continuity
of $\sum_{x\in s_{\omega}\left(t\right)}\frac{\fwe'\left(x\right)}{\fwe\left(x\right)}$
in $t$ also when $s_{\omega}\left(t\right)$ contains a vertex. Indeed,
continuity gives the same denominator on all incident edges, while
the Kirchhoff condition identifies the sums of derivatives before
and after passing through the vertex. Overall, we conclude that the
function $\phi$ is well-defined and continuous on $X_{\Omega}$.

\subsubsection*{\textbf{Step 2: Expressing $\phi\left(\omega,t\right)$ using the
nodal count of $\protect\fwe$.}}

Having defined $\phi:X_{\Omega}\rightarrow\TT$ in (\ref{eq:phi2})
we wish to apply the Schwartzman homomorphism to it via (\ref{eq:SH2}).
At this point, fix $\omega\in\Omega$ to be in the full measure set
for which (\ref{eq:SH2}) holds. Define
\begin{equation}
\Gamma_{\omega}\left(t\right):=\set{x\in\Gamma_{\omega}^{+}}{d\left(o\left(\Gamma_{\omega}\right),x\right)\leq tL},\label{eq: truncated-graph}
\end{equation}
and note that $s_{\omega}\left(t\right)$ forms part of the boundary
of $\Gamma_{\omega}\left(t\right)$. By (\ref{eq:phi2}) the function
$\phi\left(\omega,t\right)$ equals $0\in\TT$ precisely when $\fwe\left(x\right)=0$
for some $x\in s_{\omega}\left(t\right)$. With this observation we
use the values of $\phi\left(\omega,t\right)$ (or more precisely
its lift) to count the zeros of $\fwe$. To do so, recall the notation
$\phi_{(\omega,0)}(t):=\phi(\tau^{t}(\omega,0))$ and $\widetilde{\phi}_{(\omega,0)}(t)$
for its lift (see Section~\ref{subsec:Schwartzman-group}). With
this notation, the number of zeros of $\fwe$ in $\Gamma_{\omega}\left(t\right)$
(i.e. its nodal count) is equal to the number of times that the function
$\widetilde{\phi}_{(\omega,0)}$ intersects $0\text{ mod }1$ in the
interval $\left[0,t\right]$. We use this observation to connect between
the (average) zero count of $\fwe$ and the value of the Schwartzman
homomorphism $S_{\Omega}\left(\left[\phi\right]\right)$. Explicitly,
using the notation
\begin{equation}
\zc:=\#\set{x\in\Gamma_{\omega}\left(t\right)}{\fwe(x)=0},\label{eq:Zwt}
\end{equation}
 we have 
\begin{equation}
\zc=\left\lfloor \widetilde{\phi}_{(\omega,0)}(t)\right\rfloor ,\label{eq:Zwt-1}
\end{equation}
where $\left\lfloor \phantom{x}\right\rfloor $ denotes the floor
function. Therefore, by (\ref{eq:SH2}),

\begin{equation}
S_{\Omega}\left(\left[\phi\right]\right)=\lim_{t\rightarrow\infty}\frac{1}{t}\left\lfloor \widetilde{\phi}_{(\omega,0)}(t)\right\rfloor =\lim_{t\rightarrow\infty}\frac{1}{t}\zc,\label{eq: Schwarz equals zero count}
\end{equation}
where in the first equality we used that $\omega\in\Omega$ is in
the full measure set for which (\ref{eq:SH2}) holds.

Having this connection between the Schwartzman homomorphism and the
nodal count, we analyze $Z_{\omega,t}$. In what follows we decompose
the total nodal count on $\Gamma_{\omega}(t)$ via the nodal count
of its subgraphs: the decorations, and the horizontal path. Outside
a discrete set of energies $E$, the solution to the ODE (\ref{eq:E-ODE})
on each decoration $\Gamma_{a}$ is unique up to scalar multiple (as
discussed in Subsection~\ref{subsec:The-nodal-surplus}). Hence,
the nodal count on each decoration of a given type does not depend
on the location of this decoration within $\Gamma_{\omega}(t)$. Denoting
this nodal count function by $Z^{\left(a\right)}(E)$, we write
\begin{equation}
\zc=\zch+\sum_{a\in\mathcal{A}}\ct t\zca,\label{eq:nodal-split}
\end{equation}
where $\zch$ is the nodal count function of $\fwe$ on the path graph
$\left[0,tL\right]$ (which is a subgraph of $\Gamma_{\omega}(t)$),
and we extend the definition of the letter counting function (\ref{eq:counting})
to non-integer $t$ values by setting $\ct t:=\ct{\left\lfloor t\right\rfloor }$
to be the number of decorations of type $a$ in $\Gamma_{\omega}(t)$.
Note that (\ref{eq:nodal-split}) is an equality between functions
in $E$, but for brevity we omit the $E$-dependence. As already mentioned,
these functions are well-defined up to a discrete set of $E$ values.
We further use the spectral counting functions (\ref{eq:counting-1})
and nodal surplus functions (\ref{eq:surplus}) to write
\begin{align}
\zc= & \zch+\sum_{a\in\mathcal{A}}\ct t\left(n^{\left(a\right)}+\sigma^{\left(a\right)}-1\right).\label{eq:total-zeros}
\end{align}

We next express the nodal counting functions $\zc$ and $\zch$ through
spectral counting functions of suitable operators. Towards this, we
define the corresponding operators. First, consider the restriction
of $H_{\omega}$ to the finite graph $\Gamma_{\omega}(t)$. At the
vertices $u\in s_{\omega}(t)\cup o(\Gamma_{\omega})$ we impose the
Robin condition
\begin{equation}
\frac{f'\left(u\right)}{f\left(u\right)}=\frac{\fwe'\left(u\right)}{\fwe\left(u\right)},\label{eq:Robin-boundaries}
\end{equation}
and at all other vertices of $\Gamma_{\omega}(t)$ we impose the Neumann-Kirchhoff
vertex conditions as in $H_{\omega}$. We naturally denote the resulting
operator  $\left.H_{\omega}\right|_{\Gamma_{\omega}(t)}$. We describe
now an operator associated with the horizontal subgraph $\left[0,tL\right]$.
Let $v_{m}$ be an interior vertex of $\left[0,tL\right]$, which
is positioned at $mL$, where $m\in\Z\cap(0,t)$. Denote its two neighboring
edges by $e_{m}^{\pm}$. We impose at the vertex $v_{m}$ the Robin-type
conditions
\begin{align}
 & f|_{e_{v}^{+}}\left(v_{m}\right)=f|_{e_{v}^{-}}\left(v_{m}\right)=:f\left(v_{m}\right),\label{eq:Robin-internal-1}\\
 & f'|_{e_{v}^{+}}\left(v_{m}\right)+f'|_{e_{v}^{-}}\left(v_{m}\right)=-m_{\omega(m)}\left(E\right)f\left(v_{m}\right),\label{eq:Robin-internal-2}
\end{align}
where the Robin parameter $m_{\omega(m)}\left(E\right)$ is as in
(\ref{eq:m-func}), and takes into account that in $\Gamma_{\omega}(t)$
the decoration $\Gamma_{\omega(m)}$ is glued to $v_{m}$. At the
boundary vertices $u\in\{o(\Gamma_{\omega}),tL\}$ we impose the same
Robin conditions (\ref{eq:Robin-boundaries}) as were imposed for
$\left.H_{\omega}\right|_{\Gamma_{\omega}(t)}$. Overall these vertex
conditions render the one-dimensional Laplacian on $[0,tL]$ a self-adjoint
operator, which we denote by $\left.H_{\omega}\right|_{[0,tL]}$.
These particular choices of vertex conditions guarantee that $\left(E,\left.\fwe\right|_{[0,tL]}\right)$
is an eigenpair of $\left.H_{\omega}\right|_{[0,tL]}$ and $\left(E,\left.\fwe\right|_{\Gamma_{\omega}(t)}\right)$
is an eigenpair of $\left.H_{\omega}\right|_{\Gamma_{\omega}(t)}$.

Denoting the spectral counting function of $\left.H_{\omega}\right|_{[0,tL]}$
by 
\begin{equation}
\sch:=\#\set{\lambda\in\mathrm{Spec}\left(\left.H_{\omega}\right|_{[0,tL]}\right)}{\lambda\leq E},\label{eq:n-horiz}
\end{equation}
Sturm's oscillation theorem\footnote{See \cite{Berkolaiko2008} and \cite{Schapotschnikow2006a} for the
graph versions of Sturm's oscillation theorem} yields
\begin{equation}
\sch(E)=\zch(E)+1.\label{eq:counting=00003Dzeros}
\end{equation}
 Substituting this in (\ref{eq:total-zeros}) gives 
\begin{equation}
\zc=\sch-1+\sum_{a\in\mathcal{A}}\ct t\left(n^{\left(a\right)}+\sigma^{\left(a\right)}-1\right).\label{eq: nodal_count_via_spectral_count}
\end{equation}

We next relate the spectral counting functions of the three operators
$\left.H_{\omega}\right|_{\Gamma_{\omega}(t)},\left.H_{\omega}\right|_{[0,tL]},\left.H_{\omega}\right|_{\Gamma_{a}}$
discussed above (the operator $\left.H_{\omega}\right|_{\Gamma_{a}}$
was presented in Section \ref{subsec:The-nodal-surplus}, where it
was denoted by $\left.H\right|_{\Gamma_{a}}$).
\begin{lem}
\label{lem:Counting-lemma}Let $E\notin\spec{H_{\omega}}$. Assume
that for all $a\in\A$, the spectrum of the Kirchhoff Laplacian on
$\Gamma_{a}$ with Dirichlet condition imposed at $v_{a}$ does not
contain $E$. Denote the  spectral counting functions of $\left.H_{\omega}\right|_{\Gamma_{\omega}(t)}$,
$\left.H_{\omega}\right|_{[0,tL]}$ and $\left.H_{\omega}\right|_{\Gamma_{a}}$
by $\sc$ , $\sch$ and $\sca$ respectively. Then,
\begin{equation}
\sc\left(E\right)=\sch\left(E\right)+\sum_{a\in\mathcal{A}}\ct t\left(\sca\left(E\right)-1\right).\label{eq:counting-lemma}
\end{equation}
\end{lem}

The proof of the lemma involves a continuous interpolation between
the relevant operators. While this is an interesting method, the proof
is somewhat technical and is postponed to Subsection \ref{subsec:SF}.

\begin{proof}[\textbf{\textit{Step 3: Computing the Schwartzman homomorphism of
$\phi$.}}]
 Using Lemma~\ref{lem:Counting-lemma}, Equation~(\ref{eq: nodal_count_via_spectral_count})
gives
\begin{equation}
\zc=\sc-1+\sum_{a\in\mathcal{A}}\ct t\sigma^{\left(a\right)}.\label{eq:counting-lemma-2}
\end{equation}
Now, computing the Schwartzman homomorphism as in (\ref{eq: Schwarz equals zero count})
gives 
\begin{align}
S_{\Omega}\left(\left[\phi\right]\right) & =\lim_{t\rightarrow\infty}\frac{1}{t}\zc(E)\nonumber \\
 & =\lim_{t\rightarrow\infty}\frac{1}{t}\left(\sc\left(E\right)-1+\sum_{a\in\mathcal{A}}\ct t\sigma^{\left(a\right)}\left(E\right)\right)\nonumber \\
 & =\lim_{t\rightarrow\infty}\frac{\sc\left(E\right)-1}{t}+\sum_{a\in\mathcal{A}}\lim_{t\rightarrow\infty}\frac{\ct t}{t}\sigma^{\left(a\right)}\left(E\right)\nonumber \\
 & =\lim_{t\rightarrow\infty}\left[\left(\frac{\sc\left(E\right)}{\left|\Gamma_{\omega}\left(t\right)\right|}-\frac{1}{\left|\Gamma_{\omega}\left(t\right)\right|}\right)\cdot\frac{\left|\Gamma_{\omega}\left(t\right)\right|}{t}\right]+\sum_{a\in\mathcal{A}}\freq a\sigma^{\left(a\right)}\left(E\right)\nonumber \\
 & =\NHE{H_{\Omega}}\left(E\right)\cdot\overline{L}\left(\Gamma_{\Omega}\right)+\sum_{a\in\mathcal{A}}\freq a\sigma^{\left(a\right)}\left(E\right).\label{eq:SH-compute}
\end{align}
We thus finally obtain
\begin{equation}
\NHE{H_{\Omega}}\left(E\right)=\frac{S_{\Omega}\left(\left[\phi\right]\right)-\sum_{a\in\mathcal{A}}\freq a\sigma^{\left(a\right)}\left(E\right)}{\overline{L}\left(\Gamma_{\Omega}\right)}.\label{eq:GL-SH}
\end{equation}

To complete the proof we need to show that the numerator of (\ref{eq:GL-SH})
belongs to the Schwartzman group, $\S$. By definition, this group
is the image of the Schwartzman homomorphism so that $S_{\Omega}\left(\left[\phi\right]\right)\in\S$.
Since $\S$ is an additive group, it is left to prove that $\sum_{a\in\mathcal{A}}\freq a\sigma^{\left(a\right)}\left(E\right)\in\S$.
From \cite[thm. 7.1]{Damanik[2023]copyright2023}, we know that $\S$
is the $\Z$-module generated by
\begin{equation}
\left\{ \mu\left(\Xi\right):\Xi\text{ is a cylinder set in }\Omega\right\} ,\label{eq:cylinders}
\end{equation}
where a cylinder set is a subset of $\Omega$ of the form
\begin{equation}
\Xi_{W}:=\left\{ \omega\in\Omega:\omega|_{\left[0,...,k-1\right]}=W\right\} \label{eq:cylinder}
\end{equation}
for some finite word $W=W\left(0\right)...W\left(k-1\right)$. In
particular, we consider cylinder sets with a single letter being fixed,
which are of the form
\begin{equation}
\Xi_{a}:=\left\{ \omega\in\Omega:\omega\left(0\right)=a\right\} ,a\in\A.\label{eq:letter-cylinder}
\end{equation}
Since $\mu\left(\Xi_{a}\right)=\nu_{a}$ for a uniquely ergodic subshift,
and $\sigma^{\left(a\right)}\left(E\right)$ is an integer for all
$a\in\A$, we get $\sum_{a\in\mathcal{A}}\freq a\sigma^{\left(a\right)}\left(E\right)\in\S$,
as required.
\end{proof}

\subsection{Comparison of spectral counting functions\label{subsec:SF}}

\subsubsection{Continuously decoupling the graph}

This subsection proves Lemma\ \ref{lem:Counting-lemma}, which compares
the spectral counting functions for the Kirchhoff Laplacian on the
graph $\Gamma_{\omega}\left(t\right)$ with those of some of its subgraphs.

The proof relies on continuously interpolating between the full graph
$\Gamma_{\omega}\left(t\right)$ and the disjoint union of the corresponding
subgraphs of $\Gamma_{\omega}\left(t\right)$. This is done using
a one-parameter family of operators $\left(H_{\tau}\right)_{\tau\in\left[0,\pi\right]}$.
This operator family transitions between the full graph and the split
graph while keeping $\fwe$ an eigenfunction for all the graphs in
the transition process, allowing to relate the spectral counting functions
of the graphs.

\begin{figure}
\includegraphics[scale=0.8]{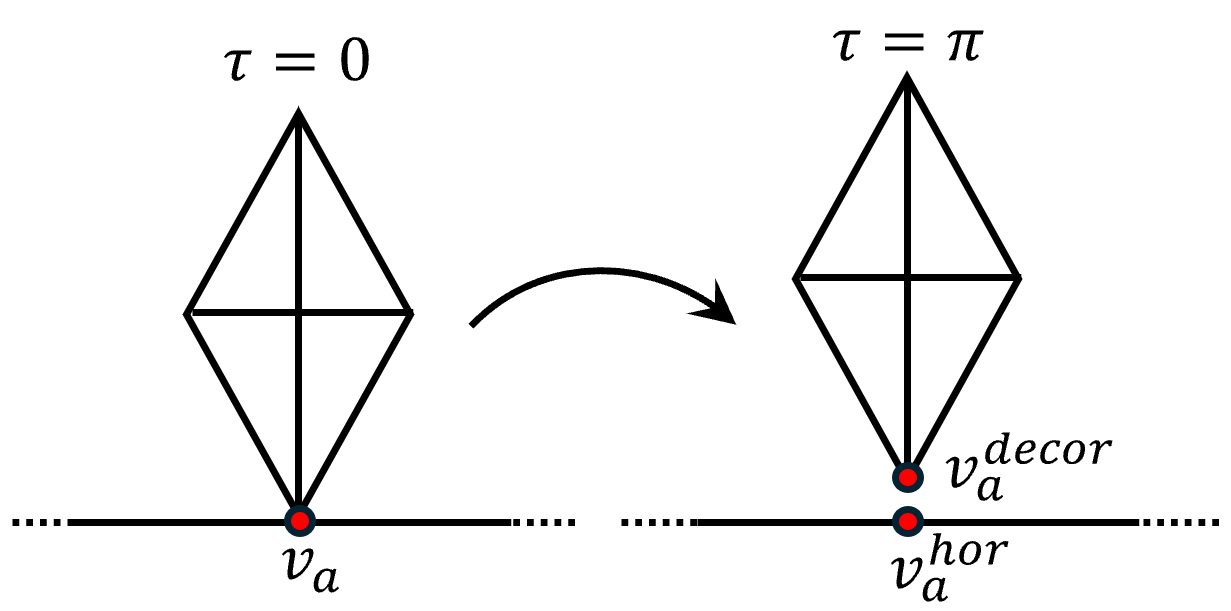}

\caption[{The family $\left(H_{\tau}\right)_{\tau\in\left[0,\pi\right]}$.}]{The family $\left(H_{\tau}\right)_{\tau\in\left[0,\pi\right]}$ which
continuously disconnects the graph $\Gamma_{\omega}^{t}$ into two
subgraphs as $\tau\rightarrow\pi$.  \label{fig: cutgraph-1}}
\end{figure}

To describe the construction, we start by fixing some $E\notin\spec{H_{\omega}}$
such that for each $a\in\A$, it holds that $E$ does not belong to
the spectrum of $\Gamma_{a}$ with a Dirichlet condition at the base
vertex $v_{a}$ of $\Gamma_{a}$ (as in the statement of Lemma~\ref{lem:Counting-lemma}).
Take on $\Gamma_{\omega}$ the unique solution $f_{\omega,E}$, as
is described in step one of the proof of Theorem~\ref{thm:GLT}.
Consider a compact truncation $\Gamma_{\omega}\left(t\right)$ of
the infinite graph $\Gamma_{\omega}$ (see (\ref{eq: truncated-graph})).
The vertices at which the cut is made are $s_{\omega}(t)\cup o(\Gamma_{\omega})$
(see (\ref{eq: spheres def})). For each vertex $u\in s_{\omega}(t)\cup o(\Gamma_{\omega})$
we denote
\begin{equation}
\alpha(u):=\frac{\sum_{e\in\E_{u}\cap\Gamma_{\omega}(t)}\left.f_{\omega,E}'\right|_{e}(u)}{f_{\omega,E}(u)},\label{eq:a-u}
\end{equation}
which is considered to be the Robin parameter of $f_{\omega,E}$ at
$u$ when restricted to $\Gamma_{\omega}(t)$. We consider an arbitrary
decoration $\Gamma_{a}$ which is attached to $\Gamma_{\omega}(t)$.
We can split the graph $\Gamma_{\omega}\left(t\right)$ into two subgraphs:
$\Gamma_{a}$ and $\Gamma_{\omega}\left(t\right)\backslash\Gamma_{a}$,
each containing a respective copy of the base vertex $v_{a}$ of the
decoration. We label these copies $v_{a}^{\mathrm{dec}},v_{a}^{\mathrm{hor}}$,
as in Figure \ref{fig: cutgraph-1}.

We describe a family of operators $\left(H_{\tau}\right)_{\tau\in\left[0,\pi\right]}$
on $\Gamma_{\omega}(t)$ which is now considered as a disjoint union
of $\Gamma_{a}$ and $\Gamma_{\omega}(t)\backslash\Gamma_{a}$. Each
operator $H_{\tau}$ acts as the Laplacian on each edge; at the vertices
$s_{\omega}(t)\cup o(\Gamma_{\omega})$ and $v_{a}^{\mathrm{dec}},v_{a}^{\mathrm{hor}}$
it satisfies continuity conditions (for this purpose $v_{a}^{\mathrm{dec}}$
and $v_{a}^{\mathrm{hor}}$ are considered as separate vertices) and
also the following vertex conditions: 
\begin{align}
 & \sum_{e\in\E_{u}}f'|_{e}\left(u\right)=\alpha(u)f(u),\quad\forall u\in s_{\omega}(t)\cup o(\Gamma_{\omega}),\label{eq: operator-family-VC-1-1}\\
 & \sum_{e\sim v_{a}^{\mathrm{dec}}}f_{e}'\left(v_{a}^{\mathrm{dec}}\right)=m_{a}(E)f\left(v_{a}^{\mathrm{dec}}\right)+\cot\left(\tau/2\right)\left(f\left(v_{a}^{\mathrm{hor}}\right)-f\left(v_{a}^{\mathrm{dec}}\right)\right),\label{eq: operator-family-VC-2-1}\\
 & \sum_{e\sim v_{a}^{\mathrm{hor}}}f_{e}'\left(v_{a}^{\mathrm{hor}}\right)=-m_{a}(E)f\left(v_{a}^{\mathrm{hor}}\right)-\cot\left(\tau/2\right)\left(f\left(v_{a}^{\mathrm{hor}}\right)-f\left(v_{a}^{\mathrm{dec}}\right)\right),\label{eq: operator-family-VC-3-1}
\end{align}
where the Robin parameter $m_{a}\left(E\right)$ is a fixed number
determined from $f_{\omega,E}$ restricted to $\Gamma_{a}$ as in
(\ref{eq:m-func}), and the Kirchhoff conditions are imposed at all
other vertices of $\Gamma_{\omega}(t)$. At $\tau=0$ we also add
the requirement $f(v_{a}^{\mathrm{dec}})=f(v_{a}^{\mathrm{hor}})$.
At $\tau=\pi$, the graph $\Gamma_{\omega}\left(t\right)$ is effectively
split at the vertex $v_{a}$, with the Robin conditions imposed at
both $v_{a}^{\mathrm{dec}}$ and $v_{a}^{\mathrm{hor}}$ (but with
opposite signs of the coupling coefficient).

One can verify that $\left.f_{\omega,E}\right|_{\Gamma_{\omega}(t)}$
satisfies the vertex conditions (\ref{eq: operator-family-VC-1-1}),(\ref{eq: operator-family-VC-2-1}),(\ref{eq: operator-family-VC-3-1})
for all $\tau\neq0$ and so it is an eigenfunction of $H_{\tau}$
for all $\tau\neq0$. At $\tau=0$, the additional condition together
with (\ref{eq: operator-family-VC-2-1}),(\ref{eq: operator-family-VC-3-1})
simplifies to Kirchhoff, and therefore $\left.f_{\omega,E}\right|_{\Gamma_{\omega}(t)}$
is an eigenfunction of $H_{0}$ as well. This is in fact the part
of the rationale behind the particular choice of the operator family
$\left(H_{\tau}\right)_{\tau\in\left[0,\pi\right]}$. We may also
describe this operator family via its quadratic form (the connection
between vertex conditions and quadratic forms for quantum graphs is
standard, see e.g. \cite{BerKuc_graphs}):
\begin{align}
Q_{\tau}\left[f\right]= & \int_{\Gamma_{\omega}(t)}\left|f'\right|^{2}dx~+\sum_{u\in s_{\omega}(t)\cup o(\Gamma_{\omega})}\alpha(u)\left|f(u)\right|^{2}\nonumber \\
 & +m_{a}(E)|f(v_{a}^{\mathrm{hor}})|^{2}-m_{a}(E)|f(v_{a}^{\mathrm{dec}})|^{2}\nonumber \\
 & +\cot\left(\tau/2\right)|f(v_{a}^{\mathrm{hor}})-f(v_{a}^{\mathrm{dec}})|^{2},\label{eq:quadratic-1}
\end{align}
where the Robin parameter $m_{a}\left(E\right)$ is fixed as above
and the domain of $Q_{\tau}$ is taken as all functions in $H^{1}\left(\Gamma_{\omega}\left(t\right)\right)$
which are continuous on $\Gamma_{a}$ and continuous on $\Gamma_{\omega}\left(t\right)\backslash\Gamma_{a}$
(but without requiring continuity at $v_{a}$, i.e., that $f(v_{a}^{\mathrm{hor}})=f(v_{a}^{\mathrm{dec}})$).
For $\tau=0$, the domain further restricts to functions satisfying
$f(v_{a}^{\mathrm{hor}})=f(v_{a}^{\mathrm{dec}})$ as well. Thus at
$\tau=0$ the operator satisfies also Kirchhoff conditions at $v_{a}$
(without splitting it into $v_{a}^{\mathrm{dec}}$ and $v_{a}^{\mathrm{hor}}$).

Thus the family $\left(H_{\tau}\right)_{\tau\in\left[0,\pi\right]}$
continuously interpolates between an operator on the full graph $\Gamma_{\omega}\left(t\right)$
(at $\tau=0$) and an operator on the cut graph $\Gamma_{\omega}\left(t\right)\backslash\Gamma_{a}$
(at $\tau=\pi$). Now, we consider all the other decorations (in addition
to $\Gamma_{a}$ discussed above) which are connected to $\Gamma_{\omega}(t)$.
We redefine the operator family $\left(H_{\tau}\right)_{\tau\in\left[0,\pi\right]}$
such that the vertex conditions at all vertices where the decorations
are attached are changed simultaneously in the same manner as for
$v_{a}$. Namely, the vertex conditions (\ref{eq: operator-family-VC-2-1}),
(\ref{eq: operator-family-VC-3-1}) are imposed at all these decoration
attachment vertices. The effect of this redefined operator family
$\left(H_{\tau}\right)_{\tau\in\left[0,\pi\right]}$ is equivalent
to disconnecting $\Gamma_{\omega}\left(t\right)$ at all these vertices
simultaneously at $\tau=\pi$. At $\tau=0$, we get the operator $\left.H_{\omega}\right|_{\Gamma_{\omega}(t)}$,
namely, the vertex conditions at the vertices of $\Gamma_{\omega}\left(t\right)$
are all Kirchhoff, except for the boundary vertices $s_{\omega}(t)\cup o(\Gamma_{\omega})$,
where the Robin conditions (\ref{eq: operator-family-VC-1-1}) are
imposed. What is important to emphasize is that exactly as above $\left.f_{\omega,E}\right|_{\Gamma_{\omega}(t)}$
is an eigenfunction of $H_{\tau}$ for all $\tau\in\left[0,\pi\right]$,
as follows directly from the definition of the vertex conditions (\ref{eq: operator-family-VC-1-1})-(\ref{eq: operator-family-VC-3-1}).

By standard methods (cf. \cite[thm. 1.4.4]{BerKuc_graphs}), $H_{\tau}$
are all self-adjoint with compact resolvents. Noticing that the map
$\tau\mapsto Q_{\tau}\left[f\right]$ is piecewise analytic with non-positive
derivative for all fixed $f$, a standard Kato-type \cite{Kato1976}
argument (see e.g. \cite{Sofer2022,Bandh} for detailed proofs in
similar systems) can be used to prove the following:
\begin{lem}
\label{lem:analytic-1}The eigenvalue branches $\left(\lambda_{n}\left(\tau\right)\right)_{n\in\N}$
of $\left(H_{\tau}\right)_{\tau\in\left[0,\pi\right]}$ are piecewise
real-analytic, and monotone non-increasing in any interval where they
are differentiable.
\end{lem}

\subsubsection{Proof of Lemma \ref{lem:Counting-lemma}\label{subsec:special-lemma}}
\begin{proof}
The lemma considers the operators $\left.H_{\omega}\right|_{\Gamma_{\omega}(t)}$,
$\left.H_{\omega}\right|_{[0,tL]}$ and $\left.H_{\omega}\right|_{\Gamma_{a}}$.
Their spectral counting functions are denoted by $\sc$ , $\sch$
and $\sca$, respectively. We note that the operator $\left.H_{\omega}\right|_{\Gamma_{\omega}(t)}$
is exactly $H_{0}$ of the operator family $\left(H_{\tau}\right)_{\tau\in\left[0,\pi\right]}$
defined above. Denoting the spectral counting functions of this family
by $\sc^{(\tau)}(E)$ (i.e., $\sc^{(0)}=\sc$), the statement of Lemma
\ref{lem:Counting-lemma} translates to
\begin{equation}
\sc^{(0)}\left(E\right)=\sch\left(E\right)+\sum_{a\in\mathcal{A}}\ct t\left(\sca\left(E\right)-1\right).\label{eq:nwt}
\end{equation}
Now, consider the operator $H_{\pi}$. It is an operator on the cut
version of $\Gamma_{\omega}(t)$, namely the disjoint union of the
horizontal graph with the individual decorations. As such, its spectral
counting function equals the sum of the individual counting functions,
\begin{equation}
\sc^{(\pi)}\left(E\right)=\sch\left(E\right)+\sum_{a\in\mathcal{A}}\ct t\sca\left(E\right).\label{eq: lem-countin - spectral pi}
\end{equation}
 Given (\ref{eq: lem-countin - spectral pi}), we can prove (\ref{eq:nwt})
by showing the following properties on the eigenvalue curves:
\begin{enumerate}
\item \label{enu: lem-counting - prop-1} The operator family $\left(H_{\tau}\right)_{\tau\in[0,\pi]}$
is uniformly bounded from below.
\item \label{enu: lem-counting - prop-2} For $\varepsilon>0$ small enough,
there are exactly $\left\lfloor t\right\rfloor $ crossings of the
eigenvalue curves with the horizontal line $\lambda=E+\varepsilon$
in the interval $\tau\in[0,\pi]$.
\end{enumerate}
We first explain why the two properties above together with (\ref{eq: lem-countin - spectral pi})
imply (\ref{eq:nwt}), and then prove these properties. Consider the
rectangle bounded by $\tau=0$, $\tau=\pi$, $\lambda=E+\varepsilon$
and $\lambda=-C$, where $-C$ is a uniform lower bound of the family
$\left(H_{\tau}\right)_{\tau\in[0,\pi]}$ (see Figure~\ref{fig: SCurves-1}).
It is clear that $\sc^{(0)}\left(E\right)$ and $\sc^{(\pi)}\left(E\right)$
equal the number of intersections of the eigenvalue curves with the
left and right sides of the rectangle respectively. Due to property
(\ref{enu: lem-counting - prop-1}) there are no intersections with
the lower side of the rectangle. By Lemma~\ref{lem:analytic-1} the
eigenvalue curves are monotone non-increasing, and hence the number
of intersections with the top side equals $\sc^{(\pi)}\left(E\right)-\sc^{(0)}\left(E\right)$.
By property (\ref{enu: lem-counting - prop-2}) the number of these
intersections is $\left\lfloor t\right\rfloor $, so that 
\begin{equation}
\sc^{(\pi)}\left(E\right)-\sc^{(0)}\left(E\right)=\left\lfloor t\right\rfloor =\sum_{a\in\mathcal{A}}\ct t,\label{eq:spectral-shift}
\end{equation}
which, combined with (\ref{eq: lem-countin - spectral pi}), proves
(\ref{eq:nwt}).

\begin{figure}
\includegraphics[scale=0.48]{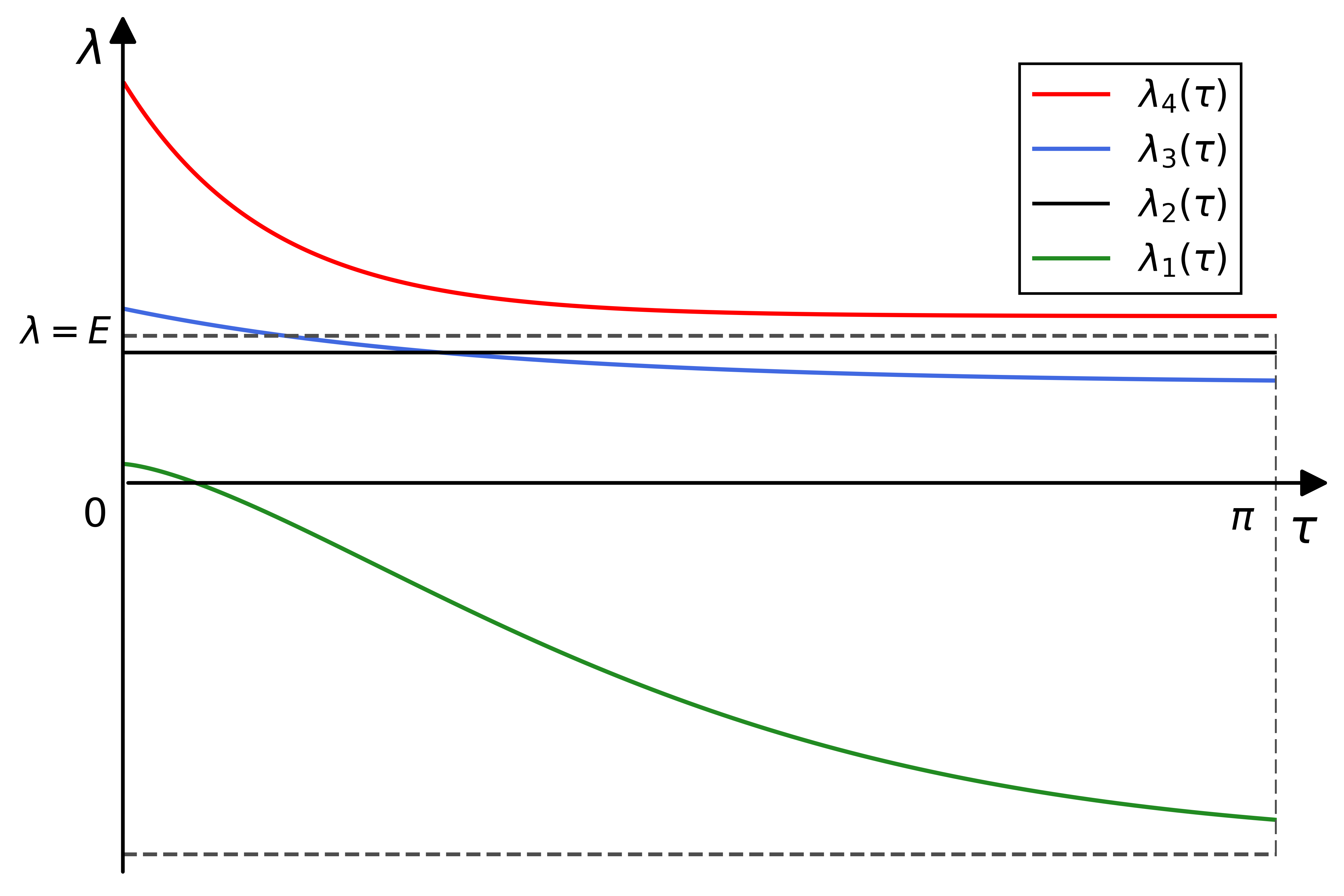}

\caption[Demonstration of the rectangle argument.]{Demonstration of the rectangle argument: the number of intersections
through the top side is equal to the spectral shift.  \label{fig: SCurves-1}}
\end{figure}

\ 

Next, we prove the two properties mentioned above. Property (\ref{enu: lem-counting - prop-1})
follows from the quadratic form (\ref{eq:quadratic-1}): $\left(H_{\tau}\right)_{\tau\in\left[0,\pi\right]}$
is uniformly bounded from below since $\cot(\tau/2)$ is bounded from
below there.

For property (\ref{enu: lem-counting - prop-2}), first recall that
$E$ is as in the proof of Theorem \ref{thm:GLT} (and was used to
define the family $\left(H_{\tau}\right)_{\tau\in\left[0,\pi\right]}$).
We have already observed that the function $\fwe$ is an eigenfunction
of $H_{\tau}$ for all $\tau\in\left[0,\pi\right]$. Hence, there
exists a `flat' eigenvalue branch at the constant value of $\lambda=E$.
Denote by $M$ the multiplicity of $E$ as an eigenvalue of $H_{0}$
(we know that $M\geq1$, since $f_{\omega,E}$ is an eigenfunction).
 Next, we will show the following two statements: (a) The multiplicity
of $E$ as an eigenvalue of $H_{\pi}$ is $M+\left\lfloor t\right\rfloor $
and (b) The multiplicity of $E$ as an eigenvalue of $H_{\tau}$ for
any $\tau\in\left[0,\pi\right)$ is $M$. This will imply that there
are exactly $\left\lfloor t\right\rfloor $ eigenvalue branches which
cross $\lambda=E$ and these crossings occur at $\tau=\pi$. This
immediately yields property (\ref{enu: lem-counting - prop-2}) which
finishes the proof.

(a) Denote by $\{f^{(j)}\}_{j=1}^{M}$ a basis to the $E$-eigenspace
of $H_{0}$. We use the functions $\{f^{(j)}\}_{j=1}^{M}$ to construct
$M+\left\lfloor t\right\rfloor $ functions on $\Gamma_{\omega}(t)$.
Note that given the value $E$, at each of the $\left\lfloor t\right\rfloor $
decorations there is a unique (up to scalar) function which is a solution
of the ODE (\ref{eq:E-ODE}) with Kirchhoff conditions imposed at
all vertices, except at the base vertex where only a continuity condition
is imposed. The uniqueness is guaranteed since we demanded that for
all decorations $\gra$, $E$ is not in the spectrum of the decoration
with Dirichlet condition (if uniqueness is violated, one may construct
such a function which vanishes at the base vertex). We denote these
unique functions on the decorations by $\{f_{a}\}_{a\in\A}$. We thus
get that for each $f^{(j)}$ ($1\leq j\leq M$), its restriction to
each of the decorations either equals the corresponding $f_{a}$,
or identically vanishes at the decoration (in the case where $f^{(j)}(v_{a})=0$).

Given the above observations, we construct $M+\left\lfloor t\right\rfloor $
functions on $\Gamma_{\omega}(t)$ as follows. For each of the $\left\lfloor t\right\rfloor $
decorations, construct a function which is supported only at this
decoration and vanishes everywhere else (i.e., it vanishes at all
the other decorations and at the horizontal line). This gives $\left\lfloor t\right\rfloor $
 eigenfunctions of $H_{\pi}$. We additionally take $\{\left.f^{(j)}\right|_{\left[0,tL\right]}\}{}_{j=1}^{M}$,
which are also eigenfunctions of $H_{\pi}$. We thus get $M+\left\lfloor t\right\rfloor $
$E$-eigenfunctions of $H_{\pi}$. Clearly these functions are linearly
independent and we now show that there are no other $E$-eigenfunctions
of $H_{\pi}$. Assume by contradiction that there is another $E$-eigenfunction
of $H_{\pi}$, denoted by $g$, which is not a linear combination
of the $M+\left\lfloor t\right\rfloor $ functions mentioned above.
To get a contradiction, we construct a function $h$ on $\Gamma_{\omega}(t)$,
such that $\left.h\right|_{\left[0,tL\right]}=\left.g\right|_{\left[0,tL\right]}$
and at each decoration $\gra$ which is included in $\Gamma_{\omega}(t)$
we set $\left.h\right|_{\gra}$ to equal $f_{a}$ up to a scalar multiple
which is chosen to guarantee that $h$ is continuous at the base vertex
of $\gra$. In this way, we obtain that $h$ is an $E$-eigenfunction
of $H_{0}$. At every decoration $\Gamma_{a}$, either $\left.g\right|_{\Gamma_{a}}$
equals $f_{a}$ up to a scalar multiple, or $\left.g\right|_{\Gamma_{a}}\equiv0$
(by the uniqueness mentioned above) . Therefore, $g$ is a linear
combination of $h$ and the $\left\lfloor t\right\rfloor $ functions
which are supported solely at the decorations of $\Gamma_{\omega}(t)$,
but this is a contradiction.

(b) Let $\tau\neq\pi$. Note that $\{f^{(j)}\}_{j=1}^{M}$ are $E$-eigenfunctions
of $H_{\tau}$. We should only show that there are no additional eigenfunctions.
Assume by contradiction that there is an eigenfunction $g$ of $H_{\tau}$
which is linearly independent of $\{f^{(j)}\}_{j=1}^{M}$. In particular,
on every decoration $\gra$ which is included in $\Gamma_{\omega}(t)$
we have that $g$ is a solution of the same ODE as $\left.f_{\omega,E}\right|_{\gra}$
(and as all of $\{\left.f^{(j)}\right|_{\gra}\}_{j=1}^{M}$). This
implies that at the base vertex $v_{a}$ of the decoration we get
$\sum_{e\sim v_{a}^{\mathrm{dec}}}\left.g'\right|_{e}\left(v_{a}^{\mathrm{dec}}\right)=m_{a}\left(E\right)g\left(v_{a}^{\mathrm{dec}}\right)$.
Comparing this to the vertex condition (\ref{eq: operator-family-VC-2-1})
and using $\cot(\tau/2)\neq0$ we get that $g(v_{a}^{\mathrm{dec}})=g(v_{a}^{\mathrm{hor}})$,
i.e., that $g$ is continuous at the base vertex of the decoration.
Since this is valid for all the decorations contained in $\Gamma_{\omega}(t)$
we get that $g$ is an $E$-eigenfunction of $H_{0}$ which is linearly
independent of $\{f^{(j)}\}_{j=1}^{M}$. A contradiction.
\end{proof}
\begin{rem*}
The proof of Lemma~\ref{lem:Counting-lemma} is based on the notion
of spectral flow. The spectral flow of an operator family such as
$\left(H_{\tau}\right)_{\tau\in[0,\pi]}$ is informally defined as
the number of oriented intersections of the eigenvalue curves of $H_{\tau}$
with some horizontal line $E=\mathrm{const}$. The spectral flow is
a topological invariant and an interesting framework in its own right.
To keep the proof self-contained, we instead used a direct argument.
We refer to \cite{BoossBavnbek2005,BoossBavnbek2018,BoossBavnbek2013}
and references therein for a thorough background about the spectral
flow, and also to \cite{Prokhorova,Bandh,Latushkin2020} for applications
of the spectral flow specifically in the context of quantum graphs.
\end{rem*}

\subsection{Gap labelling for discrete graphs\label{subsec:discrete-GLT}}

\subsubsection{Relation between metric and discrete Laplacians}

In this section we prove the GLT for discrete decorated graphs (Theorem~\ref{thm:Discrete-GLT}).
The main tool is the well-known spectral relation between the discrete
Laplacian and the Kirchhoff Laplacian on the corresponding equilateral
metric graph, summarized below.
\begin{thm}
\label{thm:discrete metric} Let $\Gamma$ be an equilateral metric
graph with all edge lengths equal to one, equipped with the Kirchhoff
Laplacian $H$. Let $G$ be the associated discrete graph, equipped
with the normalized discrete Laplacian $\Delta$.
\begin{enumerate}
\item \label{enu: thm-discrete metric-1} For all $k\notin\left\{ \pi m:m\in\N\right\} $,
\begin{equation}
k^{2}\in\spec H\iff1-\cos\left(k\right)\in\spec{\Delta}.\label{eq:discrete-metric}
\end{equation}
Furthermore, if the corresponding points in the spectrum ($k^{2}$
and $1-\cos(k)$) are eigenvalues, then they have the same multiplicities.
\item \label{enu: thm-discrete metric-2} If in addition $\Gamma$ is compact
and connected, then its spectral counting function at $k^{2}=\pi^{2}m^{2}$
equals
\begin{equation}
\#\set{\lambda\in\spec H}{\lambda\leq\pi^{2}m^{2}}=\left|\E_{\Gamma}\right|m+M,\label{eq:spectral-counting}
\end{equation}
where $M\in\{0,1\}$ is the multiplicity of $1-\cos\left(\pi m\right)\in\{0,2\}$
in $\spec{\Delta}$.
\end{enumerate}
\end{thm}

The first part of the theorem is standard (see e.g., \cite{Cattaneo1997,Below1985,Lledo2008,Pankrashkin2006}).
The second part follows from the case-by-case eigenvalue count in
\cite[prop. 6.2]{Lledo2008}, together with some basic properties
of the NDL.

Using Theorem \ref{thm:discrete metric}, we relate the IDS of the
metric and discrete decorated graphs. The ``conversion factor''
which connects between the discrete and metric IDS is given by
\begin{equation}
C\left(G_{\Omega}\right):=\frac{\overline{\V}\left(G_{\Omega}\right)}{\overline{\E}\left(G_{\Omega}\right)}=\frac{\sum_{a\in\A}\freq a\left|\V\left(G_{a}\right)\right|}{1+\sum_{a\in\A}\freq a\left|\E\left(G_{a}\right)\right|},\label{eq:conversion-ratio}
\end{equation}
which is the ratio between the  average number of vertices and the
average number of edges.

\subsubsection{Conversion between metric and discrete IDS}
\begin{prop}
\label{prop:counting-functions}Let $(\Omega,\sft)$ be a uniquely
ergodic subshift. Let $\left\{ \Gamma_{\omega}\right\} _{\omega\in\Omega}$
be a family of decorated $\Z$-graphs, such that each $\Gamma_{\omega}$
is an equilateral graph with all edge lengths equal to $1$. Let $\left\{ G_{\omega}\right\} _{\omega\in\Omega}$
be the associated discrete graphs. Denote the corresponding IDS functions
by $\NHE{H_{\Omega}}\left(E\right),\NHE{\Delta_{\Omega}}\left(E\right)$.
Then, at every point $E$, where $\NHE{H_{\Omega}}$ is continuous,
we have
\begin{equation}
\NHE{H_{\Omega}}\left(E\right)=\left\lfloor \frac{\sqrt{E}}{\pi}\right\rfloor +C\left(G_{\Omega}\right)\cdot\begin{cases}
\NHE{\Delta_{\Omega}}\left(1-\cos\left(\sqrt{E}\right)\right), & \left\lfloor \frac{\sqrt{E}}{\pi}\right\rfloor \text{ is even,}\\
\\1-\NHE{\Delta_{\Omega}}\left(1-\cos\left(\sqrt{E}\right)\right), & \left\lfloor \frac{\sqrt{E}}{\pi}\right\rfloor \text{ is odd.}
\end{cases}\label{eq:disc-metric-relation}
\end{equation}
\end{prop}

\begin{proof}
We first relate the spectral counting functions of compact metric
and discrete graphs. We then take the limit as in Proposition~\ref{prop:IDS-existence}
and (\ref{eq:truncation-IDS-2}) in order to compare the corresponding
IDS.

Let $\Gamma$ be an equilateral compact metric graph with all edge
lengths equal to $1$, and equipped with the Kirchhoff Laplacian $H$.
Let $G$ be the associated discrete graph equipped with $\Delta$.

\begin{figure}
\includegraphics[scale=0.4]{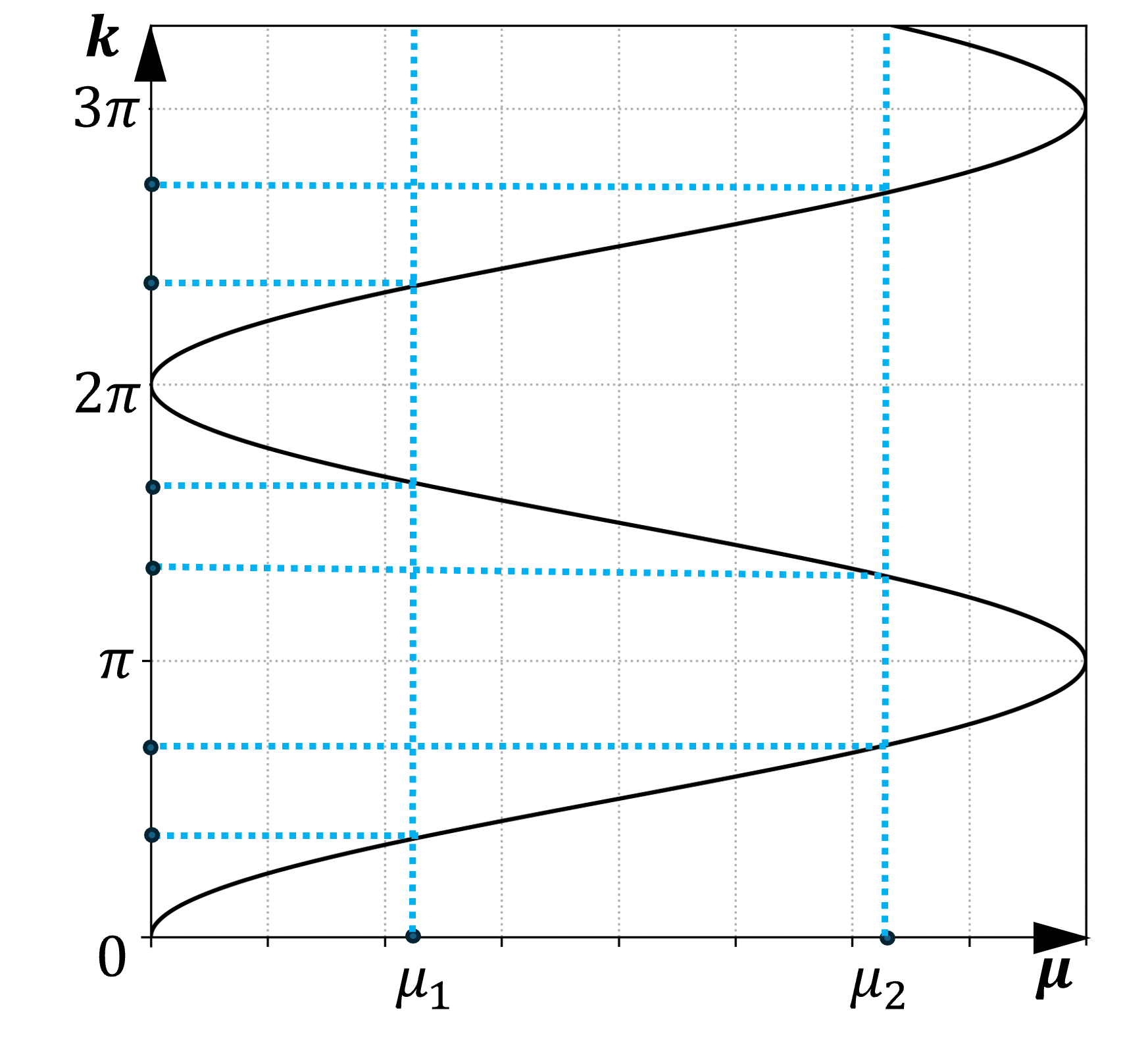}

\caption[The dispersion relation.]{The dispersion relation $k\left(\mu\right)=\arccos\left(1-\mu\right)$
from Theorem \ref{thm:discrete metric} relating $\protect\spec{\Delta}$
(horizontal) and $\protect\spec H$ (vertical). Each $\mu\in\protect\spec{\Delta}$
corresponds to a point $k\left(\mu\right)^{2}\in\protect\spec H$.
The dispersion relation $k\left(\mu\right)$ is either monotone increasing
or monotone decreasing, depending on the parity of the branch $\left\lfloor \frac{k\left(\mu\right)}{\pi}\right\rfloor $.
 \label{fig: dispersion}}
\end{figure}
Let $E\geq0$. Write $E=k^{2}$, and count the total number of square
roots of eigenvalues of $H$ in $\left[0,k\right]$ (with multiplicity).
We write $k=\pi m+r$, with $m\in\N$ and $r\in\left[0,\pi\right)$.
We use Theorem\,\ref{thm:discrete metric} to present the number
of (square roots of) eigenvalues in $\left[0,k\right]$ as the sum
of eigenvalue counts in $\left[0,\pi m\right]$ and in $\left[\pi m,k\right]$.
Since $m=\left\lfloor \frac{k}{\pi}\right\rfloor $, Theorem~\ref{thm:discrete metric}
gives
\begin{align}
 & \#\left\{ \lambda\in\spec H\right.:\left.\lambda\le k^{2}\right\} -\left\lfloor \frac{k}{\pi}\right\rfloor \left|\E_{G}\right|\nonumber \\
= & \begin{cases}
\#\set{\mu\in\spec{\Delta}}{\mu\le1-\cos\left(k\right)}, & \left\lfloor \frac{k}{\pi}\right\rfloor \text{ is even,}\\
\#\set{\mu\in\spec{\Delta}}{\mu\ge1-\cos\left(k\right)}, & \left\lfloor \frac{k}{\pi}\right\rfloor \text{ is odd,}
\end{cases}\nonumber \\
= & \begin{cases}
\#\set{\mu\in\spec{\Delta}}{\mu\le1-\cos\left(k\right)}, & \left\lfloor \frac{k}{\pi}\right\rfloor \text{ is even,}\\
\left|\V_{G}\right|-\#\set{\mu\in\spec{\Delta}}{\mu<1-\cos\left(k\right)}, & \left\lfloor \frac{k}{\pi}\right\rfloor \text{ is odd.}
\end{cases}\label{eq:disc-met-2}
\end{align}
Here, all eigenvalue counts above are with multiplicities (i.e., $\#\left\{ \phantom{\lambda}\right\} $
is considered as element counting of a multi-set). In the first equality
above we need to separate cases according to the parity of $m=\left\lfloor \frac{k}{\pi}\right\rfloor $,
since the dispersion relation $1-\cos(k)$ in (\ref{eq:discrete-metric})
is monotone increasing when $\left\lfloor \frac{k}{\pi}\right\rfloor $
is even and decreasing when $\left\lfloor \frac{k}{\pi}\right\rfloor $
is odd. See Figure~\ref{fig: dispersion}, where the inverse dispersion
relation $\lambda\left(\mu\right)=\arccos\left(1-\mu\right)$ is depicted.

Next, let $\omega\in\Omega$ be in the full measure set for which
Proposition\ \ref{prop:IDS-existence} holds and choose the sequences
of compact graphs $\Gamma^{(n)}:=\pa$ and $G^{(n)}:=\dpa$ as in
(\ref{eq:truncation-IDS}),(\ref{eq:truncation-IDS-2}) and take the
limit $n\rightarrow\infty$ to get the IDS. We perform the computation
only for the case of odd $\left\lfloor \frac{k}{\pi}\right\rfloor $.
The complementary case involves a similar (and slightly simpler) computation.
Using the convergence stated in Proposition\ \ref{prop:IDS-existence}
and applying (\ref{eq:disc-met-2}) we compute:
\begin{align}
\NHE{H_{\Omega}}\left(k^{2}\right)= & \lim_{n\rightarrow\infty}\frac{\#\set{\lambda\in\spec{\left.H_{\omega}\right|_{\Gamma^{(n)}}}}{\lambda\le k^{2}}}{\left|\Gamma^{(n)}\right|}\nonumber \\
= & \lim_{n\rightarrow\infty}\frac{\#\set{\lambda\in\spec{\left.H_{\omega}\right|_{\Gamma^{(n)}}}}{\lambda\le k^{2}}}{\left|\E_{G^{(n)}}\right|}\nonumber \\
= & \lim_{n\rightarrow\infty}\frac{1}{\left|\E\left(G^{(n)}\right)\right|}\cdot\left\lfloor \frac{k}{\pi}\right\rfloor \text{\ensuremath{\left|\E\left(G^{(n)}\right)\right|}}\nonumber \\
\nonumber \\\  & +\lim_{n\rightarrow\infty}\frac{\left|\V\left(G^{(n)}\right)\right|-\#\set{\mu\in\spec{\left.\Delta\right|_{G^{(n)}}}}{\mu<1-\cos\left(k\right)}}{\left|\E\left(G^{(n)}\right)\right|}\nonumber \\
\nonumber \\= & \left\lfloor \frac{k}{\pi}\right\rfloor +\lim_{n\rightarrow\infty}\frac{\left|\V\left(G^{(n)}\right)\right|}{\left|\E\left(G^{(n)}\right)\right|}\thinspace\frac{\left|\V\left(G^{(n)}\right)\right|-\#\set{\mu\in\spec{\left.\Delta\right|_{G^{(n)}}}}{\mu<1-\cos\left(k\right)}}{\left|\V\left(G^{(n)}\right)\right|},\label{eq:disc-met-3}
\end{align}
where in the first equality we used that $\Gamma^{(n)}$ has all edge
lengths equal to one and so $\left|\Gamma^{(n)}\right|=\left|\E\left(G^{(n)}\right)\right|$.
The prefactor inside the limit in the last line above is
\begin{align}
\lim_{n\rightarrow\infty}\frac{\left|\V\left(G^{(n)}\right)\right|}{\left|\E\left(G^{(n)}\right)\right|}= & \lim_{n\rightarrow\infty}\frac{\sum_{a\in\A}\#_{a}^{n}\left(\omega\right)\left|\V\left(G_{a}\right)\right|}{n+\sum_{a\in\A}\#_{a}^{n}\left(\omega\right)\left|\E\left(G_{a}\right)\right|}\nonumber \\
= & \lim_{n\rightarrow\infty}\frac{\sum_{a\in\A}\frac{\#_{a}^{n}\left(\omega\right)}{n}\left|\V\left(G_{a}\right)\right|}{1+\sum_{a\in\A}\frac{\#_{a}^{n}\left(\omega\right)}{n}\left|\E\left(G_{a}\right)\right|}=\frac{\sum_{a\in\A}\freq a\left|\V\left(G_{a}\right)\right|}{1+\sum_{a\in\A}\freq a\left|\E\left(G_{a}\right)\right|}=C\left(G_{\Omega}\right).\label{eq:Betti-avg}
\end{align}
Substituting this above and recalling that $k=\sqrt{E}$ gives
\begin{align}
\NHE{H_{\Omega}}\left(E\right)= & \left\lfloor \frac{\sqrt{E}}{\pi}\right\rfloor +C\left(G_{\Omega}\right)\lim_{n\rightarrow\infty}\frac{\left|\V\left(G^{(n)}\right)\right|-\set{\mu\in\spec{\left.\Delta\right|_{G^{(n)}}}}{\mu<1-\cos\left(\sqrt{E}\right)}}{\left|\V\left(G^{(n)}\right)\right|}\nonumber \\
= & \left\lfloor \frac{\sqrt{E}}{\pi}\right\rfloor +C\left(G_{\Omega}\right)\left[1-\NHE{\Delta_{\Omega}}\left(1-\cos\left(\sqrt{E}\right)\right)\right].\label{eq:final-cases}
\end{align}
Note that by definition $\NHE{\Delta_{\Omega}}\left(1-\cos\left(\sqrt{E}\right)\right)=\lim_{n\rightarrow\infty}\frac{\set{\mu\in\spec{\left.\Delta\right|_{G^{(n)}}}}{\mu\le1-\cos\left(\sqrt{E}\right)}}{\left|\V\left(G^{(n)}\right)\right|}$.
Nevertheless, the last equality above is justified (even though the
strict inequality $\mu<1-\cos\left(\sqrt{E}\right)$ appears), since
we assume that $\NHE{H_{\Omega}}\left(E\right)$ is continuous at
$E$. The distinction between strict and non-strict inequality in
the spectral counting functions matters only when there is a discontinuity
in the IDS (see more on jump discontinuities of the IDS in Subsection\ \ref{subsec:discont-IDS}).

For the case when $\left\lfloor \frac{\sqrt{E}}{\pi}\right\rfloor $
is even a similar computation gives $\NHE{H_{\Omega}}\left(E\right)=\left\lfloor \frac{\sqrt{E}}{\pi}\right\rfloor +C\left(G_{\Omega}\right)\cdot\NHE{\Delta_{\Omega}}\left(1-\cos\left(\sqrt{E}\right)\right)$.
\end{proof}

\subsubsection{Proof of discrete GLT}

\begin{proof}[Proof of Theorem \ref{thm:Discrete-GLT}]
 By Proposition \ref{prop:counting-functions}, 
\begin{equation}
\NHE{\Delta_{\Omega}}\left(1-\cos\left(\sqrt{E}\right)\right)=\begin{cases}
\frac{1}{C\left(G_{\Omega}\right)}\left(\NHE{H_{\Omega}}\left(E\right)-\left\lfloor \frac{\sqrt{E}}{\pi}\right\rfloor \right), & \left\lfloor \frac{\sqrt{E}}{\pi}\right\rfloor \text{ is even,}\\
\\1-\frac{1}{C\left(G_{\Omega}\right)}\left(\NHE{H_{\Omega}}\left(E\right)-\left\lfloor \frac{\sqrt{E}}{\pi}\right\rfloor \right), & \left\lfloor \frac{\sqrt{E}}{\pi}\right\rfloor \text{ is odd.}
\end{cases}\label{eq:ap-relation}
\end{equation}

By Theorem~\ref{thm:GLT}, if $E\notin\spec{H_{\Omega}}$ then 
\begin{equation}
\NHE{H_{\Omega}}\left(E\right)\in\frac{\S}{\overline{\E}(G_{\Omega})},\label{eq:disc-glt-0}
\end{equation}
where we used $\overline{L}(\Gamma_{\Omega})=\overline{\E}(G_{\Omega})$
which holds since the metric graphs are equilateral with each edge
length equal to $1$. We fix $\omega\in\Omega$ to be in the full
measure set for which $\spec{H_{\omega}}=\spec{H_{\Omega}}$ and $\spec{\Delta_{\omega}}=\spec{\Delta_{\Omega}}$
and for which the spectral counting functions converge to the IDS
as in Proposition~\ref{prop:IDS-existence} and Equation~(\ref{eq:truncation-IDS-2}).
From Theorem~\ref{thm:discrete metric} we conclude that $E$ is
inside a spectral gap of $H_{\omega}$ if and only if $1-\cos(\sqrt{E})$
is inside a spectral gap of $\Delta_{\omega}$. Therefore, the possible
gap labels of $\Delta_{\omega}$ (and hence of $\spec{\Delta_{\Omega}}$)
may be obtained by substituting (\ref{eq:disc-glt-0}) in (\ref{eq:ap-relation}).
For this, recall that $\overline{\E}(G_{\Omega})=\sum_{a\in\A}\nu_{a}\E(G_{a})$
and that $\nu_{a}\in\S$ for all $a\in\A$ (as is explained in the
end of the proof of Theorem\ \ref{thm:GLT}). Therefore for $E\notin\spec{H_{\Omega}}$,
\begin{equation}
\NHE{H_{\Omega}}\left(E\right)+\Z\subset\frac{\S+\Z\thinspace\overline{\E}(G_{\Omega})}{\overline{\E}(G_{\Omega})}\subset\frac{\S}{\overline{\E}(G_{\Omega})},\label{eq:N+Z}
\end{equation}
and using $C\left(G_{\Omega}\right)=\frac{\overline{\V}(G_{\Omega})}{\overline{\E}(G_{\Omega})}$
we get 
\begin{equation}
\frac{1}{C\left(G_{\Omega}\right)}\left(\NHE{H_{\Omega}}\left(E\right)+\Z\right)\subset\frac{\S}{\overline{\V}(G_{\Omega})}.\label{eq:disc-glt-1}
\end{equation}
Using again that $\nu_{a}\in\S$ for all $a\in\A$, we get $\overline{\V}(G_{\Omega})\in\S$,
which yields that 
\begin{equation}
1+\frac{1}{C\left(G_{\Omega}\right)}\left(\NHE{H_{\Omega}}\left(E\right)+\Z\right)\subset\frac{\S}{\overline{\V}(G_{\Omega})}.\label{eq:disc-glt-2}
\end{equation}

From (\ref{eq:disc-glt-1}) and (\ref{eq:disc-glt-2}), both cases
in (\ref{eq:ap-relation}) yield the same gap labels,
\begin{equation}
\mathcal{GL}\left(\NHE{\Delta_{\Omega}}\right)\subset\frac{\S}{\overline{\V}(G_{\Omega})}\cap[0,1].\label{GL}
\end{equation}
\end{proof}
\newpage{}

\section{Detailed example -- Sturmian comb graphs\label{sec:example-comb-graphs}}

The purpose of this section is to illustrate the theory presented
in Sections \ref{sec:Kotani}-\ref{sec:GLT-proofs} through a detailed
case study of metric Sturmian comb graphs, and to prepare the ground
for the analysis of the DTMP in the following sections. These are
the decorated $\Z$-graphs presented in Example \ref{exa: Sturm-comb}
(see Figure \ref{fig: TilingGraphs}), which consist of two possible
decorations -- a trivial decoration (a single vertex) and a nontrivial
decoration (a dangling edge, or tooth).

We study this example from two complementary perspectives. First,
in the spirit of Theorems \ref{thm:TMP} and  \ref{thm:Cantor}, we
analyze the global spectral structure, highlighting special features
of this model such as periodic approximations and Hausdorff convergence
of the spectrum. We then turn to gap labelling, following Theorem
\ref{thm:GLT}, and study the jump discontinuities of the IDS which
lead to closed gaps. This simple example will serve as preparation
to the DTMP, where the same mechanisms will be used to determine which
gap labels are realized for discrete decorated $\Z$-graphs.

\subsection{Global spectral structure}

As in Example \ref{exa: Sturm-comb}, we fix $\alpha\in\left(0,1\right)\backslash\Q$,
and consider the Sturmian subshift $\Omega_{\alpha}$. For $\theta\in[0,1)$,
denote by $H_{\alpha,\theta}$ the Kirchhoff Laplacian on the graph
corresponding to the Sturmian sequence $\omega_{\alpha,\theta}\in\Omega_{\alpha}$.
By Theorem \ref{thm:TMP}, $\spec{H_{\alpha,\theta}}$ is independent
of $\theta$, and we therefore focus on $H_{\alpha}:=H_{\alpha,0}$.
For this operator, we define a family of periodic approximations for
$H_{\alpha}$ using the (infinite) continued fraction expansion for
$\alpha$:
\begin{equation}
\alpha=\frac{1}{a_{1}+\frac{1}{a_{2}+...}}.\label{eq:cont-frac}
\end{equation}
By truncating the continued fraction expansion at the $n$th step,
we obtain a rational number $\alpha_{n}:=\frac{p_{n}}{q_{n}}$, which
through (\ref{eq:sturm}) corresponds to a periodic metric graph operator
$H_{\alpha_{n}}$ (see Figure \ref{fig: periodic-approx}).

\begin{figure}
\includegraphics[scale=0.5]{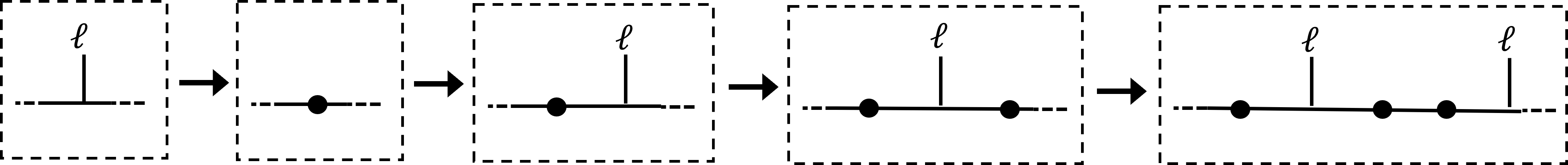}

\caption[Periodic approximations for Fibonacci comb.]{Fundamental domains for the sequence of periodic approximations for
the Fibonacci comb with $\alpha=\frac{\sqrt{5}-1}{2}$. The $n$th
fundamental domain contains $F_{n}$ decorations, where $F_{n}$ is
the $n$th Fibonacci number.  \label{fig: periodic-approx}}
\end{figure}

A key advantage of this example is that the transfer matrices introduced
in Subsection \ref{subsec:Transfer-matrices} can be computed explicitly
and used to study the periodic approximations. Through direct computation,
using the explicit solutions of the eigenvalue equation on each edge
together with the Kirchhoff vertex conditions, one obtains the following
expression for the one-step transfer matrix of the possible tiles:
\begin{equation}
\mathcal{M}_{0}\left(E=k^{2}\right)=\left(\begin{array}{cc}
\cos\left(kL\right) & \frac{1}{k}\sin\left(kL\right)\\
-\frac{1}{k}\sin\left(kL\right) & \cos\left(kL\right)
\end{array}\right),\label{eq:fibo-scatter-2}
\end{equation}
for the trivial tile, and
\begin{align}
 & \mathcal{M}_{1}\left(E=k^{2}\right)\nonumber \\
 & =\left(\begin{array}{cc}
\cos\left(2kL\right)-\frac{1}{2}\sin\left(2kL\right)\tan\left(k\ell\right) & -k\sin\left(2kL\right)-k\cos\left(kL\right)^{2}\tan\left(k\ell\right)\\
\frac{1}{k}\sin\left(2kL\right)-\frac{1}{k}\sin^{2}\left(kL\right)\tan\left(k\ell\right) & \cos\left(2kL\right)-\frac{1}{2}\sin\left(2kL\right)\tan\left(k\ell\right)
\end{array}\right),\label{eq:Fibo-scatter}
\end{align}
for the nontrivial tile. Note that the cocycle $\mathcal{M}\left(k^{2}\right)$
is well-defined and analytic for all $k$, except for a discrete subset
of energies, as in Definition \ref{def:Bad-energies}:
\begin{equation}
\B_{\ell}:=\left\{ \left(\frac{\pi/2+\pi m}{\ell}\right)^{2}:m\in\Z\right\} .\label{eq:B-comb}
\end{equation}
To make our analysis more convenient, we assume from now on that these
``forbidden'' values of $k^{2}$ are excluded from the spectrum,
similar to the Baire-generic assumption in Theorem \ref{thm:Generic-assumption}:
\begin{assumption}
\label{assu:bad-energies}$\spec{H_{\alpha}}\cap\B_{\ell}=\varnothing$.
\end{assumption}

By Theorem \ref{thm:Generic-assumption}, this assumption holds for
a Baire-generic choice of the decoration length $\ell$. Under this
generic assumption, not only does Theorem \ref{thm:TMP} ensure that
$\spec{H_{\alpha}}$ is of zero Lebesgue measure, but the standard
Gordon lemma (see, e.g., \cite[prop. 2]{Suetoe1987}) implies that
$\spec{H_{\alpha}}$ contains no eigenvalues, and is thus a generalized
Cantor set and the spectral measures are purely singular continuous.
If one omits Assumption \ref{assu:bad-energies}, then $\spec{H_{\alpha}}$
may contain flat bands at the spectral points where $\mw{}$ does
not exist (and in particular, will not be a Cantor set). These flat
bands lead to jump discontinuities of the IDS, which are analyzed
explicitly in the next subsection.

Denoting the Sturmian sequence corresponding to $\theta=0$ by $\omega_{\alpha}\left(n\right)$,
we define the following sequence of transfer matrices:
\begin{equation}
\mathcal{M}_{n}^{\alpha}\left(E\right)=\mathcal{M}_{q_{n}}\left(\omega_{\alpha},E\right),\label{eq:mqn}
\end{equation}
where $\alpha_{n}=:p_{n}/q_{n}$. The following identities may be
proven from analysis of the associated Sturmian sequence using arguments
similar to \cite{Bellissard1989}:
\begin{prop}
\label{prop:Sturmian-properties}Define the sequence $\left(t_{n}^{\alpha}\left(E\right)\right)_{n=1}^{\infty}$
by 
\begin{equation}
t_{n}^{\alpha}\left(E\right)=tr\left(\mathcal{M}_{n}^{\alpha}\left(E\right)\right).\label{eq:tnE}
\end{equation}
Then the following recursion relation for $\left(t_{n}^{\alpha}\left(E\right)\right)_{n=1}^{\infty}$
holds:
\begin{align}
t_{n+1}^{\alpha}\left(E\right)= & \frac{C_{a_{n+1}-1}\left(t_{n}^{\alpha}\left(E\right)\right)C_{a_{n}}\left(t_{n-1}^{\alpha}\left(E\right)\right)}{C_{a_{n}-1}\left(t_{n-1}^{\alpha}\left(E\right)\right)}t_{n}^{\alpha}\left(E\right)\nonumber \\
 & -C_{a_{n+1}-2}\left(t_{n}^{\alpha}\left(E\right)\right)t_{n-1}^{\alpha}\left(E\right)-\frac{C_{a_{n+1}-1}\left(t_{n}^{\alpha}\left(E\right)\right)}{C_{a_{n}-1}\left(t_{n-1}^{\alpha}\left(E\right)\right)}t_{n-2}^{\alpha}\left(E\right),\label{eq:recursion}
\end{align}
where $C_{m}\left(x\right)$ are the Chebyshev polynomials:
\begin{align}
 & C_{-1}\left(x\right)=0,\,\,C_{0}\left(x\right)=1,\,\,C_{1}\left(x\right)=x,\label{eq:Chebysev1}\\
 & C_{m}\left(x\right)=C_{m-1}\left(x\right)x-C_{m-2}\left(x\right).\label{eq:Chebyshev2}
\end{align}
\end{prop}

With the proposition above, we can apply the methods presented in
\cite{Bellissard1989,Suetoe1987} to prove the following:
\begin{thm}
\label{thm:Per-approx}We have
\begin{equation}
\spec{H_{\alpha}}=\bigcap_{n\in\N}\left(\spec{H_{\alpha_{n}}}\cup\spec{H_{\alpha_{n+1}}}\right).\label{eq:approx1}
\end{equation}
\end{thm}

\begin{proof}
For completeness, we outline the main steps in the proof. First, one
shows that
\begin{equation}
\spec{H_{\alpha}}=\left\{ E\in\R:t_{n}^{\alpha}\left(E\right)\text{ is bounded}\right\} ,\label{eq:B-infty}
\end{equation}
which can be done by constructing an appropriate Weyl sequence. Using
Proposition \ref{prop:Sturmian-properties}, one can then show that
the sequence $t_{n}^{\alpha}\left(E\right)$ is unbounded if and only
if $\left|t_{n}^{\alpha}\left(E\right)\right|,\left|t_{n+1}^{\alpha}\left(E\right)\right|$
are both larger than $2$ for some $n$. On the other hand, standard
Bloch-Floquet theory shows that
\begin{equation}
E\in\spec{H_{\alpha_{n}}}\iff\left|t_{n}^{\alpha}\left(E\right)\right|\leq2.\label{eq:Bloch}
\end{equation}
Denoting the resolvent set of an operator $H$ by $\rho\left(H\right)$,
we overall get
\begin{align}
 & \rho\left(H_{\alpha}\right)=\left\{ E:t_{n}^{\alpha}\left(E\right)\text{ is unbounded}\right\} \nonumber \\
 & =\left\{ E:\left|t_{n}^{\alpha}\left(E\right)\right|,\left|t_{n+1}^{\alpha}\left(E\right)\right|>2\text{ for some \ensuremath{n}}\right\} \nonumber \\
 & =\bigcup_{n\in\N}\left(\rho\left(H_{\alpha_{n}}\right)\cap\rho\left(H_{\alpha_{n+1}}\right)\right),\label{eq:rho-est}
\end{align}
completing the proof.
\end{proof}
Using the result in \cite[thm. 2.5.1]{Beckus2016a}, one can even
prove a stronger notion of periodic approximations:
\begin{thm}
Denoting $\Sigma_{\alpha_{n}}:=\spec{H_{\alpha_{n}}}\cup\spec{H_{\alpha_{n+1}}}$,
then
\begin{equation}
\Sigma_{\alpha_{n}}\rightarrow\spec{H_{\alpha}}\label{eq:Hausdorff-conv}
\end{equation}
in the Hausdorff topology.
\end{thm}

This result is noteworthy, considering the fact that $\Sigma_{\alpha_{n}}$
are not compact sets (albeit closed).
\begin{proof}
By \cite[thm. 2.5.1]{Beckus2016a}, it suffices to verify that
\begin{equation}
\spec{R\left(H_{\alpha_{n}},z\right)}\cup\spec{R\left(H_{\alpha_{n+1}},z\right)}\rightarrow\spec{R\left(H_{\alpha},z\right)},\label{eq:res-convergence}
\end{equation}
in the Hausdorff topology, where $R\left(H,z\right)$ is the resolvent
of $H$ at some $z\in\C\backslash\R$. By Theorem \ref{thm:Per-approx},
\begin{equation}
\spec{H_{\alpha}}=\bigcap_{n\in\N}\left(\spec{H_{\alpha_{n}}}\cup\spec{H_{\alpha_{n+1}}}\right).\label{eq:approx1-2}
\end{equation}
Applying the spectral mapping theorem, and using injectivity of $\lambda\mapsto\left(\lambda-z\right)^{-1}$
to pass this map through the intersection, we obtain
\begin{equation}
\spec{R\left(H_{\alpha},z\right)}=\bigcap_{n\in\N}\left(\spec{R\left(H_{\alpha_{n}},z\right)}\cup\spec{R\left(H_{\alpha_{n+1}},z\right)}\right).\label{eq:resolvent-approx}
\end{equation}
Since the sets $\spec{R\left(H_{\bullet},z\right)}$ are all compact
for $\bullet\in\left\{ \alpha,\alpha_{n},\alpha_{n+1}\right\} $,
Equation \ref{eq:resolvent-approx} combined with a standard topological
argument then shows that (\ref{eq:res-convergence}) indeed holds.
\end{proof}
Figure \ref{fig: periodic-approx-2} demonstrates how the sequence
$\Sigma_{\alpha_{n}}$ converges to the limiting set $\spec{H_{\alpha}}$.
It also illustrates the Cantor structure of the spectrum and its zero
Lebesgue measure. Figures \ref{fig: butterfly-1} and \ref{fig: butterfly}
display the fractal nature of the spectrum $\spec{H_{\alpha}}$ as
a function of $\alpha$; this can be viewed as a quantum-graph analogue
of the Kohmoto butterfly (cf. \cite{Band2024,band2026complete,Kohmoto1983,Kohmoto1984}).
We also refer to \cite{Baradaran2023a}, where a plot similar to \ref{fig: butterfly}
was produced for a Harper-like quantum graph model, showing resemblance
to the Hofstadter butterfly \cite{Hofstadter1976}.

\begin{figure}
\includegraphics[scale=0.55]{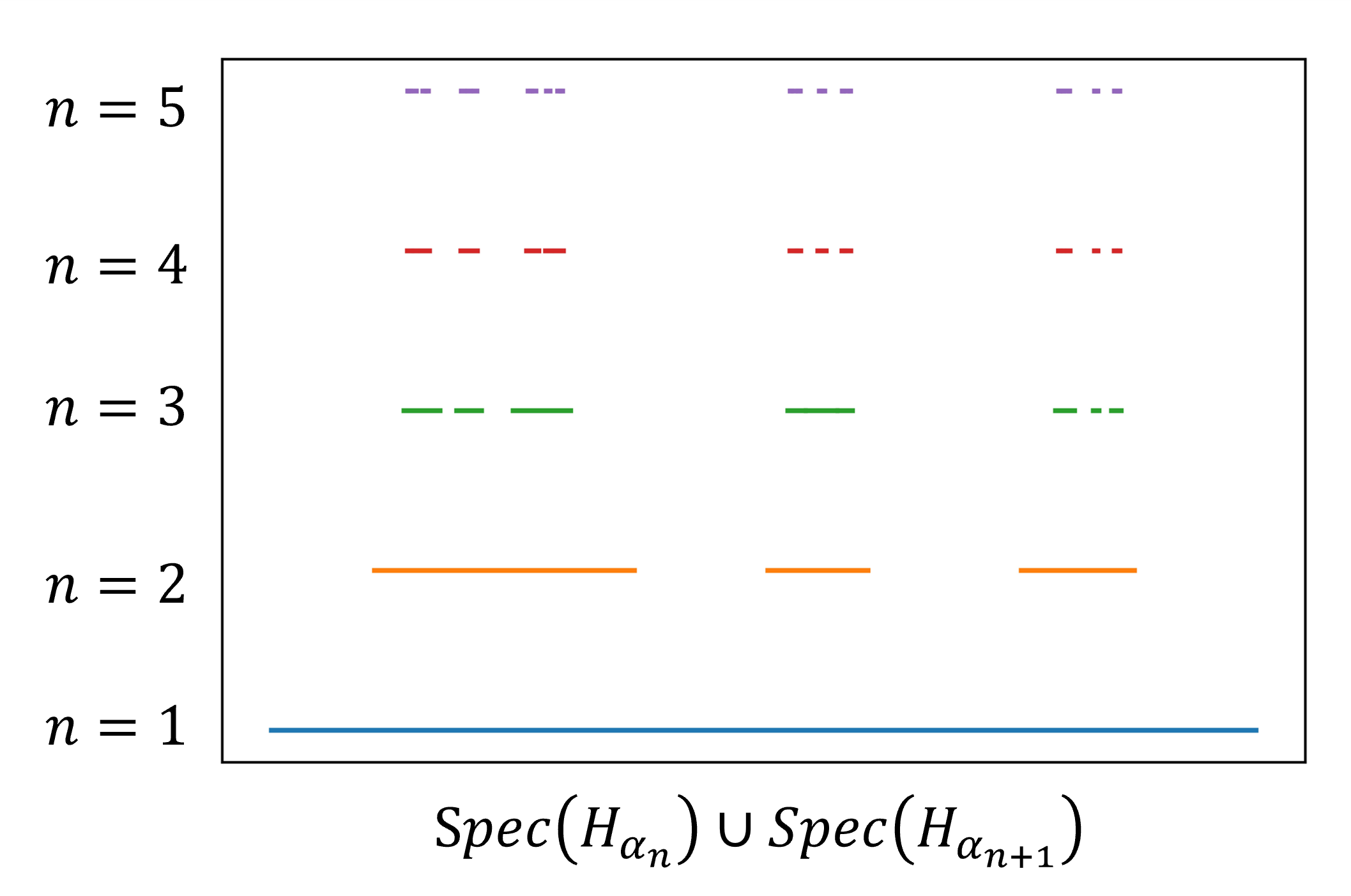}

\caption[Spectra of periodic approximations for Fibonacci comb.]{Parts of the sequence of periodic approximations $\protect\spec{H_{\alpha_{n}}}\cup\protect\spec{H_{\alpha_{n+1}}}$
for the spectrum $\protect\spec{H_{\alpha}}$ of the Fibonacci comb.
\label{fig: periodic-approx-2}}
\end{figure}

\begin{figure}
\includegraphics[scale=0.75]{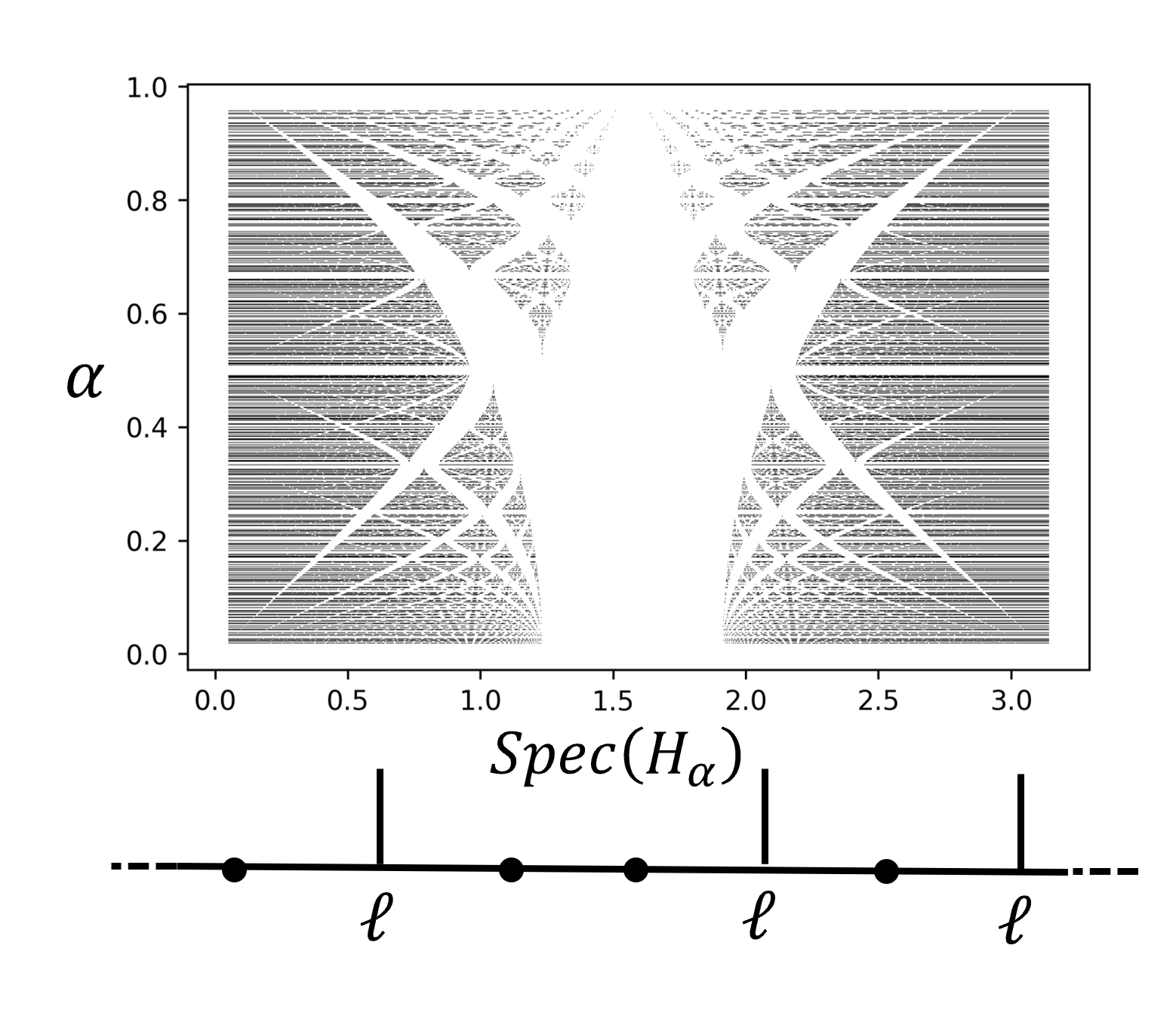}

\caption[Kohmoto butterfly for Sturmian comb.]{The \textquotedblleft Kohmoto butterfly\textquotedblright{} for Sturmian
comb graphs -- parts of $\protect\spec{H_{\alpha}}$ as a function
of $\alpha$. The horizontal axis represents $k\in\protect\R$ such
that $k^{2}\in\protect\spec{H_{\alpha}}$.  \label{fig: butterfly-1}}
\end{figure}

\begin{figure}
\includegraphics[scale=0.75]{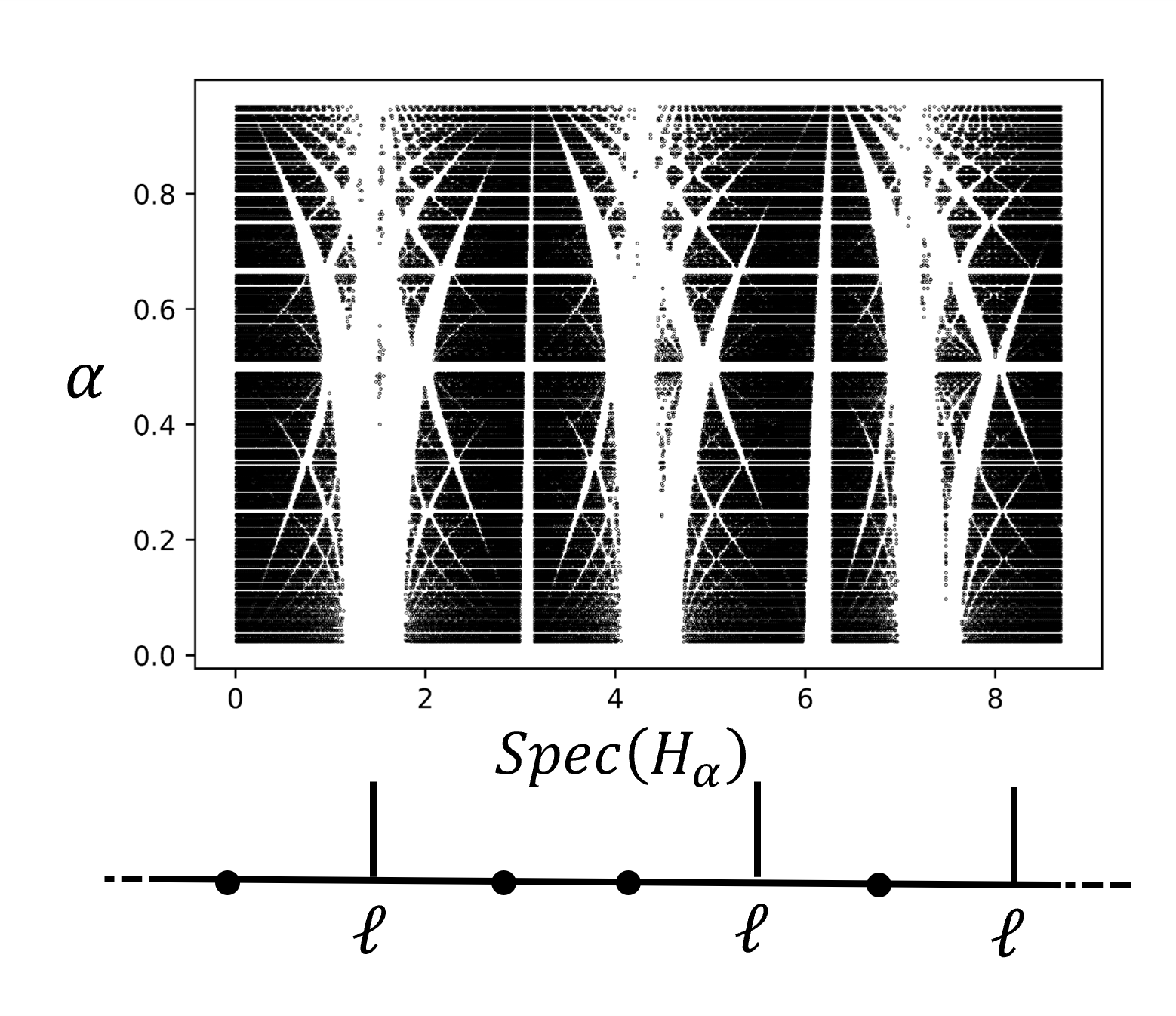}

\caption[Extended Kohmoto butterfly for Sturmian comb.]{The \textquotedblleft Kohmoto butterfly\textquotedblright{} for Sturmian
combs, showing the fractal structure of $\protect\spec{H_{\alpha}}$
changing gradually as the energy increases. \label{fig: butterfly}}
\end{figure}

\subsection{Jump discontinuities of the IDS and closed gaps\label{subsec:discont-IDS}}

In preparation for analyzing the Dry Ten Martini Problem in the following
sections, we now turn to a key obstruction to realizing all gap labels
predicted by Theorem \ref{thm:GLT}. For Sturmian combs, this obstruction
comes from jump discontinuities of the IDS, which force intervals
of allowed labels to remain unattained. Since the predicted gap labels
form a dense set, any discontinuity of the IDS forces the existence
of labels that are not realized (also known as closed gaps). We illustrate
this by completely analyzing the IDS jumps for metric Sturmian combs.
The necessary and sufficient conditions for IDS discontinuities of
these graphs are given in Theorem~\ref{thm: IDS jumps Sturmian combs}.
In addition, the theorem explicitly states all the energies at which
such discontinuities occur and the size of the IDS jump at those energies.
\begin{thm}
\label{thm: IDS jumps Sturmian combs} Let $\alpha\in\left(0,1\right)\backslash\Q$,
written as the following infinite continued fraction:
\begin{equation}
\alpha=\frac{1}{a_{1}+\frac{1}{a_{2}+...}}.\label{eq:cf-expansion-2-1}
\end{equation}
and let $\Omega_{\alpha}$ be the corresponding Sturmian subshift.
Let $\Gamma_{\Omega_{\alpha}}$ be the associated family of metric
Sturmian comb graphs with tooth length $\ell$ and spacing $L$. Then
the IDS $\NHE{H_{\Omega_{\alpha}}}$ has discontinuities if and only
if one of the following holds:
\begin{enumerate}
\item $\frac{\ell}{L}=\frac{2m+1}{2n}a_{1}$ for some $m,n\in\N$. \\
In this case, the IDS is discontinuous at $E=\left(\frac{\pi n}{La_{1}}\right)^{2}$,
and the associated jump in the IDS value is
\begin{equation}
\Delta N_{H_{\Omega_{\alpha}}}\left(E\right)=\frac{(a_{1}+1)\alpha-1}{L+\alpha\ell},\label{eq:smalljump}
\end{equation}
or\\
~
\item $\frac{\ell}{L}=\frac{2m+1}{2n}(a_{1}+1)$ for some $m,n\in\N$.\\
In this case the IDS is discontinuous at $E=\left(\frac{\pi n}{L\left(a_{1}+1\right)}\right)^{2}$,
and the associated jump in the IDS value is
\begin{equation}
\Delta N_{H_{\Omega_{\alpha}}}\left(E\right)=\frac{1-a_{1}\alpha}{L+\alpha\ell}.\label{eq:bigjump}
\end{equation}
\end{enumerate}
If both conditions on $\ell/L$ above hold simultaneously, i.e., $\frac{\ell}{L}=\frac{2m_{1}+1}{2n_{1}}a_{1}=\frac{2m_{2}+1}{2n_{2}}(a_{1}+1)$
for $m_{1},n_{1},m_{2},n_{2}\in\N$, then the IDS is discontinuous
at $E=\left(\frac{\pi n_{1}}{La_{1}}\right)^{2}=\left(\frac{\pi n_{2}}{L\left(a_{1}+1\right)}\right)^{2}$,
and the associated jump in the IDS value is the sum of (\ref{eq:smalljump})
and (\ref{eq:bigjump}), i.e., 
\begin{equation}
\Delta N_{H_{\Omega_{\alpha}}}\left(E\right)=\frac{\alpha}{L+\alpha\ell}.\label{eq: sumjump}
\end{equation}

\end{thm}

\begin{rem*}
Note that if either case in the theorem occurs, it holds for infinitely
many pairs $\left(m,n\right)$, hence the IDS has jumps at infinitely
many energies.
\end{rem*}
This completely characterizes the IDS discontinuities for this model,
which we show are caused by compactly supported eigenfunctions. Two
intriguing recent works \cite{Damanik2023a,Schirmann2024} explore
IDS discontinuities in aperiodic discrete graphs, which are also due
to compactly supported eigenfunctions. Some fundamental results on
this phenomenon for random operators on aperiodic discrete graphs
appeared already in \cite{Klassert2003}. Similar phenomena are also
observed in periodic graphs models, as was analyzed for discrete graphs
\cite{Peyerimhoff2021} and metric graphs \cite{Lenz2009} (see also
\cite{Peyerimhoff2021} where continuous models are confronted with
metric and discrete graphs). The most recent works on the IDS of quantum
graphs and their discontinuities appear in \cite{Breuer2026,Levi2026},
where periodic metric trees are analyzed.

\ 

To prove Theorem\ \ref{thm: IDS jumps Sturmian combs} we need two
lemmas. Lemma\ \ref{lem: compactly supported between teeth} shows
that all the compactly supported eigenfunctions of $\Gamma_{\omega}$
are supported on specific subgraphs. These subgraphs are associated
with particular subwords of $\omega\in\Omega_{\alpha}$ and Lemma\ \ref{lem:Frequencies}
expresses the frequencies of these subwords.
\begin{lem}
\label{lem: compactly supported between teeth} Let $\alpha\in\left(0,1\right)\backslash\Q$,
$\omega\in\Omega_{\alpha}$ and $E\in\R$. A compactly supported $E$-eigenfunction
of the Sturmian comb $\Gamma_{\omega}$ exists if and only if there
exists an $E$-eigenfunction which is supported between two adjacent
teeth in $\Gamma_{\omega}$.
\end{lem}

\begin{proof}
One direction is trivial. For the converse, let $f$ be a compactly
supported solution to $-\frac{d^{2}f}{dx^{2}}=Ef$ on $\Gamma_{\omega}$,
satisfying Neumann-Kirchhoff vertex conditions. Since $f$ is compactly
supported, choose the two outermost teeth on which it is supported,
and get that $f$ must vanish at the base of each of these two teeth
(i.e., the vertex which connects them to the $\Z$-graph). Furthermore,
the Neumann condition $f'=0$ at the tooth boundary vertex gives the
form $f\left(x\right)=C\cdot\cos\left(k\left(x-\ell\right)\right)$
on each tooth, where $k:=\sqrt{E}$ and $\ell$ is the tooth length.
Since $f$ must also vanish at the base of the tooth, we deduce that
$k\ell=\frac{\pi}{2}+\pi m$ for some $m\in\mathbb{N}$. The cosine
form of $f$ on the teeth then implies that $f$ in fact vanishes
at the base of all teeth of the comb. Now, choose a (horizontal) path
$\tilde{e}$ between two adjacent teeth $e_{1},e_{2},$ such that
$f$ does not identically vanish on this path. At the bases of these
teeth $v_{1},v_{2}$ we have that $f\left(v_{1}\right)=f\left(v_{2}\right)=0$,
by the argument given above. Construct a new eigenfunction $\tilde{f}$
as follows:

1. At the horizontal path set $\tilde{f}=f$.

2. At the two mentioned teeth $e_{1},e_{2}$, set $\left.\tilde{f}\right|_{e_{i}}\left(x\right)=A_{i}\sin\left(\left(\frac{\frac{\pi}{2}+\pi m}{\ell}\right)x\right)$
for $i\in\left\{ 1,2\right\} $. Choose $A_{1},A_{2}$ so that $\left.\tilde{f}'\right|_{e_{i}}(v_{i})+\left.\tilde{f}'\right|_{\tilde{e}}(v_{i})=0$.

3. Extend $\tilde{f}$ to be identically $0$ everywhere else.

The resulting function is an $E$-eigenfunction  supported between
two adjacent teeth of $\Gamma_{\omega}$.
\end{proof}
Towards the next lemma, we define the frequency of a subword as follows.
Let $\Omega_{\alpha}$ be a Sturmian subshift, and $W=W_{1}...W_{k}$
a finite subword over the alphabet $\A=\{0,1\}$. Let $\omega\in\Omega_{\alpha}$,
and denote
\begin{equation}
\nu_{W}:=\lim_{N\rightarrow\infty}\frac{\#\set{n\in\left\{ 0,...,N-1\right\} }{\left.\omega\right|_{\left[n,n+k-1\right]}=W}}{N}.\label{eq:word-freq-1}
\end{equation}
By unique ergodicity, this limit exists uniformly in $\omega\in\Omega_{\alpha}$
and is independent of $\omega$. Equivalently, $\nu_{W}=\mu\left(V_{W}\right)$,
where $\mu$ is the unique invariant measure and $V_{W}$ is the cylinder
set corresponding to $W$ (see Subsection \ref{subsec:Dynamics} and
\cite[prop. 4.4]{Baake2013}). We therefore refer to $\nu_{W}$\LyXZeroWidthSpace{}
as the frequency with which $W$ occurs in the subshift $\Omega_{\alpha}$.
For Sturmian subshifts, these frequencies can be computed explicitly
from the irrational rotation representation, as used in the proof
below.
\begin{lem}
\label{lem:Frequencies}Let $\alpha\in\left(0,1\right)\backslash\Q$
with the continued fraction expansion (\ref{eq:cf-expansion-2-1}).
Then there exist only two subwords of the form 
\begin{equation}
W=1\underset{k}{\underbrace{0....0}}1\label{eq:Word}
\end{equation}
which occur in the subshift $\Omega_{\alpha}$:
\end{lem}

\begin{enumerate}
\item A subword $W$ with $k=a_{1}$ zeros, which appear with frequency
$1-a_{1}\alpha$ in $\Omega_{\alpha}$.
\item A subword $W$ with $k=a_{1}-1$ zeros, which appears with frequency
$(a_{1}+1)\alpha-1$.
\end{enumerate}
\begin{proof}
Given a finite word $W$ we consider the following subset of $S^{1}$:
\begin{equation}
I_{W}:=\set{\theta\in S^{1}}{\left.\omega_{\alpha,\theta}\right|_{\left[0,...,\left|W\right|-1\right]}=W},\label{eq:Iw}
\end{equation}
where $\omega_{\alpha,\theta}(n):=\chi_{[1-\alpha,1)}\left(n\alpha+\theta\text{ mod \ensuremath{1}}\right)$
is a Sturmian (infinite) word such that $\omega_{\alpha,\theta}\in\Omega_{\alpha}$.
By \cite[sec. 2.2.3]{Lothaire2002} (see also \cite[sec. 5]{Baake2024}),
the frequency of the subword $W$ in $\Omega_{\alpha}$ is equal to
the Lebesgue measure of $I_{W}$. We therefore compute the Lebesgue
measure $I_{W}$ for all admissible subwords of the form $W=10....01$.
We accompany the proof with Figure~\ref{fig: subword}.
\begin{figure}
\includegraphics[scale=0.48]{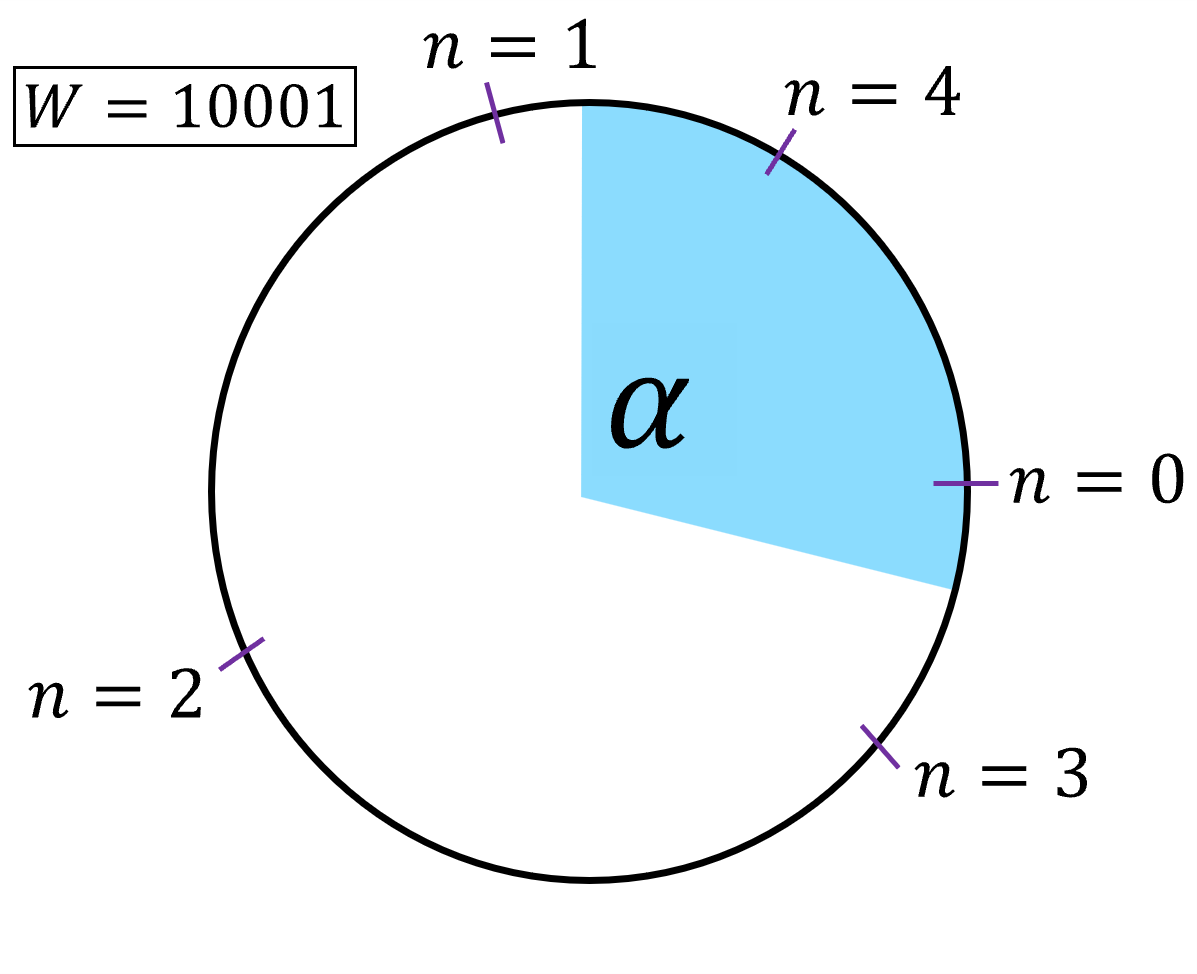}

\caption[Sturmian subword as irrational rotation.]{The subword $W=10001$ appearing inside the Sturmian word with $\alpha\approx0.29$.
Taking the initial angle to be $\theta=1-\alpha+\varepsilon$, the
first subword of length $5$ of the associated Sturmian sequence (\ref{eq:sturm})
is as displayed in the figure.  \label{fig: subword}}
\end{figure}

First, note that 
\begin{equation}
\frac{1}{a_{1}+1}<\alpha<\frac{1}{a_{1}}.\label{eq:alpha-c1-bounds}
\end{equation}
By the definition of the sequence $\omega_{\alpha,\theta}$ we have
that $\omega_{\alpha,\theta}(0)=1$ iff $\theta\in[1-\alpha,1)$.
By (\ref{eq:alpha-c1-bounds}) we get that $n\alpha+\theta\in\left[\frac{n-1}{a_{1}+1}+1,\frac{n}{a_{1}}+1\right)$
for all $\theta\in[1-\alpha,1)$. In particular we get that $n\alpha+\theta\mod 1\in\left[0,1-\alpha\right)$,
for all $1\leq n\leq a_{1}-1$. We conclude the argument above by
\begin{equation}
\omega_{\alpha,\theta}(0)=1\quad\Leftrightarrow\quad\theta\in(1-\alpha,1]\quad\Leftrightarrow\quad\left.\omega_{\alpha,\theta}\right|_{\left[0,...,a_{1}-1\right]}=1\underset{a_{1}-1}{\underbrace{0....0}}.\label{eq:given-word}
\end{equation}
As we wish that $\omega_{\alpha,\theta}(0)=W(0)=1$, we may assume
the above equivalent conditions and split into two cases:
\begin{enumerate}
\item Assume $\omega_{\alpha,\theta}(a_{1})=0$, which is equivalent to
$a_{1}\alpha+\theta\mod 1\in[0,1-\alpha)$. In addition, from $\theta\in[1-\alpha,1)$
we have $a_{1}\alpha+\theta\in[1+(a_{1}-1)\alpha,1+a_{1}\alpha)$
and so $a_{1}\alpha+\theta\mod 1\in[(a_{1}-1)\alpha,a_{1}\alpha)$.
Intersecting both intervals gives $a_{1}\alpha+\theta\mod 1\in$ $[(a_{1}-1)\alpha,1-\alpha)$.
From here we get $(a_{1}+1)\alpha+\theta\mod 1\in[a_{1}\alpha,1)$,
and this implies $\omega_{\alpha,\theta}(a_{1}+1)=1$. Concluding
we get that in this case
\begin{equation}
\left.\omega_{\alpha,\theta}\right|_{\left[0,...,a_{1}+1\right]}=W=1\underset{a_{1}}{\underbrace{0....0}}1.\label{eq:given-word-1}
\end{equation}
We need also to know the range of $\theta$ for this case, namely
what is $I_{W}$ for the subword $W$ above. In the current case,
we got $a_{1}\alpha+\theta\mod 1\in$ $[(a_{1}-1)\alpha,1-\alpha)$.
This means that $\theta\in[-\alpha,1-(a_{1}+1)\alpha)$. The Lebesgue
measure of this interval is $1-a_{1}\alpha$, which is the frequency
of the word $W$ above.
\item Assume $\omega_{\alpha,\theta}(a_{1})=1$, which is equivalent to
$a_{1}\alpha+\theta\mod 1\in[1-\alpha,1)$. Repeating the arguments
as in the case above we get that $I_{W}=[2-(a_{1}+1)\alpha,1)$ for
$W=1\underset{a_{1}-1}{\underbrace{0....0}}1$. The Lebesgue measure
of this interval is $(a_{1}+1)\alpha-1$, which is the frequency of
that word.
\end{enumerate}
The two cases above exhaust all subwords of the form $W=10....01$
occurring in $\Omega_{\alpha}$.
\end{proof}
\begin{proof}[Proof of Theorem \ref{thm: IDS jumps Sturmian combs}]
 We start the proof by referring to Corollary \ref{cor:freq-jump},
whose hypothesis holds because all finite subwords of a Sturmian subshift
have positive frequency (see beginning of proof of Lemma \ref{lem:Frequencies},
or similar arguments in \cite[sec. 2.2.3]{Lothaire2002} and \cite[sec. 5]{Baake2024}).
We conclude from Corollary \ref{cor:freq-jump} that $N_{\Omega_{\alpha}}$
has a jump discontinuity at energy $E$ if and only if there exists
a compactly supported $E$-eigenfunction. By Lemma~\ref{lem: compactly supported between teeth}
compactly supported $E$-eigenfunctions exist precisely when there
is one supported between two adjacent teeth of the graph, see Figure
\ref{fig:cpt-efun}. 
\begin{figure}
\includegraphics[scale=0.7]{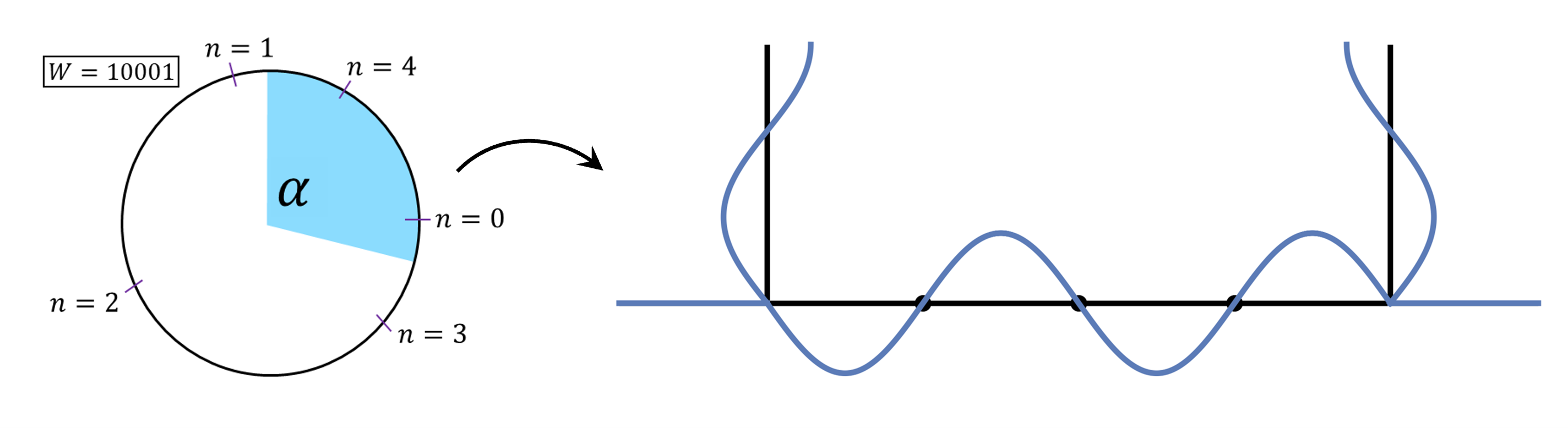}

\caption[Compactly supported eigenfunction arising from Sturmian subword.]{Illustration of how subwords of the form $W=10...01$ give rise to
a compactly supported eigenfunction.  \label{fig:cpt-efun}}
\end{figure}

Let $f$ be an $E$-eigenfunction which is supported between adjacent
teeth and denote $k:=\sqrt{E}$. The following holds:

(1) By the proof of Lemma\ \ref{lem: compactly supported between teeth},
$f$ vanishes at the base of each tooth, and $k\ell=\frac{\pi}{2}+\pi m$
for some $m\in\mathbb{N}$.

(2) By Lemma~\ref{lem:Frequencies}, the (horizontal) distance between
adjacent teeth in $\Gamma_{\omega}$ is either $a_{1}L$ or $(a_{1}+1)L$.
By (1) above, this implies that $ka_{1}L=\pi n$ or $k\left(a_{1}+1\right)L=\pi n$
for $n\in\N$. It may be that both equalities hold for the same value
of $k$, but with different $n$ values.

We now examine the two cases in (2). First consider $ka_{1}L=\pi n$.
Combining this with the $k\ell=\frac{\pi}{2}+\pi m$, translates into
the following condition:
\begin{equation}
\frac{\ell}{L}=\frac{\left(2m+1\right)a_{1}}{2n},\,m,n\in\N,\label{eq:lL1}
\end{equation}
and the corresponding eigenvalue is $E=k^{2}=\left(\frac{\pi n}{La_{1}}\right)^{2}$.

Similarly, the case $k\left(a_{1}+1\right)L=\pi n$ translates into
the following condition:
\begin{equation}
\frac{\ell}{L}=\frac{\left(2m+1\right)\left(a_{1}+1\right)}{2n},\,m,n\in\N,\label{eq:lL2}
\end{equation}
and the corresponding eigenvalue is $E=k^{2}=\left(\frac{\pi n}{L(a_{1}+1)}\right)^{2}$.

These are exactly the two possible conditions on $\ell/L$ and the
corresponding energies in the theorem. It remains to compute the jump
size. Using (\ref{eq:truncation-IDS}), we count the number of compactly
supported eigenfunctions for the finite truncations $\left.H_{\omega}\right|_{[0,N]}$.
For the energy $E=\left(\frac{\pi n}{La_{1}}\right)^{2}$, we need
to consider the subword $W=1\underset{a_{1}-1}{\underbrace{0....0}}1$
and the eigenfunctions supported on the corresponding subgraphs. These
eigenfunctions are linearly independent. Conversely, by Lemma \ref{lem: compactly supported between teeth},
every compactly supported $E$-eigenfunction vanishes at the base
of every tooth, and therefore decomposes into eigenfunctions supported
between consecutive teeth, and such a component is nonzero only when
the corresponding subword is $W$. Hence these eigenfunctions span
the $E$-eigenspace, and so
\begin{equation}
\dim\ker\left(\left.H_{\omega}\right|_{[0,N]}-E\right)=\#\set{j\in\left\{ 0,...,N-a_{1}\right\} }{\left.\omega\right|_{[j,j+a_{1}]}=W}.\label{eq:dimker}
\end{equation}
 The jump in the IDS at $E$ is given by 
\begin{align}
 & \Delta N_{H_{\Omega_{\alpha}}}\left(E\right)\nonumber \\
 & =\lim_{N\rightarrow\infty}\frac{\#\set{\lambda\in\spec{\left.H_{\alpha}\right|_{[0,N]}}}{\lambda\leq E}-\#\set{\lambda\in\spec{\left.H_{\alpha}\right|_{[0,N]}}}{\lambda<E}}{\left|\left.\Gamma_{\alpha}\right|_{[0,N]}\right|}\nonumber \\
 & =\lim_{N\rightarrow\infty}\frac{\dim\ker\left(\left.H_{\alpha}\right|_{[0,N]}-E\right)}{\left|\left.\Gamma_{\alpha}\right|_{[0,N]}\right|}\nonumber \\
 & =\lim_{N\rightarrow\infty}\frac{\#\set{j\in\left\{ 0,...,N-c_{1}-1\right\} }{\left.\omega_{\alpha}\right|_{[j,j+c_{1}+1]}=W}}{\left|\left.\Gamma_{\alpha}\right|_{[0,N]}\right|}\nonumber \\
 & =\lim_{N\rightarrow\infty}\frac{\left(N-c_{1}\right)\nu_{W}}{NL+\alpha N\ell}=\frac{1-c_{1}\alpha}{L+\alpha\ell},\label{eq:jump-computation}
\end{align}
where the last line is obtained by Lemma \ref{lem:Frequencies} according
to which $\nu_{W}=(a_{1}+1)\alpha-1$(see also the definition of word
frequency, (\ref{eq:word-freq-1})).

Repeating the same computation for the energy $E=\left(\frac{\pi n}{L(a_{1}+1)}\right)^{2}$
whose eigenfunctions correspond to the subword $W=1\underset{a_{1}}{\underbrace{0....0}}1$.
The only change which is required in the computation is in using the
word frequency which is now $\nu_{W}=1-a_{1}\alpha$, and we get 
\begin{align}
\Delta N_{H_{\Omega_{\alpha}}}\left(E\right) & =\frac{1-a_{1}\alpha}{L+\alpha\cdot\ell}.\label{eq:jump-computation-2}
\end{align}
It may happen that both (\ref{eq:lL1}) and (\ref{eq:lL2}) hold (but
for different $n,m$ values). Namely, 
\begin{equation}
\frac{\ell}{L}=\frac{\left(2m_{1}+1\right)a_{1}}{2n_{1}}=\frac{\left(2m_{2}+1\right)\left(a_{1}+1\right)}{2n_{2}}\label{eq:l/L}
\end{equation}
 for some $m_{1},n_{1},m_{2},n_{2}\in\N$. The corresponding energy
is then $E=\left(\frac{\pi n_{1}}{La_{1}}\right)^{2}=\left(\frac{\pi n_{2}}{L(a_{1}+1)}\right)^{2}$
and the associated eigenfunctions are supported on subgraphs corresponding
to both subwords $1\underset{a_{1}-1}{\underbrace{0\cdot....\cdot0}}1$
and $1\underset{a_{1}}{\underbrace{0\cdot....\cdot0}}1$. These eigenfunctions
are linearly independent and so the dimensions of the corresponding
eigenspaces sum up (and the same holds for the frequencies). Therefore,
the IDS jump at such energies is the sum of (\ref{eq:jump-computation})
and (\ref{eq:jump-computation-2}),
\begin{equation}
\Delta N_{H_{\Omega_{\alpha}}}\left(E\right)=\frac{(a_{1}+1)\alpha-1}{L+\alpha\cdot\ell}+\frac{1-a_{1}\alpha}{L+\alpha\cdot\ell}=\frac{\alpha}{L+\alpha\cdot\ell}.\label{eq:DN}
\end{equation}
\end{proof}
This shows explicitly how local geometry can remove gap labels allowed
by the GLT for this example.

\newpage{}

\section{Dry Ten Martini Problem -- preliminaries \label{sec:DTMP}}

\subsection{Introduction}

We now begin the analysis of the DTMP for discrete Sturmian decorated
$\Z$-graphs. In the classical Sturmian setting, the DTMP asks whether
every gap label permitted by the GLT is realized by an open spectral
gap. For Sturmian Hamiltonians on $\Z$, the answer is affirmative:
all allowed gaps are open \cite{Band2024}. A central part of the
proof relies on periodic approximations for Sturmian Hamiltonians.
More precisely, one describes the relative combinatorial structure
of their spectral bands through an associated spectral tree, identifies
infinite paths along this tree with spectral points, and then uses
this identification to determine which gap labels are realized.

Our goal in Sections \ref{sec:DTMP}--\ref{sec:DTMP-computation}
is to analyze the DTMP for discrete Sturmian decorated $\Z$-graphs
using similar methods: we study the band structure of the periodic
approximants via transfer matrices, encode it through a spectral tree,
and read the gap labels from infinite paths. The main difference is
that the effective transfer matrices depend on the energy, and may
have poles or degenerate at isolated energies. Away from these special
energies, we show that the periodic approximants exhibit the same
Sturmian combinatorial structure as in the classical case, while the
singular energies are responsible for the possible failure of gap
opening.

We denote our discrete Sturmian decorated $\Z$-graphs by $G_{\alpha}$
with $\alpha\in\left(0,1\right)\backslash\Q$. These graphs have two
possible decorations, and are equipped with a Jacobi operator of Sturmian
type,
\begin{align}
 & \Ja:\ell^{2}\left(G_{\alpha}\right)\rightarrow\ell^{2}\left(G_{\alpha}\right),\label{eq:Jacobi}\\
 & \Ja\psi\left(v\right)=b_{v}\psi\left(v\right)+\sum_{u\in\Ev}c_{uv}\psi\left(u\right),\label{eq:Jacobi2}
\end{align}
where the coefficients $b_{v},c_{uv}\in\R$ satisfy

1. \textit{Symmetry:} $c_{uv}=c_{vu}$.

2. \textit{Equivariance:} If $v,u\in\V_{G_{i}}$ appear as both $v_{1},u_{1}\in\V_{G_{\alpha}^{\left(n\right)}}$
and $v_{2},u_{2}\in\V_{G_{\alpha}^{\left(m\right)}}$ (i.e. in the
$n$th and $m$th decorations), then $c_{u_{1}v_{1}}=c_{u_{2}v_{2}}$,
$b_{v_{1}}=b_{v_{2}}$, and $b_{u_{1}}=b_{u_{2}}$.

3. \textit{Horizontal homogeneity:} There exists $0\neq c\in\R$ such
that if $u,v$ are base vertices of adjacent decorations, then $c_{uv}=c_{vu}=c$.

We also assume that the given operator locally distinguishes between
the different decorations, in the following sense:
\begin{defn}
\label{def:Green}For a given decoration $G$, denote the restriction
of $\Ja$ to $G$ by $\mathcal{J}_{G}$:
\begin{align}
 & \mathcal{J}_{G}:\C^{\V\left(G\right)}\rightarrow\C^{\V\left(G\right)},\label{eq:J_G}\\
 & \mathcal{J}_{G}\psi\left(v\right)=b_{v}\psi\left(v\right)+\sum_{u\in\Ev}c_{uv}\psi\left(u\right),
\end{align}
where the coefficients $c_{uv},b_{v}$ are as above. Denote the eigenpairs
of $\mathcal{J}_{G}$ by $\left(\lambda_{k},\psi_{k}\right)$, and
define the diagonal Green's function for the pointed graph $\left(G,v\right)$
as
\begin{equation}
\mathcal{G}_{\mathcal{J}_{G}}\left(E\right):=\sum_{k}\frac{|\psi_{k}(v)|^{2}}{E-\lambda_{k}},\label{eq:diag-Green}
\end{equation}
for $E\notin\spec{\mathcal{J}_{G}}$. 
\end{defn}

\begin{assumption}
\label{assu:distinct-decorations}We assume that the two pointed decorations
$\left(G_{0},v_{0}\right),\left(G_{1},v_{1}\right)$ in the Sturmian
decorated $\Z$-graph have distinct Green's functions:
\[
\mathcal{G}_{\mathcal{J}_{G_{0}}}\neq\mathcal{G}_{\mathcal{J}_{G_{1}}}.
\]
\end{assumption}

\begin{rem}
Assumption \ref{assu:distinct-decorations} is necessary for our main
results to hold. As we shall see later, the Green's functions determine
the transfer matrices associated with the system. Indeed, if the two
Green's functions are identical, then the transfer matrices do not
distinguish the two decorations; the associated transfer-matrix cocycle
is then constant, and the spectral problem reduces to a periodic one,
with only finitely many gap labels.
\end{rem}

We note that the normalized discrete Laplacian $\Delta_{\alpha}$
does not satisfy the horizontal homogeneity assumption, as the coupling
between adjacent base vertices of decorations depends on their degree
(which may vary from decoration to decoration). Nevertheless, in Subsection
\ref{subsec:NDL-DTMP} we briefly describe how the gap labels for
$\Delta_{\alpha}$ can still be computed after applying an appropriate
subdivision trick.

The main result of this part completely characterizes the gap labels
of the adjacency matrix for discrete Sturmian combs:
\begin{thm}
\label{thm:DTMP-Adj} Let $\alpha\in\left(0,1\right)\backslash\Q$
and let
\begin{equation}
\alpha=\frac{1}{a_{1}+\frac{1}{a_{2}+...}}\label{eq:cont-frac-2}
\end{equation}
be the continued fraction expansion of $\alpha$. Let $G_{\alpha}$
be a discrete Sturmian comb, equipped with the adjacency matrix $\mathcal{A}_{\alpha}$
and satisfying Assumption \ref{assu:distinct-decorations}. Then the
gap labels of $\mathcal{A}_{\alpha}$ are given by
\begin{equation}
\GL{\Aa}=\begin{cases}
\left\{ \frac{\alpha n+m}{1+\alpha}:m,n\in\Z\right\} \cap\left(\left[0,\frac{(a_{1}+1)\alpha}{2(1+\alpha)}\right]\cup\left[1-\frac{(a_{1}+1)\alpha}{2(1+\alpha)},1\right]\right), & a_{1}\ \text{odd},\\[2mm]
\left\{ \frac{\alpha n+m}{1+\alpha}:m,n\in\Z\right\} \cap\left(\left[0,\frac{2-a_{1}\alpha}{2(1+\alpha)}\right]\cup\left[\frac{(a_{1}+2)\alpha}{2(1+\alpha)},1\right]\right), & a_{1}\ \text{even}.
\end{cases}\label{eq:GLAa}
\end{equation}
\end{thm}

By analogy with the GLT for the normalized discrete Laplacian, the
natural set of labels for the Sturmian comb is
\begin{equation}
\left\{ \frac{\alpha n+m}{1+\alpha}:m,n\in\mathbb{Z}\right\} \cap[0,1],\label{eq:S-GLT}
\end{equation}
Theorem \ref{thm:DTMP-Adj} identifies precisely which of these labels
are actually realized by the IDS. In particular, the answer is not
the full arithmetic set: an interval of labels is missing, and this
interval depends on the parity of the first continued fraction digit
$a_{1}$. We shall see that these missing labels are precisely those
which are not realized due to a jump discontinuity of the IDS.

\subsection{Transfer matrices}

Our analysis of the DTMP relies heavily on the transfer matrices introduced
in Subsection \ref{subsec:Transfer-matrices-and-Lyapunov}. We now
specialize transfer matrices for the Sturmian setting, beginning with
the standard Sturmian Hamiltonians $H_{\alpha}^{V}$ acting on $\ell^{2}\left(\Z\right)$.
These arise from the eigenvalue equation
\begin{equation}
\varphi(n+1)+\varphi(n-1)+V\omega_{\alpha}(n)\varphi(n)=E\varphi(n),\quad E\in\mathbb{R}.\label{eq:sturmian-eigen}
\end{equation}
The equation above implies that each solution $\varphi$ of (\ref{eq:sturmian-eigen})
satisfies the recurrence
\begin{equation}
\begin{pmatrix}\varphi(n+1)\\
\varphi(n)
\end{pmatrix}=\mathcal{M}_{\omega_{\alpha}\left(n\right)}^{V}(E)\begin{pmatrix}\varphi(n)\\
\varphi(n-1)
\end{pmatrix},\label{eq:M-recurrence}
\end{equation}
where the one-step transfer matrix is defined by
\begin{equation}
\mathcal{M}_{\omega_{\alpha}\left(n\right)}^{V}(E):=\begin{pmatrix}E-V\omega_{\alpha}(n) & -1\\
1 & 0
\end{pmatrix},\label{eq:MwE}
\end{equation}
meaning that it takes one of two possible forms, depending on the
value of $\omega_{\alpha}\left(n\right)$. The transfer matrices are
central in studying the spectra of the periodic approximants for Sturmian
Hamiltonians, and hence in computing gap labels.

The definition above naturally extends to decorated $\Z$-graphs.
Let $G_{\alpha}$ be a discrete Sturmian decorated $\Z$-graph equipped
with a Jacobi operator $\Ja$, and consider the formal eigenvalue
equation on $G_{\alpha}$:
\begin{equation}
b_{v}\varphi\left(v\right)+\sum_{u\in\Ev}c_{uv}\varphi\left(u\right)=E\varphi\left(v\right)\,\,\,\,(E\in\mathbb{R}).\label{eq:ODE-transfer-1}
\end{equation}
For a (not necessarily $L^{2}$) solution $\varphi$ to (\ref{eq:ODE-transfer-1}),
we denote
\begin{equation}
\Phi_{n}=\left(\begin{array}{c}
\varphi\left(n\right)\\
\varphi\left(n-1\right)
\end{array}\right)\in\R^{2}.\label{eq:initial-ode-1}
\end{equation}
Applying the same transfer-matrix construction as in Subsection \ref{subsec:Transfer-matrices-and-Lyapunov}
to the present discrete setting, for all but a discrete set of $E\in\mathbb{R}$,
there exists a unique matrix $\mathcal{M}_{\omega_{\alpha}\left(n\right)}\left(E\right)\in SL\left(2,\mathbb{\R}\right)$,
depending only on the decoration $G_{\alpha}^{\left(n\right)}$, such
that
\begin{equation}
\Phi_{n+1}=\mathcal{M}_{\omega_{\alpha}\left(n\right)}\left(E\right)\Phi_{n}.\label{eq:trans-mat-1}
\end{equation}

For a decorated $\Z$-graph equipped with the operator $\Ja$, the
one-step transfer matrix takes the form
\begin{equation}
\mathcal{M}_{\omega_{\alpha}\left(n\right)}\left(E\right)=\left(\begin{array}{cc}
E-f_{\omega_{\alpha}\left(n\right)}\left(E\right) & -1\\
1 & 0
\end{array}\right).\label{eq:transfer-general}
\end{equation}
Notably, for a Sturmian Hamiltonian on $\Z$, the function $f_{\omega_{\alpha}\left(n\right)}\left(E\right)$
is simply the potential, and hence independent of $E$. However, for
a nontrivial decorated $\Z$-graph, $f_{\omega_{\alpha}\left(n\right)}\left(E\right)$
will explicitly depend on $E$, and a main part of the analysis in
Subsection \ref{subsec:combinatorial-structure} will be to extend
certain spectral properties of Sturmian Hamiltonians to systems with
transfer matrices of the form (\ref{eq:transfer-general}). We mention
that some cases of interest have transfer matrices not of this form,
most notably the NDL. The general method of treating the NDL is briefly
described in Subsection \ref{subsec:NDL-DTMP}.

\subsection{Sturmian periodic approximations and Floquet--Bloch theory\label{subsec:FB}}

As in Section \ref{sec:example-comb-graphs}, given $\alpha\notin\Q$,
we can define a family of periodic approximations for $H_{\alpha}^{V}$
using the (infinite) continued fraction expansion of $\alpha$:
\begin{equation}
\alpha=\frac{1}{a_{1}+\frac{1}{a_{2}+...}}.\label{eq:cont-frac-1}
\end{equation}
By truncating the continued fraction expansion at the $n$th step,
we obtain a rational number $\alpha_{n}:=\frac{p_{n}}{q_{n}}$, which
corresponds to a periodic Schrödinger operator $H_{\alpha_{n}}^{V}$.
We define the following sequence of transfer matrices by
\begin{equation}
\MNV nV\left(E\right):=\mathcal{M}_{\omega_{\alpha}\left(q_{n}-1\right)}^{V}\left(E\right)\cdot...\cdot\mathcal{M}_{\omega_{\alpha}\left(0\right)}^{V}\left(E\right).\label{eq:mqn-1}
\end{equation}
Similarly, given a decorated $\Z$-graph, we may define a sequence
of periodic graph operators $\Jan n$, with an associated sequence
of transfer matrices $\MNJ n\left(E\right)$.

For such periodic operators, Floquet--Bloch theory describes the
spectrum in terms of finite-dimensional operators. Namely, $H_{\alpha_{n}}^{V}$
decomposes as a direct integral of finite-dimensional operators $\left(H_{\alpha_{n}}^{V}(\theta)\right)_{\theta\in\left[0,\pi\right]}$
acting on the fundamental domain (see \cite{BerKuc_graphs}). The
spectrum satisfies
\begin{equation}
\spec{H_{\alpha_{n}}^{V}}=\bigcup_{\theta\in\left[0,\pi\right]}\spec{H_{\alpha_{n}}^{V}(\theta)}.\label{eq:FB-spec}
\end{equation}
The connected components in the spectral decomposition above are the
spectral bands, and their endpoints are attained at $\theta\in\{0,\pi\}$.

When the transfer matrices are well-defined, this description can
be expressed in terms of the trace. The following version of Theorem
\ref{thm:Per-approx} holds for both Sturmian Hamiltonians (cf. \cite{Bellissard1989,Suetoe1987})
and, by similar arguments, decorated $\Z$-graphs \cite{Band2024a}.
For definiteness, we state it for decorated $\Z$-graphs:
\begin{prop}
\label{prop:Sturmian-properties-1} Define the sequence $\left(t_{n}^{\alpha}\left(E\right)\right)_{n=1}^{\infty}$
by 
\begin{equation}
t_{n}^{\alpha}\left(E\right)=tr\left(\MNJ n\left(E\right)\right).\label{eq:tnE-1}
\end{equation}
Then for all $E\in\R$ such that $\mathcal{M}_{n}^{\Ja}\left(E\right)$
is well-defined, we have that
\begin{equation}
E\in\spec{\Jan n}\iff\left|t_{n}^{\alpha}\left(E\right)\right|\leq2.\label{eq:FB-theory}
\end{equation}
Furthermore,
\begin{equation}
\spec{\Ja}=\bigcap_{n\in\N}\left(\spec{\Jan n}\cup\spec{\Jan{n+1}}\right).\label{eq:approx1-1}
\end{equation}
\end{prop}

Notably, the characterization (\ref{eq:FB-theory}) does not capture
the entire spectrum in general. At energies where the transfer matrix
is not defined (i.e., at poles of the functions $f_{0}(E),f_{1}(E)$
from (\ref{eq:transfer-general})), the operator may still have isolated
eigenvalues, known as \textit{flat bands}. These correspond to joint
eigenvalues of all Floquet--Bloch operators, which are not detected
by the trace condition. Such energies will play a special role in
the analysis below.

\subsection{Band types and the spectral tree\label{subsec:Sturmian-tree-1}}

We now recall some existing results on the band structure of the usual
Sturmian Hamiltonians. The point of this subsection is to set up the
terminology of band types and the spectral tree, which will later
be adapted to decorated $\mathbb{Z}$-graphs. By Proposition \ref{prop:Sturmian-properties-1},
each $E\in\spec{\Jan n}$ necessarily satisfies $E\in\spec{\Jan{n-1}}$
or $E\in\spec{\Jan{n-2}}$. For the usual Sturmian Hamiltonians, a
stronger statement holds at the level of spectral bands. We say that
$\left[a,b\right]\subset_{\text{strict}}\left[c,d\right]$ if $c<a$
and $b<d$.
\begin{defn}
A spectral band $I\subset\spec{H_{\alpha_{n}}^{V}}$ is said to be\\
1. of \textit{backward type $A$} if there exists a spectral band
$J$ of $\spec{H_{\alpha_{n-1}}^{V}}$ such that $I\subset_{\text{strict}}J$.\\
2. of \textit{backward type $B$} if there exists a spectral band
$J$ of $\spec{H_{\alpha_{n-2}}^{V}}$ such that $I\subset_{\text{strict}}J$
and $J\not\subset\spec{H_{\alpha_{n-1}}^{V}}$.

By standard theory (see \cite[prop. 2.3]{band2026complete}), the
spectral band edges of $H_{\alpha_{n}}^{V}$ depend continuously on
$V$, and one can globally extend a given spectral band $I\subset\spec{H_{\alpha_{n}}^{V_{0}}}$
to a continuous function $I\left(V\right):\mathbb{R}\rightarrow\mathcal{K}\left(\R\right)$
(compact subsets of $\R$) with respect to the Hausdorff distance
for all $V\neq0$ such that $I\left(V\right)$ is the associated spectral
band in $\spec{H_{\alpha_{n}}^{V}}$.
\end{defn}

\begin{prop}
\cite[thm. 1.1]{band2026complete}\label{prop:distinct-type} For
all $V\neq0$, each spectral band $I\subset\spec{H_{\alpha_{n}}^{V}}$
has a distinct backward type -- $A$ or $B$. Furthermore, the backward
type of $I\left(V\right)$ does not depend on the choice of $V\neq0$.
\end{prop}

Once the backward type is well-defined, one can prove an even stronger
result about the relative combinatorial structure between them. Given
two intervals $I_{1}=\left[a_{1},b_{1}\right],I_{2}=\left[a_{2},b_{2}\right]\subset\R$,
we say that $I_{1}\prec I_{2}$ if $b_{1}<a_{2}$.
\begin{defn}
\label{def:forward-type}Write
\begin{equation}
\alpha=\frac{1}{a_{1}+\frac{1}{a_{2}+...}}.\label{eq:cont-frac-1}
\end{equation}
A spectral band $I\subset\spec{H_{\alpha_{n}}^{V}}$ is said to be
of forward type $A$ with $M=a_{n+1}-1$ (resp. forward type $B$
with $M=a_{n+1}$) if \\
1. There exist exactly $M$ spectral bands $\left(I_{\alpha_{n+1}}^{j}\right)_{j=1}^{M}$
of $\spec{H_{\alpha_{n+1}}^{V}}$ such that\\
$I_{\alpha_{n+1}}^{j}\subset_{\text{strict}}I$ for all $j$ and they
are not contained in $\spec{H_{\alpha_{n-1}}^{V}}$ (i.e., they are
of backward type $A$).\\
2. There exist exactly $M+1$ spectral bands $\left(I_{\alpha_{n+2}}^{j}\right)_{j=1}^{M+1}$
of $\spec{H_{\alpha_{n+2}}^{V}}$ such that $I_{\alpha_{n+2}}^{j}\subset_{\text{strict}}I$
for all $j$ and they are not contained in $\spec{H_{\alpha_{n+1}}^{V}}$
(i.e., they are of backward type $B$).\\
3. \label{def:interlacing}The bands above interlace, meaning that
\begin{equation}
I_{\alpha_{n+2}}^{1}\prec I_{\alpha_{n+1}}^{1}\prec I_{\alpha_{n+2}}^{2}\prec...\prec I_{\alpha_{n+1}}^{M}\prec I_{\alpha_{n+2}}^{M+1}.\label{eq:interlacing}
\end{equation}
\end{defn}

While the backward type determines how a spectral band sits inside
previous levels, the forward type determines how it branches into
subsequent levels. In the classical Sturmian case these notions coincide,
as established in \cite[prop. 3.4]{band2026complete}:
\begin{prop}
\label{prop:type-existence}A spectral band $I$ is of backward type
$A$ (resp. $B$) if and only if it is of forward type $A$ (resp.
$B$). Each spectral band $I\subset\spec{H_{\alpha_{n}}^{V}}$ has
a distinct backward type ($A$ or $B$), and this type is independent
of $V$.
\end{prop}

The forward type of the bands allows one to encode the structure of
the spectrum for the periodic approximations via a directed rooted
tree $\mathcal{T}\left(H_{\alpha}\right)$ with labeled vertices,
as demonstrated in Figure \ref{fig:Sturmian-tree}. The vertices in
the $n$th level of the tree correspond to spectral bands of $\spec{H_{\alpha_{n}}^{V}}$,
and are labeled by their type. Each band at level $n$ is connected
by directed edge to all bands at levels $n+1$ or $n+2$ which are
contained in it. By Propositions \ref{prop:distinct-type} and \ref{prop:type-existence},
the resulting tree is independent of $V$, and is completely determined
by its first two levels. In the case of Sturmian Hamiltonians, the
first level consists of a single $A$-type band, and the second level
consists of a single $B$-type band. We shall refer to the vertices
corresponding to $\spec{H_{\alpha_{n}}^{V}}$ as level $\alpha_{n}$
of the tree. A descendant of a vertex is a vertex reached from it
by a directed path, and the branch rooted at a vertex consists of
this vertex together with all of its descendants.

\begin{figure}
\includegraphics[scale=0.3]{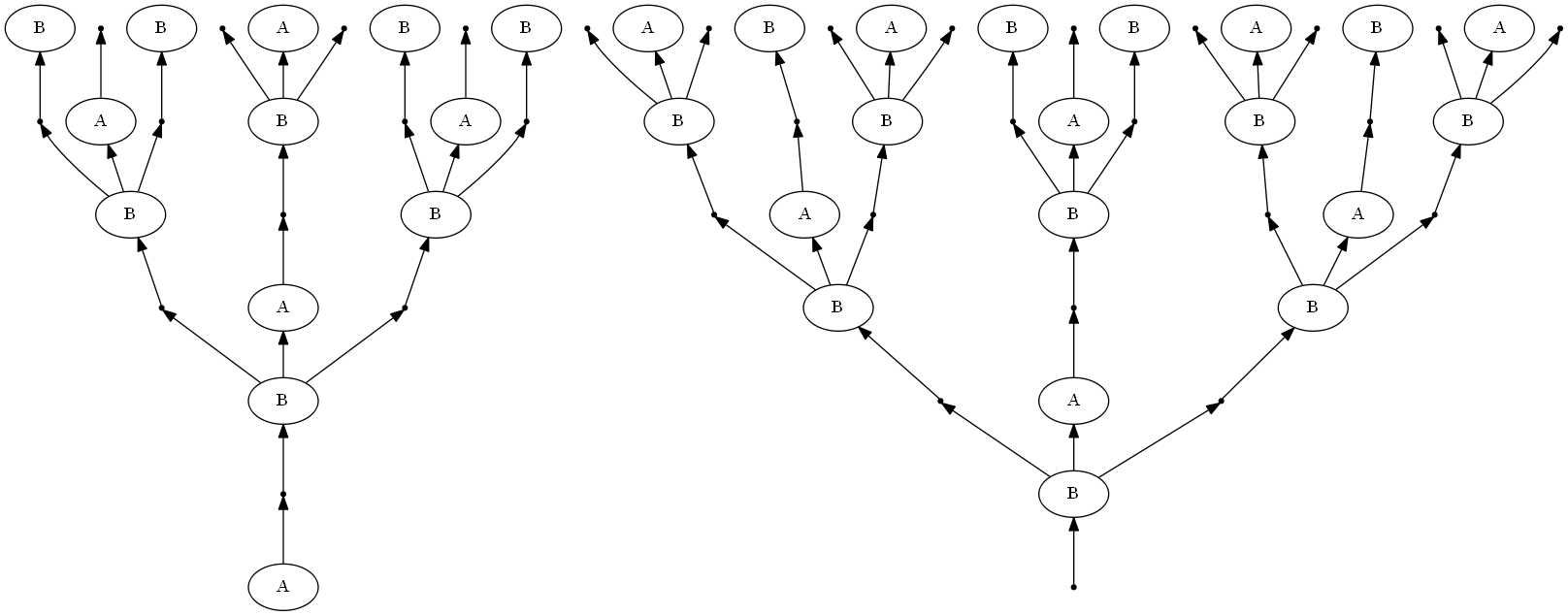}

\caption[Spectral tree for the Fibonacci Hamiltonian.]{The combinatorial spectral structure of the periodic approximations
$H_{\alpha_{n}}^{V}$ through the directed tree $\mathcal{T}\left(H_{\alpha}\right)$
for the Fibonacci Hamiltonian $\alpha=\frac{\sqrt{5}-1}{2}$.\label{fig:Sturmian-tree}}
\end{figure}

The boundary $\partial\mathcal{T}\left(H_{\alpha}\right)$ is the
set of infinite directed paths starting from the root. A path $\gamma\in\partial\T(H_{\alpha})$
is an infinite sequence $\gamma=(I_{1},I_{2},\ldots),$ where $I_{n}$
is a spectral band of $\spec{H_{\alpha_{n}}^{V}}$, and the tree structure
means that these bands are nested:
\begin{equation}
I_{n+1}\subset I_{n},\qquad n\in\mathbb{N}.\label{eq:nested}
\end{equation}
For Sturmian Hamiltonians these nested bands shrink to a single point,
\begin{equation}
\bigcap_{n=1}^{\infty}I_{n}=\left\{ E_{\alpha}(\gamma)\right\} .\label{eq:nesting}
\end{equation}
since their intersection is an interval contained in the zero-measure
set $\spec{H_{\alpha}}$. Moreover, it is shown in \cite[thm. 1.10]{Band2024}
that every point in $\spec{H_{\alpha}}$ is obtained in this way from
a unique path $\gamma\in\partial T(H_{\alpha})$, so that the map
\begin{equation}
\partial T(H_{\alpha})\ni\gamma\mapsto E_{\alpha}(\gamma)\in\spec{H_{\alpha}}\label{eq:pathmap}
\end{equation}
is a bijection, see \cite[thm. 1.10]{Band2024}. We therefore identify
$\partial\T(H_{\alpha})$ with $\spec{H_{\alpha}}$ through this map.
Composing this identification with the IDS defines $N_{H_{\alpha}}(\gamma):=N_{H_{\alpha}}(E_{\alpha}(\gamma))$.
In \cite{Band2024}, certain paths $\gamma\in\partial\T(H_{\alpha})$,
corresponding to boundaries of gaps, are identified, and the value
$N_{H_{\alpha}}(\gamma)$ is shown to equal the corresponding gap
label.

Our goal in the next section is to prove a similar existence of backward
and forward types for decorated $\Z$-graphs, and use this  to identify
an analogous spectral tree structure for these graphs. The main new
difficulty is that the effective potentials depend on the energy,
and may coincide at isolated energies. Because of this, we will need
a slightly weaker notion of backward and forward type. We will first
show that ``most'' spectral bands of $\Jan n$ satisfy the usual
Sturmian band structure, and then identify a restricted collection
of exceptional bands where a slightly weaker combinatorial structure
holds.

\newpage{}

\section{Dry Ten Martini Problem -- combinatorial spectral structure\label{sec:DTMP-proofs}}

The purpose of this section is to show that the periodic approximants
of a Sturmian decorated $\Z$-graph admit the same local combinatorial
band structure as the classical Sturmian Hamiltonian, and that this
structure controls the attained gap labels up to a finite set of exceptional
energies which may lead to closed gaps. This is the content of \hyperref[subsec:informal-results]{Main Result $4$},
and is made precise in Theorems \ref{thm:bad-type}, \ref{lem:Sturmian-type},
and \ref{thm:closed-gaps} below.

The argument proceeds in three stages. First, we reduce the spectral
analysis of $\Ja$ at a fixed energy to that of a classical Sturmian
Hamiltonian with effective (energy dependent) potentials, using the
transfer matrix formalism. Second, on spectral bands where the two
effective potentials are distinct, we show that the standard Sturmian
combinatorial structure is inherited. Finally, we treat the finitely
many energies where the effective potentials coincide, and prove that
a slightly weaker combinatorial structure still holds. In particular,
after artificially separating bands that may share endpoints, the
resulting spectral tree displays the usual Sturmian branching rule
from Definition \ref{def:forward-type}, and any possible discrepancy
between tree-predicted and open gap labels can occur only at those
exceptional energies.

\subsection{Backward and forward band types for $\protect\Ja$\label{subsec:combinatorial-structure}}

We wish to prove that the spectral bands of $\Jan n$ possess well-defined
backward and forward type, similar to Proposition \ref{prop:distinct-type},
which will later be used for computing the gap labels. 
\begin{defn}
A spectral band $I\subset\spec{\Jan n}$ is said to be

1. Of backward type $A$ if there exists a spectral band $J$ of $\spec{\Jan{n-1}}$
such that $I\subset_{\text{strict}}J$.

2. Of backward type $B$ if there exists a spectral band $J$ of $\spec{\Jan{n-2}}$
such that $I\subset_{\text{strict}}J$ and $I\not\subset\spec{\Jan{n-1}}$.
\end{defn}

Similarly, Definition \ref{def:forward-type} of forward band type
can be extended verbatim to the spectral bands of $\Jan n$.

Our strategy for showing the existence of band types for decorated
$\Z$-graphs is to reduce the spectral problem for $\Jan n$ at a
fixed energy $E$ to that of a classical Sturmian Hamiltonian with
effective energy-dependent potentials, and then transfer the Sturmian
band structure back to $\Jan n$. As discussed before, the one-step
transfer matrices for this model are of the form
\begin{equation}
\mathcal{M}_{\omega_{\alpha}\left(n\right)}\left(E\right)=\left(\begin{array}{cc}
E-f_{\omega_{\alpha}\left(n\right)}\left(E\right) & -1\\
1 & 0
\end{array}\right).\label{eq:Sturmian-M-f}
\end{equation}
In the setting considered here, the functions $f_{0},f_{1}:\R\backslash S\rightarrow\R$
are distinct and real-analytic, where $S\subset\R$ is a finite (possibly
empty) set of poles where the transfer matrices are not defined (this
will be justified in Subsection \ref{subsec:adj-general}). 
\begin{assumption}
\label{assu:no-flat-bands}Throughout this section, we assume that
$S\cap\spec{\Ja}=\varnothing$.
\end{assumption}

The role of this assumption is only to exclude singular spectral points
from the transfer-matrix analysis. The energies in $S\cap\spec{\Ja}$
play a special role, and their effect is discussed in Subsection \ref{subsec:singular-energies}.
\begin{defn}
For $\left(V_{0},V_{1}\right)\in\R^{2}$ fixed, denote by $H_{\alpha}^{V_{0},V_{1}}$
the Sturmian-like Hamiltonian on $\ell^{2}\left(\Z\right)$ given
by
\begin{equation}
(H_{\alpha}^{V_{0},V_{1}}\psi)(n)=\psi(n+1)+\psi(n-1)+\left(V_{0}+\left(V_{1}-V_{0}\right)\omega_{\alpha}\left(n\right)\right)\psi(n).\label{eq:modified-Sturmian}
\end{equation}
\end{defn}

Since $H_{\alpha}^{V_{0},V_{1}}=H_{\alpha}^{V_{1}-V_{0}}+V_{0}$,
Proposition \ref{prop:distinct-type} implies that every spectral
band of its periodic approximants has a distinct backward type independent
of $\left(V_{0},V_{1}\right)$ as long as $V_{0}\neq V_{1}$. Note
that the one-step transfer matrices for this model are given by
\begin{equation}
\mathcal{M}_{\omega_{\alpha}\left(n\right)}^{V_{0},V_{1}}\left(E\right)=\begin{cases}
\left(\begin{array}{cc}
E-V_{0} & -1\\
1 & 0
\end{array}\right), & \omega_{\alpha}\left(n\right)=0,\\
\left(\begin{array}{cc}
E-V_{1} & -1\\
1 & 0
\end{array}\right), & \omega_{\alpha}\left(n\right)=1.
\end{cases}\label{eq:MV0V1}
\end{equation}
The key observation here is that for $\overline{E}\in\R$ fixed, the
one-step transfer matrices of $\Ja$ \emph{at energy $\overline{E}$}
are exactly those of the Sturmian-like Hamiltonian $H_{\alpha}^{f_{0}\left(\overline{E}\right),f_{1}\left(\overline{E}\right)}$.
Thus, at a fixed energy $\overline{E}\in\R$, the spectral properties
of $\Ja$ are governed by those of a classical Sturmian Hamiltonian
with potentials depending on $\overline{E}$.
\begin{lem}
\label{lem:T-H-correspondence} For every $\alpha\in\left(0,1\right)$
and $\overline{E}\in\R\backslash S$,
\begin{equation}
\overline{E}\in\spec{\Ja}\iff\overline{E}\in\spec{H_{\alpha}^{f_{0}\left(\overline{E}\right),f_{1}\left(\overline{E}\right)}}.\label{eq:T-H-energies}
\end{equation}
\end{lem}

\begin{proof}
As noted above, the one-step transfer matrices of $H_{\alpha}^{f_{0}\left(\overline{E}\right),f_{1}\left(\overline{E}\right)}$
and $\Ja$ at energy $\overline{E}$ are the same, and consequently
$\MNV n{f_{0}\left(\overline{E}\right),f_{1}\left(\overline{E}\right)}\left(\overline{E}\right)=\MNJ n\left(\overline{E}\right)$
for all $n\in\N$. Denote the associated sequence of traces (of either
transfer matrix) by $\left(t_{n}^{\alpha}\left(\overline{E}\right)\right)_{n=1}^{\infty}$.
By (\ref{eq:FB-theory}) and (\ref{eq:approx1-1}) in Proposition
\ref{prop:Sturmian-properties-1}, we have that $\overline{E}\in\spec{\Ja}\backslash S$
if and only if for all $n\in\N$,
\begin{equation}
\left|t_{n}^{\alpha}\left(\overline{E}\right)\right|\leq2\text{ or }\left|t_{n+1}^{\alpha}\left(\overline{E}\right)\right|\leq2,\label{eq:B-infty-1}
\end{equation}
which by Proposition \ref{prop:Sturmian-properties-1} then also corresponds
to $\overline{E}\in\spec{H_{\alpha}^{f_{0}\left(\overline{E}\right),f_{1}\left(\overline{E}\right)}}$. 
\end{proof}
The lemma above is demonstrated in Figure \ref{fig:SpecVsV}, which
displays how $\spec{\Aa}$ is obtained from $\spec{H_{\alpha}^{V}}$
for a Sturmian comb. Motivated by Lemma \ref{lem:T-H-correspondence},
we introduce from now on the following notation:

\begin{figure}
\includegraphics[scale=0.6]{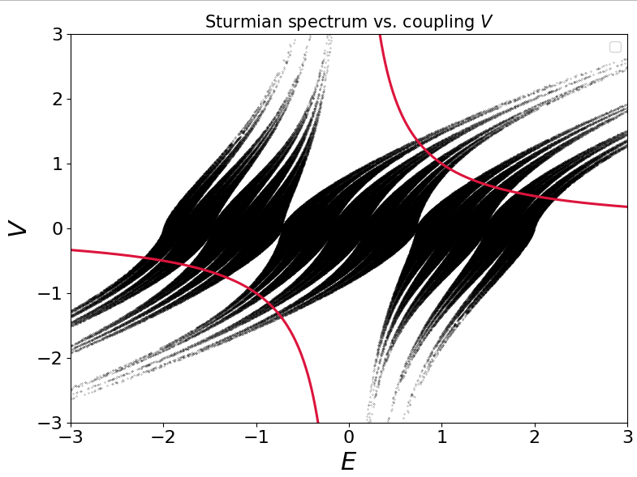}

\caption[Demonstration of Lemma \ref{lem:T-H-correspondence}.]{Demonstration of Lemma \ref{lem:T-H-correspondence} for a Sturmian
comb through a plot of $\protect\spec{H_{\alpha}^{V}}$ as a function
of $V$. The curve $V=1/E$ (red), corresponding to the effective
coupling pair $V\left(E\right)=\left(0,1/E\right)$ for $E\protect\neq0$,
intersects the spectral sets, giving rise to the points in $\protect\spec{\protect\Aa}$
\label{fig:SpecVsV}}
\end{figure}

\begin{notation}
For $E\notin S$, denote $V\left(E\right):=\left(f_{0}\left(E\right),f_{1}\left(E\right)\right)$.
\end{notation}

\begin{defn}
\label{def:frozen band}Let $E\in\spec{\Jan n}$. By Lemma \ref{lem:T-H-correspondence},
we have that $E\in\spec{H_{\alpha_{n}}^{V\left(E\right)}}$ (which
is well-defined by Assumption \ref{assu:no-flat-bands}). We denote
by $I_{E}^{H_{\alpha_{n}}}:\mathbb{R}^{2}\rightarrow\mathcal{K}\left(\R\right)$
(nonempty compact subsets of $\R$) the unique map such that:\\
1. $I_{E}^{H_{\alpha_{n}}}$ is continuous with respect to the Hausdorff
distance on $\mathcal{K}\left(\R\right)$:
\begin{equation}
d\left(A,B\right)=\max\left\{ \max_{a\in A}\left\{ \min_{b\in B}\left|a-b\right|\right\} ,\max_{b\in B}\left\{ \min_{a\in A}\left|a-b\right|\right\} \right\} .\label{eq:hausdorff}
\end{equation}
\\
2. For each $\left(V_{0},V_{1}\right)$, the set $I_{E}^{H_{\alpha_{n}}}\left(V_{0},V_{1}\right)$
is a spectral band in $\spec{H_{\alpha_{n}}^{V_{0},V_{1}}}$;\\
3. The spectral band $I_{E}^{H_{\alpha_{n}}}\left(V\left(E\right)\right)$
contains $E$.

The map $I_{E}^{H_{\alpha_{n}}}$ is well-defined and unique by \cite[def. 2.2, prop. 2.3]{band2026complete}.
Heuristically, the map $I_{E}^{H_{\alpha_{n}}}$ allows us to ``follow''
how a fixed spectral band of $H_{\alpha_{n}}$ changes with respect
to the coupling $V$.
\end{defn}

\begin{lem}
\label{lem:same-components}Let $\tilde{E}_{1},\tilde{E}_{2}$ be
in the same spectral band of $\spec{\Jan n}$. Then $I_{\tilde{E}_{1}}^{H_{\alpha_{n}}}=I_{\tilde{E}_{2}}^{H_{\alpha_{n}}}$
(as maps).
\end{lem}

\begin{proof}[Proof of Lemma \ref{lem:same-components}]
 Let $\tilde{E}\in\spec{\Jan n}$ be in some spectral band $I\subset\spec{\Jan n}$.
Consider the set:
\begin{equation}
\A_{\tilde{E}}:=\left\{ E'\in I:I_{E'}^{H_{\alpha_{n}}}=I_{\tilde{E}}^{H_{\alpha_{n}}}\right\} .\label{eq:agree-set}
\end{equation}
We claim that $\A_{\tilde{E}}=I$. Note that $\A_{\tilde{E}}$ is
nonempty, as it contains at least $\tilde{E}$. It is also closed:
$E'\in\A_{\tilde{E}}$ if and only if $E'\in I_{\tilde{E}}^{H_{\alpha_{n}}}\left(V\left(E'\right)\right)$,
and the latter condition is closed by continuity of $V$ and $I_{\tilde{E}}^{H_{\alpha_{n}}}$.
Since $I$ is connected, it is enough to show that $\A_{\tilde{E}}$
is also open. In other words, that there is some open neighborhood
$U_{\tilde{E}}$ (in $I$) such that $I_{E}^{H_{\alpha_{n}}}=I_{\tilde{E}}^{H_{\alpha_{n}}}$
for all $E$ in $U_{\tilde{E}}$.

Recall that $\tilde{E}\in\spec{\Jan n}$, and namely
\begin{equation}
\left|tr\left(\MNV n{V\left(\tilde{E}\right)}\left(\tilde{E}\right)\right)\right|\leq2\label{eq:small-trace}
\end{equation}
 by Proposition \ref{prop:Sturmian-properties-1}. Note that the function
$E\mapsto tr\left(\MNV n{V\left(E\right)}\left(E\right)\right)$ is
well-defined and continuous throughout $I$. There are two possible
cases: either $\tilde{E}\in I$ is an interior point and so 
\begin{equation}
\left|tr\left(\MNV n{V\left(\tilde{E}\right)}\left(\tilde{E}\right)\right)\right|<2;\label{eq:t<2}
\end{equation}
or $\tilde{E}\in I$ is a boundary point, and thus 
\begin{equation}
\left|tr\left(\MNV n{V\left(\tilde{E}\right)}\left(\tilde{E}\right)\right)\right|=2.\label{eq:t=00003D2}
\end{equation}
In the second case, fixing the coupling at $V\left(\tilde{E}\right)$,
\cite[prop. 4.1]{band2024review} implies that $\tilde{E}$ has a
one-sided neighborhood in $I$ where the trace is strictly monotone.
In either case, continuity of the trace implies the existence of an
open neighborhood $U_{\tilde{E}}^{1}$ (in $I$) of $\tilde{E}$ such
that
\begin{equation}
\left|tr\left(\MNV n{V\left(\tilde{E}\right)}\left(E\right)\right)\right|\leq2,\,\,\forall E\in U_{E'}^{1}.\label{eq:small-trace-1}
\end{equation}
In particular, $U_{\tilde{E}}^{1}\subset I_{\tilde{E}}^{H_{\alpha_{n}}}\left(V\left(\tilde{E}\right)\right)$.
Now, continuity of the map $I_{\tilde{E}}^{H_{\alpha_{n}}}\left(V\left(E\right)\right)$
with respect to $E$ tells us that there is some (smaller) open neighborhood
$U_{\tilde{E}}^{2}$ of $\tilde{E}$ such that $E\in I_{\tilde{E}}^{H_{\alpha_{n}}}\left(V\left(E\right)\right)$
for all $E\in U_{\tilde{E}}^{2}$. On the other hand, $I_{E}^{H_{\alpha_{n}}}$
is defined as the \emph{unique} map such that $E\in I_{E}^{H_{\alpha_{n}}}\left(V\left(E\right)\right)$,
and so $E\in I_{\tilde{E}}^{H_{\alpha_{n}}}\left(V\left(E\right)\right)$
implies that $I_{E}^{H_{\alpha_{n}}}=I_{\tilde{E}}^{H_{\alpha_{n}}}$
for all $E\in U_{E'}^{2}$, concluding the proof.
\end{proof}
At this point we must distinguish between spectral bands on which
the two effective potentials coincide and those on which they do not.
Since the functions $f_{0},f_{1}$ are real-analytic and distinct
(by Assumption \ref{assu:distinct-decorations}), the equation $f_{0}(E)=f_{1}(E)$
has at most finitely many solutions within $\spec{\Jan n}$. 
\begin{defn}
\label{def:goodband}Let $I\subset\spec{\Jan n}$ be a spectral band.
We call $I$ \emph{good} if 
\begin{equation}
f_{0}(E)\neq f_{1}(E),\,\,\forall E\in I.\label{eq:goodband}
\end{equation}
Otherwise, we call $I$ \emph{bad}.
\end{defn}

In this subsection we treat good bands and show that they inherit
the usual Sturmian combinatorial structure. Bad bands will be treated
separately in Subsection \ref{subsec:Sturmian-tree}.
\begin{thm}
\label{thm:Backwards-type}Let $I\subset\spec{\Jan n}$ be a good
spectral band. Then $I$ possesses a distinct backward type ($A$
or $B$).
\end{thm}

\begin{proof}[Proof of Theorem \ref{thm:Backwards-type}]
 We fix a good spectral band $I\subset\spec{\Jan n}$. Our goal is
to show that the entire band $I$ is contained either in a single
band of $\spec{\Jan{n-1}}$ (type $A$) or in a single band of $\spec{\Jan{n-2}}$
but not in any band of $\spec{\Jan{n-1}}$ (type $B$), and that this
inclusion is strict. The idea of the proof is to use Lemma \ref{lem:T-H-correspondence}
to show that the backward type of $I$ is inherited from the type
of the associated band in $\spec{H_{\alpha_{n}}}$.

Fix some $\overline{E}\in I\subset\spec{\Jan n}$. By Lemma \ref{lem:T-H-correspondence}
\begin{equation}
\overline{E}\in\spec{H_{\alpha_{n}}^{V\left(\overline{E}\right)}},\label{eq:E-bar}
\end{equation}
and let $I_{\overline{E}}^{H_{\alpha_{n}}}(V\left(\overline{E}\right))$
denote the corresponding spectral band. By Lemma \ref{lem:same-components},
the map $E\mapsto I_{E}^{H_{\alpha_{n}}}$ is constant along $I$.
By Definition \ref{def:goodband}, since $I$ is good, we have $f_{0}(E)\neq f_{1}(E)$
for all $E\in I$. Therefore Proposition \ref{prop:distinct-type}
applies to all bands $I_{E}^{H_{\alpha_{n}}}(V\left(E\right))$, and
the backward type of $I_{E}^{H_{\alpha_{n}}}(V_{0},V_{1})$ is independent
of the parameters $(V_{0},V_{1})$ as long as $V_{0}\neq V_{1}$.
It follows that all bands $I_{E}^{H_{\alpha_{n}}}(V\left(E\right))$
with $E\in I$ have the same backward type. We distinguish the two
possible cases.

$\bullet$ If the type of $I_{\overline{E}}^{H_{\alpha_{n}}}\left(V\left(\overline{E}\right)\right)$
is $A$, then by the above we deduce that $I_{E}^{H_{\alpha_{n}}}\left(V\left(E\right)\right)$
are all type $A$. By definition, this means that
\begin{equation}
E\in\spec{H_{\alpha_{n-1}}^{V\left(E\right)}}.\label{eq:type-A-1}
\end{equation}
Applying Lemma \ref{lem:T-H-correspondence} then yields
\begin{equation}
E\in\spec{\Jan{n-1}},\,\,\forall E\in I,\label{eq:type-A-1-1}
\end{equation}
and connectedness of the spectral bands then implies that  $I$ itself
is contained in a distinct spectral band in $\spec{\Jan{n-1}}$. Furthermore,
this inclusion is strict. Indeed, for every $E\in I$ the  band $I_{H}^{E}(V\left(E\right))$
is strictly contained in a band of $\spec{H_{\alpha_{n-1}}^{V\left(E\right)}}$.
If the inclusion $I\subset\spec{\Jan{n-1}}$ were not strict, then
some endpoint $E_{*}$ of $I$ would coincide with an endpoint of
a band of $\spec{\Jan{n-1}}$. Lemma \ref{lem:T-H-correspondence}
would then also imply this for the spectral band $I_{E}^{H_{\alpha_{n}}}(V\left(E_{*}\right))$
, which is impossible since it is of backward type $A$ (and thus
satisfies strict inclusion). Therefore the inclusion of $I$ in $\spec{\Jan{n-1}}$
is strict, and $I$ is of backward type $A$.

$\bullet$ If the type of $I_{\overline{E}}^{H_{\alpha_{n}}}\left(V\left(\overline{E}\right)\right)$
is $B$, then using the same arguments as in the previous case, we
conclude that $I$ is strictly contained in a spectral band of $\spec{\Jan{n-2}}$.
To deduce that $I$ is of backward type $B$, we just need to show
that it is not contained in any spectral band of $\spec{\Jan{n-1}}$.
In other words, that there is some $E_{*}\in I$ such that $E_{*}\notin\spec{\Jan{n-1}}$.
By Lemma \ref{lem:T-H-correspondence} combined with (\ref{eq:FB-theory}),
this is equivalent to
\begin{equation}
\left|tr\left(\MNV{n-1}{V\left(E_{*}\right)}\left(E_{*}\right)\right)\right|>2.\label{eq:larg-trace}
\end{equation}
This does not follow automatically from the fact that $I_{\overline{E}}^{H_{\alpha_{n}}}\left(V\left(\overline{E}\right)\right)$
is of type $B$, since here the same energy $E_{*}$ both determines
the effective potentials and serves as the spectral parameter. We
now choose $E_{*}$ so that this happens. For ease of notation, denote
\begin{equation}
T\left(E,V_{0},V_{1}\right):=tr\left(\MNV{n-1}{V_{0},V_{1}}\left(E\right)\right).\label{eq:short-trace}
\end{equation}

For each $E\in I$, consider the spectral band $I_{\overline{E}}^{H_{\alpha_{n}}}\left(V\left(E\right)\right)$.
Since $I_{\overline{E}}^{H_{\alpha_{n}}}\left(V\left(\overline{E}\right)\right)$
is of type $B$ by assumption, then, since $I$ is good and namely
$f_{0}\left(E\right)\neq f_{1}\left(E\right)$ for all $E\in I$,
the discussion above shows that $I_{\overline{E}}^{H_{\alpha_{n}}}\left(V\left(E\right)\right)$
are all of type $B$ for $E\in I$. This means that for all $E\in I$,
the set
\begin{align}
J\left(E\right) & :=\left\{ E'\in I_{\overline{E}}^{H_{\alpha_{n}}}\left(V\left(E\right)\right):\left|T\left(E',V\left(E\right)\right)\right|>2\right\} \label{eq:strong-B-set}\\
 & =\left\{ E'\in I_{\overline{E}}^{H_{\alpha_{n}}}\left(V\left(E\right)\right):E'\notin\spec{H_{\alpha_{n-1}}^{V\left(E\right)}}\right\} \nonumber 
\end{align}
is nonempty. $J\left(E\right)$ is also open, and consists of a finite
union of disjoint open intervals. The two endpoints of each connected
component of $J\left(E\right)$ are characterized by the equalities
$\left\{ T\left(E',V\left(E\right)\right)=\pm2\right\} $. By \cite[prop. 4.1]{band2024review},
these endpoints satisfy $\partial_{E}T\neq0$, and so by the implicit
function theorem, they depend smoothly on $E$. We select one such
connected component of $J\left(E\right)$, and parameterize its closure
as $\left[a\left(E\right),b\left(E\right)\right]$, where $a,b:I\rightarrow\R$
are smooth functions.

Define $\varphi\left(E\right)=\frac{b\left(E\right)+a\left(E\right)}{2}$
to be the (smooth) middle point of $\left[a\left(E\right),b\left(E\right)\right]$.
We seek to find $E_{*}\in I$ such that $\varphi\left(E_{*}\right)=E_{*}$.
Consider the (continuous) auxiliary function
\begin{align}
 & g:I\rightarrow\R,\label{eq:midpoint1}\\
 & g\left(E\right)=E-\varphi\left(E\right).\label{eq:midpoint2}
\end{align}
Writing $I=\left[E_{\min},E_{\max}\right]$, we have that 
\begin{equation}
tr\left(\MNV n{V\left(E_{\min}\right)}\left(E_{\min}\right)\right)=tr\left(\MNJ n\left(E_{\min}\right)\right)\in\left\{ \pm2\right\} ,\label{eq:trace-equality}
\end{equation}
and hence $E_{\min}$ is an endpoint of $I_{\overline{E}}^{H_{\alpha_{n}}}\left(V\left(E_{\min}\right)\right)$.
Without loss of generality, it is the left endpoint. Hence, $E_{\min}<\varphi\left(E_{\min}\right)$,
which means  that $g\left(E_{\min}\right)<0$. Similarly, $E_{\max}$
is the right endpoint of $I_{\overline{E}}^{H_{\alpha_{n}}}\left(V\left(E_{\max}\right)\right)$,
implying that $\varphi\left(E_{\max}\right)<E_{\max}$ and so $g\left(E_{\max}\right)>0$.
By the intermediate value theorem, there exists some $E_{*}\in I$
such that $\varphi\left(E_{*}\right)=E_{*}$.

By construction, $E_{*}=\varphi\left(E_{*}\right)\in\left[a\left(E_{*}\right),b\left(E_{*}\right)\right]\subset J\left(E_{*}\right)$,
hence
\begin{equation}
\left|T\left(E_{*},V\left(E_{*}\right)\right)\right|>2.\label{eq:large-trace}
\end{equation}
We thus obtain an energy in $I$ which does not belong to $\spec{\Jan{n-1}}$,
showing that $I$ is of backward type $B$.
\end{proof}
In the next subsection, we show that, although certain bad bands may
cause degeneracies, a weaker version of forward (and backward) type
holds, and the same combinatorial tree structure persists up to possible
identification of vertices.

\subsection{Sturmian tree structure near bad bands\label{subsec:Sturmian-tree}}

We now turn to the existence of a well-defined forward type for the
periodic approximants. The additional subtlety arises at energies
where the two effective potentials coincide. We first record what
the equality $f_{0}(E)=f_{1}(E)$ implies for the periodic approximants.
We then introduce a weaker notion of forward type and show that it
holds for all bands.

Define the set of bad energies:
\begin{equation}
\E_{\mathrm{bad}}:=\{E\in\mathbb{R}:f_{0}(E)=f_{1}(E)\}.\label{eq:bad-energy}
\end{equation}
Thus, away from $\E_{\text{bad}}$, the associated effective Sturmian
Hamiltonian has nonzero coupling.
\begin{lem}
\label{lem:spectral-edges}Let $E\in\mathcal{E}_{\mathrm{bad}}$.
Then
\[
E\in\spec{\Jan{n_{0}}}\text{ for some }n_{0}\quad\Longleftrightarrow\quad E\in\spec{\Jan n}\ \forall n\in\mathbb{N}\quad\Longleftrightarrow\quad E\in\spec{\Ja}.
\]
\end{lem}

\begin{proof}
Assume that $f_{0}(E)=f_{1}(E)=:f(E)$. Then the one-step transfer
matrices of $\Ja$ at energy $E$ are independent of the letter $\omega_{\alpha}(n)$,
and are given by
\begin{equation}
\mathcal{M}(E):=\begin{pmatrix}E-f(E) & -1\\
1 & 0
\end{pmatrix}.\label{eq:M(E)}
\end{equation}
Consequently,
\begin{equation}
\MNJ n(E)=\mathcal{M}(E)^{q_{n}},\,\,\,\forall n\in\N.\label{eq:MnE}
\end{equation}

If $E\in\spec{\Jan{n_{0}}}$, then by Proposition \ref{prop:Sturmian-properties-1}
we have
\begin{equation}
\big|tr\big(\mathcal{M}(E)^{q_{n_{0}}}\big)\big|\le2.\label{eq:t<2-1}
\end{equation}
Since $\mathcal{M}(E)\in SL\left(2,\R\right)$, the eigenvalues of
$\mathcal{M}(E)$ are of the form $\lambda,\lambda^{-1}$ for some
$\lambda\in\mathbb{C}\setminus\{0\}$, and either $|\lambda|=1$,
or $\lambda,\lambda^{-1}\in\mathbb{R}$. In the second case,
\begin{equation}
\left|\lambda^{q}+\lambda^{-q}\right|>2,\,\,\,\forall q\in\N.\label{eq:hyperbolic}
\end{equation}
Hence (\ref{eq:t<2-1}) already implies that $|\lambda|=1$, and therefore
\begin{equation}
tr\big(\mathcal{M}(E)^{q_{n}}\big)=\left|\lambda^{q_{n}}+\lambda^{-q_{n}}\right|\leq2\qquad\forall n\in\mathbb{N}.\label{eq:smalltrace}
\end{equation}
It follows from Proposition \ref{prop:Sturmian-properties-1} that
\begin{equation}
E\in\spec{\Jan n}\quad\forall n\in\N,\label{eq:EinS}
\end{equation}
and therefore $E\in\spec{\Ja}$.

Conversely, if $E\in\spec{\Ja}$, then by Proposition \ref{prop:Sturmian-properties-1},
\begin{equation}
E\in\spec{\Jan n}\cup\spec{\Jan{n+1}}\label{eq:E-in-union}
\end{equation}
for every $n$. Hence $E$ belongs to at least one periodic spectrum,
and the first part applies.
\end{proof}
\begin{figure}
\includegraphics[scale=0.6]{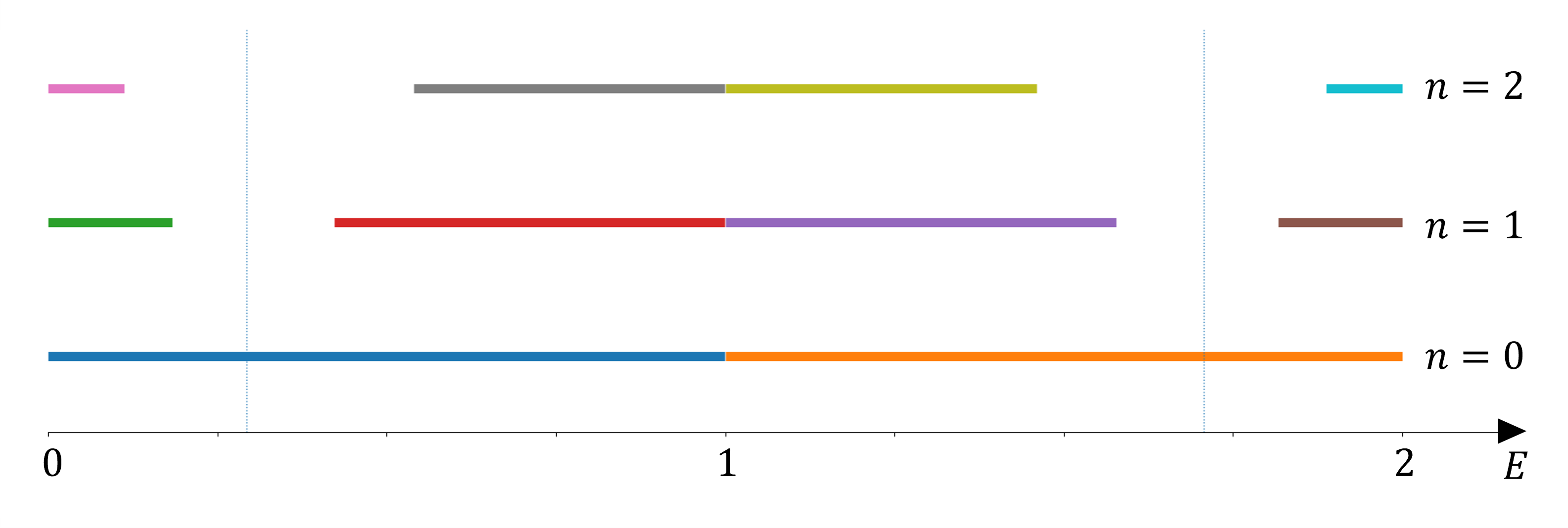}

\caption[Periodic approximations for NDL on Sturmian comb.]{The first few spectral approximations for a Fibonacci comb graph ($\alpha=\frac{\sqrt{5}-1}{2}$)
equipped with the NDL $\Delta_{\alpha}$. Most bands display the usual
forward structure, while the bands that touch $E=0,1,2$ only display
the weaker forward structure. \label{fig:Bad bands}}
\end{figure}

At energies in $\mathcal{E}_{\mathrm{bad}}$, the usual strict band
structure may degenerate, as illustrated in Figure \ref{fig:Bad bands}
through an example from the normalized discrete Laplacian\footnote{Recall that the normalized discrete Laplacian is not horizontally
homogeneous. Nevertheless, as explained in Subsection \ref{subsec:NDL-DTMP}
it still has a well-defined spectral tree, and we use it here as an
instructive example of the possible behavior of bad bands.}. To deal with this, we slightly modify the definition of spectral
bands in a way which preserves the usual Sturmian tree structure. 
\begin{defn}
For each $n\ge0$, consider the connected components of the set
\begin{equation}
\mathcal{O}_{n}^{\alpha}:=\left\{ E\in\R:\left|t_{n}^{\alpha}\left(E\right)\right|<2\right\} .\label{eq:open-band-1}
\end{equation}
Define
\begin{equation}
\mathcal{B}_{n}^{\alpha}:=\left\{ \overline{C}:C\text{ is a connected component of }\mathcal{O}_{n}^{\alpha}\right\} .\label{eq:closed-tree-1}
\end{equation}
\end{defn}

By construction, distinct elements of $\mathcal{B}_{n}^{\alpha}$
have disjoint interiors (although they may share endpoints), and we
think of these elements as the ``genuine'' spectral bands which
will be used to define forward types. 

For two intervals $I_{1}=\left[a_{1},b_{1}\right],I_{2}=\left[a_{2},b_{2}\right]$,
we say that $I_{1}\preceq I_{2}$ if $b_{1}\leq a_{2}$.
\begin{defn}
\label{def:weak-forward}Let $I\in\mathcal{B}_{n}^{\alpha}$ be a
spectral band.

1. We say that $I$ is of strong forward type $A$ (resp. $B$) with
$M=a_{n+1}-1$ (resp. $M=a_{n+1}$) if Definition \ref{def:forward-type}
holds verbatim for $\Jan n$.

2. We say that $I$ is of weak forward type $A$ (resp. $B$) with
$M=a_{n+1}-1$ (resp. $M=a_{n+1}$) if Definition \ref{def:forward-type}
holds verbatim for $\Jan n$, except that all inclusions $\subset_{\text{strict}}$
are replaced with $\subset$, and the interlacing relation $\prec$
is replaced with $\preceq$.
\end{defn}

In other words, weak type allows a band to touch (share an endpoint
with) its predecessor, but without change in the ordering or in the
number of descendants. We now prove that all spectral bands admit
a well-defined forward type (either weak or strong). The idea of the
proof is similar to that of Theorem \ref{thm:Backwards-type} --
using the forward structure of $H_{\alpha_{n}}^{f_{0}\left(E\right),f_{1}\left(E\right)}$
in order to ``pull back'' a similar forward structure to $\Jan n$.
The only delicate point in the proof is the fact that possibly $f_{0}\left(E\right)=f_{1}\left(E\right)$
at some of the bad spectral band endpoints, which may cause spectral
band edges to touch.

We first record a monotonicity observation.
\begin{lem}
\label{lem:monotone-pullback}Let $I^{H}\left(V\right)$ be a band
map for $H_{\alpha_{m}}^{V_{0},V_{1}}$ as in Definition \ref{def:frozen band}.
On any interval where $V\left(E\right)$ is defined, the set
\begin{equation}
\left\{ E:E\in I^{H}\left(V\left(E\right)\right)^{\circ}\right\} \label{eq:pullback-band}
\end{equation}
is a (possibly empty) open interval.
\end{lem}

\begin{proof}
By (\ref{eq:f-G}), for $i\in\left\{ 0,1\right\} $,
\begin{equation}
\frac{d}{dE}\left(E-f_{i}\left(E\right)\right)=-\frac{\mathcal{G}_{i}'(E)}{c\mathcal{G}_{i}(E)^{2}},\qquad\mathcal{G}_{i}'(E)=-\left\Vert (E-\mathcal{J}_{G_{i}})^{-1}e_{v_{i}}\right\Vert ^{2}<0.\label{eq:monotone-green}
\end{equation}
Hence $E-f_{0}\left(E\right)$ and $E-f_{1}\left(E\right)$ are both
monotone in the same direction.

Write $I^{H}\left(V\right)=\left[a\left(V\right),b\left(V\right)\right]$.
By the min-max principle, for either band edge $d\in\left\{ a,b\right\} $,
\begin{equation}
\min_{i}(V'_{i}-V_{i})\le d(V')-d(V)\le\max_{i}(V'_{i}-V_{i}).\label{eq:perturbation}
\end{equation}
It follows that $E-a\left(V\left(E\right)\right)$ and $E-b\left(V\left(E\right)\right)$
are both strictly monotone in the same direction. Therefore the set
\begin{equation}
\left\{ E:E\in I^{H}\left(V\left(E\right)\right)^{\circ}\right\} =\left\{ E:a\left(V\left(E\right)\right)<E<b\left(V\left(E\right)\right)\right\} \label{eq:pullback-band-1}
\end{equation}
is an open interval, possibly empty.
\end{proof}
\begin{thm}
\label{thm:bad-type}Let $I\in\mathcal{B}_{n}^{\alpha}$ be a spectral
band. Then $I$ has a well-defined forward type $\tau\in\{A,B\}$.
More precisely, writing
\begin{equation}
M=\begin{cases}
a_{n+1}-1, & \tau=A,\\
a_{n+1}, & \tau=B,
\end{cases}\label{eq:M-type}
\end{equation}
there are exactly $M$ spectral bands of $\mathcal{B}_{n+1}^{\alpha}$
contained in $I$ and not contained in $\spec{\Jan{n-1}}$, and exactly
$M+1$ spectral bands of $\mathcal{B}_{n+2}^{\alpha}$ contained in
$I$ and not contained in $\spec{\Jan{n+1}}$. 

These bands satisfy the interlacing property of Definition \ref{def:forward-type},
with $\prec$ possibly replaced by $\preceq$, and with $\subset_{\mathrm{strict}}$
possibly replaced by $\subset$. If one of these relations is not
strict, then the two corresponding bands share an endpoint in $\ebad$.
In particular, if $I\cap\ebad=\varnothing$, then the type is strong.
\end{thm}

\begin{proof}
Choose $E_{0}\in I^{\circ}\setminus\ebad$, and let
\begin{equation}
I^{H}(V):=I_{E_{0}}^{H_{\alpha_{n}}}(V)\label{eq:bandmap}
\end{equation}
be the corresponding band map. We first claim that
\begin{equation}
I^{\circ}=\left\{ E:E\in I^{H}(V(E))^{\circ}\right\} .\label{eq:monotonepullback-I}
\end{equation}
Indeed, by Lemma \ref{lem:monotone-pullback}, the set on the right-hand
side is an interval. It contains $E_{0}$ and is contained in $\mathcal{O}_{n}^{\alpha}$,
and hence in $I^{\circ}$. If the inclusion were strict, one of its
endpoints would lie in $I^{\circ}$. At this endpoint, $\left|t_{n}^{\alpha}(E)\right|=2$,
contrary to the definition of $\mathcal{O}_{n}^{\alpha}$. This proves
the claim.

Since $E_{0}\notin\ebad$, Proposition \ref{prop:type-existence}
gives that $I^{H}(V(E_{0}))$ has a forward type $\tau\in\{A,B\}$
independent of $V_{0}\ne V_{1}$. We define this to be the forward
type of $I$, and let $M$ be as in the statement. We first prove
the assertion at level $\alpha_{n+1}$. For $V_{0}\ne V_{1}$, denote
by
\begin{equation}
I_{1}^{(n+1)}(V),\ldots,I_{M}^{(n+1)}(V)\label{eq:descendant-bands}
\end{equation}
the $M$ level-$\alpha_{n+1}$ Sturmian bands contained in $I^{H}(V)$
and not contained in $\spec{H_{\alpha_{n-1}}^{V}}$, ordered from
left to right. Define
\begin{equation}
D_{j}:=\left\{ E\in I^{\circ}:E\in I_{j}^{(n+1)}(V(E))^{\circ}\right\} .\label{eq:pullback-Dj}
\end{equation}
By Lemma \ref{lem:monotone-pullback}, each $D_{j}$ is either empty
or an open interval. We claim that it is nonempty. Writing
\begin{equation}
I_{j}^{(n+1)}(V(E))=[a_{j}(E),b_{j}(E)],\label{eq:I_j-parametrization}
\end{equation}
the monotonicity argument from Lemma \ref{lem:monotone-pullback}
shows that
\begin{equation}
E-\frac{a_{j}(E)+b_{j}(E)}{2}\label{eq:monotone-center}
\end{equation}
is strictly monotone. As $E$ approaches the two endpoints of $I^{\circ}$,
this expression has opposite signs, since 
\begin{equation}
I_{j}^{(n+1)}(V(E))\subset I^{H}(V(E)).\label{eq:forward-inclusion}
\end{equation}
It therefore vanishes at some $E\in I^{\circ}$, and hence $D_{j}\ne\varnothing$.
We next show that $D_{j}$ is a connected component of $\mathcal{O}_{n+1}^{\alpha}$.
Every endpoint $E_{*}$ of $D_{j}$ is a band edge of $I_{j}^{(n+1)}(V(E_{*}))$.
This is immediate if $E_{*}\in I^{\circ}$; if $E_{*}\in\partial I$,
it follows from
\begin{equation}
I_{j}^{(n+1)}(V(E))\subset I^{H}(V(E)).\label{eq:forward-inclusion-1}
\end{equation}
Thus $\left|t_{n+1}^{\alpha}(E_{*})\right|=2$. Consequently, $D_{j}$
is a connected component of $\mathcal{O}_{n+1}^{\alpha}$, and
\begin{equation}
\overline{D_{j}}\in\mathcal{B}_{n+1}^{\alpha},\qquad\overline{D_{j}}\subset I.\label{eq:component-inclusion}
\end{equation}

Since $I_{j}^{(n+1)}(V)$ is not contained in $\spec{H_{\alpha_{n-1}}^{V}}$,
choose one of the corresponding components
\begin{equation}
K_{j}(V)\subset I_{j}^{(n+1)}(V)^{\circ}\setminus\spec{H_{\alpha_{n-1}}^{V}},\label{eq:K_j}
\end{equation}
continued through its band edges as $V$ varies. Applying the same
pullback argument to $K_{j}(V)$ gives an energy
\begin{equation}
E\in D_{j}\quad\text{such that}\quad E\notin\spec{\Jan{n-1}}.\label{eq:D_j-energies}
\end{equation}
Therefore
\begin{equation}
\overline{D_{j}}\not\subset\spec{\Jan{n-1}}.\label{eq:D_j-not-included}
\end{equation}

It remains to show that there are no other such bands. Let $J\in\mathcal{B}_{n+1}^{\alpha}$
be contained in $I$ but not contained in $\spec{\Jan{n-1}}$. Since
$\ebad$ is finite, we may choose
\begin{equation}
E\in J^{\circ}\setminus\left(\spec{\Jan{n-1}}\cup\ebad\right).\label{eq:good-energy}
\end{equation}
At this energy the usual Sturmian forward structure applies, and hence
$E\in D_{j}$ for some $j$. Since $J^{\circ}$ and $D_{j}$ are connected
components of $\mathcal{O}_{n+1}^{\alpha}$, we obtain $J=\overline{D_{j}}$.
This proves the assertion at level $\alpha_{n+1}$.

The proof at level $\alpha_{n+2}$ is identical, using the $M+1$
level-$\alpha_{n+2}$ Sturmian bands contained in $I^{H}(V)$ and
not contained in $\spec{H_{\alpha_{n+1}}^{V}}$. This gives exactly
$M+1$ elements of $\mathcal{B}_{n+2}^{\alpha}$ contained in $I$
and not contained in $\spec{\Jan{n+1}}$. Away from $\ebad$, the
corresponding Sturmian bands satisfy the usual strict inclusions and
interlacing. Their pullbacks therefore satisfy the same ordering,
with equality possible only when two of the corresponding closed bands
share an endpoint in $\ebad$. Thus the forward type is weak in general,
and is strong when $I\cap\ebad=\varnothing$.
\end{proof}
Combining Theorems \ref{thm:Backwards-type} and \ref{thm:bad-type},
the backward and forward types agree whenever both are defined. Motivated
by this and Proposition \ref{prop:type-existence}, we define the
weak backward type of a band to be equal to its forward type, weak
or strong, and simply refer to the common label as the type of the
band.
\begin{thm}
\label{lem:Sturmian-type} Every element of $\mathcal{B}_{n}^{\alpha}$
has a well-defined weak type. If $I\in\mathcal{B}_{n}^{\alpha}$ is
good, then this type is strong.
\end{thm}

This will allow us to use the Sturmian tree structure to define a
spectral tree, as described in the following subsection.

\subsection{The modified Sturmian tree and possible closed gaps\label{subsec:closed-gaps}}

In the preceding subsection we established that for the family $\Jan n$,
the local combinatorial structure of band types still holds, where
around the set $\ebad$ one should possibly consider the weak type.
We now use this structure to define the associated spectral tree and
compare its labels with the gap labels of the IDS. 

To separate the combinatorics from the possible identification of
bands that share endpoints, we define the spectral tree $\T\left(\Ja\right)$
as follows: the vertices at level $\alpha_{n}$ are the elements of
$\B_{n}^{\alpha}$, and there is a directed edge $I\to I'$ whenever
$I'\subset I$. The forward structure proved in Theorem \ref{thm:bad-type}
implies that the branching rule of $\T\left(\Ja\right)$ coincides
with the classical Sturmian branching rule at every level.

An infinite path $\gamma\in\partial\T\left(\Ja\right)$ is a nested
sequence
\begin{equation}
I_{0}\supset I_{1}\supset I_{2}\supset\dots,\qquad I_{n}\in\mathcal{B}_{n}^{\alpha}.\label{eq:tree-path}
\end{equation}
As in the Sturmian case described in Subsection \ref{subsec:Sturmian-tree-1},
the intersection $\bigcap_{n\geq0}I_{n}$ consists of a single energy
$E\in\spec{\Ja}$, which allows us to identify paths in $\partial\T\left(\Ja\right)$
with energies in $\spec{\Ja}$, and also to assign a gap label through
$N_{\Ja}\left(\gamma\right):=N_{\Ja}\left(E\left(\gamma\right)\right)$.
In the usual Sturmian case, this identification $\gamma\mapsto E\left(\gamma\right)$
is a bijection \cite[thm. 1.10]{Band2024}. In our present setting
it remains surjective, but may fail to be injective at paths that
pass through bad bands (see Lemma \ref{lem:not-injective} below).
By standard properties of the IDS, a number $C\in\left(0,1\right)$
is a gap label if and only if there exist two distinct energies $E_{1,}E_{2}\in\R$
such that $N_{\Ja}\left(E_{1}\right)=N_{\Ja}\left(E_{2}\right)=C$.
With this in mind, we define the set of tree gap labels by
\begin{align}
 & \mathcal{GL}\left(\T\left(\Ja\right)\right):=\label{eq:TREE-GL}\\
 & \left\{ C\in\left(0,1\right):\exists\gamma_{1}\neq\gamma_{2}\in\partial\T\left(\Ja\right),\NHE{\Ja}\left(\gamma_{1}\right)=\NHE{\Ja}\left(\gamma_{2}\right)=C\right\} \cup\left\{ 0,1\right\} .\nonumber 
\end{align}
A tree gap label fails to correspond to an actual open spectral gap
in $\GL{\Ja}$ only if the two distinct tree paths above correspond
to the same energy $E\in\spec{\Ja}$, i.e. if the map $\gamma\mapsto E_{\alpha}\left(\gamma\right)$
is not injective. We now analyze when this can occur.
\begin{lem}
\label{lem:not-injective}Let $\gamma_{1}=(I_{n})$ and $\gamma_{2}=(J_{n})$
be two distinct boundary paths. If $E_{\alpha}\left(\gamma_{1}\right)=E_{\alpha}\left(\gamma_{2}\right)=E$,
then there exists $n_{0}\in\N$ such that for every $n\geq n_{0}$
\begin{equation}
I_{n}\ne J_{n},\qquad E\in\partial I_{n}\cap\partial J_{n}.\label{eq:touching-bands}
\end{equation}
Moreover, $E\in\ebad$.
\end{lem}

\begin{proof}
Since the two paths are distinct, there exists $n_{0}$ such that
$I_{n_{0}}\ne J_{n_{0}}$. For $n\ge n_{0}$, the intervals $I_{n}$
and $J_{n}$ remain distinct. Indeed, an element of $\mathcal{B}_{n+1}^{\alpha}$
has nonempty interior and cannot be contained in two distinct elements
of $\mathcal{B}_{n}^{\alpha}$, whose interiors are disjoint. Both
$I_{n}$ and $J_{n}$ contain $E$. Since they are distinct elements
of $\mathcal{B}_{n}^{\alpha}$, their interiors are disjoint. Hence
$E$ cannot lie in the interior of either one of them, and therefore
\begin{equation}
E\in\partial I_{n}\cap\partial J_{n},\qquad n\ge n_{0}.\label{eq:touching-boundaries}
\end{equation}

Lastly, by Theorem \ref{thm:bad-type}, two distinct elements of $\mathcal{B}_{n}^{\alpha}$
can share an endpoint only at an energy in $\ebad$. Since
\[
E\in\partial I_{n}\cap\partial J_{n},\qquad n\ge n_{0},
\]
we conclude that $E\in\ebad$.
\end{proof}
We may now compare the tree gap labels $\mathcal{GL}\left(\T\left(\Ja\right)\right)$
with the actual gap labels $\GL{\Ja}$ of the IDS.
\begin{thm}
\label{thm:closed-gaps}$\mathcal{GL}\left(\T\left(\Ja\right)\right)\setminus\GL{\Ja}\subset\{N_{\Ja}(E):E\in\E_{\mathrm{bad}}\}.$
\end{thm}

\begin{proof}
Let
\begin{equation}
C\in\mathcal{GL}\left(\T\left(\Ja\right)\right)\setminus\GL{\Ja}.\label{eq:unattained-label}
\end{equation}
By definition of $\mathcal{GL}\left(\T\left(\Ja\right)\right)$, there
exist two distinct paths $\gamma_{1},\gamma_{2}\in\partial\T(\Ja)$
such that
\begin{equation}
\NHE{\Ja}(\gamma_{1})=\NHE{\Ja}(\gamma_{2})=C.\label{eq:joint-label}
\end{equation}
If $E_{\alpha}(\gamma_{1})\ne E_{\alpha}(\gamma_{2})$, then $C\in\GL{\Ja}$,
a contradiction. Hence
\begin{equation}
E_{\alpha}(\gamma_{1})=E_{\alpha}(\gamma_{2})=:E.\label{eq:Same-label}
\end{equation}
By Lemma \ref{lem:not-injective}, $E\in\ebad$. Therefore $C=\NHE{\Ja}(E)$
for some $E\in\ebad$.
\end{proof}
The theorem above shows that the tree $\T\left(\Ja\right)$ (which
displays the usual Sturmian combinatorial structure) can be used to
compute the gap labels $\GL{\Ja}$, up to the bad energies in $\E_{\mathrm{bad}}$,
which may correspond to closed gaps. For a Sturmian comb graph equipped
with the adjacency matrix $\Aa$, we will show in Subsection \ref{subsec:tree-description-adj}
that no such energies exist, and so the analysis can be carried out
entirely through the spectral tree.

\subsection{Singular energies and flat bands\label{subsec:singular-energies}}

Recall Assumption \ref{assu:no-flat-bands}, which asserted that the
set $S$ of poles for the transfer matrices satisfies $S\cap\spec{\Ja}=\varnothing$.
We call the set $S$ the set of singular energies. At the points of
$S\cap\spec{\Ja}$, which we call singular spectral points, the analysis
developed above using the transfer matrices cannot be applied. We
now outline the possible effect of such points on the spectrum and
its combinatorial structure.

Let $E\in S\cap\spec{\Ja}$ be a singular spectral point. By the results
of Subsection \ref{subsec:Generic-Cantor-spectrum}, $E$ corresponds
to a flat band. Denoting the set of poles of $f_{0}$ and $f_{1}$
by $S_{0}$ and $S_{1}$ respectively, we have that $S=S_{0}\cup S_{1}$,
and these points will generally correspond to poles of the transfer
matrices $\MNJ n$. Moreover, it may happen that a spectral band of
some periodic approximant $\Jan n$ contains an energy which is singular
only for a higher approximant $\Jan m$ with $m>n$. In that case,
the transfer matrix description breaks down inside that band at the
higher level, and the arguments from previous subsections proving
the existence of a well-defined backward/forward type can no longer
be applied across the whole band. On the other hand, let $I\subset\R$
such that the transfer matrices $\MNJ n\left(E\right)$ are well-defined
for all $E\in I$ and all $n\geq N$ for some $N$. Then, restricted
to $I$, the arguments of the previous subsections apply verbatim
from level $N$ onward, and the usual forward structure holds for
all spectral bands in $I$.

Thus, for the purpose of analyzing the gap labels using the Sturmian
tree, one may apply the analysis above separately to each spectral
band $I\subset\spec{\Jan N}$ along which all higher transfer matrices
are well-defined. The exceptional bands for which this fails must
be treated separately. In particular, one must check individually
whether they contain singular spectral points (flat bands), and how
these points affect the local branching of the resulting spectral
tree. We will explicitly show how this can be done for Sturmian combs
equipped with the adjacency matrix in Section \ref{sec:DTMP-computation}.

\newpage{}

\section{Dry Ten Martini Problem -- the adjacency matrix on Sturmian combs\label{sec:DTMP-computation}}

In this section we prove Theorem \ref{thm:DTMP-Adj}, which determines
the gap labels of the adjacency matrix on the discrete Sturmian combs
from Example \ref{exa: Sturm-comb}. The techniques used here can
be extended to arbitrary Sturmian decorated $\Z$-graphs. We outline
how such a generalization can be carried out at the end of the section.

In the previous section, we constructed an infinite tree $\mathcal{T}\left(\A_{\alpha}\right)$,
which (away from the singular spectral points) exhibits the same local
combinatorial band structure as the usual Sturmian tree $\mathcal{T}\left(H_{\alpha}\right)$.
While the two trees share the same local structure, their global structure
can differ due to two reasons: firstly, the first two levels of the
tree (which serve as an ``initial condition'' that determines the
entire tree) are different. Secondly, as we shall see, $\Aa$ has
a flat band at $E=0$, where the local branching of the tree behaves
differently. We first account for these two factors, and then combine
those with the gap labels of the usual Sturmian Hamiltonians in order
to compute the gap labels of $\Aa$.

\subsection{The spectral tree $\protect\T\left(\protect\Aa\right)$\label{subsec:tree-description-adj}}

The spectral tree of the adjacency matrix on the Sturmian comb can
be computed using the transfer matrices of $\A_{\alpha}$. As follows
from the computation in Example \ref{exa:comb-transfer}, the one-step
transfer matrices for this model are given by:
\begin{align}
 & \M_{\text{trivial}}\left(E\right)=\left(\begin{array}{cc}
E & -1\\
1 & 0
\end{array}\right),\label{eq:M-triv}\\
 & \M_{\text{tooth}}\left(E\right)=\left(\begin{array}{cc}
E-\frac{1}{E} & -1\\
1 & 0
\end{array}\right).\label{eq:M-tooth}
\end{align}
From these expressions, one obtains the following important facts:

1. $\spec{\Aan n}$ is symmetric about $E=0$ for each $n\in\N$,
and consequently so is $\spec{\Aa}$ (this follows since $tr\left(\mathcal{M}_{n}^{\Aa}\left(-E\right)\right)=\left(-1\right)^{n}tr\left(\mathcal{M}_{n}^{\Aa}\left(E\right)\right)$,
preserving the condition (\ref{eq:FB-theory}))

2. $E=0$ is the unique singular energy for this model.

3. The set of bad energies $\E_{\text{bad}}\left(\Aa\right)$ is defined
through the equality
\begin{equation}
\E_{\text{bad}}\left(\Aa\right)=\left\{ E\in\R:E=E-\frac{1}{E}\right\} ,\label{eq:bad-energies}
\end{equation}
and is thus empty. Consequently, $\mathcal{GL}\left(\T\left(\Aa\right)\right)=\GL{\Aa}$,
so the gap labels can be computed entirely from the tree (i.e. no
bands touch).

4. For all $\alpha\notin\Q$, the first periodic approximant satisfies
$\spec{\Aan 1}=\left[-2,2\right]$, and we assign this unique spectral
band backward type $A$ by convention. However, since it contains
the singular energy $E=0$, its forward type is not well-defined and
must be studied separately.

Our goal is therefore to compute the initial structure of the tree
$\mathcal{T}\left(\A_{\alpha}\right)$ and compare it to that of $\mathcal{T}\left(H_{\alpha}\right)$.
In particular, we need to understand the branching of the initial
band $\left[-2,2\right]\subset\spec{\Aan 1}$ compared to $\left[-2,2\right]\subset\spec{H_{\alpha_{1}}}$
(see Figure \ref{fig:Sturmian-tree}), as well as compute $\spec{\Aan 2}$,
which gives the next level of $\mathcal{T}\left(\A_{\alpha}\right)$.
As we shall see, both of these factors will depend on the first continued
fraction expansion digit of $\alpha$ (and mainly on its parity).
Write
\begin{equation}
\alpha=[0;a_{1},a_{2},\dots]:=\frac{1}{a_{1}+\frac{1}{a_{2}+...}}.\label{eq:cf-full}
\end{equation}

\begin{thm}
\textbf{\label{full-gen-2-2}} The initial condition for the spectral
tree $\T(\Aa)$ is as follows.

At level $\alpha_{0}$ there is a single nonflat band $I^{(0)}=[-2,2]$.

At level $\alpha_{1}=1/a_{1}$:

1. If $a_{1}$ is odd, then there are $a_{1}+1$ nonflat bands in
total, and no flat bands. Among them, $a_{1}-1$ bands are contained
in $I^{(0)}$, with the two remaining nonflat bands placed symmetrically
on each side of this set of bands.

2. If $a_{1}$ is even, then there are $a_{1}$ nonflat bands in total,
together with a flat band at $E=0$ of multiplicity one. Among them,
$a_{1}-2$ bands are contained in $I^{(0)}$, with the two remaining
nonflat bands placed symmetrically on each side of this set of bands.

Apart from the band $I^{(0)}=[-2,2]$, all nonflat bands in $\T(\Aa)$
have a well-defined forward type, and display the usual Sturmian branching
rule from Definition \ref{def:forward-type}.
\end{thm}

\begin{figure}
\includegraphics[scale=0.42]{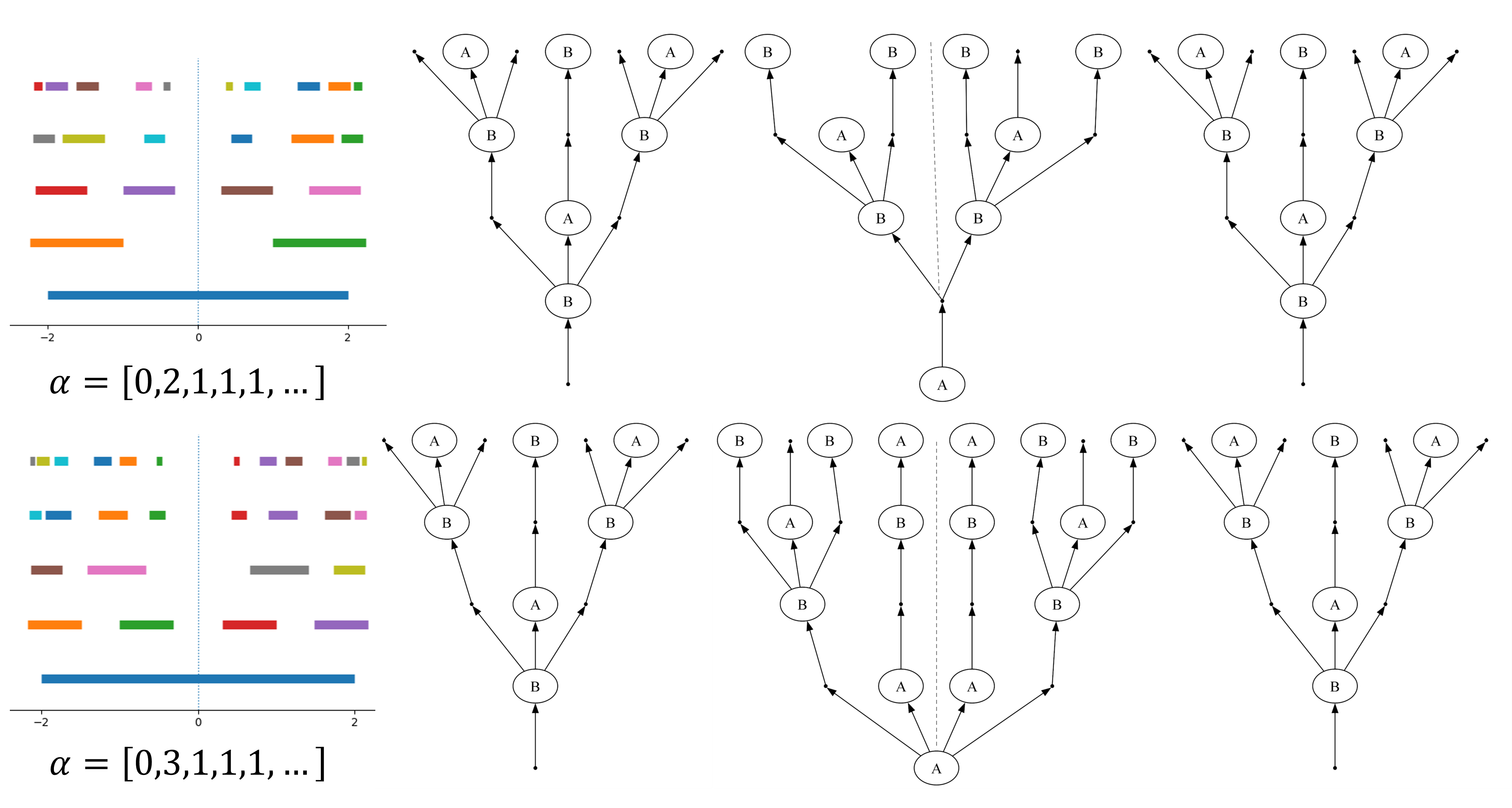}

\caption[~~The spectral tree $\T\left(\Aa\right)$.]{The spectral bands of the periodic approximants and the corresponding
spectral tree $\protect\T\left(\protect\Aa\right)$ for two different
choices of $\alpha$ -- one with $a_{1}$ even, and one with $a_{1}$\LyXZeroWidthSpace{}
odd. The vertices of the tree correspond to the spectral bands shown
on the left, and the edges describe how these bands are connected
across the different periodic approximants. The parity of $a_{1}$
affects the initial condition of the tree and, consequently, the resulting
missing branch. The dashed vertical lines indicate the same obstruction
in the two pictures: on the left, they mark the flat bands at $E=0$,
while on the right, they mark the location of the missing branch which
collapses into these flat bands. \label{fig:Adjacency-tree}}
\end{figure}

Thus, the only deviation from the usual forward structure occurs in
the initial $A$ band, and is entirely controlled by the parity of
$a_{1}$. The proof of Theorem \ref{full-gen-2-2} appears in the
next subsection. For now, we use the theorem to explicitly describe
the full structure of the spectral tree $\T\left(\Aa\right)$.

The first level consists of a single type $A$ band $I^{(0)}=[-2,2]$.
At the next level (corresponding to $\alpha_{1}=1/a_{1}$), the nonflat
spectrum splits into three parts: two outer bands lying outside $[-2,2]$
(of type $B$), and a collection of bands contained in $(-2,2)$,
which are connected to $I^{(0)}$ in the tree. The spectrum is symmetric
with respect to $E=0$, and all nonflat bands occur in $\pm$-pairs.

The part of the tree above $I^{(0)}$ is determined by the parity
of $a_{1}$, and is most naturally described in terms of a single
missing central branch. At level $\alpha_{1}$, the number of nonflat
bands contained in $(-2,2)$ is:
\begin{itemize}
\item $a_{1}-1$ if $a_{1}$ is odd,
\item $a_{1}-2$ if $a_{1}$ is even,
\end{itemize}
in accordance with Theorem \ref{full-gen-2-2}. 

In terms of the Sturmian combinatorial structure, the modification
is simple: relative to $\T\left(H_{\alpha}\right)$, one removes a
single central branch of the tree above $I^{(0)}$. The parity of
$a_{1}$ determines where this branch begins. If $a_{1}$ is even,
the removed branch starts already at level $\alpha_{1}$; if $a_{1}$
is odd, all level-$\alpha_{1}$ vertices remain, and the removed branch
starts only at level $\alpha_{2}$. Apart from this single missing
branch, all remaining vertices propagate according to the usual Sturmian
forward branching rule.

We may think of $\T\left(\Aa\right)$ as consisting of two parts:
one isomorphic to the right subtree of $\T\left(H_{\alpha}\right)$
(cf. \cite[sec. 1.3]{Band2024} and Figure \ref{fig:Sturmian-tree}),
and one obtained from the Sturmian tree $\T(H_{\alpha})$ by removing
a single central branch, whose root depends on the parity of $a_{1}$
(see Figure \ref{fig:Adjacency-tree}). We now make this precise.
To do so, we first define the following subtree of $\T\left(H_{\alpha}\right)$:
\begin{defn}
\label{def:T-red}Write $\alpha=[0;a_{1},a_{2},\ldots]$. Let $v_{0}$
be the initial $A$-vertex of $\T(H_{\alpha})$, and equip each descendant
level with the natural left-to-right ordering. Denote by $u_{1}\prec\cdots\prec u_{a_{1}-1}$
the descendant vertices of $v_{0}$ at level $\alpha_{1}$, and by
$w_{1}\prec\cdots\prec w_{a_{1}}$ the descendants of $v_{0}$ at
level $\alpha_{2}$ (so that the $u_{j}$ and $w_{j}$ interlace).
We define the reduced subtree $\T^{red}(H_{\alpha})\subset\T(H_{\alpha})$
by removing a single central branch rooted at
\begin{equation}
r:=\begin{cases}
u_{a_{1}/2}, & a_{1}\text{ is even,}\\
w_{(a_{1}-1)/2+1}, & a_{1}\text{ is odd.}
\end{cases}\label{eq:branch-root}
\end{equation}
That is, we remove $r$, together with all of its descendants, and
retain all remaining vertices with the usual Sturmian forward branching
rule. The two possible locations of $r$, according to the parity
of $a_{1}$, are illustrated in Figure \ref{fig:Adjacency-tree}.
\end{defn}

\begin{lem}
\label{lem:tree-embedding}The spectral tree $\T(\Aa)$ decomposes
into two disjoint subtrees
\begin{equation}
\T(\Aa)=\T_{L}(\Aa)\sqcup\T_{R}(\Aa),\label{eq:tree-decomposition}
\end{equation}
ordered from left to right according to the natural ordering of the
corresponding spectral bands, with the following properties.

1. Left component: The subtree $\T_{L}(\Aa)$ consists of the descendants
of the left outer band described in Theorem \ref{full-gen-2-2}. It
is isomorphic to the right subtree of the Sturmian tree $\T(H_{\alpha})$.
More precisely, if $\T^{(R)}(H_{\alpha})\subset\T(H_{\alpha})$ denotes
the subtree generated by the unique $B$-vertex in the second level,
then there exists an embedding
\begin{equation}
\iota_{L}:\T^{(R)}(H_{\alpha})\longrightarrow\T(\Aa)\label{eq:iL}
\end{equation}
which is an isomorphism onto $\T_{L}(\Aa)$.

2. Right component: The subtree $\T_{R}(\Aa)$ consists of all remaining
nonflat bands. Then there exists an embedding
\begin{equation}
\iota_{R}:\T^{red}(H_{\alpha})\longrightarrow\T(\Aa)\label{iR}
\end{equation}
which is an isomorphism onto $\T_{R}(\Aa)$.

In particular, all deviations from the classical Sturmian structure
are confined to the descendants of the initial $A$-band, and consist
precisely of the removal of a single central branch as described in
Definition \ref{def:T-red}.
\end{lem}

\begin{proof}
The first two levels of $\mathcal{T}(\Aa)$ are given explicitly by
Theorem \ref{full-gen-2-2}. In particular, the vertices at level
$\alpha_{0}$ and $\alpha_{1}$ can be identified with the corresponding
vertices in $\mathcal{T}(H_{\alpha})$ (for the outer branches) and
in $\mathcal{T}_{\mathrm{red}}(H_{\alpha})$ (for the descendants
of the initial $A$-band), preserving their left-to-right order and
their types.

We now extend this identification inductively. Suppose that the vertices
at level $\alpha_{n}$ have been identified with those of the corresponding
Sturmian trees, in a way that preserves order and types. By Theorem
\ref{thm:bad-type}, all nonflat bands at levels $\alpha_{n}$ with
$n\geq1$ have a well-defined forward type and satisfy the usual Sturmian
branching rule. It follows that the collection of descendants of each
vertex in the next level is uniquely determined by its type and by
the next continued fraction digit. Therefore, the identification at
level $\alpha_{n}$ induces uniquely an identification at level $\alpha_{n+1}$,
again preserving order and types. By induction, this defines an isomorphism
of labeled rooted trees between $\mathcal{T}(\Aa)$ and the union
of the two subtrees described above.
\end{proof}
This shows that the tree structure is completely determined by the
initial condition, and hence has the form described in Definition
\ref{def:T-red} and Lemma \ref{lem:tree-embedding}.

This identification will allow us to relate the eigenvalue counting
functions for the periodic approximations of $\Aa$ and $H_{\alpha}$
in what follows. We remark that the derivation of the initial condition
in Theorem \ref{full-gen-2-2} is somewhat technical. The next subsection
may be skipped on a first reading (continuing directly to Subsection
\ref{subsec:flat-bands}), taking the resulting tree structure as
given.

\subsection{Initial conditions for the spectral tree\label{subsec:adj-initial}}

In this subsection we prove Theorem \ref{full-gen-2-2}, thereby verifying
the description of the modified tree from the previous subsection.
Throughout this subsection, we mainly restrict attention to nonflat
bands, as they determine the tree structure. Flat bands and their
multiplicities will be treated separately in Subsection \ref{subsec:flat-bands}.

We first analyze the first periodic approximant $\alpha_{1}$ and
determine the number and location of its nonflat bands. We then study
the second periodic approximant $\alpha_{2}$ in order to control
the behavior of the initial $A$-band and show that, from that level
onward, the usual Sturmian branching rule applies.

The computations in this subsection are somewhat technical, and only
their conclusions are needed for the construction of the tree. The
full details are given in Appendix \ref{sec:trace-computations}.

\subsubsection{First periodic approximant}

Let
\begin{equation}
\alpha=[0;a,b,...],\label{eq:cf-ab}
\end{equation}
and consider the first periodic approximant $\alpha_{1}=\frac{1}{a}$,
whose word is $0^{a-1}1$. The associated transfer matrix is
\begin{equation}
\M_{a}(E):=\M_{\mathrm{tooth}}(E)\M_{\mathrm{trivial}}(E)^{a-1},\label{eq:word-a}
\end{equation}
and we denote its trace by $t_{a}(E):=tr(\M_{a}(E))$. The nonflat
spectrum is given by
\begin{equation}
\{E\in\mathbb{R}\setminus\{0\}:|t_{a}(E)|\le2\}.\label{eq:nonflat}
\end{equation}

By (\ref{eq:M-tooth}), the traces $t_{a}$ have a pole at $E=0$,
which makes it inconvenient to work with them directly. In order to
normalize this singularity, we define
\begin{equation}
h_{a}(E):=Et_{a}(E),\label{eq:ha}
\end{equation}
which will cause the resulting functions to be polynomials. The following
lemma summarizes the main structural properties of $t_{a}$ and $h_{a}$:
\begin{lem}
\label{lem:trace-recursion}The functions $t_{a}$ and $h_{a}$ satisfy:

1. The recursion
\begin{align}
 & t_{a+1}(E)=Et_{a}(E)-t_{a-1}(E),\qquad t_{0}(E)=2,\quad t_{1}(E)=E-\frac{1}{E},\label{eq:ta-rec}\\
 & h_{a+1}(E)=Eh_{a}(E)-h_{a-1}(E),\qquad h_{0}(E)=2E,\quad h_{1}(E)=E^{2}-1.\label{eq:ha-rec}
\end{align}

2. For every $a\ge0$, $h_{a}$ is a monic polynomial of degree $a+1$.

3. The symmetry relations
\begin{equation}
h_{a}(-E)=(-1)^{a+1}h_{a}(E)\label{eq:sym}
\end{equation}
hold.
\end{lem}

\begin{proof}
The proof is based on standard trace recursion relations -- see Appendix
\ref{sec:trace-computations}, Subsection \ref{subsec:The-first-periodic}.
\end{proof}
The polynomials $h_{a}$ play the role of characteristic polynomials
for the Floquet--Bloch operators associated with $\Aan 1$ (described
in Subsection \ref{subsec:FB}), and allow us to describe the band
edges explicitly.
\begin{lem}
\label{lem:FB}Let $H_{a}(\theta)$ be the Floquet--Bloch operator
of the first periodic approximant, with $\theta\in\{0,\pi\}$. Then
its characteristic polynomial is (up to a nonzero scalar)
\begin{equation}
P_{a,\theta}(E)=h_{a}(E)-2(\cos\theta)E.\label{eq:fb-pol}
\end{equation}
In particular, the band edges of $\Aan 1$ are given by the roots
of $P_{a,0}(E)$ and $P_{a,\pi}(E)$.

Consequently,

1. if $a$ is odd, then $E=0$ is not a band edge, and is therefore
not a flat band.

2. if $a$ is even, then $E=0$ is a simple common root of $P_{a,0}$
and $P_{a,\pi}$, and therefore corresponds to a flat band.
\end{lem}

\begin{proof}
The formula for $P_{a,\theta}$ follows from standard Floquet--Bloch
theory, while the claim about $E=0$ follows from Lemma \ref{lem:trace-recursion}
-- see Appendix \ref{sec:trace-computations}, Subsection \ref{subsec:The-first-periodic}. 
\end{proof}
\begin{rem}
The parity distinction in the lemma above can also be seen directly
from the eigenvalue equation. At energy $E=0$, the equation at a
tooth decoration forces the eigenfunction to vanish at the base vertex.
Thus, between two consecutive teeth one obtains the zero-eigenvalue
equation on a finite path with Dirichlet endpoints. This path has
a nontrivial zero mode exactly when the number of interior vertices
is odd, i.e. when $a$ is even, which is precisely the flat band appearing
above.
\end{rem}

Using the results above, one can determine the structure of $\spec{\Aan 1}$:
\begin{prop}
\label{prop:nonflat-bands-gen1} $\Aan 1$ has $a+1$  bands in total.
Among them:

1. If $a$ is odd, all $a+1$ bands are nonflat.

2. If $a$ is even, there is exactly one flat band at $E=0$, and
the remaining $a$ bands are nonflat.
\end{prop}

\begin{proof}
The total number of bands equals the  size of our unit cell, which
consists of $a-1$ trivial decorations (with a single vertex) and
an additional tooth (with two vertices), which overall gives $a+1$
bands. Lemma \ref{lem:FB} shows that $E=0$ is a flat band if and
only if $a$ is even, and that this root is simple. This produces
exactly one flat band in the even case, and all remaining bands are
nonflat.
\end{proof}
It remains to determine the type of the nonflat bands described above.
We do this by identifying which of these bands are contained in the
initial band $I^{(0)}=[-2,2]$, and are hence of backward type $A$.
To do so, we apply the following results:
\begin{lem}
\label{lem:B-bands}The Floquet--Bloch operators $H_{a}(\theta)$
satisfy:

1. For every $\theta\in\{0,\pi\}$, the operator $H_{a}(\theta)$
has at most one eigenvalue in $(2,\infty)$ and at most one eigenvalue
in $(-\infty,-2)$.

2. There exists $\theta\in\{0,\pi\}$ such that $H_{a}(\theta)$ has
an eigenvalue strictly larger than $2$, and similarly one strictly
smaller than $-2$.
\end{lem}

\begin{proof}
The proof follows from standard perturbation theory arguments --
see Lemma \ref{lem:perturbation} in Appendix \ref{sec:trace-computations}.
\end{proof}
\begin{cor}
\label{cor:level-1}The spectrum of $\Aan 1$ satisfies:

1. If $a$ is odd, then $\spec{\Aan 1}$ contains $a+1$ nonflat bands.
Among them, exactly two are not contained in $[-2,2]$, and the remaining
$a-1$ lie inside $(-2,2)$.

2. If $a$ is even, then $\spec{\Aan 1}$ contains $a$ nonflat bands.
Among them, exactly two are not contained in $[-2,2]$, and the remaining
$a-2$ lie inside $(-2,2)$.
\end{cor}

\begin{proof}
The statements regarding the number of nonflat bands follow from Proposition
\ref{prop:nonflat-bands-gen1}. Furthermore, part (1) of Lemma \ref{lem:B-bands}
says that at most one Bloch band of $\spec{\A_{\alpha_{1}}}$ can
intersect $\left(2,\infty\right)$ and $\left(-\infty,-2\right)$
respectively, and part (2) shows that both such bands indeed exist.
These bands are not contained in $I^{(0)}$, while all other nonflat
bands must be contained in $I^{(0)}$, proving the claim.
\end{proof}

\subsubsection{Completion of the proof of Theorem \ref{full-gen-2-2}}

While Corollary \ref{cor:level-1} identifies the first two levels,
the initial band $I^{(0)}=[-2,2]$ is not covered by the general forward
structure from Section \ref{sec:DTMP}, since it contains the singular
energy $E=0$. In particular, its branching into the next level $\alpha_{2}$
cannot be deduced directly from the Sturmian combinatorial structure,
and must be studied separately in order to prove Theorem \ref{full-gen-2-2}.

We therefore study the second periodic approximant, corresponding
to the next continued fraction digit, in order to understand the behavior
of $I^{(0)}$ and explicitly show that from this level onward, all
nonflat bands propagate according to the usual Sturmian branching
rule. 

Let $\alpha_{2}=[0;a,b]$ and let $t_{a,b}(E)$ be the trace of the
corresponding transfer matrix. We first determine whether the singular
energy $E=0$ can belong to the nonflat spectral bands at this level.
Indeed, if $E=0$ is excluded from the nonflat spectrum of $\Aan 2$,
then all nonflat bands at levels $\alpha_{n}$ with $n\geq2$ lie
away from the singular energy, and are therefore good bands in the
sense of Subsection \ref{subsec:combinatorial-structure}. In particular,
the usual Sturmian forward branching rule will apply to all such bands.
\begin{prop}
\label{prop:no-E=00003D0}The function $t_{a,b}(E)$ either has a
pole at $E=0$, or satisfies $|t_{a,b}(0)|>2$. In particular, $E=0$
does not belong to any nonflat spectral band of $\Aan 2$.
\end{prop}

\begin{proof}
The proof is based on the explicit recursion for $t_{a,b}$ and a
comparison of pole orders at $E=0$ (which depends on the parity of
$a$) -- see Appendix \ref{sec:trace-computations}, Subsection \ref{subsec:Proof-of-Proposition}. 
\end{proof}
It follows that there exists $\varepsilon>0$ such that $(-\varepsilon,\varepsilon)\setminus\{0\}$
is disjoint from the nonflat spectrum of both the first and second
periodic approximants. Using the inclusion
\begin{equation}
\spec{\Aan n}\subset\spec{\Aan{n-1}}\cup\spec{\Aan{n-2}},\label{eq:inclusion}
\end{equation}
we conclude that the nonflat spectrum of all higher approximants is
also contained in $\mathbb{R}\setminus(-\varepsilon,\varepsilon)$.
In particular, every nonflat band is a good band in the sense of Subsection
\ref{subsec:combinatorial-structure}, and therefore the usual Sturmian
forward branching rule applies to all nonflat bands from this point
onward.

It remains to determine how the initial band $I^{(0)}$ branches into
the next two levels. We do this by computing the total number of flat
and nonflat bands of $\Aan 2$, and comparing this with the structure
of $\Aan 1$ obtained above. 
\begin{lem}
\label{lem:gen2-bands}$\Aan 2$ has $ab+b+1$ bands. Among them,
the number of nonflat bands is
\begin{equation}
\begin{cases}
ab+b, & a\text{ odd},\\[1mm]
ab+2, & a\text{ even}.
\end{cases}\label{eq:nonflat-1}
\end{equation}
\end{lem}

\begin{proof}
The computation is similar to the one leading to Corollary \ref{cor:level-1}
-- see Appendix \ref{sec:trace-computations}, Subsection \ref{subsec:Proof-of-Lemma}.
\end{proof}
Combining this with Corollary \ref{cor:level-1}, we obtain the following
description of the branching of the initial $A$-band.
\begin{prop}
\label{prop:gen2-branching}Let $I^{(0)}=[-2,2]$ be the initial $A$-band.
Then:

1. The number of its nonflat descendants at level $\alpha_{1}$ is
\begin{equation}
\begin{cases}
a-1, & a\text{ odd},\\
a-2, & a\text{ even}.
\end{cases}\label{eq:nonflat-2}
\end{equation}

2. The number of its nonflat descendants at level $\alpha_{2}$ is
\begin{equation}
\begin{cases}
ab-b, & a\text{ odd},\\
ab, & a\text{ even}.
\end{cases}\label{eq:nonflat-gen2}
\end{equation}

3. From that point onward, the descendant structure agrees with the
Sturmian forward branching rule.
\end{prop}

\begin{proof}
The number of nonflat descendants of $I^{(0)}$ at level $\alpha_{1}$
follows directly from Corollary \ref{cor:level-1}. We now compute
the number of nonflat descendants at level $\alpha_{2}$. By Lemma
\ref{lem:gen2-bands}, the total number of nonflat bands of $\Aan 2$
is
\begin{equation}
\begin{cases}
ab+b, & a\text{ odd},\\
ab+2, & a\text{ even}.
\end{cases}\label{eq:nonflat-a2}
\end{equation}
On the other hand, the two outer bands at level $\alpha_{1}$ are
of backward type $B$ by Corollary \ref{cor:level-1}, and therefore
each contributes exactly $b$ nonflat descendants at level $\alpha_{2}$.
Thus the total number of nonflat descendants coming from the outer
branches is $2b$. Subtracting these from the total number of nonflat
bands gives the number of nonflat descendants of $I^{(0)}$ in (\ref{eq:nonflat-gen2}).

Finally, by Proposition \ref{prop:no-E=00003D0}, all nonflat bands
from level $\alpha_{2}$ onward lie away from the singular energy
$E=0$, and are therefore good bands. By Theorem \ref{thm:bad-type},
they have a well-defined forward type and satisfy the usual Sturmian
branching rule.
\end{proof}
Proposition \ref{prop:gen2-branching} gives exactly the branch deletion
described in Definition \ref{def:T-red}. In the usual Sturmian tree
$\mathcal{T}(H_{\alpha})$, the initial $A$-vertex has $a-1$ descendants
in level $\alpha_{1}$ and $a$ descendants at level $\alpha_{2}$,
with the next levels determined by the Sturmian branching rule. For
the comb, the counts in Proposition \ref{prop:gen2-branching} show
that precisely one central branch is missing, as demonstrated in Figure
\ref{fig:Adjacency-tree}. More precisely:

1. If $a$ is odd, the level-$\alpha_{1}$ descendants are unchanged,
while the number of descendants at level $\alpha_{2}$ is reduced
by $b$; this is exactly the removal of one central level-$\alpha_{2}$
branch.

2. If $a$ is even, the number of descendants at level $\alpha_{1}$
is already reduced by one, and the corresponding entire branch is
removed.

Apart from this missing branch, the remaining nonflat bands propagate
according to the usual Sturmian branching rule.

We can now finally prove Theorem \ref{full-gen-2-2}:
\begin{proof}[Proof of Theorem \ref{full-gen-2-2}]
The description of the first periodic approximant follows from Corollary
\ref{cor:level-1}. The analysis of the second periodic approximant
carried out above shows that:

1. The initial band $I^{(0)}$ produces the stated number of nonflat
descendants at level $\alpha_{1}$,

2. Its descendants at level $\alpha_{2}$ consist of $ab-b$ nonflat
bands if $a$ is odd, and $ab$ nonflat bands if $a$ is even.

3. From that point onward, all nonflat bands propagate according to
the usual Sturmian branching rule.

This determines the first two levels of the spectral tree $\T(\Aa)$,
and hence its entire structure. The statement of the theorem now follows.
\end{proof}

\subsection{Flat bands at $E=0$ and the IDS jump\label{subsec:flat-bands}}

Having determined the tree $\T\left(\Aa\right)$, our next goal is
to study the flat band of $\Aa$ and the resulting IDS jump at $E=0$.
In Subsection \ref{subsec:Proof-of-Theorem} we will see that these
are exactly the ``closed gaps'' for this system (i.e. gap labels
predicted by the GLT that are not attained by the IDS). We begin by
computing the multiplicity of the flat band at $E=0$ for the periodic
approximants $A_{\alpha_{n}}$, and then pass to the limit to obtain
the size of the jump of the IDS.

Let
\begin{equation}
N_{n}:=\#\{\text{nonflat bands of \ensuremath{\Aan n}}\}.\label{eq:Nn}
\end{equation}

\begin{lem}
\label{flat-multiplicity} For every $n\ge2$, the numbers $N_{n}$
satisfy
\begin{equation}
N_{n+1}=a_{n+1}N_{n}+N_{n-1}.\label{eq:Nn-rec}
\end{equation}
\end{lem}

\begin{proof}
Let $A_{n},B_{n}$ denote the number of type $A$ and type $B$ nonflat
bands at level $\alpha_{n}$, so that $N_{n}=A_{n}+B_{n}$. Since
all nonflat bands under consideration are good, the usual Sturmian
forward branching rule applies. Thus
\begin{equation}
A_{n+1}=(a_{n+1}-1)A_{n}+a_{n+1}B_{n}.\label{eq:An}
\end{equation}
Similarly,
\begin{equation}
A_{n}=(a_{n}-1)A_{n-1}+a_{n}B_{n-1}.\label{eq:An-identity}
\end{equation}
Rewriting the first identity as
\begin{equation}
A_{n+1}=a_{n+1}N_{n}-A_{n},\label{eq:aN}
\end{equation}
we obtain
\begin{equation}
N_{n+1}=A_{n+1}+B_{n+1}=a_{n+1}N_{n}-A_{n}+B_{n+1}.\label{eq:Nn+1}
\end{equation}

On the other hand, a band of $\spec{\A_{\alpha_{n+1}}}$ is of type
$B$ precisely when it is contained in a band of $\spec{\A_{\alpha_{n-1}}}$
but not in any band of $\spec{\A_{\alpha_{n}}}$. Thus the type $B$
bands of $\spec{\A_{\alpha_{n+1}}}$ are exactly the second-level
descendants of the bands of $\spec{\A_{\alpha_{n-1}}}$. A type $A$
band of $\spec{\A_{\alpha_{n-1}}}$ contributes $a_{n}$ such descendants,
while a type $B$ band contributes $a_{n}+1$. Hence
\begin{equation}
B_{n+1}=a_{n}A_{n-1}+(a_{n}+1)B_{n-1}.\label{eq:Bn}
\end{equation}
Using (\ref{eq:An-identity}) we obtain
\begin{equation}
B_{n+1}=\bigl((a_{n}-1)A_{n-1}+a_{n}B_{n-1}\bigr)+(A_{n-1}+B_{n-1})=A_{n}+N_{n-1}.\label{eq:Bn-1}
\end{equation}
Substituting this into (\ref{eq:Nn+1}) gives
\begin{equation}
N_{n+1}=a_{n+1}N_{n}-A_{n}+(A_{n}+N_{n-1})=a_{n+1}N_{n}+N_{n-1},\label{eq:Nn+1-1}
\end{equation}
as claimed.
\end{proof}
Proposition \ref{prop:gen2-branching} and Lemma \ref{flat-multiplicity}
can now be combined to provide a recursive formula for the number
of nonflat bands. Solving this recursion gives the following explicit
expression.
\begin{prop}
\label{non-flat-bands} For every $n\ge1$, the number of nonflat
bands is given by
\begin{equation}
N_{n}=\begin{cases}
(a_{1}+1)p_{n}, & a_{1}\text{ odd},\\[1mm]
2q_{n}-a_{1}p_{n}, & a_{1}\text{ even}.
\end{cases}\label{eq:Nn-nonflat}
\end{equation}
\end{prop}

\begin{proof}
Assume first that $a_{1}$ is odd. By Proposition \ref{prop:gen2-branching}
we have
\begin{equation}
N_{1}=a_{1}+1,\qquad N_{2}=a_{2}(a_{1}+1).\label{eq:N1}
\end{equation}
Since $p_{1}=1$ and $p_{2}=a_{2}$, this can be written as
\begin{equation}
N_{1}=(a_{1}+1)p_{1},\qquad N_{2}=(a_{1}+1)p_{2}.\label{eq:N1-1}
\end{equation}
The sequence $p_{n}$ satisfies the recursion
\begin{equation}
p_{n+1}=a_{n+1}p_{n}+p_{n-1},\label{eq:pn-1}
\end{equation}
and hence $(a_{1}+1)p_{n}$ satisfies the same recursion as $N_{n}$
with the same initial values. Thus:
\begin{equation}
N_{n}=(a_{1}+1)p_{n}.\label{eq:Nn-formula}
\end{equation}
Assume now that $a_{1}$ is even. By Proposition \ref{prop:gen2-branching}
we have 
\begin{equation}
N_{1}=a_{1},\qquad N_{2}=a_{1}a_{2}+2.\label{eq:N1-even}
\end{equation}
We claim that
\begin{equation}
N_{n}=2q_{n}-a_{1}p_{n}.\label{eq:Nn-odd}
\end{equation}
Indeed, since both $p_{n}$ and $q_{n}$ satisfy
\begin{equation}
x_{n+1}=a_{n+1}x_{n}+x_{n-1},\label{eq:qn}
\end{equation}
the sequence $2q_{n}-a_{1}p_{n}$ satisfies the same linear recursion.
Moreover, it satisfies the same initial condition as $N_{n}$:
\begin{align}
 & 2q_{1}-a_{1}p_{1}=2a_{1}-a_{1}=a_{1}=N_{1},\label{eq:=00003DN1}\\
 & 2q_{2}-a_{1}p_{2}=2(a_{1}a_{2}+1)-a_{1}a_{2}=a_{1}a_{2}+2=N_{2}.\label{eq:=00003DN2}
\end{align}
Thus the sequences $N_{n}$ and $2q_{n}-a_{1}p_{n}$  coincide for
all $n\geq1$.
\end{proof}
We now compute the flat band multiplicity.
\begin{lem}
\label{flat-band-multiplicity} For every $n\ge1$, the multiplicity
of the flat band at $E=0$ is
\begin{equation}
\dim\ker\left(\Aan n\right)=\begin{cases}
q_{n}-a_{1}p_{n}, & a_{1}\text{ odd},\\[1mm]
\left(a_{1}+1\right)p_{n}-q_{n}, & a_{1}\text{ even}.
\end{cases}\label{eq:dimker-1}
\end{equation}
\end{lem}

\begin{proof}
Recall that the self-adjoint operator $\A_{\alpha_{n}}$ acts on a
finite graph with $q_{n}+p_{n}$ vertices: $q_{n}$ sites, out of
which $p_{n}$ with an attached tooth. It thus has overall $q_{n}+p_{n}$
real eigenvalues. Since $E=0$ is the only possible flat band energy,
we have
\begin{equation}
\dim\ker\left(\Aan n\right)=(q_{n}+p_{n})-N_{n}.\label{eq:dimker-1}
\end{equation}

If $a_{1}$ is odd, then
\begin{equation}
\dim\ker\left(\Aan n\right)=q_{n}+p_{n}-(a_{1}+1)p_{n}=q_{n}-a_{1}p_{n}.\label{eq:dimker-2}
\end{equation}
If $a_{1}$ is even, then
\begin{equation}
\dim\ker\left(\Aan n\right)=q_{n}+p_{n}-(2q_{n}-a_{1}p_{n})=p_{n}\left(a_{1}+1\right)-q_{n}.\label{eq:dimker-3}
\end{equation}
\end{proof}
We now use this to compute the jump discontinuity of the IDS at $E=0$.
\begin{prop}
\label{prop:IDS-jump}The jump discontinuity of the IDS at $E=0$
is given by
\begin{equation}
\Delta\NHE{\Aa}\left(E=0\right)=\begin{cases}
\frac{1-a_{1}\alpha}{1+\alpha}, & a_{1}\text{ odd},\\[1mm]
\frac{\left(a_{1}+1\right)\alpha-1}{1+\alpha}, & a_{1}\text{ even}.
\end{cases}\label{eq:ids-jump}
\end{equation}
In particular, all gaps in the interval $\left(\frac{1}{2}-\frac{1-a_{1}\alpha}{2\left(1+\alpha\right)},\frac{1}{2}+\frac{1-a_{1}\alpha}{2\left(1+\alpha\right)}\right)$
are closed when $a_{1}$ is odd, and all gaps in the interval $\left(\frac{1}{2}-\frac{\left(a_{1}+1\right)\alpha-1}{2+2\alpha},\frac{1}{2}+\frac{\left(a_{1}+1\right)\alpha-1}{2+2\alpha}\right)$
are closed when $a_{1}$ is even.
\end{prop}

\begin{proof}
The size of the jump is determined by the relative mass of eigenvalues
at $E=0$:
\begin{align}
\Delta N\left(E=0\right) & :=\NHE{\A_{\alpha}}\left(0^{+}\right)-\NHE{\A_{\alpha}}\left(0^{-}\right)\label{eq:jump}\\
 & =\lim_{n\rightarrow\infty}\frac{\dim\ker\left(\Aan n\right)}{q_{n}+p_{n}}\nonumber \\
 & =\begin{cases}
\lim_{n\rightarrow\infty}\frac{q_{n}-a_{1}p_{n}}{q_{n}+p_{n}}, & a_{1}\text{ odd},\\[1mm]
\lim_{n\rightarrow\infty}\frac{p_{n}\left(a_{1}+1\right)-q_{n}}{q_{n}+p_{n}}, & a_{1}\text{ even}.
\end{cases}\nonumber \\
 & =\begin{cases}
\lim_{n\rightarrow\infty}\frac{1-a_{1}p_{n}/q_{n}}{1+p_{n}/q_{n}}, & a_{1}\text{ odd},\\[1mm]
\lim_{n\rightarrow\infty}\frac{\left(a_{1}+1\right)p_{n}/q_{n}-1}{1+p_{n}/q_{n}}, & a_{1}\text{ even}.
\end{cases}\nonumber \\
 & =\begin{cases}
\frac{1-a_{1}\alpha}{1+\alpha}, & a_{1}\text{ odd},\\[1mm]
\frac{\left(a_{1}+1\right)\alpha-1}{1+\alpha}, & a_{1}\text{ even}.
\end{cases}\nonumber 
\end{align}
The final statement about closed gaps follows from the symmetry of
$\spec{\Aa}$ around $E=0$, which implies that the jump interval
for $\NHE{\Aa}$ is symmetric around $\frac{1}{2}$.
\end{proof}
\bigskip

\subsection{Relating the IDS of $\protect\Aa$ and $H_{\alpha}$}

We now relate the IDS of $\A_{\alpha}$ to that of the classical Sturmian
Hamiltonian $H_{\alpha}$, using the description of $\T(\A_{\alpha})$
obtained in Subsection \ref{subsec:tree-description-adj}. Recall
from Lemma \ref{lem:tree-embedding} that
\begin{equation}
\T(\A_{\alpha})=\T_{L}(\A_{\alpha})\sqcup\T_{R}(\A_{\alpha}),\label{eq:tree-decomposition-1}
\end{equation}
where $\T_{L}(\A_{\alpha})$ is isomorphic to the right subtree $\T^{(R)}(H_{\alpha})\subset\T(H_{\alpha})$,
while $\T_{R}(\A_{\alpha})$ is isomorphic to the reduced tree $\T_{\mathrm{red}}(H_{\alpha})$
from Definition \ref{def:T-red}. This is the decomposition illustrated
in Figure \ref{fig:Adjacency-tree}: $\T_{L}(\A_{\alpha})$ is the
ordinary Sturmian component, while $\T_{R}(\A_{\alpha})$ is the component
in which the central branch has been removed. Thus, in order to determine
the attained gap labels, it remains to compute the values of the IDS
along paths in the two components of $\T(\A_{\alpha})$. In fact,
the symmetry of $\spec{\Aa}$ allows us to simplify the computation:
recall also from Proposition \ref{prop:IDS-jump} that $\A_{\alpha}$
has a flat band at $E=0$, with an IDS jump discontinuity across some
interval which is symmetric around the value $\frac{1}{2}$ (the precise
jump interval depends on the first continued fraction digit $a_{1}$).
We compute $\NHE{\Aa}$ along paths corresponding to gap labels that
are attained \textbf{to the left} of the jump discontinuity (i.e.
for $E<0$), and the symmetry will then allow us to obtain the remaining
gap labels in the next subsection. Thus, it is enough to compute the
IDS value along paths in either the left component $\T_{L}(\A_{\alpha})$,
and or in the right component $\T_{R}\left(\Aa\right)$ corresponding
to energies $E<0$.

We begin with $\T_{L}(\A_{\alpha})$, by taking paths in the right
subtree $\T^{(R)}(H_{\alpha})$, which are known to correspond to
gap labels of the form $\alpha n+m\in\left[1-\alpha,1\right]$, and
then calculating the gap labels arising from their embedding into
$\T_{L}(\A_{\alpha})$.
\begin{prop}
\label{prop:adj-left-component}Let $\gamma$ be an infinite path
in $\T^{(R)}(H_{\alpha})\subset\T(H_{\alpha})$, and denote its embedding
into $\T_{L}(\A_{\alpha})$ by $\iota_{L}(\gamma)$. If
\begin{equation}
N_{H_{\alpha}}(\gamma)=\alpha n+m\in[1-\alpha,1]\label{eq:IDS-map}
\end{equation}
for some $m,n\in\mathbb{Z}$, then
\begin{equation}
N_{\A_{\alpha}}(\iota_{L}(\gamma))=\frac{(n+1)\alpha+m-1}{1+\alpha}\in\left[0,\frac{\alpha}{1+\alpha}\right].\label{eq:IDS-map-1}
\end{equation}
\end{prop}

\begin{proof}
Fix $n\in\N$ and consider the level $\alpha_{n}$ of the Sturmian
tree $\T(H_{\alpha})$. This level contains $q_{n}$ bands in total:
$q_{n}-p_{n}$ from the left subtree and $p_{n}$ from $\T^{(R)}(H_{\alpha})$.
Since $\T_{L}(\A_{\alpha})$ is isomorphic to $\T^{(R)}(H_{\alpha})$,
the eigenvalue counting functions along paths in these trees are shifted
by $q_{n}-p_{n}$ relative to the Sturmian case. Thus, for a finite
path $\gamma_{n}\subset\T^{(R)}(H_{\alpha})$ , its embedding into
$\T_{L}(\A_{\alpha})$ satisfies
\begin{equation}
N_{n}^{\A_{\alpha}}(\iota_{L}(\gamma_{n}))=N_{n}^{H_{\alpha}}(\gamma_{n})-(q_{n}-p_{n}).\label{eq:Count}
\end{equation}

Normalizing by the total number of vertices in the unit cell $q_{n}+p_{n}$,
and letting $n\to\infty$, we obtain
\begin{equation}
N_{\A_{\alpha}}(\iota_{L}(\gamma))=\lim_{n\to\infty}\frac{N_{n}^{H_{\alpha}}(\gamma_{n})-q_{n}+p_{n}}{q_{n}+p_{n}}.\label{eq:IDS-count}
\end{equation}
Using
\begin{equation}
\lim_{n\to\infty}\frac{N_{n}^{H_{\alpha}}(\gamma_{n})}{q_{n}}=N_{H_{\alpha}}(\gamma)=\alpha n+m\in[1-\alpha,1],\label{eq:GL-1}
\end{equation}
and the fact that $p_{n}/q_{n}\to\alpha$, we conclude
\begin{equation}
N_{\A_{\alpha}}(\iota_{L}(\gamma))=\frac{N_{H_{\alpha}}(\gamma)-1+\alpha}{1+\alpha}=\frac{(n+1)\alpha+m-1}{1+\alpha}.\label{eq:GL-2}
\end{equation}
Since $N_{H_{\alpha}}(\gamma)\in[1-\alpha,1]$, it follows that
\begin{equation}
N_{\A_{\alpha}}(\iota_{L}(\gamma))\in\left[0,\frac{\alpha}{1+\alpha}\right],\label{eq:GL-interval}
\end{equation}
as claimed.
\end{proof}
We now turn to the right component $\T_{R}(\A_{\alpha})$, which is
identified with the reduced tree $\T_{\mathrm{red}}(H_{\alpha})$.
As established before, we may focus only on paths in $\T_{R}(\A_{\alpha})$
corresponding to negative energies, which are identified by Lemma
\ref{lem:tree-embedding} with paths in $\T_{\mathrm{red}}(H_{\alpha})$
that are to the left of the missing branch. Before computing this
IDS correspondence, we first identify precisely which gap labels of
$H_{\alpha}$ arise from such paths in $\T_{\mathrm{red}}(H_{\alpha})$
to the left of the missing branch.
\begin{lem}
\label{lem:jump-interval}Let $\gamma$ be an infinite path in $\T_{\mathrm{red}}(H_{\alpha})$
which is to the left of the missing branch. Then its gap label is
of the form $N_{H_{\alpha}}(\gamma)=\alpha n+m$ with $m,n\in\Z$
and
\begin{equation}
\alpha n+m\in\begin{cases}
\left[0,\frac{\left(a_{1}-1\right)\alpha}{2}\right], & a_{1}\text{ is odd,}\\
\left[0,1-\frac{\left(a_{1}+2\right)\alpha}{2}\right], & a_{1}\text{ is even,}
\end{cases}\label{eq:GL-interval-1}
\end{equation}
where $a_{1}$ is the first continued fraction digit of $\alpha$
and $m,n\in\Z$.
\end{lem}

\begin{proof}
The fact that $N_{H_{\alpha}}(\gamma)=\alpha n+m$ for some $m,n\in\Z$
follows from the gap labelling theorem. It remains to determine the
interval of possible IDS values for paths in $\T_{\mathrm{red}}(H_{\alpha})$
to the left of the missing branch. Since the subtree generated by
the initial $A$-vertex in $\T(H_{\alpha})$ gives rise to gap labels
in the interval $[0,\alpha]$, it follows that the paths in $\T_{\mathrm{red}}(H_{\alpha})$
to the left of the missing branch give rise to gap labels in an interval
of the form $[0,X]$. It therefore remains to determine the endpoint
$X$ of this interval. We do this by counting eigenvalues along the
corresponding subtree of $\T(H_{\alpha})$. Recall that at level $\alpha_{n}$,
the subtree generated by the initial $A$-vertex contains $q_{n}-p_{n}$
vertices. In $\T_{\mathrm{red}}(H_{\alpha})$, some of these vertices
are missing due to the missing central branch, and the remaining vertices
are placed symmetrically around this missing branch. Thus, to obtain
the number of vertices to the left of the missing branch (which we
denote by $M_{n}$), we can subtract the number of missing vertices
due to the central branch from the overall number of vertices in this
subtree, and then divide by $2$. 

The number of vertices missing due to the central branch is precisely
the multiplicity of the flat band at level $\alpha_{n}$, which by
Lemma \ref{flat-band-multiplicity} equals
\begin{equation}
\begin{cases}
q_{n}-a_{1}p_{n}, & a_{1}\text{ odd},\\[1mm]
\left(a_{1}+1\right)p_{n}-q_{n}, & a_{1}\text{ even}.
\end{cases}\label{eq:flat-mul}
\end{equation}
Thus, the number of vertices in the left half of the $A$-subtree
is
\begin{align}
M_{n} & =\begin{cases}
\frac{(q_{n}-p_{n})-(q_{n}-a_{1}p_{n})}{2}, & a_{1}\text{ odd},\\[1mm]
\frac{(q_{n}-p_{n})-\left((a_{1}+1)p_{n}-q_{n}\right)}{2}, & a_{1}\text{ even},
\end{cases}\label{eq:Mn}\\
 & =\begin{cases}
\frac{(a_{1}-1)p_{n}}{2}, & a_{1}\text{ odd},\\[1mm]
\frac{2q_{n}-\left(a_{1}+2\right)p_{n}}{2}, & a_{1}\text{ even}.
\end{cases}\nonumber 
\end{align}
Dividing by the size of the unit cell $q_{n}$ and taking $n\to\infty$,
we obtain
\begin{align}
X & =\lim_{n\to\infty}\frac{M_{n}}{q_{n}}\label{eq:Xlim}\\
 & =\begin{cases}
\lim_{n\to\infty}\frac{(a_{1}-1)p_{n}}{2q_{n}}, & a_{1}\text{ odd},\\[1mm]
\lim_{n\to\infty}\frac{2q_{n}-\left(a_{1}+2\right)p_{n}}{2q_{n}}, & a_{1}\text{ even},
\end{cases}\nonumber \\
 & =\begin{cases}
\lim_{n\to\infty}\frac{(a_{1}-1)p_{n}/q_{n}}{2}, & a_{1}\text{ odd},\\[1mm]
\lim_{n\to\infty}\frac{2-\left(a_{1}+2\right)p_{n}/q_{n}}{2}, & a_{1}\text{ even},
\end{cases}\nonumber \\
 & =\begin{cases}
\lim_{n\to\infty}\frac{(a_{1}-1)\alpha}{2}, & a_{1}\text{ odd},\\[1mm]
\lim_{n\to\infty}1-\frac{\left(a_{1}+2\right)\alpha}{2}, & a_{1}\text{ even}.
\end{cases}\nonumber 
\end{align}
This completes the proof.
\end{proof}
\begin{prop}
\label{prop:adj-right-component}Let $\gamma$ be an infinite path
in $\T_{\mathrm{red}}(H_{\alpha})$, and denote its embedding into
$\T_{R}(\A_{\alpha})$ by $\iota_{R}(\gamma)$. If
\begin{equation}
N_{H_{\alpha}}(\gamma)=\alpha n+m\label{eq:an+m}
\end{equation}
for some $m,n\in\Z$, then
\begin{equation}
N_{\A_{\alpha}}(\iota_{R}(\gamma))=\frac{\alpha(n+1)+m}{1+\alpha}.\label{eq:IDS-shift}
\end{equation}
\end{prop}

\begin{proof}
Fix $n\in\N$ and consider level $\alpha_{n}$ of the Sturmian tree
$\T(H_{\alpha})$. As before, this level contains $q_{n}$ bands in
total, of which $p_{n}$ belong to $\T^{(R)}(H_{\alpha})$. We now
consider paths $\gamma_{n}\subset\T_{\mathrm{red}}(H_{\alpha})$ to
the left of the missing branch, and compute the corresponding counting
function after embedding into $\T_{R}(\A_{\alpha})$. The difference
arises from the contribution of the additional bands corresponding
to $\T_{L}(\A_{\alpha})$, which lie entirely to the left of the paths
under consideration. More precisely, all bands corresponding to $\T_{L}(\A_{\alpha})$
contribute to the counting function along such paths, and since $\T_{L}(\A_{\alpha})$
is isomorphic to $\T^{(R)}(H_{\alpha})$, this contribution is exactly
$p_{n}$. Therefore,
\begin{equation}
N_{n}^{\A_{\alpha}}(\iota_{R}(\gamma_{n}))=N_{n}^{H_{\alpha}}(\gamma_{n})+p_{n}.\label{eq:pn-eig}
\end{equation}

Normalizing by the total number of vertices in the unit cell $q_{n}+p_{n}$
and letting $n\to\infty$, we obtain
\begin{equation}
N_{\A_{\alpha}}(\iota_{R}(\gamma))=\lim_{n\to\infty}\frac{N_{n}^{H_{\alpha}}(\gamma_{n})+p_{n}}{q_{n}+p_{n}}.\label{eq:GL-3}
\end{equation}
Using
\begin{equation}
\lim_{n\to\infty}\frac{N_{n}^{H_{\alpha}}(\gamma_{n})}{q_{n}}=N_{H_{\alpha}}(\gamma)=\alpha n+m,\label{eq:og-GL}
\end{equation}
together with $p_{n}/q_{n}\to\alpha$, we conclude
\begin{equation}
N_{\A_{\alpha}}(\iota_{R}(\gamma))=\frac{N_{H_{\alpha}}(\gamma)+\alpha}{1+\alpha}=\frac{\alpha(n+1)+m}{1+\alpha},\label{eq:shift-GL}
\end{equation}
as claimed.
\end{proof}
Propositions \ref{prop:adj-left-component} and \ref{prop:adj-right-component}
above will be used in order to compute the complete list of gap labels
for $\Aa$.

\subsection{Proof of Theorem \ref{thm:DTMP-Adj}\label{subsec:Proof-of-Theorem}}

We are now ready to prove Theorem \ref{thm:DTMP-Adj}, resolving the
DTMP for the adjacency matrix $\A_{\alpha}$ on the Sturmian comb. 
\begin{proof}[Proof of Theorem \ref{thm:DTMP-Adj}]
 As seen in the previous subsections, the IDS of the comb has a jump
discontinuity at $E=0$, and by Proposition \ref{prop:IDS-jump} the
corresponding interval of missing gap labels is symmetric around $\frac{1}{2}$.
It therefore remains to show that all gap labels outside this interval
are attained. By the symmetry of $\spec{\A_{\alpha}}$ with respect
to $E=0$, we can first compute the attained gap labels to the left
of $\frac{1}{2}$.

We begin with the contribution of the left component $\T_{L}(\A_{\alpha})$
to the gap labels. By Proposition \ref{prop:adj-left-component},
if $\gamma$ is an infinite path in $\T^{(R)}(H_{\alpha})$ with
\begin{equation}
N_{H_{\alpha}}(\gamma)=\alpha n+m\in[1-\alpha,1],\label{eq:int-1}
\end{equation}
then
\begin{equation}
N_{\A_{\alpha}}(\iota_{L}(\gamma))=\frac{(n+1)\alpha+m-1}{1+\alpha}.\label{eq:shift-1}
\end{equation}
Thus, the gap labels arising from $\T_{L}(\A_{\alpha})$ are precisely
\begin{equation}
\left\{ \frac{\alpha n+m}{1+\alpha}:m,n\in\Z\right\} \cap\left[0,\frac{\alpha}{1+\alpha}\right].\label{eq:GL-1-1}
\end{equation}

We now turn to the right component $\T_{R}(\A_{\alpha})$. By Lemma
\ref{lem:jump-interval}, the paths in $\T_{\mathrm{red}}(H_{\alpha})$
to the left of the missing branch correspond to gap labels
\begin{equation}
N_{H_{\alpha}}(\gamma)=\alpha n+m\in[0,X],\label{eq:GLT-int}
\end{equation}
where
\begin{equation}
X=\begin{cases}
\frac{(a_{1}-1)\alpha}{2}, & a_{1}\ \text{odd},\\[2mm]
1-\frac{(a_{1}+2)\alpha}{2}, & a_{1}\ \text{even}.
\end{cases}\label{eq:X-int}
\end{equation}
By Proposition \ref{prop:adj-right-component}, their embedding into
$\T_{R}(\A_{\alpha})$ satisfies
\begin{equation}
N_{\A_{\alpha}}(\iota_{R}(\gamma))=\frac{\alpha(n+1)+m}{1+\alpha}.\label{eq:GL-shift}
\end{equation}
It follows that the gap labels arising from $\T_{R}(\A_{\alpha})$
to the left of the jump are
\begin{equation}
\left\{ \frac{\alpha n+m}{1+\alpha}:m,n\in\Z\right\} \cap\left[\frac{\alpha}{1+\alpha},\frac{X+\alpha}{1+\alpha}\right].\label{eq:GL-int-2}
\end{equation}

Combining the contributions of $\T_{L}(\A_{\alpha})$ and $\T_{R}(\A_{\alpha})$,
we obtain that all attained gap labels to the left of the jump interval
are
\begin{equation}
\left\{ \frac{\alpha n+m}{1+\alpha}:m,n\in\Z\right\} \cap\left[0,\frac{X+\alpha}{1+\alpha}\right].\label{eq:GL-left}
\end{equation}
By symmetry, the attained gap labels to the right of the jump interval
are obtained by reflection around $\frac{1}{2}$, and hence
\begin{align}
 & \GL{\Aa}\label{eq:GL-all}\\
 & =\left\{ \frac{\alpha n+m}{1+\alpha}:m,n\in\Z\right\} \cap\left(\left[0,\frac{X+\alpha}{1+\alpha}\right]\cup\left[1-\frac{X+\alpha}{1+\alpha},1\right]\right)\nonumber \\
 & =\begin{cases}
\left\{ \frac{\alpha n+m}{1+\alpha}:m,n\in\Z\right\} \cap\left(\left[0,\frac{(a_{1}+1)\alpha}{2(1+\alpha)}\right]\cup\left[1-\frac{(a_{1}+1)\alpha}{2(1+\alpha)},1\right]\right), & a_{1}\ \text{odd},\\[2mm]
\left\{ \frac{\alpha n+m}{1+\alpha}:m,n\in\Z\right\} \cap\left(\left[0,\frac{2-a_{1}\alpha}{2(1+\alpha)}\right]\cup\left[\frac{(a_{1}+2)\alpha}{2(1+\alpha)},1\right]\right), & a_{1}\ \text{even}.
\end{cases}\nonumber 
\end{align}
In both cases, the missing interval is $\left(\frac{X+\alpha}{1+\alpha},1-\frac{X+\alpha}{1+\alpha}\right)$,
which (upon substituting the value of $X$) coincides exactly with
the jump interval from Proposition \ref{prop:IDS-jump}. We conclude
that all gap labels outside the jump interval are attained, and no
gap label inside it is attained. This proves the theorem.
\end{proof}

\subsection{Extension to arbitrary Sturmian decorated $\protect\Z$-graphs\label{subsec:adj-general}}

We now outline how the computation extends to arbitrary Sturmian decorated
$\Z$-graphs. So far, we considered only the trivial decoration and
the tooth decoration. However, the transfer matrix formalism for the
adjacency matrix (and thus the spectral analysis through the Sturmian
combinatorial structure) applies equally well to any choice of the
decorations. We outline how this analysis can be carried out for more
general decorations. We then also outline the idea of computation
for decorated $\Z$-graphs equipped with an arbitrary horizontally
homogeneous Jacobi operator.

The general idea is the following. Each decoration contributes to
the transfer matrix only through the corresponding diagonal Green
function at the base vertex. Thus, away from the poles of these Green
functions and from the energies where the two effective potentials
coincide, the operator reduces to a Sturmian Hamiltonian with energy-dependent
potentials. The results of Section \ref{sec:DTMP-proofs} can then
be applied exactly as above, giving the same Sturmian tree structure
outside the exceptional set. What remains in each concrete example
is to compute the Green functions of the two decorations, identify
the bad and singular energies, and determine which of these energies
produce flat bands or closed gaps.

\subsubsection{Computing the transfer matrix}

Given a general decorated $\Z$-graph, the first step is to compute
the transfer matrix for its decorations, which can then be used to
study the associated Sturmian tree.

Fix a finite pointed graph $\left(G,v\right)$, and consider a decorated
$\Z$-graph in which at each site $n\in\Z$ we attach a decoration
$G_{\omega\left(n\right)}\cong G$ through the vertex $v$. Let $\A_{G}$
denote the adjacency matrix of $G$. On the decorated $\Z$-graph,
the action of the full adjacency matrix can be locally written at
the $n$th site as
\begin{equation}
\A=\A_{h}+\A_{G_{\omega_{\alpha}\left(n\right)}},\label{eq:hGsplit}
\end{equation}
where $\A_{h}$ accounts for the horizontal (chain) edges, and $\A_{G_{\omega_{\alpha}\left(n\right)}}$
is the adjacency matrix on the $n$th decoration. The eigenvalue equation
becomes:
\begin{equation}
(\mathcal{A}_{h}+\mathcal{A}_{G})\varphi=E\varphi.\label{eq:split-eval-equation}
\end{equation}

Let $\varphi_{n}(w)$ denote the value of the solution at vertex $w\in\V_{G_{\omega\left(n\right)}}$.
We obtain:
\begin{equation}
E\varphi_{n}(w)=\sum_{u\sim_{G}w}\varphi_{n}(u)+\delta_{w,v}\left(\varphi_{n-1}(v)+\varphi_{n+1}(v)\right).\label{eq:Eval-split}
\end{equation}
This simplifies to:
\begin{equation}
(EI-\mathcal{A}_{G})\varphi_{n}=e_{v}\left(\varphi_{n-1}(v)+\varphi_{n+1}(v)\right),\label{eq:Eval-split-1}
\end{equation}
where $e_{v}$ is the delta function at the base vertex $v$. Solving
for $\varphi_{n}$ gives:
\begin{equation}
\varphi_{n}=(EI-\mathcal{A}_{G})^{-1}e_{v}\left(\varphi_{n-1}(v)+\varphi_{n+1}(v)\right).\label{eq:Eval-split-2}
\end{equation}
Taking the inner product with $e_{v}$ yields:
\begin{equation}
\varphi_{n}(v)=\left\langle (EI-\mathcal{A}_{G})^{-1}e_{v},e_{v}\right\rangle \left(\varphi_{n-1}(v)+\varphi_{n+1}(v)\right).\label{eq:Eval-split-3}
\end{equation}

This gives the recurrence relation:
\begin{equation}
\varphi_{n+1}(v)=\frac{1}{\left\langle (EI-\mathcal{A}_{G})^{-1}e_{v},e_{v}\right\rangle }\varphi_{n}(v)-\varphi_{n-1}(v).\label{eq:Eval-split-4}
\end{equation}
The associated one-step transfer matrix for the decoration $G$ is
therefore:
\begin{equation}
\M_{G}(E)=\begin{pmatrix}{\displaystyle \frac{1}{\left\langle (EI-\mathcal{A}_{G})^{-1}e_{v},e_{v}\right\rangle }} & -1\\
1 & 0
\end{pmatrix}.\label{eq:1-tmat-adj}
\end{equation}
To write this in standard Sturmian form, define:

\begin{equation}
f_{G}(E)=E-\frac{1}{\left\langle (EI-\mathcal{A}_{G})^{-1}e_{v},e_{v}\right\rangle },\label{eq:f-Green}
\end{equation}
so that:

\begin{equation}
\M_{G}(E)=\begin{pmatrix}E-f_{G}(E) & -1\\
1 & 0
\end{pmatrix}.\label{eq:tmat-final-adj}
\end{equation}
This shows that for any decoration, the transfer matrix still takes
Sturmian form with an energy-dependent function $f_{G}(E)$. Therefore,
the analysis in Subsection \ref{subsec:combinatorial-structure} remains
valid, and the usual Sturmian combinatorial structure holds at all
energies where the transfer matrices are well-defined. We also see
that the functions $f_{0}(E),f_{1}(E)$ corresponding to the decorations
are indeed piecewise real-analytic, as assumed in Subsection \ref{subsec:combinatorial-structure}.
\begin{example}
Let us consider the trivial decoration (a single vertex). In this
case, $\mathcal{A}_{G}=0$, so:
\begin{equation}
f(E)=E-E=0,\label{eq:Green-triv}
\end{equation}
and:
\begin{equation}
\M(E)=\begin{pmatrix}E & -1\\
1 & 0
\end{pmatrix},\label{eq:m-triv}
\end{equation}
as expected.
\end{example}

\begin{example}
\label{exa:comb-transfer}Consider the tooth decoration, i.e.,
\end{example}

\begin{equation}
\mathcal{A}_{G}=\left(\begin{array}{cc}
0 & 1\\
1 & 0
\end{array}\right).\label{eq:adj-tooth}
\end{equation}
Then:
\begin{equation}
(EI-\mathcal{A}_{G})^{-1}=\frac{1}{E^{2}-1}\left(\begin{array}{cc}
E & 1\\
1 & E
\end{array}\right),\label{eq:Green-tooth}
\end{equation}
and so
\begin{equation}
\M(E)_{11}=\left\langle \frac{1}{E^{2}-1}\left(\begin{array}{c}
E\\
1
\end{array}\right),\left(\begin{array}{c}
1\\
0
\end{array}\right)\right\rangle ^{-1}=\frac{E^{2}-1}{E}=E-\frac{1}{E},\label{eq:m-tooth}
\end{equation}
which matches the transfer matrix we presented earlier.

\subsubsection{Singular energies}

We now focus on the ``singular energies'', where the transfer matrix
is not well-defined. Define the diagonal Green's function:

\begin{equation}
\mathcal{G}_{\A_{G}}(E):=\left\langle (EI-\mathcal{A}_{G})^{-1}e_{v},e_{v}\right\rangle .\label{eq:Green}
\end{equation}
The transfer matrix becomes undefined at energies for which $\mathcal{G}_{\A_{G}}(E)=0$,
i.e., those solving:
\begin{equation}
\sum_{k}\frac{|\psi_{k}(v)|^{2}}{E-\lambda_{k}}=0,\label{eq:Green=00003D0}
\end{equation}
where $(\lambda_{k},\psi_{k})$ are the eigenpairs of $\mathcal{A}_{G}$.
These singular energies interlace between the eigenvalues of $\A_{G}$,
are not captured by the transfer matrix, and generally may correspond
to flat bands. Additionally, recalling Assumption \ref{assu:distinct-decorations}
of distinct Green's functions, we see that the functions $f_{0}(E),f_{1}(E)$
are indeed distinct, as assumed in Subsection \ref{subsec:combinatorial-structure}.

In any case, once $\mathcal{G}_{\A_{G}}(E)$ is known, the band structure
and IDS can be analyzed using the same tools as in the simpler Sturmian
comb case. First, one computes the initial condition for the tree
$\T\left(\A_{\alpha}\right)$ using the traces, which are precisely
the inverses to the diagonal Green's functions of both decorations.
Once the first two levels of the tree are known, the tree structure
is uniquely determined by Proposition \ref{prop:type-existence}.
One can then use the tree structure to compute the correspondence
between $N_{\A_{\alpha}}$ and $N_{H_{\alpha}}$ (and hence the gap
labels) as in the previous subsections, keeping in mind the IDS jump
at the special energies where the diagonal Green's functions vanish,
i.e. (\ref{eq:Green=00003D0}).

\subsubsection{Generalizing to different Jacobi operators}

Since the main results of Section \ref{sec:DTMP-proofs} were proven
for an arbitrary horizontally homogeneous Jacobi operator $\Ja$,
one can try and similarly resolve the DTMP for operators other than
the adjacency matrix.

A natural first step would be to express the transfer matrices associated
with each decoration. A direct computation appearing in Appendix \ref{sec:t-mat-computation}
shows that similar to the case of the adjacency matrix, the transfer
matrix of a given decoration $G$ can be written in the form
\begin{equation}
\M_{G}(E)=\begin{pmatrix}\frac{1}{c\mathcal{G}_{\mathcal{J}_{|G}}(E)} & -1\\
1 & 0
\end{pmatrix},\label{eq:t-mat-1}
\end{equation}
where $c\in\R$ is the horizontal coupling constant and $\mathcal{G}_{\mathcal{J}_{|G}}$
is Green's function for the associated decoration:
\begin{equation}
\mathcal{G}_{\mathcal{J}_{|G}}(E):=\langle(E-\mathcal{J}_{|G})^{-1}e_{v},e_{v}\rangle.\label{eq:Green-2-1}
\end{equation}
Once the operator $\Ja$ and the decorations are given, this can be
used in order to compute the initial condition for the spectral tree
$\T\left(\Ja\right)$, as well as the singular energies (given by
the zeros of Green's function). One may then apply the methods from
this section in order to compute the gap labels.

\subsection{Extension to the normalized discrete Laplacian and equilateral metric
graphs\label{subsec:NDL-DTMP}}

One special case of interest which is not directly covered by our
result is the normalized discrete Laplacian (NDL) from Theorem \ref{thm:Discrete-GLT}.
The gap labels here are interesting to study both in their own right,
and due to their relation to the gap labels of equilateral metric
decorated $\Z$-graphs, as described in Subsection \ref{subsec:discrete-GLT}.
A technical difficulty arises here: unlike the adjacency operator,
the normalized Laplacian assigns non-uniform edge weights, so that
the value of $\Delta_{\alpha}$ at a base vertex depends not only
on the local decoration, but also on the degrees of its neighboring
decorations. Namely, the horizontal homogeneity assumption does not
apply, and consequently, the transfer matrix is no longer determined
by a single decoration alone, but by the entire surrounding configuration.

Nevertheless, this problem may be solved using a \textit{subdivision
trick}, which we briefly describe here (and refer to \cite{Bandc}
for the full details). Instead of studying the NDL on a given decorated
$\Z$-graph $G_{\alpha}$, we consider a subdivided graph $G_{\alpha}'$,
which is obtained from $G_{\alpha}$ by placing an auxiliary vertex
at the center of every edge, see Figure \ref{fig:subdivision}.

\begin{figure}
\includegraphics[scale=0.4]{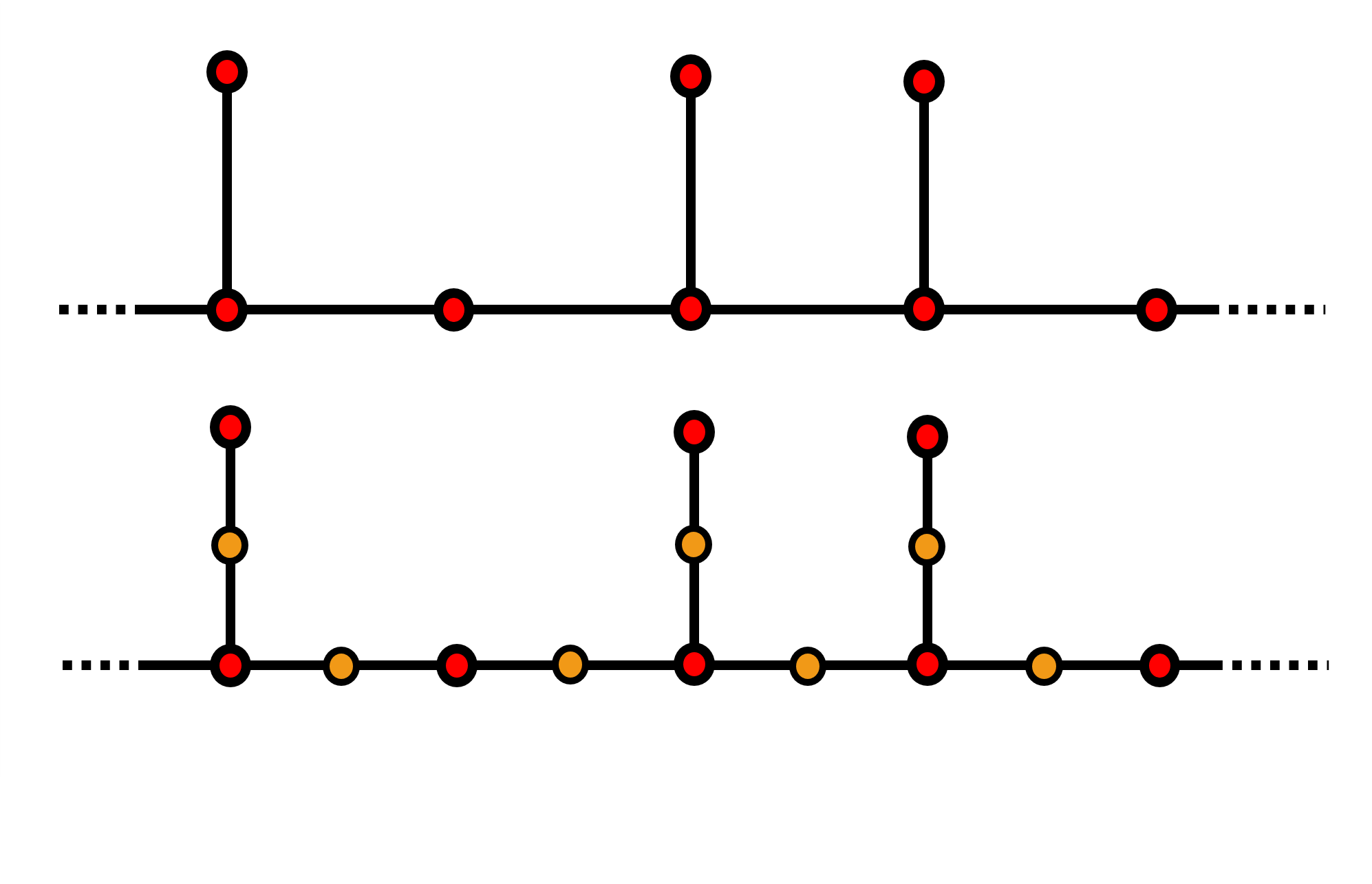}

\caption[~~A subdivision of a graph.]{A subdivision procedure for a comb graph. The auxiliary vertices (orange)
are added to each edge, artificially causing the boundary vertices
of the unit cell to be of degree 2. \label{fig:subdivision}}
\end{figure}
The result of this procedure is that the NDL is now horizontally homogeneous,
and admits well-defined transfer matrices, and so it can be studied
via our previous methods. This means that the gap labels of the subdivided
graph $G_{\alpha}'$ can be computed explicitly. Although the subdivision
procedure changes the spectrum (and gap labels) of our original graph
$G_{\alpha}$, it is still possible to relate their periodic approximants
through the following fundamental result:
\begin{thm}
\label{thm:subdivision}\cite[thm. 3.7]{Xie2016} Let $G$ be a finite
graph and let $G'$ be its subdivision graph. Define the map $Q^{-1}$
on a multiset $S\subset\mathbb{R}$ by
\begin{equation}
Q^{-1}(S):=\left\{ 1\pm\sqrt{1-\frac{x}{2}}:x\in S\setminus\{2\}\right\} .\label{eq:fold-map}
\end{equation}
Then the spectrum of the NDL on $G$ satisfies
\begin{equation}
\spec{\Delta^{G'}}=Q^{-1}(\spec{\Delta^{G}}\setminus\{2\})\cup\{1,\dots,1\},\label{eq:subdivision-spectra}
\end{equation}
where the multiplicity of the added eigenvalue $1$ depends on the
multiplicity of the eigenvalue $2$ on the graph $G$:
\begin{equation}
\mathrm{mult}_{G'}(1)=\beta-1+2\mathrm{mult}_{G}(2),\label{eq:multiplicity-1}
\end{equation}
and $\beta=|\E(G)|-|\V(G)|+1$ is the first Betti number of $G$.
\end{thm}

Essentially, each original eigenvalue $x\in\spec G$ (except $2$)
gives rise to two new eigenvalues in $G'$:
\begin{equation}
1\pm\sqrt{1-\frac{x}{2}},\label{eq:folded-evals}
\end{equation}
\\
which produces two spectral copies: one in the interval $[0,1]$ (preserving
order), and another in $[1,2]$ (with reversed order). Heuristically,
this result causes the spectral tree of $\T\left(\Delta_{G_{\alpha}'}\right)$
to have a twofold structure -- consisting of two subtrees, each isomorphic
to $\T\left(\Delta_{G_{\alpha}}\right)$. Consequently, the gap labels
of $\Delta_{G_{\alpha}}$ can be computed from those of $\Delta_{G_{\alpha}'}$.
Thus, the methods we presented here can be used to compute the gap
labels of the NDL as well.

Upon computing the gap labels for the NDL, one may further compute
the gap labels for equilateral metric decorated $\Z$-graphs as well,
by applying Proposition \ref{prop:counting-functions} which relates
the gap labels of discrete and equilateral metric graphs.

\newpage{}

\section{Discussion and further remarks \label{sec:Discussion}}

We conclude by presenting several  observations that provide context
to our main results, and indicate possible generalizations.

\subsection{Zero measure Cantor spectrum for tiling graphs}

Our main results are, on one hand, closely aligned with the classical
one-dimensional setting: under Boshernitzan\textquoteright s condition,
the spectrum has Lebesgue measure zero, and apart from a possible
discrete set, it is a generalized Cantor set. On the other hand, the
possible appearance of only an \textquotedblleft almost\textquotedblright{}
Cantor structure marks a genuine difference between the geometric
setting and the one-dimensional potential setting. This difference
is due to the possibility of flat bands, which have no analogue in
the standard one-dimensional models. The same phenomenon is also responsible
for the gap-closing mechanism observed in the DTMP.

At the same time, the current proofs rely partly on generic assumptions
on the geometry (for instance on the edge lengths), and it is natural
to ask to what extent these assumptions are essential. In particular,
it is not clear whether the conclusions of Theorem \ref{thm:TMP}
regarding zero measure Cantor spectrum require such assumptions, or
whether they reflect limitations of the method.

\subsubsection*{Weakening the geometric assumptions of Theorem \ref{thm:TMP}.}

Theorem \ref{thm:TMP} proves zero measure spectrum under a generic
assumption on the edge lengths, which is imposed in order to ensure
the validity of the Borg--Marchenko Theorem \ref{thm:BM}. The proof
suggests that the essential point is not genericity itself, but rather
a distinguishability condition on the periodic orbits associated with
the different tiles. Indeed, the TPO expansion of the $m$-function
developed in Section \ref{sec:Kotani-proofs} shows that the Borg--Marchenko
argument relies on recovering the sequence of tiles from the associated
periodic orbits.

This suggests that Theorem \ref{thm:BM}, and hence also Theorem \ref{thm:TMP},
should remain valid under substantially weaker geometric assumptions,
provided such a distinguishability condition holds. Extending the
proof in this direction is, however, technically nontrivial, as it
requires controlling possible degeneracies in the periodic orbit expansion.

It is also clear that Theorem \ref{thm:BM} does not hold in complete
generality. For example, if one chooses the tiles as in Figure \ref{fig:bad-tiles},
different sequences will always generate the same metric graph. In
this case, the $m$-function cannot distinguish between the underlying
sequences, and the Borg--Marchenko property fails. Consequently,
the conclusions of Theorem \ref{thm:TMP} do not hold in this setting.
Apart from such degeneracies, one expects the main results to persist
for a broader class of tiling graphs.

\begin{figure}
\includegraphics[scale=0.5]{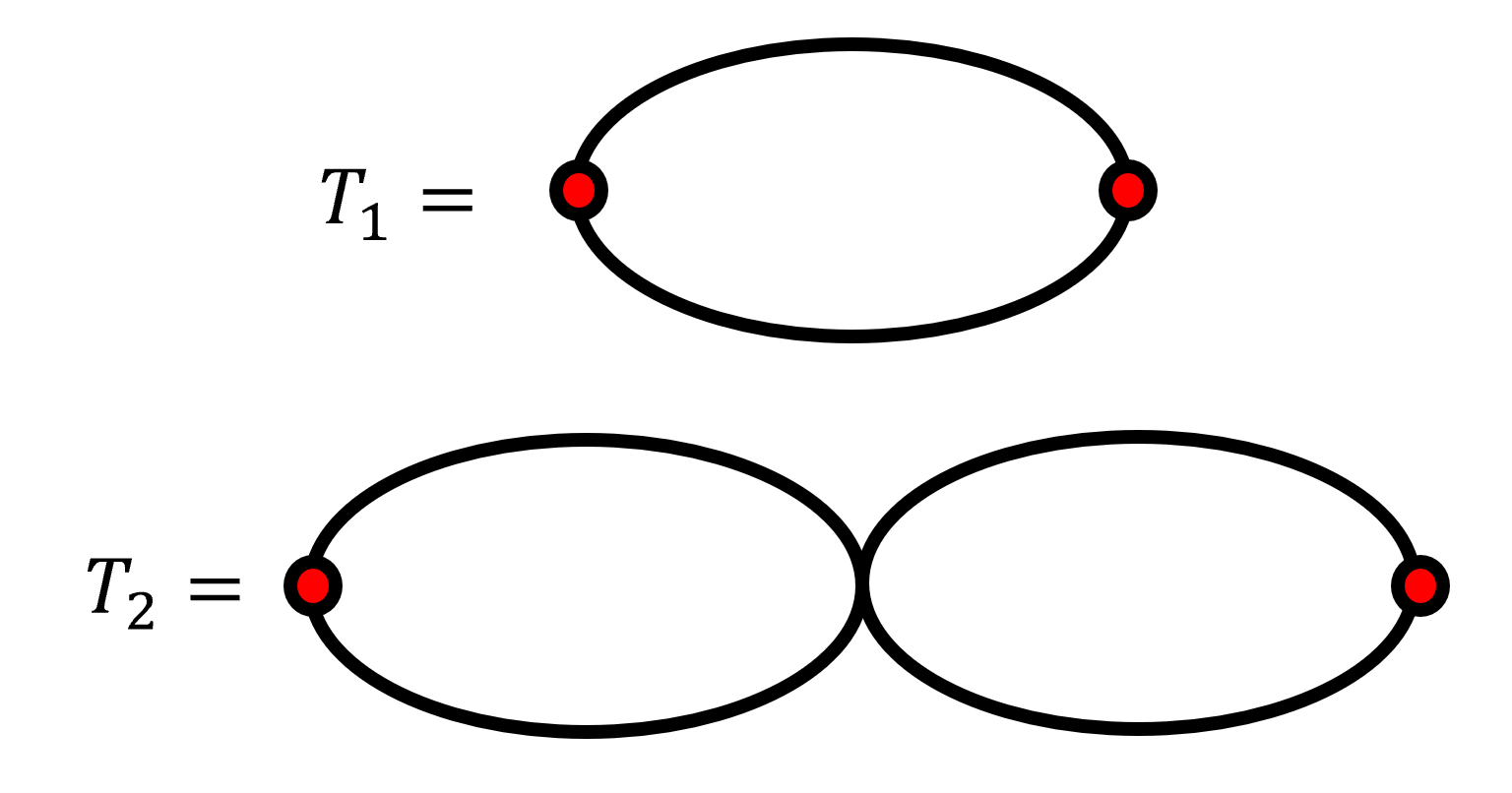}

\caption[~~Graph tiles without Cantor spectrum.]{Two choices of graph tiles (attachment vertices $v_{1},v_{2}$ are
marked in red), where $T_{2}$ is obtained by concatenating two copies
of $T_{1}$. Although the tiles $T_{1}$ and $T_{2}$ are distinct,
the resulting tiling graph is periodic, independently of $\Omega$.
Consequently, the associated Kirchhoff Laplacian has band-gap spectrum,
which is neither of Lebesgue measure zero nor a Cantor set. \label{fig:bad-tiles}}
\end{figure}

\subsubsection*{Increasing the complexity of the tiles.}

In this work, the tiles are connected through a single vertex. A natural
generalization is to allow tiles to be attached through several vertices.
This generalizes the model currently studied in \cite{Becker}. The
main technical difficulty in studying such a model is the fact that
the cocycle $\mnw$ now consists of matrices that are larger than
$2\times2$, which causes significant technical complications when
applying Kotani theory and results about hyperbolicity of the cocycle.
Nevertheless, to our understanding, such a generalization should be
possible using the methods presented in \cite{Becker}.

\subsubsection*{Introducing a Schrödinger operator.}

Another natural extension is to replace the Kirchhoff Laplacian with
a Schrödinger operator, $H_{\omega}:=-\frac{d^{2}}{dx^{2}}+V_{\omega}\left(x\right)$,
where the potential $V_{\omega}\left(x\right)\in L_{\text{loc,unif}}^{2}\left(\Gamma_{\omega}\right)$
is ``compatible'' with the given tiling; i.e., $V_{\omega}\left(x\right)|_{\Gamma_{\omega}^{\left(n\right)}}$
depends only on $\omega\left(n\right)$. This extends some aspects
of the continuum Schrödinger operator models on $\R$ studied in \cite{Damanik2014,damanik_embree_fillman_gorodetski_mei_2026,Damanik2017a,Fillman2018}.
While the presence of a potential introduces additional complications,
we expect that the main spectral features of Theorem \ref{thm:TMP}
should still hold under suitable assumptions on the induced cocycle.

\subsection{Johnson--Schwartzman gap labelling}

Our results show that, up to a natural geometric normalization, the
algebraic structure of the gap labels is entirely determined by the
underlying dynamical system. More precisely, for both metric and discrete
decorated $\Z$-graphs, the set of possible gap labels is given by
the Schwartzman group associated with the subshift, scaled by either
the normalized length in the metric case or the average number of
vertices in the discrete case. In particular, the passage from potentials
to geometry does not alter the dynamical part of the GLT, but only
introduces a geometric normalization factor.

\subsubsection{Moving to tiling graphs}

A natural next step is to extend the GLT from decorated $\Z$-graphs
to the more general class of tiling graphs. From a conceptual point
of view, one expects the same result to hold, namely that the allowed
gap labels are given by the Schwartzman group of the underlying dynamical
system, up to a suitable normalization. However, the current proof
does not extend directly to this setting.

The main difficulty is that, in the general tiling graph case, the
nodal surplus associated with a tile eigenfunction is no longer determined
solely by the tile type. Instead, it depends on the value of the $m$-function
at the attachment vertex, which may vary between different occurrences
of the same tile along the graph. As a result, the argument used in
Section \ref{sec:GLT-proofs}, which reduces the IDS to an average
over contributions of individual tiles, breaks down. To extend the
proof, one would need to understand the statistical behavior of the
nodal surplus along $\Gamma_{\omega}\left(t\right)$, and in particular
to control its asymptotic averages.

\subsubsection{Higher-dimensional analogues}

Another natural direction is to extend the GLT to higher-dimensional
decorated $\Z^{d}$-graphs, with $d>1$. One possible approach is
to mimic the one-dimensional argument by introducing a suitable notion
of Prüfer angle and computing its rotation number over large regions
of the graph. However, this requires the existence of a canonical
choice of generalized eigenfunction, analogous to $\fwe$ in the one-dimensional
setting. In higher dimensions, such a canonical choice is not available
in general, and this presents a substantial obstacle. The main difficulty
is therefore to define and control an appropriate phase function in
higher dimensions.

\subsection{The Dry Ten Martini Problem}

The results of Sections \ref{sec:DTMP}--\ref{sec:DTMP-computation}
show that, for Sturmian decorated $\Z$-graphs, the obstruction to
realizing all gap labels is geometric rather than dynamical. Although
the periodic approximants for this model exhibit the same combinatorial
band structure as the usual Sturmian setting, flat bands produce jump
discontinuities of the IDS, and these discontinuities remove entire
intervals of allowed labels. 

In the case of Sturmian comb graphs equipped with the adjacency matrix,
this mechanism can be described explicitly. The occurrence of flat
bands, and hence of gap closing, depends on arithmetic properties
of the underlying irrational rotation, and in particular on the first
digit in the continued fraction expansion of $\alpha$ (and in particular
its parity).

\subsubsection*{Generalizing to non-equilateral metric graphs}

A natural next step is to treat non-equilateral metric graphs, where
the spectral correspondence between the Kirchhoff Laplacian and the
NDL via the map
\begin{equation}
E\mapsto1-\cos\left(\sqrt{E}\right)\label{eq:spectral-branches}
\end{equation}
is no longer available. There are known spectral correspondences in
the non-equilateral case (see, e.g., \cite{Lledo2008,Exner2018}),
but they are energy-dependent: for each spectral parameter $E$, the
metric operator is related to a different discrete graph operator.
Thus the reduction used here cannot be applied directly.

One possible approach is to extend Theorem \ref{thm:Backwards-type}
to more general transfer matrices, in which several (or even all)
entries depend nontrivially on the energy. In the discrete decorated
setting we reduced the problem to a Sturmian Hamiltonian with an energy-dependent
potential in the $\M_{11}$ entry. It is possible that, after an appropriate
normalization or conjugation, a similar reduction could be carried
out even when the off-diagonal entries are energy-dependent as well.

Our current conjecture is that for a generic choice of edge lengths,
all gap labels permitted by the GLT are attained. The heuristic is
that the flat bands responsible for the IDS jumps in the equilateral
case arise from eigenvalue degeneracies. A generic perturbation of
the edge lengths should remove these degeneracies, and hence eliminate
the jump discontinuities. If, in addition, the periodic approximants
still exhibit the same local $A/B$ combinatorial structure (which
numerically seems to be the case), then the mechanism which produces
``missing gaps'' disappears and one expects the DTMP to hold in
full generality. At present, however, we can only establish the existence
of the Sturmian combinatorial structure in certain spectral regions
for non-equilateral graphs -- mainly where the Fricke-Vogt invariant
is sufficiently large. Understanding the behavior in other regions
of the spectrum is still an open problem.

\subsubsection*{Unattained gap labels and IDS jumps}

For comb graphs equipped with the adjacency matrix, we were able to
show that apart from the gap labels removed by the jump discontinuity
of the IDS, all gap labels predicted by the GLT are attained. Although
the same transfer-matrix mechanism applies to more general decorated
$\Z$-graphs and the same strategy can be carried out for those as
well, we do not currently know whether this phenomenon is universal,
i.e., the only possible source of missing labels is the presence of
IDS jump discontinuities. Relatedly, we currently do not know whether
flat bands (and hence IDS jumps) always occur for decorated $\Z$-graphs.
Two natural problems in this context are:

1. Construct an explicit example of a decorated $\Z$-graph (with
a suitable graph operator) for which the IDS is continuous and all
gap labels are open.

2. Alternatively, prove that, at least generically, IDS jumps are
unavoidable and quantify their size.

3. More generally, determine to what extent the choice of decorations
and operator can control the gap labels, by forcing prescribed labels
either to be realized by open gaps or to disappear through gap closing
or IDS jumps.

In this context, the diagonal Green\textquoteright s function seems
to be the main relevant object. Its zeros determine precisely the
energies at which the transfer matrix becomes singular and flat bands
may occur. It is possible that a more precise analysis of these zeros
could allow one to compute directly the mass concentrated at each
flat band and hence the size of the IDS jump. Ideally, one would like
a formulation in which the initial condition of the Sturmian spectral
tree (and therefore the jump size) is read off directly from the analytic
structure of $\mathcal{G}\left(E\right)$, instead of having to apply
the technical methods presented in Subsection \ref{subsec:adj-initial}.

\subsubsection*{Horizontally non-homogeneous Jacobi operators and more complicated
graph structures}

Throughout most of the DTMP analysis we assumed horizontal homogeneity,
so that the off-diagonal entries of the transfer matrices are constant.
This simplifies the reduction to Sturmian form and allows a direct
comparison with the usual Sturmian Hamiltonians. 

From the point of view of gap labelling, this assumption is not essential.
The NDL already provides an example where horizontal homogeneity fails
but the gap labels can still be computed explicitly. However, for
arbitrary non-homogeneous Jacobi operators, there is currently no
general framework for computing the gap labels. This case is genuinely
more complicated, since now the interaction between sites no longer
depends only on the letters of the Sturmian word, but rather on subwords
of length $2$, which suggests the possibility that in such models
new phenomena may occur. Understanding whether the Sturmian tree structure
still holds in this setting is an interesting open problem.

Alternatively, one may consider more complicated graph models leading
to nontrivial off-diagonal entries of the transfer matrices, such
as the aperiodic tiling graph models studied in this thesis.

\newpage{}

\appendix

\section{\label{sec:Appendices} Borg--Marchenko result for decorated $\protect\Z$-graphs}

In this appendix, we present the proof of the Borg--Marchenko result
for decorated $\Z$-graphs in Proposition \ref{prop:BM-decorated},
which follows the same strategy as the proof for tiling graphs.
\begin{proof}[Proof of Proposition \ref{prop:BM-decorated}]
Let $\left(\omega\left(n\right)\right)_{n=0}^{\infty}\in\A^{\N}$.
We use the tailed periodic orbit (TPO) expansion (\ref{eq:periodic-m})
of $m_{\omega}^{+}\left(z\right)$ in order to inductively recover
each element of $\omega$.

\uline{Base case -- \mbox{$n=0$}:}

Consider the TPO expansion of $m_{\omega}^{+}\left(z\right)$. Firstly,
note that if $\overline{a}\neq a\in\A$ such that $\omega\left(0\right)=a$,
then the term $-\sqrt{-z}A_{a}e^{-\sqrt{-z}2\left(1+\ell_{\min}^{a}\right)}$
(with $\ell_{\min}^{a}$ as in the statement of Proposition \ref{prop:BM-decorated})
should appear in this TPO expansion with $A_{a}\neq0$. Indeed, in
this case, there are exactly $\deg\left(v_{a}\right)$ (where $v_{a}$
is the base vertex of $\gra$) TPOs $p\in\mathcal{P}_{\omega}^{+}$
of length $2\left(1+\ell_{\min}^{a}\right)$ -- exactly the TPOs
that move from the origin $o$ to the base vertex $v_{a}$ along the
horizontal edge of length $1$, traverse the (unique) edge $e_{\min}\sim v_{a}$
of length $\ell_{\min}^{a}$, backscatter and return to the origin
along the same (reversed) path, and then scatter to some other edge
$e'\sim o$. If we denote by $v_{\min}$ the neighboring vertex to
$v_{a}$ such that $e_{\min}=\left\{ v_{a},v_{\min}\right\} $, then
a computation directly analogous to that in the proof of Theorem \ref{thm:BM}
shows that $A_{a}$ must be nonzero.

Now, iterate over the list $\left(\ell_{\min}^{a}\right)_{\overline{a}\neq a\in\A}$
in increasing order, and examine the corresponding term $-\sqrt{-z}A_{a}e^{-\sqrt{-z}2\left(1+\ell_{\min}^{a}\right)}$
in the TPO expansion of $m_{\omega}^{+}\left(z\right)$. By the discussion
above, for each $\overline{a}\neq a\in\A$ such that $A_{a}=0$, we
can immediately deduce that $\omega\left(0\right)\neq a$. Suppose
first that for some $\overline{a}\neq a\in\A$, the prefactor $A_{a}$
is nonzero. Then there exists a TPO of length $2\left(1+\ell_{\min}^{a}\right)$
around $o$, and hence the unique edge of length $\ell_{\min}^{a}$
must lie in the first decoration, implying that $\omega\left(0\right)=a$.
If, on the other hand, $A_{a}=0$ for all $a\neq\overline{a}$, then
the only remaining possibility is $\omega\left(0\right)=\overline{a}$.

\uline{Induction step:}

Suppose that we have recovered $\omega\left(0\right),...,\omega\left(n-1\right)$,
and wish to identify $\omega\left(n\right)$. As in the base case,
we iterate over the list $\left(\ell_{\min}^{a}\right)_{\overline{a}\neq a\in\A}$
in increasing order. For each $a\neq\overline{a}$, focus on the term
$-\sqrt{-z}A_{a}e^{-\sqrt{-z}2\left(\left(2n+1\right)+\ell_{\min}^{a}\right)}$
in the TPO expansion of $m_{\omega}^{+}\left(z\right)$. Since we
already know the first $n$ decorations, we can list all TPOs of length
$2\left(\left(2n+1\right)+\ell_{\min}^{a}\right)$ in $\mathcal{\P}_{\omega}^{+}$
passing only through the first $n$ decorations (if such exist), and
explicitly compute their contribution to the TPO expansion of $m_{\omega}^{+}\left(z\right)$
through a term of the form $-\sqrt{-z}\tilde{A_{a}}e^{-\sqrt{-z}2\left(\left(2n+1\right)+\ell_{\min}^{a}\right)}$
(and again, $\tilde{A_{a}}$ may be $0$).

Once again, if for some $\overline{a}\neq a\in\A$ we have that $A_{a}\neq\tilde{A_{a}}$,
then there must be additional TPOs of length $2\left(\left(2n+1\right)+\ell_{\min}^{a}\right)$
in $\mathcal{\P}_{\omega}^{+}$ which pass through $\TG$. By the
choice of $\ell_{\min}^{a}$, such a TPO must consist of the shortest
path from the origin to $u_{\omega\left(n\right)}$ (namely the horizontal
path of length $2n+1$), going along and back the unique edge of length
$\ell_{\min}^{a}$ in $\Gamma_{\omega\left(n\right)}$ adjacent to
the base vertex $u_{\omega\left(n\right)}$, and then returning to
the origin through the initial path (reversed). Since $\ell_{\min}^{a'}\neq\ell_{\min}^{a}$
for $a\neq a'\in\A$, we deduce that $\omega\left(n\right)=a$. If
instead $A_{a}=\tilde{A_{a}}$ for all $a\neq\overline{a}$, we are
only left with the possibility $\omega\left(n\right)=\overline{a}$.
\end{proof}

\section{Computation of transfer matrices\label{sec:t-mat-computation}}

In this appendix we calculate the transfer matrices used in Section
\ref{sec:DTMP-computation}. We carry this out for a general horizontally
homogeneous Jacobi operator.

Let $G_{\alpha}$ be a decorated $\Z$-graph and let $\Ja$ be a Jacobi
operator satisfying symmetry, equivariance, and horizontal homogeneity.
Fix one decoration $(G,v)$ and denote by $\mathcal{J}_{G}$ the restriction
of $\Ja$ to $G$. Let $c\neq0$ be the horizontal coupling constant
between adjacent base vertices.

Let $\varphi$ be a (formal) solution of the eigenvalue equation
\begin{equation}
\Ja\varphi=E\varphi.\label{eq:eval-equation}
\end{equation}
Denote by $\varphi_{n}$ the restriction of $\varphi$ to the $n$-th
decoration, and by $\varphi(n):=\varphi_{n}(v)$ the value at the
base vertex. For $w\in\V\left(G\right)$, the eigenvalue equation
reads
\begin{equation}
(\mathcal{J}_{G}\varphi_{n})(w)+c\delta_{w,v}\big(\varphi(n-1)+\varphi(n+1)\big)=E\varphi_{n}(w).\label{eq:eval-equation-1}
\end{equation}
Equivalently,
\begin{equation}
(E-\mathcal{J}_{G})\varphi_{n}=c\big(\varphi(n-1)+\varphi(n+1)\big)e_{v}.\label{eq:eval-equation-2}
\end{equation}

Assume $E$ is such that $E-\mathcal{J}_{G}$ is invertible. The exceptional
energies where $E-\mathcal{J}_{G}$ is not invertible are precisely
the singular energies discussed in Subsection \ref{subsec:singular-energies}.
Solving gives
\begin{equation}
\varphi_{n}=c\big(\varphi(n-1)+\varphi(n+1)\big)(E-\mathcal{J}_{G})^{-1}e_{v}.\label{eq:eval-equation-3}
\end{equation}
Taking the inner product with $e_{v}$ yields
\begin{equation}
\varphi(n)=c\mathcal{G}_{\mathcal{J}_{G}}(E)\big(\varphi(n-1)+\varphi(n+1)\big),\label{eq:Green-1}
\end{equation}
where
\begin{equation}
\mathcal{G}_{\mathcal{J}_{G}}(E):=\langle(E-\mathcal{J}_{G})^{-1}e_{v},e_{v}\rangle\label{eq:Green-2}
\end{equation}
is the diagonal Green's function of the pointed decoration. Rearranging
gives the second-order recurrence
\begin{equation}
\varphi(n+1)=\frac{1}{c\mathcal{G}_{\mathcal{J}_{G}}(E)}\varphi(n)-\varphi(n-1).\label{eq:Green-3}
\end{equation}

Hence the one-step transfer matrix is
\begin{equation}
\M_{G}(E)=\begin{pmatrix}\frac{1}{c\mathcal{G}_{\mathcal{J}_{G}}(E)} & -1\\
1 & 0
\end{pmatrix}.\label{eq:t-mat}
\end{equation}
Defining
\begin{equation}
f_{G}(E):=E-\frac{1}{c\mathcal{G}_{\mathcal{J}_{G}}(E)},\label{eq:f-G}
\end{equation}
we obtain the Sturmian form
\begin{equation}
\M_{G}(E)=\begin{pmatrix}E-f_{G}(E) & -1\\
1 & 0
\end{pmatrix}.\label{eq:final-tmat}
\end{equation}

\section{Computation of the initial condition for $\protect\T\left(\protect\Aa\right)$
on Sturmian combs\label{sec:trace-computations}}

In this appendix we prove the technical results used in Subsection
\ref{subsec:adj-initial} concerning the first two periodic approximants.
More precisely, Subsection \ref{subsec:The-first-periodic} contains
the details behind Lemmas \ref{lem:trace-recursion}, \ref{lem:FB},
and \ref{lem:B-bands}. Subsection \ref{subsec:Proof-of-Proposition}
proves Proposition \ref{prop:no-E=00003D0}, and Subsection \ref{subsec:Proof-of-Lemma}
proves Lemma \ref{lem:gen2-bands}.

\subsection{The first periodic approximant\label{subsec:The-first-periodic}}

Let
\begin{equation}
\alpha=[0;a,b,\ldots],\label{eq:0ab}
\end{equation}
and consider the first periodic approximant $\alpha_{1}=\frac{1}{a}$.
We first prove the trace recursion stated as Lemma \ref{lem:trace-recursion}
in the main text.
\begin{lem}
\label{lem:traces}The functions $t_{a}$ and $h_{a}$ satisfy:

1. The recursion
\begin{align}
 & t_{a+1}(E)=Et_{a}(E)-t_{a-1}(E),\qquad t_{0}(E)=2,\quad t_{1}(E)=E-\frac{1}{E},\label{eq:ta}\\
 & h_{a+1}(E)=Eh_{a}(E)-h_{a-1}(E),\qquad h_{0}(E)=2E,\quad h_{1}(E)=E^{2}-1.\label{eq:ha-1}
\end{align}

2. For every $a\ge0$, $h_{a}$ is a monic polynomial of degree $a+1$.

3. The symmetry relations
\begin{equation}
h_{a}(-E)=(-1)^{a+1}h_{a}(E)\label{eq:symmetry-1}
\end{equation}
hold.
\end{lem}

\begin{proof}
The recursion for $t_{a}$ is standard for Sturmian transfer matrices,
see e.g. \cite{Bellissard1989,band2024review}. The initial conditions
follow directly from the explicit matrices above. Multiplying the
recursion by $E$ gives the recursion for $h_{a}$ and the initial
conditions.

We now prove the degree statement by induction. For $a=0,1$ we have
\begin{equation}
h_{0}(E)=2E,\qquad h_{1}(E)=E^{2}-1,\label{eq:initial}
\end{equation}
so the claim holds. Assume $h_{a-1}$ and $h_{a}$ are monic of degrees
$a$ and $a+1$. Then $Eh_{a}(E)$ is monic of degree $a+2$, while
$h_{a-1}(E)$ has strictly lower degree. Hence
\begin{equation}
h_{a+1}(E)=Eh_{a}(E)-h_{a-1}(E)\label{eq:h-rec}
\end{equation}
is monic of degree $a+2$, completing the induction.

Finally, the symmetry follows by a similar induction, using the recursion
and the initial values. Indeed,
\begin{equation}
h_{a+1}\left(-E\right)=-Eh_{a}\left(-E\right)-h_{a-1}\left(-E\right)=\left(-1\right)^{a+2}h_{a+1}\left(E\right).\label{eq:symmetry-2}
\end{equation}
\end{proof}
The next three lemmas provide the details behind Lemma \ref{lem:FB}.
\begin{lem}
Let $H_{a}(\theta)$ be the Floquet--Bloch operator of $\Aan 1$,
with $\theta\in\{0,\pi\}$. Then its characteristic polynomial is
(up to a nonzero scalar)
\begin{equation}
P_{a,\theta}(E)=h_{a}(E)-2(\cos\theta)E.\label{eq:char-pol}
\end{equation}
In particular, the band edges of $\Aan 1$ are given by the roots
of $P_{a,0}(E)$ and $P_{a,\pi}(E)$.
\end{lem}

\begin{proof}
For $E\ne0$, the transfer matrix is well-defined and the Floquet--Bloch
condition gives
\begin{equation}
t_{a}(E)=2\cos\theta.\label{eq:FB-condition}
\end{equation}
Multiplying by $E$ yields the stated polynomial. By Lemma \ref{lem:traces},
$h_{a}$ has degree $a+1$, and hence $P_{a,\theta}$ has degree $a+1$.
Since the Floquet--Bloch operator acts on a space of dimension $a+1$,
its characteristic polynomial must have the same degree, and the two
expressions coincide (up to scaling).
\end{proof}
\begin{lem}
\label{lem:parity}We have
\begin{equation}
h_{a}(0)=0\quad\Longleftrightarrow\quad a\text{ is even}.\label{eq:zeros-ha}
\end{equation}
\end{lem}

\begin{proof}
Evaluating the recursion at $E=0$ gives
\begin{equation}
h_{a+1}(0)=-h_{a-1}(0),\label{eq:symmetry}
\end{equation}
with initial values
\begin{equation}
h_{0}(0)=0,\qquad h_{1}(0)=-1.\label{eq:initial-E=00003D0}
\end{equation}
Induction on $a$ using (\ref{eq:symmetry}) then gives
\begin{equation}
h_{2m}(0)=0,\qquad h_{2m+1}(0)=(-1)^{m+1},\label{eq:h2m0}
\end{equation}
which proves the claim.
\end{proof}
\begin{lem}
\label{lem:simple-root}Assume that $a$ is even. Then $E=0$ is a
simple root of both $P_{a,0}$ and $P_{a,\pi}$.
\end{lem}

\begin{proof}
Recalling that $h_{a}(E)=Et_{a}(E)$, Lemma \ref{lem:parity} shows
that for $a$ even, $t_{a}\left(0\right)$ is well-defined and $t_{a}\left(0\right)\neq0$.
This implies that $E=0$ is a simple zero of $h_{a}$. Since $P_{a,\theta}$
all differ from $h_{a}$ by a linear term, the claim follows.
\end{proof}
We conclude by proving Lemma \ref{lem:B-bands}.
\begin{lem}
\label{lem:perturbation}The Floquet--Bloch operators $H_{a}(\theta)$
satisfy:

1. For every $\theta\in\{0,\pi\}$, the operator $H_{a}(\theta)$
has at most one eigenvalue in $(2,\infty)$ and at most one eigenvalue
in $(-\infty,-2)$.

2. There exists $\theta\in\{0,\pi\}$ such that $H_{a}(\theta)$ has
an eigenvalue strictly larger than $2$, and similarly one strictly
smaller than $-2$.
\end{lem}

\begin{proof}
We compare $H_{a}(\theta)$ with the Floquet--Bloch operator of the
underlying periodic chain (i.e., without the tooth). The latter has
spectrum contained in $[-2,2]$ for all $\theta$. The operator $H_{a}(\theta)$
is obtained from this chain operator by attaching a single additional
vertex (the tooth) to the unit cell. This amounts to a rank-one perturbation
of the corresponding matrix. By eigenvalue interlacing (see, e.g.,
\cite[chap. III]{Bhatia1997}), such a perturbation can introduce
at most one eigenvalue above the top of the spectrum and at most one
below the bottom. This proves part (1).

For part (2), fix $\theta=0$. Let $v_{0}$ be the base vertex to
which the tooth is attached, and let $w$ denote the tooth vertex.
For given $0<\varepsilon<1$, consider the trial vector $\psi_{\varepsilon}$
which takes the value $1$ on every vertex of the underlying chain
and the value $\varepsilon$ at $w$. Then
\begin{equation}
\langle\psi_{\varepsilon},H_{a}(0)\psi_{\varepsilon}\rangle=2a+2\varepsilon,\qquad\|\psi_{\varepsilon}\|^{2}=a+\varepsilon^{2}.\label{trial-vec}
\end{equation}
Hence
\begin{equation}
\frac{\langle\psi_{\varepsilon},H_{a}(0)\psi_{\varepsilon}\rangle}{\|\psi_{\varepsilon}\|^{2}}=\frac{2a+2\varepsilon}{a+\varepsilon^{2}}>2,\label{eq:rayleigh}
\end{equation}
since $0<\varepsilon<1$. This shows that $H_{a}(0)$ has an eigenvalue
larger than $2$. As noted in Subsection \ref{subsec:combinatorial-structure},
the spectrum is symmetric with respect to $E=0$, and therefore there
is also an eigenvalue strictly smaller than $-2$.
\end{proof}

\subsection{Proof of Proposition \ref{prop:no-E=00003D0}\label{subsec:Proof-of-Proposition}}

Let $\alpha_{2}=[0;a,b]$, and denote by $t_{a,b}(E)$ the trace of
the corresponding transfer matrix. In this subsection we prove Proposition
\ref{prop:no-E=00003D0}, namely that $E=0$ does not belong to the
nonflat spectrum of $\Aan 2$. It follows that from level $\alpha_{2}$
onward, all nonflat bands satisfy the usual Sturmian forward branching
rule.

The next three lemmas are auxiliary computations. Together they lead
to Proposition \ref{prop:no-E0}, which is the statement of Proposition
\ref{prop:no-E=00003D0}.
\begin{lem}
\label{lem:recursion}For fixed $a\in\mathbb{N}$, the functions $t_{a,b}$
satisfy
\begin{equation}
t_{a,0}(E)=E,\qquad t_{a,1}(E)=t_{a+1}(E),\label{eq:rec-2}
\end{equation}
and for every $b\ge1$,
\begin{equation}
t_{a,b+1}(E)=t_{a}(E)t_{a,b}(E)-t_{a,b-1}(E).\label{eq:ta-rec-1}
\end{equation}
Equivalently, if $(S_{n})_{n\ge-1}$ are the dilated Chebyshev polynomials
defined by
\begin{equation}
S_{-1}=0,\qquad S_{0}=1,\qquad S_{n+1}(x)=xS_{n}(x)-S_{n-1}(x),\label{eq:chebyshev}
\end{equation}
then for every $b\ge1$,
\begin{equation}
t_{a,b}(E)=S_{b-1}(t_{a}(E))t_{a+1}(E)-S_{b-2}(t_{a}(E))E.\label{eq:cheb-trace}
\end{equation}
\end{lem}

\begin{proof}
The recursion is the standard Sturmian trace recursion, cf. Proposition
\ref{prop:Sturmian-properties} and \cite{Bellissard1989,band2024review}.
The initial conditions are immediate from the definition of the second
periodic approximant. 
\end{proof}
We now analyze the behavior of $t_{a,b}$ at the singular energy $E=0$,
distinguishing according to the parity of $a$.
\begin{lem}
\label{lem:a-even}Assume that $a$ is even. Then for every $b\geq1$,
the function $t_{a,b}(E)$ has a pole at $E=0$.
\end{lem}

\begin{proof}
Since $a$ is even, Lemma \ref{lem:parity} shows that $h_{a}(0)=0$,
and Lemma \ref{lem:simple-root} shows that this zero is simple. Hence
$t_{a}(E)=\frac{h_{a}(E)}{E}$ is regular at $E=0$.

We claim that $|t_{a}(0)|>2$. To see this, write $a=2m$. Using the
trace recursion (\ref{eq:ta}), and passing to the limit as $E\to0$,
we obtain
\begin{equation}
t_{2m+2}(0)=h_{2m+1}(0)-t_{2m}(0).\label{eq:rec-mixed}
\end{equation}
By Lemma \ref{lem:parity},
\begin{equation}
h_{2m+1}(0)=(-1)^{m+1},\label{eq:h-parity}
\end{equation}
and therefore
\begin{equation}
t_{2m+2}(0)=(-1)^{m+1}-t_{2m}(0).\label{eq:ta-parity}
\end{equation}
Since
\begin{equation}
t_{2}(0)=\lim_{E\to0}\left(E\Big(E-\frac{1}{E}\Big)-2\right)=-3,\label{eq:t2-lim}
\end{equation}
it follows inductively that $|t_{2m}(0)|>2$ for every $m\geq1$.
Thus indeed $|t_{a}(0)|>2$.

On the other hand, since $a+1$ is odd, Lemma \ref{lem:parity} gives
\begin{equation}
h_{a+1}(0)\neq0,\label{eq:hneq0}
\end{equation}
and hence $t_{a+1}(E)=\frac{h_{a+1}(E)}{E}$ has a simple pole at
$E=0$.

Now use the Chebyshev representation:
\begin{equation}
t_{a,b}(E)=S_{b-1}(t_{a}(E))\,t_{a+1}(E)-S_{b-2}(t_{a}(E))E.\label{eq:cheb-rep}
\end{equation}
Since $t_{a}(E)\to t_{a}(0)$ as $E\to0$, and $|t_{a}(0)|>2$ while
all zeros of $S_{b-1}$ lie in $(-2,2)$, we have $S_{b-1}(t_{a}(0))\neq0$.
Therefore the first term has the same simple pole at $E=0$ as $t_{a+1}(E)$,
whereas the second term is regular and vanishes at $E=0$. Hence no
cancellation is possible, and $t_{a,b}(E)$ has a pole at $E=0$.
\end{proof}
\begin{lem}
\label{lem:a-odd}Assume that $a$ is odd. Then:

1. If $b=1$, then $t_{a,1}(E)=t_{a+1}(E)$ is regular at $E=0$,
and $|t_{a,1}(0)|>2$.

2. If $b\geq2$, then $t_{a,b}(E)$ has a pole at $E=0$.
\end{lem}

\begin{proof}
If $b=1$, then $t_{a,1}(E)=t_{a+1}(E)$. Since $a+1$ is even, Lemma
\ref{lem:simple-root} shows that $t_{a+1}$ is regular at $E=0$,
and the argument from the proof of Lemma \ref{lem:a-even} gives $|t_{a+1}(0)|>2$.
Now assume $b\geq2$. Since $a$ is odd, Lemma \ref{lem:parity} gives
$h_{a}(0)\neq0$, and hence $t_{a}(E)=\frac{h_{a}(E)}{E}$ has a simple
pole at $E=0$. On the other hand, $t_{a+1}(E)$ is regular at $0$.

Using again the Chebyshev representation (\ref{eq:cheb-trace}), we
compare the pole orders of the two terms. Since $t_{a}(E)$ has a
simple pole at $E=0$, the first term has pole order $b-1$, while
the second term has pole order $b-3$ (and is regular when $b=2$).
In particular, the second term has strictly lower pole order, so it
cannot cancel the first. Therefore $t_{a,b}(E)$ has a pole at $E=0$
for every $b\geq2$.
\end{proof}
Combining the two lemmas above gives the following:
\begin{prop}
\label{prop:no-E0}The function $t_{a,b}(E)$ either has a pole at
$E=0$, or satisfies $|t_{a,b}(0)|>2$. In particular, $E=0$ does
not belong to any nonflat spectral band of $\Aan 2$.
\end{prop}

\begin{proof}
If $a$ is even, Lemma \ref{lem:a-even} shows that $t_{a,b}(E)$
has a pole at $E=0$. If $a$ is odd and $b=1$, Lemma \ref{lem:a-odd}
shows that $t_{a,b}(E)=t_{a+1}(E)$ is regular at $E=0$ and satisfies
$|t_{a,b}(0)|>2$. If $a$ is odd and $b\geq2$, Lemma \ref{lem:a-odd}
again shows that $t_{a,b}(E)$ has a pole at $E=0$.

Thus in all cases, $t_{a,b}(E)$ either has a pole at $E=0$, or is
regular there with value of magnitude strictly larger than $2$. The
last statement follows immediately.
\end{proof}
This gives the following corollary, used immediately after Proposition
\ref{prop:no-E=00003D0}:
\begin{cor}
There exists $\varepsilon>0$ such that $(-\varepsilon,\varepsilon)\setminus\{0\}$
is disjoint from the nonflat spectrum of both the first and second
periodic approximants.
\end{cor}

\begin{proof}
For the first periodic approximant, the unique singularity is at $E=0$,
and all nonflat bands are separated from $0$. For the second periodic
approximant, Proposition \ref{prop:no-E0} above shows that $0$ does
not belong to any nonflat band. Since the nonflat spectrum (which
is the essential spectrum of $\Aan n$) is closed and does not contain
$E=0$, there exists $\varepsilon>0$ such that $(-\varepsilon,\varepsilon)\setminus\{0\}$
is disjoint from all nonflat spectra.
\end{proof}

\subsection{Proof of Lemma \ref{lem:gen2-bands}\label{subsec:Proof-of-Lemma}}

We now compute the number of flat and nonflat bands of $\Aan 2$,
completing the proof of Lemma \ref{lem:gen2-bands}. The intermediate
results below are used only for this computation. The final statement
is Lemma \ref{lem:c.14}, which is the statement quoted as Lemma \ref{lem:gen2-bands}.
As before, we remove the poles of $t_{a,b}$ by defining
\begin{equation}
H_{a,b}(E):=E^{b}t_{a,b}(E).\label{eq:Ha,b}
\end{equation}

\begin{lem}
\textbf{\label{Ha-recursion}} The functions $H_{a,b}$ are polynomials
satisfying
\begin{equation}
H_{a,0}(E)=E,\qquad H_{a,1}(E)=h_{a+1}(E),\label{eq:H-initial}
\end{equation}
and
\[
H_{a,b+1}(E)=h_{a}(E)H_{a,b}(E)-E^{2}H_{a,b-1}(E).
\]
Moreover,
\begin{equation}
\deg H_{a,b}=ab+b+1.\label{eq:H-deg}
\end{equation}
\end{lem}

\begin{proof}
Multiplying the recursion for $t_{a,b}$ by $E^{b+1}$ gives the stated
relation. The degree statement follows by induction. Since $\deg h_{a}=a+1$,
we have
\begin{equation}
\deg\bigl(h_{a}H_{a,b}\bigr)=(a+1)+\deg H_{a,b},\qquad\deg\bigl(E^{2}H_{a,b-1}\bigr)=2+\deg H_{a,b-1}.\label{eq:degrees}
\end{equation}
By the induction hypothesis,
\begin{equation}
2+\deg H_{a,b-1}=ab-a+b+2<ab+a+b+2=(a+1)+\deg H_{a,b},\label{eq:degrees-1}
\end{equation}
so the leading term comes from $h_{a}H_{a,b}$ and no cancellation
occurs. Hence
\[
\deg H_{a,b+1}=(a+1)+\deg H_{a,b}.
\]
Starting from $\deg H_{a,0}=1$, this gives $\deg H_{a,b}=b(a+1)+1=ab+b+1$.
\end{proof}
As in Subsection \ref{subsec:adj-initial}, the periodic and antiperiodic
band-edge polynomials are given (up to scaling) by
\begin{equation}
P_{a,b,\theta}(E)=H_{a,b}(E)-2(\cos\theta)E^{b},\qquad\theta\in\{0,\pi\}.\label{eq:Pab}
\end{equation}
We now determine their order of vanishing at $E=0$.
\begin{lem}
\label{lem:vanishing-order}The order of vanishing of $H_{a,b}$ at
$E=0$ is
\begin{equation}
\text{ord}_{E=0}H_{a,b}=\begin{cases}
1, & a\text{ odd},\\[1mm]
b-1, & a\text{ even}.
\end{cases}\label{eq:ordH}
\end{equation}
\end{lem}

\begin{proof}
If $a$ is odd, then $t_{a}$ has a simple pole at $0$, while $t_{a+1}$
is regular. The Chebyshev formula shows that $t_{a,b}$ has pole order
$b-1$ for $b\ge2$, and is regular for $b=1$. In both cases, multiplying
by $E^{b}$ yields a simple zero. If $a$ is even, then $t_{a}$ is
regular at $0$ with $|t_{a}(0)|>2$, while $t_{a+1}$ has a simple
pole. The leading term in the Chebyshev expression therefore produces
a simple pole for $t_{a,b}$, and hence $H_{a,b}$ vanishes to order
$b-1$. 
\end{proof}
\begin{cor}
\label{cor:bandmul}The multiplicity of the flat band at $E=0$ is
\begin{equation}
\dim\ker\left(\Aan 2\right)=\begin{cases}
1, & a\text{ odd},\\[1mm]
b-1, & a\text{ even}.
\end{cases}\label{eq:dimker-2-1}
\end{equation}
\end{cor}

\begin{proof}
This follows from the fact that $E=0$ is a common root of the periodic
and antiperiodic polynomials, and its multiplicity equals the order
of vanishing of $H_{a,b}$. 
\end{proof}
We now prove the statement corresponding to Lemma \ref{lem:gen2-bands}
in the main text.
\begin{lem}
\label{lem:c.14}$\Aan 2$ has $ab+b+1$  bands. Among them, the number
of nonflat bands is
\begin{equation}
\begin{cases}
ab+b, & a\text{ odd},\\[1mm]
ab+2, & a\text{ even}.
\end{cases}\label{eq:nonflat-A2}
\end{equation}
\end{lem}

\begin{proof}
The total number of bands is $ab+b+1$, and Corollary \ref{cor:bandmul}
gives the multiplicity of the flat band at $E=0$. Subtracting this
multiplicity gives the stated number of nonflat bands.
\end{proof}
\newpage{}

\bibliographystyle{plain}
\bibliography{GlobalBib_2026}

\end{document}